\documentclass[10pt]{article}
\usepackage[lmargin=1in,rmargin=1in,tmargin=1in,bmargin=1in]{geometry}
\usepackage{amsmath,amsthm,amsfonts,amssymb,comment,slashed,mathtools}
\usepackage[utf8]{inputenc}
\usepackage[T1]{fontenc}  
\usepackage[english]{babel}

\usepackage{csquotes}
\usepackage{bm}
\usepackage[shortlabels]{enumitem}
\usepackage[affil-it]{authblk}

\usepackage[section]{placeins}
\usepackage{soul} 

\usepackage{bm}
\usepackage{mathrsfs, graphicx} 
\usepackage{stmaryrd}
\usepackage{epigraph}

\graphicspath{{figures/}}

\usepackage[backend=biber,style=alphabetic,giveninits=true,doi=false,isbn=false,url=false,sorting=nyt,maxbibnames=99,date=year]{biblatex}
\renewbibmacro{in:}{}
\usepackage{chngcntr}

\usepackage[%
bookmarks=true,
colorlinks=true,
linkcolor=black,
urlcolor=black,
citecolor=black,
plainpages=false,
pdfpagelabels,
final]{hyperref}
\usepackage[capitalise,nameinlink]{cleveref}

\numberwithin{equation}{section}

\theoremstyle{plain}
\newtheorem{thm}{Theorem}
\newtheorem{prop}{Proposition}[section]
\newtheorem{lem}[prop]{Lemma}
\newtheorem*{thm*}{Theorem}
\newtheorem{cor}{Corollary}
\newtheorem*{conj*}{Conjecture}
\newtheorem{conj}{Conjecture}
\theoremstyle{definition}
\newtheorem{defn}[prop]{Definition}
\newtheorem{ass}[prop]{Assumption}
\theoremstyle{remark}
\newtheorem{rk}[prop]{Remark}
\crefname{thm}{Theorem}{Theorems} 
\crefname{lem}{Lemma}{Lemmas} 
\crefname{conj}{Conjecture}{Conjectures}
\crefname{prop}{Proposition}{Propositions}

\renewcommand{\Bbb}{\mathbb}
\newcommand{\ve}{\varepsilon}
\newcommand{\les}{\lesssim}

\newcommand{\sign}{\operatorname{sign}}

\newcommand{\id}{\operatorname{id}}

\newcommand{\tr}{\operatorname{tr}}

\newcommand{\Ric}{\mathrm{Ric}}

\renewcommand{\d}{\mathrm{d}}

\newcommand{\loc}{\mathrm{loc}}
\newcommand{\T}{\mathbf T}
\newcommand{\spt}{\operatorname{spt}}

\newcommand{\R}{\Bbb R}
\newcommand{\Div}{\operatorname{div}}

\renewcommand{\b}{\mathrm{b}}

\makeatletter
\renewcommand{\paragraph}{%
  \@startsection{paragraph}{4}%
  {\z@}{1.25ex \@plus 1ex \@minus .2ex}{-1em}%
  {\normalfont\normalsize\bfseries}%
}
\makeatother

\begin{document}

\title{Extremal black hole formation as a critical phenomenon}

\author[1]{Christoph~Kehle\thanks{christoph.kehle@eth-its.ethz.ch}}
\author[2]{Ryan Unger\thanks{runger@math.princeton.edu}}
\affil[1]{\small  Institute~for~Theoretical~Studies \& Department of Mathematics,~ETH~Zürich,

Clausiusstrasse~47,~8092~Zürich,~Switzerland \vskip.1pc \ 
}
\affil[2]{\small  Department of Mathematics, Princeton University, 
	Washington~Road,~Princeton~NJ~08544,~United~States~of~America \vskip.1pc \  
	}

 \date{February 15, 2024}
\maketitle

\begin{abstract}
In this paper, we prove that extremal black holes arise on the threshold of gravitational collapse. More precisely, we construct smooth one-parameter families of smooth, spherically symmetric solutions to the Einstein--Maxwell--Vlasov system which interpolate between dispersion and collapse and for which the critical solution is an extremal black hole. Physically, these solutions can be understood as beams of gravitationally self-interacting collisionless charged particles fired into Minkowski space from past infinity. Depending on the precise value of the parameter, we show that the Vlasov matter either disperses due to the combined effects of angular momentum and electromagnetic repulsion, or undergoes gravitational collapse. At the critical value of the parameter, an extremal Reissner--Nordstr\"om black hole is formed. No naked singularities occur as the extremal threshold is crossed. We call this critical phenomenon \emph{extremal critical collapse} and the present work constitutes the first rigorous result on the black hole formation threshold in general relativity. 
\end{abstract}

\thispagestyle{empty}
\newpage 
\tableofcontents
\thispagestyle{empty}
\newpage

\section{Introduction}

One of the most spectacular predictions of general relativity is the existence and formation of black holes. Solutions of the Einstein field equations,
\begin{equation}
    \Ric(g)-\tfrac 12 R(g)g=2\mathbf T,\label{eq:EFE}
\end{equation}
 can undergo \emph{gravitational collapse} to form a black hole dynamically, starting from regular, one-ended Cauchy data. Building on earlier work of Lema\^itre \cite{Lem}, Oppenheimer and Snyder \cite{OS39} produced the first mathematical example of gravitational collapse with the collapse of a homogeneous dust cloud. In contrast, for reasonable matter models, solutions with small initial data \emph{disperse} without a black hole forming.
It is a fundamental problem in classical general relativity to understand how these different classes of spacetimes---collapsing and dispersing---fit together in the moduli space of solutions. The interface between collapse and dispersion is known as the \emph{black hole formation threshold} and families of solutions crossing this threshold are said to exhibit \emph{critical collapse}. Spacetimes lying on the threshold are called \emph{critical solutions}.

Critical collapse has been extensively studied numerically, starting with the influential work of Choptuik \cite{choptuik1993universality} on the spherically symmetric Einstein-scalar field model, in a regime where the critical solutions are believed to be naked singularities. The Einstein--Vlasov system is believed to have static ``star-like'' critical solutions \cite{RRS98,OC02}, but critical collapse involving naked singularities has so far not been observed. These numerical studies on critical collapse (see also the survey \cite{Gundlach2007}) have yet to be made rigorous.

At first glance, the \emph{Reissner--Nordstr\"om} family of metrics (indexed by the \emph{mass} $M>0$ and \emph{charge}~$e$) appears to exhibit a type of critical behavior: the solution contains a black hole when $|e|<M$ (\emph{subextremal}) or $|e|=M$ (\emph{extremal}) and does not contain a black hole when $|e|>M$ (\emph{superextremal}). However, the Reissner--Nordstr\"om black holes are eternal and arise from two-ended Cauchy data, while the superextremal variants contain an eternal ``naked singularity'' that has historically caused much confusion. Moreover, it was long thought that extremal black holes could not form dynamically (a consideration closely related to the \emph{third law of black hole thermodynamics} \cite{BCH,Israel-third-law}). Were this true, it would seem to rule out extremal Reissner--Nordstr\"om as the late-time behavior of any critical solution. In \cite{KU22}, the present authors \emph{disproved} the third law in the Einstein--Maxwell-charged scalar field model and showed that an exactly extremal Reissner--Nordstr\"om domain of outer communication can indeed form in gravitational collapse.

In this paper, we continue our investigation of extremal black hole formation by showing that extremal Reissner--Nordstr\"om \emph{does} arise as a critical solution in gravitational collapse for the Einstein--Maxwell--Vlasov model, giving an example of a new phenomenon that we call \emph{extremal~critical~collapse}. 

\begin{thm} \label{thm:main} 
There exist extremal black holes on the threshold between collapsing and dispersing smooth configurations of charged matter. More precisely, for any mass $M>0$, fundamental charge $\mathfrak e\ne 0$, and particle mass $0\le \mathfrak m\le \mathfrak m_0$, where $0<\mathfrak m_0\ll 1$ depends only on $M$ and $\mathfrak e$, there exists a \underline{smooth one-parameter family} of smooth, spherically symmetric, one-ended asymptotically flat Cauchy data $\{\Psi_\lambda\}_{\lambda\in[0,1]}$ for the Einstein--Maxwell--Vlasov system for particles of fundamental charge $\mathfrak e$ and mass $\mathfrak m$, such that the maximal globally hyperbolic development of $\Psi_\lambda$, denoted by $\mathcal D_\lambda$, has the following properties.
\begin{enumerate}
\item $\mathcal D_0$ is  isometric to Minkowski space and there exists $\lambda_*\in(0,1)$ such that for $\lambda < \lambda_*$, $\mathcal D_\lambda$ is future causally geodesically complete and \ul{disperses towards Minkowski space}. 
In particular, $\mathcal D_\lambda$ \ul{does not contain a black hole or naked singularity}. If $\lambda<\lambda_*$ is sufficiently close to $\lambda_*$, then for sufficiently large advanced times and sufficiently small retarded times, $\mathcal D_\lambda$ is isometric to an appropriate causal diamond in a \underline{superextremal} Reissner--Nordstr\"om solution. 

\item $\mathcal D_{\lambda_*}$ contains a nonempty black hole region $\mathcal{BH}\doteq\mathcal M\setminus J^-(\mathcal I^+)$ and for sufficiently large advanced times, the domain of outer communication, including the event horizon $\mathcal H^+\doteq \partial(\mathcal{BH})$, is isometric to that of an \underline{extremal} Reissner--Nordstr\"om solution of mass $M$. The spacetime contains \ul{no trapped surfaces}.

\item For $\lambda>\lambda_*$, $\mathcal D_\lambda$ contains a nonempty black hole region $\mathcal{BH}$ and for sufficiently large advanced times, the domain of outer communication, including the event horizon $\mathcal H^+$, is isometric to that of a  \underline{subextremal} Reissner--Nordstr\"om solution. The spacetime contains an \ul{open set of trapped surfaces}. 
     
   \end{enumerate}
In addition, for every $\lambda\in [0,1]$, $\mathcal D_\lambda$ is past causally geodesically complete, possesses complete null infinities $\mathcal I^+$ and $\mathcal I^-$, and is isometric to Minkowski space near the center $\{r=0\}$ for all time. 
  \end{thm}

\begin{figure}[ht]
\centering{
\def\svgwidth{30pc}
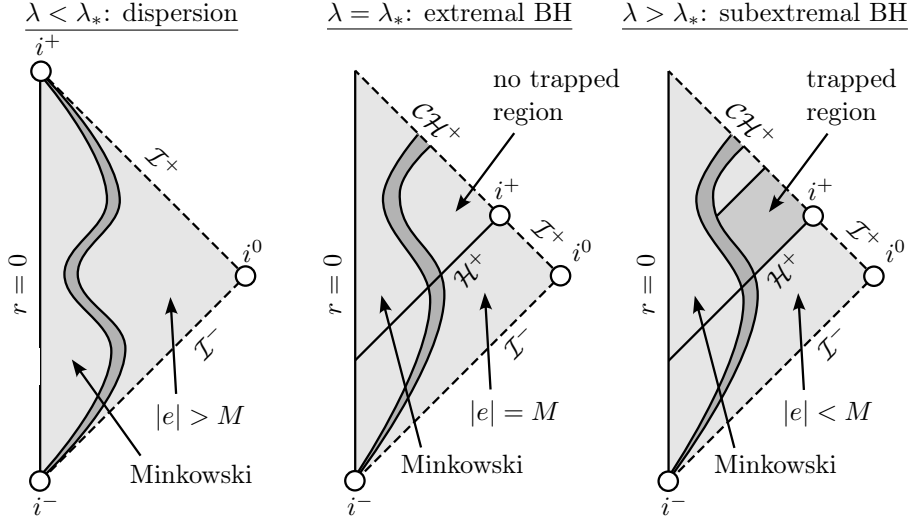}
\caption{Penrose diagrams of the one-parameter family $\{\mathcal D_\lambda\}$ from \cref{thm:main} in the case of massive particles. The dark gray region depicts the physical space support of the Vlasov matter beam. The region of spacetime to the left of the beam is exactly Minkowski space and the region to the right of the beam is exactly Reissner--Nordstr\"om with the parameter ratio as depicted. In every case, the beam ``bounces'' before it hits the center $\{r=0\}$ due to the repulsive effects of angular momentum and the electromagnetic field. When $\lambda<\lambda_*$, the beam bounces before a black hole is formed.  We note already that the beams actually have more structure than is depicted here in these ``zoomed out'' pictures. See already \cref{fig:zoomed-1}.}
\label{fig:bounce-1}
\end{figure}

\begin{figure}[ht]
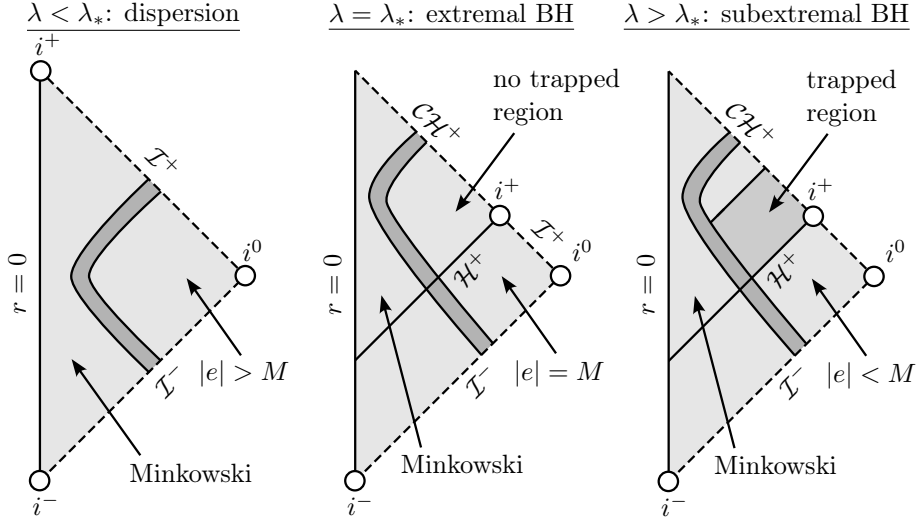

\centering{
\def\svgwidth{30pc}
}
\caption{Penrose diagrams of the one-parameter family $\{\mathcal D_\lambda\}$ from \cref{thm:main} in the massless case. In the ingoing (resp., outgoing) phases, the massless beams are entirely contained in slabs of finite advanced (resp., retarded) time. Therefore, for sufficiently early advanced times and sufficiently late retarded times, the solutions are vacuum and isometric to Minkowski space.}
\label{fig:bounce-2}
\end{figure}

In the proof of \cref{thm:main}, we construct one-parameter families of charged Vlasov beams coming in from past timelike infinity (if $\mathfrak m>0$, cf.~\cref{fig:bounce-1}) or from past null infinity (if $\mathfrak m=0$, cf.~\cref{fig:bounce-2}).
In the dispersive case $\lambda <\lambda_*$, the area-radius $r$ of the beam grows linearly in time as the matter expands towards the future. Moreover, the macroscopic observables of the Vlasov matter (the particle current $N$ and energy momentum tensor $T$) decay at the sharp $t^{-3}$ rate in the massive case and at the sharp $t^{-2}$ rate in the massless case (with faster decay for certain null components), see already \eqref{eq:massive-decay}--\eqref{eq:massless-decay-3} in \cref{prop:approx-1}. In fact, this same dispersive behavior happens in the \emph{past} for every $\lambda\in[0,1]$.  

As a direct consequence of \cref{thm:main}, we obtain
\begin{cor}\label{cor:blackholeness-instab}
    The very ``black hole-ness'' of an extremal black hole arising in gravitational collapse can be unstable: There exist one-ended asymptotically flat Cauchy data for the Einstein--Maxwell--Vlasov system, leading to the formation of an extremal black hole, such that an arbitrarily small smooth perturbation of the data leads to a future causally geodesically complete, dispersive spacetime. 
\end{cor}

This is in stark contrast to the subextremal case, where formation of trapped surfaces behind the event horizon---and hence \emph{stable} geodesic incompleteness \cite{Penrose}---is expected. Despite this inherent instability of the critical solution, we expect extremal critical collapse itself to be a stable phenomenon: We conjecture that there exists a telelologically determined ``hypersurface'' $\mathfrak B_\mathrm{crit}$ in moduli space which consists of asymptotically extremal black holes, contains $\mathcal D_{\lambda_*}$, and locally delimits the boundary in moduli space between future complete and collapsing spacetimes. This ``codimension-one'' property is expected to hold for other variants of critical collapse and will be discussed in detail in \cref{subsec:stability-of-ecc}.

We further expect extremal critical collapse to be a more general phenomenon: we conjecture it to occur in the spherically symmetric Einstein--Maxwell-charged~scalar~field model and also for the Einstein vacuum equations, where extremal Kerr is the model critical solution. In this paper, we also prove (see already \cref{thm:ECC-Ori} in \cref{sec:Ori-ECC}) that extremal critical collapse already occurs in the simpler---but singular---\emph{bouncing charged null dust model}, which was first introduced by Ori in \cite{Ori91}. The proof of \cref{thm:main}, which will be outlined in \cref{sec:proof-outline}, can be viewed as a global-in-time desingularization of these extremal critical collapse families in dust. 

Besides the Einstein--Maxwell--Vlasov and bouncing charged null dust models, it turns out that the \emph{thin charged shell model} \cite{israel1966singular, Delacruz-Israel} also exhibits extremal critical collapse: Pr\'oszy\'nski observed in \cite{Proszynski} that if a thin charged shell is injected into Minkowski space (so the interior of the shell is always flat), the parameters can be continuously varied so that the exterior of the shell goes from forming a subextremal Reissner--Nordstr\"om black hole, to forming an extremal Reissner--Nordstr\"om black hole, to forming no black hole or naked singularity at all: the shell ``bounces'' off to future timelike infinity. Because the thin shell model is quite singular (the energy-momentum tensor is merely a distribution and the metric can fail to be $C^1$ across the shell), it seems to have been discounted as a serious matter model. We refer to the discussion in \cite[Section 1.4.1]{KU22} in reference to the thin charged shell counterexample to the third law by Farrugia and Hajicek \cite{Farrugia-Hajicek}. \cref{thm:main} can be viewed as a vindication of \cite{Proszynski}, since our smooth Einstein--Maxwell--massive Vlasov solutions exhibit all of the qualitative features of Pr\'oszy\'nski's dust shells. In particular, \cref{fig:time-symmetry} below is strikingly similar to Fig.~3 in \cite{Proszynski}.

\subsection{The Einstein--Maxwell--Vlasov system}\label{sec:intro-general}

In this paper, we consider the \emph{Einstein--Maxwell--Vlasov} system, which models a distribution of collisionless, self-gravitating charged particles with \emph{mass} $\mathfrak m\ge 0$ and \emph{fundamental charge} $\mathfrak e\in\Bbb R\setminus\{0\}$. The model consists of a quadruple $(\mathcal M^4,g,F,f)$, where $(\mathcal M^4,g)$ is a spacetime, $F$ is a closed 2-form representing the electromagnetic field strength, and $f=f(x,p)$, called the \emph{distribution function} of the Vlasov matter, is  a smooth nonnegative function defined on the \emph{mass shell} 
\begin{equation*}
    P^\mathfrak m\doteq \{(x,p)\in T\mathcal M:\text{$p$ is future-directed causal and }g(p,p)=-\mathfrak m^2\}.
\end{equation*} 
The equations of motion are
\begin{align}
 \label{eq:intro-EMV}    R_{\mu\nu}-\tfrac 12 Rg_{\mu\nu}&=2\left(T^{\mathrm{EM}}_{\mu\nu}+T_{\mu\nu}\right),\\   
\label{eq:intro-Max}  \nabla_\mu F^{\mu\nu}&=-\mathfrak eN^\nu,\\
\label{eq:intro-MV} Xf&=0,
\end{align} where $T^\mathrm{EM}_{\mu\nu}\doteq F_\mu{}^\alpha F_{\nu\alpha}-\frac 14g_{\mu\nu}F_{\alpha\beta}F^{\alpha\beta}$ is the energy-momentum tensor of the electromagnetic field, $N$ and $T$ are the number current and energy-momentum tensor of the Vlasov matter, defined by
\begin{equation}\label{eq:T-and-N-intro}
     N^\mu[f](x)\doteq \int_{P^\mathfrak m_x}p^\mu f(x,p)\,d\mu_x^\mathfrak m(p),\quad T^{\mu\nu}[f](x)\doteq \int_{P^\mathfrak m_x}p^\mu p^\nu f(x,p)\,d\mu_x^\mathfrak m(p),
\end{equation}
and $X\in\Gamma(TT\mathcal M)$ is the \emph{electromagnetic geodesic spray} vector field, defined relative to canonical coordinates $(x^\mu,p^\mu)$ on $T\mathcal M$ by
    \begin{equation}\label{eq:Lor-spray}
    X\doteq  p^\mu \frac{\partial}{\partial x^\mu} -\left( \Gamma^\mu_{\alpha\beta}p^\alpha p^\beta-\mathfrak e F^\mu{}_\alpha p^\alpha\right) \frac{\partial}{\partial p^\mu}.
\end{equation}
For the definition of the family of measures $d\mu_x^\mathfrak m$ on $P^\mathfrak m$ and a proof of the consistency of the system \eqref{eq:intro-EMV}--\eqref{eq:intro-MV}, we refer to \cref{app:A}. 

The integral curves of the vector field $X$ consist of curves of the form $s\mapsto (\gamma(s),p(s))\in T\mathcal M$, where $p=d\gamma/ds$ and $p$ satisfies the \emph{Lorentz force equation}
\begin{equation*}
    \frac{Dp^\mu}{ds}=\mathfrak eF^\mu{}_\nu p^\nu.
\end{equation*}
We refer to such curves $\gamma$ as \emph{electromagnetic geodesics}. The vector field $X$ is tangent to $P^\mathfrak m$ for any $\mathfrak m\ge 0$, and the Vlasov equation \eqref{eq:intro-MV} implies that $f$ is conserved along electromagnetic geodesics. Since $f\ge 0$, $N$ is a future-directed causal vector field on $\mathcal M$ and the model satisfies the dominant energy condition. 

When $\mathfrak m>0$, the system \eqref{eq:intro-EMV}--\eqref{eq:intro-MV} is locally well-posed outside of symmetry, which can be seen as a special case of results in \cite{Blancel} or by applying the general methods of \cite{Ringstrom-topology}. Well-posedness when $\mathfrak m=0$ is conditional and is a delicate issue that we will return to in \cref{sec:local-WP}. We emphasize at this point that \cref{thm:main} produces examples of extremal critical collapse for any sufficiently small positive particle mass, where well-posedness is unconditional and valid outside of spherical symmetry.

 \subsection{The problem of critical collapse}

We would like to place \cref{thm:main} into the larger picture of \emph{critical collapse}, the general study of the black hole formation threshold. In particular, we conjecture that our examples in \cref{thm:main} have a suitable \emph{codimension-one} property as is expected to hold for other, so far only numerically observed, critical phenomena in gravitational collapse.

 In order to discuss the general concept of critical collapse, it is very helpful to have a notion of ``phase space'' or \emph{moduli space} for initial data (or maximal Cauchy developments) for the Einstein equations. Consider, formally, the set $\mathfrak M$ of one-ended asymptotically flat Cauchy data for the Einstein equations with a fixed matter model (or vacuum) and perhaps with an additional symmetry assumption.
We will be intentionally vague about what regularity elements of $\mathfrak M$ have, what decay conditions to impose, or what topology to endow $\mathfrak M$ with. We will also not discuss gauge conditions, which could be viewed as taking specific quotients of $\mathfrak M$. These questions are related to several fundamental issues in general relativity, see for instance \cite{Cnaked,C-WCC,christodoulou2002global,dafermos2018rough,LukOhI,diophantine,RSR,LeonhardPhD,KM24,Jay-linear}.\footnote{In particular, it would actually be most natural to define $\mathfrak M$ in terms of (perhaps a quotient space of) \emph{scattering data} on past infinity (past null infinity $\mathcal I^-$ in the case of massless fields). However, since a nonlinear scattering theory for the full Einstein equations has not yet been developed in any regime, we limit ourselves to the Cauchy problem for now.} Indeed, it seems likely that there is no single ``correct'' definition---it is doubtful that a single moduli space will capture every interesting phenomenon.

Nevertheless, we will pretend in this section that a ``reasonable'' definition of $\mathfrak M$ exists. At the very least, $\mathfrak M$ ought to consist of initial data possessing a well-posed initial value problem. For each element $\Psi=(\bar g,\bar k,\dotsc)\in\mathfrak M$ (where $\bar g$ is a Riemannian metric on $\Bbb R^3$, $\bar k$ the induced second fundamental form, and $\dotsc$ denotes possible matter fields), we have a unique maximal globally hyperbolic development $\mathcal D=(\mathcal M,g,\dotsc)$ of $\Psi$, where $\mathcal M\cong \Bbb R^4$ \cite{MR53338, CBG69, Zorn-slayed}.\footnote{By an abuse of terminology, we will interchangeably refer to either $\Psi$ or its development $\mathcal D$, which is of course only unique up to isometry.} We assume that $(\mathcal M,g)$ is asymptotically flat. In particular, we assume that we have a well-defined notion of future null infinity $\mathcal I^+$ and past null infinity $\mathcal I^-$. 

\begin{rk}
    In the proof of Theorem 1, we define a ``naive moduli space'' $\mathfrak M_\infty$ consisting of all smooth solutions of the Einstein--Maxwell--Vlasov constraint equations on $\Bbb R^3$, equipped with the $C^\infty_\mathrm{loc}$ topology, and with no identifications made. See already \cref{def:moduli-space-coarse}.
    This topology is inadequate for addressing asymptotic stability questions but since our families are electrovacuum outside a fixed large compact set anyway, they will be continuous in any ``reasonable'' topology that respects asymptotic flatness. 
\end{rk}

Let $\mathfrak C\subset\mathfrak M$ denote the subset of initial data with \emph{future causally geodesically complete} developments. We also highlight the special class $\mathfrak D\subset\mathfrak C$ of initial data with \emph{dispersive} developments, i.e, those solutions whose geometry asymptotically converges to Minkowski space in the far future and matter fields decay suitably.\footnote{Again, we are being intentionally vague here.} Nontrivial stationary states, if they exist, lie in $\mathfrak C\setminus\mathfrak D$ since they do not decay.\footnote{According to a famous theorem of Lichnerowicz, the Einstein vacuum equations do not admit nontrivial asymptotically flat stationary solutions on $\Bbb R^3\times \Bbb R$ (with an everywhere timelike Killing field) \cite{lichnerowicz1955theories}. On the other hand, the Einstein--Vlasov and Einstein--Maxwell--Vlasov systems, for example, have many asymptotically flat stationary solutions \cite{rein1993smooth,Thaller-thin-shell,Thaller-rotating}.}  Let $\mathfrak B\subset\mathfrak M$ denote the set of initial data leading to the formation of a nonempty \emph{black hole region}, i.e., $\mathcal{BH}\doteq \mathcal M\setminus J^-(\mathcal I^+)\ne\emptyset$. The question of \emph{critical collapse} is concerned with the study of phase transitions between $\mathfrak C$, $\mathfrak D$, and $\mathfrak B$, that is, the structure of the boundaries $\partial \mathfrak C$, $\partial\mathfrak D$, and $\partial \mathfrak B$, how they interact, and characterizing solutions lying on the threshold.  

A natural way of exploring this phase transition is by studying continuous paths of initial data interpolating between future complete and black hole forming solutions.
\begin{defn} An \emph{interpolating family} is a continuous one-parameter family $\{\Psi_\lambda\}_{\lambda\in[-1,1]}\subset\mathfrak M$ such that $\Psi_{0}\in\mathfrak C$ and $\Psi_1\in\mathfrak B$. Given such a family, we may define the \emph{critical parameter} $\lambda_*$ and the \emph{critical solution} $\mathcal D_{\lambda_*}$ (the development of $\Psi_{\lambda_*}$) by \begin{equation*}\lambda_*\doteq \sup\{\lambda\in[0,1]: \Psi_\lambda\in\mathfrak C\}.\end{equation*}
\end{defn}

The prototypical critical collapse scenario consists of a spherically symmetric self-gravitating massless scalar field pulse with fixed profile and ``total energy'' $\sim \lambda$. At $\lambda=0$, the solution is Minkowski space and for $\lambda$ very close to $0$, the solution disperses and is future complete \cite{Christ-stab-mink}. As $\lambda$ approaches $1$, a trapped surface forms in evolution, signaling the formation of a black hole \cite{C-ss-formation}. This is precisely the scenario first studied numerically by Christodoulou in his thesis \cite{Chr71} and then later by Choptuik in the influential work \cite{choptuik1993universality}. Based on numerical evidence, it is believed that the critical solutions for these types of families are \emph{naked singularities} that form a codimension-one ``submanifold'' in moduli space. For discussion of Choptuik's results we refer to the survey \cite{Gundlach2007}. 

\begin{rk}
    A codimension-one submanifold of naked singularities is nongeneric and therefore compatible with the weak cosmic censorship conjecture, which has been proved in this model by Christodoulou \cite{C-WCC}.
\end{rk}

\begin{rk}
 A rigorous understanding of Choptuik's critical collapse scenario would in particular give a construction of naked singularities in the Einstein-scalar field system 
 starting from \emph{smooth} initial data, in contrast to Christodoulou's examples in \cite{Cnaked}. It already follows from work of Christodoulou \cite{C-ss-formation} that a critical solution cannot be a black hole in this model and from work of Luk and Oh that a critical solution cannot ``scatter in BV norm'' \cite{luk2015quantitative}. This leaves the possibility of either a first singularity along the center not hidden behind an event horizon\footnote{See \cite[Page 10]{Kommemi13} for a catalog of the possible Penrose diagrams in this case.} or a solution in $\mathfrak C\setminus\mathfrak D$ which ``blows up at infinity.'' Ruling out this latter case is an interesting open problem.
\end{rk}

When massive fields are introduced, such as in the spherically symmetric Einstein--massive Klein--Gordon or Einstein--massive Vlasov systems, then static ``star-like'' critical solutions can be observed numerically \cite{brady1997phases,RRS98,OC02,AR06,AAR21}.  These static solutions are nonsingular and lie in $\mathfrak C\setminus\mathfrak D$. It is interesting to note that while Einstein--Klein--Gordon also displays Choptuik-like naked singularity critical solutions, there is no numerical evidence for the existence of naked singularities in the Einstein--Vlasov system. We again refer to \cite{Gundlach2007} for references and would also like to point out the new development \cite{Baumgarte2023-ad} on numerical critical collapse in vacuum. 

\subsection{Extremal critical collapse} 

So far, all numerically observed critical solutions are believed to be either naked singularities or complete and nondispersive. It follows at once from Penrose's incompleteness theorem \cite{Penrose} and Cauchy stability that a critical solution cannot contain a trapped surface. While a generic black hole is expected to contain trapped surfaces,\footnote{By the celebrated redshift effect, one expects a spacetime asymptoting to a subextremal Kerr--Newman black hole to contain trapped surfaces asymptoting to future timelike infinity $i^+$. See \cite{dafermos2005interior,Price-law,van2018stability,An2023-zt}.} members of the extremal Kerr--Newman family do not. In view of this, we raise the question of whether extremal black holes can arise on the black hole formation threshold:

\begin{defn}\label{def:ECC}
    An interpolating family $\{\Psi_\lambda\}_{\lambda\in[0,1]}$ exhibits \emph{extremal critical collapse} if the critical solution $\mathcal D_{\lambda_*}$ asymptotically settles down to an extremal black hole.
\end{defn}

Our main result, \cref{thm:main}, proves that the Einstein--Maxwell--Vlasov system exhibits extremal critical collapse, with critical solution $\mathcal D_{\lambda_*}$ exactly isometric to extremal Reissner--Nordstr\"om in the domain of outer communication at late advanced times. As shown by Pr\'oszy\'nski \cite{Proszynski} and the present authors in \cref{thm:ECC-Ori}, the fundamentally singular thin charged shell and charged null dust models, respectively, exhibit extremal critical collapse, also with extremal Reissner--Nordstr\"om as the critical solution. We expect this phenomenon to also occur in the spherically symmetric Einstein--Maxwell-charged scalar field system and even for the Einstein vacuum equations, where the critical solution is expected to be based on the extremal Kerr solution. Note that we only require the \emph{asymptotic geometry} of the critical solution to be an extremal black hole in \cref{def:ECC}, which is a much weaker condition than being exactly extremal as in \cref{thm:main}.

\begin{rk}
    Because black holes in the spherically symmetric Einstein-scalar field model always contain trapped surfaces \cite{C-ss-formation}, this model \emph{does not} exhibit extremal critical collapse. In particular, since the presence of a trapped surface in this model already implies completeness of null infinity and the existence of a black hole \cite{dafermos-trapped-surface}, \emph{$\mathfrak B$ is open in the spherically symmetric Einstein-scalar field model}.
\end{rk}

\begin{rk}\label{rk:axisymmetry}
It is not possible for a Kerr solution with nonzero angular momentum (i.e., not Schwarzschild) to appear as the asymptotic state in axisymmetric vacuum gravitational collapse. This is because the \emph{Komar angular momentum} $(16\pi)^{-1}\int_{S}\star dZ^\flat$, where $Z$ is the axial Killing vector field, is independent of the sphere $S$, which is nullhomologous.
Similarly, it is not possible for a Kerr--Newman solution with nonzero charge (i.e., not Kerr) to appear as the asymptotic state in gravitational collapse for the Einstein--Maxwell system. This is because the charge $(4\pi)^{-1}\int_S\star F$ is independent of the sphere $S$, which is nullhomologous. The presence of charged matter is essential in \cref{thm:main}. 
\end{rk}

\begin{rk}[Stationary solutions and the extremal limit] In the 1960s and '70s, it was suggested that astrophysical black holes could form through quasistationary accretion processes. In a landmark work, Bardeen and Wagoner \cite{bardeen1970kerr,bardeen1971relativistic} numerically studied axisymmetric stationary states of the Einstein-dust system (modeling accretion disks) and found that a ``black hole limit'' was only possible in the ``extremal limit'' of the dust configuration.\footnote{Recall that the classical Buchdahl inequality states that a spherically symmetric stationary fluid ball is always ``far away'' from being a black hole in the sense that $\frac{2m}{r}<\frac 89$, which quantitatively forbids (even marginally) trapped surfaces \cite{buchdahl1959general}. This bound is relaxed outside of spherical symmetry or in the presence of charge. In particular, the sharp charged Buchdahl inequality in \cite{andreasson2009sharp} is consistent with becoming arbitrarily close to extremality and forming a marginally trapped surface.} In this limit, the exterior metric of the disk converges, in a certain sense, to the metric of the domain of outer communication of extremal Kerr. 

However, the event horizon of a stationary black hole is necessarily a Killing horizon and therefore an exactly stationary black hole solution cannot admit a one-ended asymptotically flat Cauchy hypersurface. It follows that a sequence of one-ended stationary states cannot actually smoothly converge to a black hole spacetime up to and including the event horizon, and that the black hole threshold cannot be directly probed by studying limits of stationary states---black hole formation is a fundamentally dynamical process. 

Nevertheless, there is a substantial body of numerical and heuristic literature exploring ``extremal black hole limits'' of stationary solutions in dust models 
\cite{neugebauer1995general,meinel2006black,meinel2008relativistic,meinel2011black,kleinwachter2011black} and using Einstein--Yang--Mills--Higgs magnetic monopoles \cite{lue1999magnetic,lue2000gravitational}; see also references therein. In particular, we refer the reader to \cite{meinel2011black} for a cogent explanation of the exact nature of the convergence of these stationary states to extremal Reissner--Nordstr\"om/Kerr exteriors and throats. It would be interesting to see if perturbing these ``near-extremal'' non-black hole stationary states can provide another route to extremal critical collapse (and also perhaps to new examples of third law violating solutions), but this seems to be a difficult and fully dynamical problem.
\end{rk}

\begin{rk}[Overcharging and overspinning] Extremal critical collapse should not be confused with the attempt to overcharge or overspin a black hole, i.e., the attempt to destroy the event horizon and create a ``superextremal naked singularity'' by throwing charged or spinning matter into a (near-)extremal black hole. The fear of forming such a naked singularity provided some impetus for the original formulation of the third law in \cite{BCH}\footnote{With this in mind, the formulation of the third law in \cite{BCH} can be thought of as simply outright forbidding the formation of extremal black holes. The formulation in Israel's work \cite{Israel-third-law} is more refined and specifically refers to subextremal black holes ``becoming'' extremal in a dynamical process. In any case, both formulations are false as shown in \cite{KU22} and again in the present paper.} and many arguments for and against have appeared in the literature, see \cite{Wald-Gedanken-I,hubeny1999overcharging,jacobson2009overspinning,Wald-Gedanken-II} and references therein. Overcharging has been definitively disproved in spherical symmetry for the class of ``weakly tame'' matter models \cite{dafermos-trapped-surface,Kommemi13}, which includes the Einstein--Maxwell--Vlasov system considered in this paper. We expect overcharging and overspinning to be definitively disproved with a positive resolution of the black hole stability problem for extremal black holes, to be discussed in \cref{subsec:stability-of-ecc} below. 
\end{rk}

\subsection{Stability of extremal critical collapse} \label{subsec:stability-of-ecc}

Before discussing the stability of our interpolating families in \cref{thm:main}, we must first address the expected notion of stability for the domain of outer communication of the extremal Reissner--Nordstr\"om solution. 

Firstly, since the asymptotic parameter ratio of the black hole is inherently unstable, we can at most expect a \emph{positive codimension} stability statement for extremal Reissner--Nordstr\"om. This should be compared with the codimension-three nonlinear stability theorem of the Schwarzschild solution by Dafermos, Holzegel, Rodnianski, and Taylor \cite{DHRT}: Only a codimension-three ``submanifold'' of moduli space can be expected to asymptote to Schwarzschild, which has codimension three in the Kerr family (parametrized by the mass and specific angular momentum vector). In the case of Reissner--Nordstr\"om, the set of extremal solutions has \emph{codimension one} in the full family. Indeed, any fixed parameter ratio subfamily of the Reissner--Nordstr\"om family has codimension one. See already \cref{rk:subextremal-foliation}.

Secondly, and far less trivially, the stability problem for extremal black holes is complicated by the absence of the celebrated \emph{redshift effect}, which acts as a stabilizing mechanism for the event horizon of subextremal black holes. The event horizon of extremal Reissner--Nordstr\"om (and axisymmetric extremal black holes in general) suffers from a linear instability known as the \emph{Aretakis instability} \cite{Aretakis-instability-1,Aretakis-instability-2,Aretakis-instability-3,apetroaie}, which causes ingoing translation invariant null derivatives of solutions to the linear wave equation to (generically) either not decay, or to blow up polynomially along the event horizon as $v\to\infty$. Weissenbacher has recently shown that a similar instability (non-decay of the first derivative of the energy-momentum tensor) occurs for the linear massless Vlasov equation on extremal Reissner--Nordstr\"om \cite{W23}. 

However, the Aretakis instability is weak and does not preclude asymptotic stability and decay \emph{away from the event horizon}. Including the horizon, we expect a degenerate type of stability, with decay in directions tangent to it, and possible non-decay and growth transverse to it (so-called \emph{horizon hair}). This behavior has been shown rigorously for a semilinear model problem on a fixed background \cite{A16, AAG20} and numerically for the coupled spherically symmetric nonlinear Einstein--Maxwell-(massless and neutral) scalar field system \cite{Reall-numerical}.

To further complicate matters, the massive and massless Vlasov equations behave fundamentally differently and we state two separate conjectures. In these statements, we consider characteristic data posed on a bifurcate null hypersurface $ C_\mathrm{out}\cup \underline C_\mathrm{in}$, where $C_\mathrm{out}$ is complete and $\underline C_\mathrm{in}$ penetrates the event horizon in the case of trivial data. Solutions of the linear massless Vlasov equation decay exponentially on subextremal Reissner--Nordstr\"om black holes \cite{Bigorgne-Schw,W23} and Velozo Ruiz has proved nonlinear asymptotic stability of Schwarzschild for the spherically symmetric Einstein--massless Vlasov system \cite{VelozoPhD}. Based on this, \cite{Reall-numerical,A16,AAG20}, and \cite[Conjecture IV.2]{DHRT}, we make the

\begin{conj}\label{conj:massless-stab}
    The extremal Reissner--Nordstr\"om solution is nonlinearly asymptotically stable to spherically symmetric perturbations in the Einstein--Maxwell--massless Vlasov model in the following sense: Given sufficiently small characteristic data posed on a bifurcate null hypersurface $ C_\mathrm{out}\cup \underline C_\mathrm{in}$ and lying on a ``codimension-one submanifold'' $\mathfrak M_\mathrm{stab}$ (which contains the trivial solution) of the moduli space of such initial data, the maximal Cauchy development contains a black hole which asymptotically settles down to the domain of outer communication of an extremal Reissner--Nordstr\"om solution, away from the event horizon $\mathcal H^+$. Moreover, along the horizon, the solution decays towards extremal Reissner--Nordstr\"om in tangential directions, with possibly growing ``Vlasov hair'' transverse to the horizon. 
\end{conj}

\begin{rk}
 There exist nontrivial spherically symmetric static solutions of the Einstein--massless Vlasov system containing a black hole which are isometric to a Schwarzschild solution in a neighborhood of the event horizon   \cite{Hakan-massless-steady}. However, these are not small perturbations of Schwarzschild (as the structure of trapping for null geodesics has to be significantly modified) and their existence is therefore consistent with \cite{VelozoPhD} and \cref{conj:massless-stab}.
\end{rk}

The massive Vlasov equation admits many nontrivial stationary states on black hole backgrounds, which is an obstruction to decay and we do not expect a general asymptotic stability statement to hold, even in the subextremal case. In fact, it has been shown that there exist spherically symmetric static solutions of Einstein--massive Vlasov bifurcating off of Schwarzschild \cite{Rein-static,Schw-bifurcation}. We refer to \cite{VelozoPhD} for a characterization of the ``largest'' region of phase space on which one can expect decay for the massive Vlasov energy-momentum tensor on a Schwarzschild background. However, one might still hope for \emph{orbital} stability of the exterior, with a non-decaying Vlasov matter \emph{atmosphere}, and that the horizon itself decays to that of extremal Reissner--Nordstr\"om:

\begin{conj}\label{conj:massive-stab}
    The extremal Reissner--Nordstr\"om solution is nonlinearly orbitally stable to spherically symmetric perturbations in the Einstein--Maxwell--massive Vlasov model in the following sense: Given
    sufficiently small characteristic data posed on a bifurcate null hypersurface $ C_\mathrm{out}\cup \underline C_\mathrm{in}$ and lying on a ``codimension-one submanifold'' $\mathfrak M_\mathrm{stab}$ of the moduli space of such initial data, the maximal Cauchy development contains a black hole which remains close to an extremal Reissner--Nordstr\"om solution in the domain of outer communication and asymptotically settles down to extremal Reissner--Nordstr\"om tangentially along the horizon, with possibly growing ``Vlasov hair'' transverse to the horizon. 
\end{conj}

\begin{rk}
  We emphasize that this type of nonlinear orbital stability for massive Vlasov has not yet been proven even in the subextremal case, where we do not expect horizon hair to occur.    
\end{rk}

With the conjectured description of the stability properties of the \emph{exterior} of the critical solution at hand, we are now ready to state our conjecture for the \emph{global} stability of the extremal critical collapse families in \cref{thm:main}. Refer to \cref{fig:ecc-stability} for a schematic depiction of this conjecture. 

\begin{figure}
\centering{
\def\svgwidth{15pc}
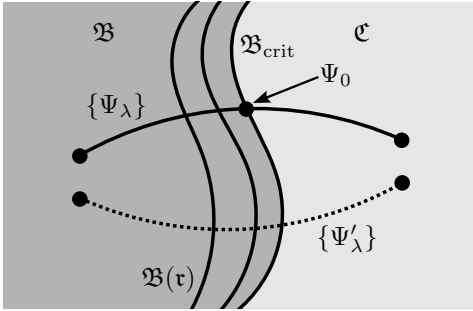}
\caption{A cartoon depiction of the conjectured structure of a neighborhood of moduli space near an interpolating family $\{\Psi_\lambda\}$ from \cref{thm:main}. We have suppressed infinitely many dimensions and emphasize the codimension-one property of the critical ``submanifold'' $\mathfrak B_\mathrm{crit}$ which consists of asymptotically extremal black holes in accordance with \cref{conj:massless-stab,conj:massive-stab}. The interpolating family $\{\Psi_\lambda'\}$ is a small perturbation of $\{\Psi_\lambda\}$ which also crosses $\mathfrak B_\mathrm{crit}$ and exhibits extremal critical collapse. Locally, $\mathfrak B$ is foliated by ``hypersurfaces'' $\mathfrak B(\mathfrak r)$ consisting of black hole spacetimes with asymptotic parameter ratio $\mathfrak r$ close to $1$.}\label{fig:ecc-stability} 
\end{figure}

\begin{conj}\label{conj:ecc-stab} Extremal critical collapse is stable in the following sense:
    Consider the moduli space $\mathfrak M$ of the spherically symmetric Einstein--Maxwell--Vlasov system for particles of mass $\mathfrak m$. Let $\{\Psi_\lambda\}$ be one of the interpolating families given by \cref{thm:main}. Then there exists a ``codimension-one submanifold'' $\mathfrak B_\mathrm{crit}$ of $\mathfrak M$ such that $\Psi_0\in\mathfrak B_\mathrm{crit}\subset\mathfrak B$, which has the following properties:
    \begin{enumerate}
        \item $\mathfrak B_\mathrm{crit}$  is critical in the sense that $\mathfrak B$ and $\mathfrak C$ locally lie on opposite sides of $\mathfrak B_\mathrm{crit}$.  
        \item If $\mathfrak m=0$ and $\Psi\in \mathfrak B_\mathrm{crit}$, the domain of outer communication of the maximal Cauchy development of $\Psi$ asymptotically settles down to an extremal Reissner--Nordstr\"om black hole as in \cref{conj:massless-stab}. 
        \item If $\mathfrak m>0$ and $\Psi\in \mathfrak B_\mathrm{crit}$, the domain of outer communication of the maximal Cauchy development of $\Psi$ remains close to an extremal Reissner--Nordstr\"om black hole and the event horizon asymptotically settles down to an extremal Reissner--Nordstr\"om event horizon as in \cref{conj:massive-stab}.
    \end{enumerate}
    Therefore, any nearby interpolating family $\{\Psi_\lambda'\}$ also intersects $\mathfrak B_\mathrm{crit}$ and exhibits extremal critical collapse. 
\end{conj}

\begin{rk}\label{rk:subextremal-foliation}
   We further conjecture that given $\mathfrak r\in [1-\ve,1]$ for some $\ve>0$, there exists a one-parameter family of disjoint ``codimension-one submanifolds''  $\mathfrak B(\mathfrak r)\subset\mathfrak B$, varying ``continuously'' in $\mathfrak r$, such that $\mathfrak B(1)=\mathfrak B_\mathrm{crit}$ and if $\Psi\in \mathfrak B(\mathfrak r)$, then the maximal Cauchy development of $\Psi$ contains a black hole which asymptotes to a Reissner--Nordstr\"om black hole with parameter ratio $\mathfrak r=e_f/M_f$, where $M_f$ is the final renormalized Hawking mass and $e_f$ is the final charge. One can then interpret equation \eqref{eq:parameter-ratio-5} below as saying that the families $\{\Psi_\lambda\}$ in \cref{thm:main} intersect the foliation $\{\mathfrak B(\mathfrak r)\}$ \emph{transversally}, as depicted in \cref{fig:ecc-stability}.
\end{rk}

While one should think that $\mathfrak B_\mathrm{crit}$ in \cref{conj:ecc-stab} corresponds to $\mathfrak M_\mathrm{stab}$ in \cref{conj:massless-stab,conj:massive-stab}, Part~1 of \cref{conj:ecc-stab} is also a highly nontrivial statement about the \emph{interiors} of the black holes arising from $\mathfrak B_\mathrm{crit}$. In particular, by the incompleteness theorem, it would imply that there are no trapped surfaces in the maximal developments of any member of $\mathfrak B_\mathrm{crit}$; see \cite[Remark 1.8]{gajic-luk} and the following remark.

\begin{rk}
  Using arguments from \cite[Appendix A]{LukOhI}, one can show the following statement in the spherically symmetric Einstein--Maxwell-(neutral and massless) scalar field model: If the maximal Cauchy development of a partial Cauchy hypersurface\footnote{By this, we mean an asymptotically flat spacelike hypersurface which terminates at a symmetric sphere with positive area-radius. If the charge is nonzero and nondynamical (as in the neutral scalar field model), one cannot have a regular center.} with $\partial_ur<0$ contains a black hole with asymptotically extremal parameter ratio, then the development does not contain trapped symmetry spheres. The argument uses crucially the constancy of charge and absence of $T_{uv}$ in this model.
\end{rk}

\subsection{Extremal critical collapse of a charged scalar field and in vacuum} 

It is natural to conjecture the analog of \cref{thm:main} for a massless charged scalar field in spherical symmetry:

\begin{conj}
    Extremal critical collapse occurs in the spherically symmetric Einstein--Maxwell-charged scalar field model and there exist critical solutions which are isometric to extremal Reissner--Nordstr\"om in the domain of outer communication after sufficiently large advanced time. 
\end{conj}

In \cite{KU22}, the present authors showed that a black hole with an extremal Reissner--Nordstr\"om domain of outer communication and containing no trapped surfaces can arise from regular one-ended Cauchy data in the spherically symmetric charged scalar field model (see Corollary 3 of \cite{KU22}). The proof is based on a \emph{characteristic gluing} argument, in which we glue a late ingoing cone in the interior of extremal Reissner--Nordstr\"om to an ingoing cone in Minkowski space. The desired properties of the spacetime are obtained softly by Cauchy stability arguments. In particular, the method is inadequate to address whether the solution constructed in \cite[Corollary 3]{KU22} is critical. 

It is also natural to conjecture the analog of \cref{thm:main} for the Einstein vacuum equations,
\begin{equation}
    \Ric(g)=0,
\end{equation}
where the role of extremal Reissner--Nordstr\"om is played by the rotating \emph{extremal Kerr} solution.\footnote{Recall also \cref{rk:axisymmetry}: replacing ``vacuum'' with ``electrovacuum'' and ``Kerr'' with ``Kerr--Newman with nonzero charge'' in \cref{conj:Kerr} is not possible!}
\begin{conj}\label{conj:Kerr}
    Extremal critical collapse occurs in vacuum gravitational collapse and there exist critical solutions which are isometric to extremal Kerr in the domain of outer communication after sufficiently large advanced time.
\end{conj}

In \cite{KU23}, the present authors constructed examples of vacuum gravitational collapse which are isometric to Kerr black holes with prescribed mass $M$ and specific angular momentum $a$, where $M$ and $a$ are any Kerr parameters satisfying $0\le |a|/M\le \mathfrak a_0$ for some small positive constant $\mathfrak a_0$. Like in \cite{KU22}, the proof is by characteristic gluing and uses the perturbative and obstruction-free gluing results of Aretakis--Czimek--Rodnianski \cite{ACR1} and Czimek--Rodnianski \cite{Czimek2022-cl} as a black box. In particular, this provided a fundamentally new proof of Christodoulou's seminal theorem on black hole formation in vacuum \cite{Christo09}. Our proof does not work for large values of $a$ and whether extremal Kerr black holes can form in gravitational collapse remains open.

\begin{rk} In \cite{KU23}, the Cauchy data $(\bar g,\bar k)$ are constructed with regularity $H^{7/2-}_\loc\times H^{5/2-}_\loc$, which is well above the threshold for classical existence and uniqueness for the Einstein vacuum equations \cite{HKM,Planchon-Rodnianski,C-low-reg}. This limited regularity is because the characteristic gluing results \cite{ACR1,Czimek2022-cl} which we use as a black box are limited to $C^2$ regularity of transverse derivatives in the non-bifurcate case. Using the more recent \emph{spacelike} gluing results of Mao--Oh--Tao \cite{Mao2023-cm}, it is possible to construct suitable Cauchy data in $H^s_\loc\times H^{s-1}_\loc$ for any $s$. 
\end{rk}

If extremal critical collapse involving the Kerr solution does occur, then one may also ask about stability as in \cref{subsec:stability-of-ecc}. In this case, the question hangs on the stability properties of extremal Kerr, which are more delicate than for extremal Reissner--Nordstr\"om. While extremal Kerr is mode-stable \cite{Rita-mode-stability}, axisymmetric scalar perturbations have been shown to exhibit the same non-decay and growth hierarchy as general scalar perturbations of extremal Reissner--Nordstr\"om \cite{Aretakis-Kerr,Aretakis-instability-3}. In light of the newly discovered \emph{azimuthal instabilities} of extremal Kerr by Gajic \cite{Gajic23}, in which growth of scalar perturbations already occurs at first order of differentiability, the full (in)stability picture of extremal Kerr may be one of spectacular complexity!

\subsection{The third law of black hole thermodynamics and event horizon jumping at extremality}\label{sec:EMV-third-law}

 \begin{figure}
\centering{
\def\svgwidth{14pc}
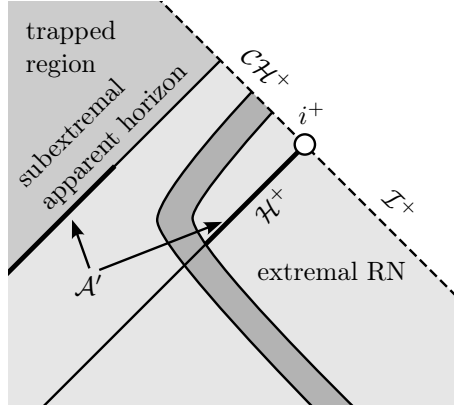}
\caption{Penrose diagram of a counterexample to the third law of black hole thermodynamics in the Einstein--Maxwell--Vlasov model from \cref{thm:third-law-Vlasov}. The broken curve $\mathcal A'$ is the outermost apparent horizon of the spacetime. This view is zoomed in on the Vlasov beam that charges up the subextremal black hole to extremality. We refer to \cref{fig:third-law-Vlasov-proof} in \cref{sec:Vlasov-third-law} for diagrams of the entire spacetime.}
\label{fig:third-law-Vlasov}
\end{figure}

The techniques used to prove \cref{thm:main} can also be immediately used to disprove the third law in the Einstein--Maxwell--Vlasov model, which complements our previous disproof in the Einstein--Maxwell-charged scalar field model \cite{KU22}. The present method has the advantage of constructing counterexamples which are past causally geodesically complete, like the spacetimes in \cref{thm:main}.

\begin{thm}\label{thm:third-law-Vlasov}
There exist smooth solutions of the Einstein--Maxwell--Vlasov system for either massless or massive particles that violate the third law of black hole thermodynamics: a subextremal Reissner--Nordstr\"om apparent horizon can evolve into an extremal Reissner--Nordstr\"om event horizon in finite advanced time due to the incidence of charged Vlasov matter. 

More precisely, there exist smooth, spherically symmetric, one-ended asymptotically flat Cauchy data for the Einstein--Maxwell--Vlasov system for either massive or massless particles such that the maximal globally hyperbolic development $\mathcal D$ has the following properties. 
\begin{enumerate}
    \item $\mathcal D$ contains a nonempty black hole region and for sufficiently large advanced times, the domain of outer communication, including the event horizon $\mathcal H^+$, is isometric to that of an extremal Reissner--Nordstr\"om solution.

    \item $\mathcal D$ contains a causal diamond which is isometric to a causal diamond in a subextremal Reissner--Nordstr\"om black hole, including an appropriate portion of the subextremal apparent horizon. This subextremal region contains an open set of trapped surfaces.

    \item The outermost apparent horizon $\mathcal A'$ of $\mathcal D$ has at least two connected components. One component of $\mathcal A'$ coincides in part with the subextremal apparent horizon and the last component (with respect to $v$) coincides with the extremal event horizon.
    
    \item $\mathcal D$ is past causally goedesically complete, possesses complete null infinities $\mathcal I^+$ and $\mathcal I^-$, and is isometric to Minkowski space near the center $\{r=0\}$ for all time.
\end{enumerate}
\end{thm}

Refer to \cref{fig:third-law-Vlasov} for a Penrose diagram of one of these solutions. Note the disconnectedness of the outermost apparent horizon $\mathcal A'$, which is necessary in third law violating spacetimes---see the discussion in Section 1.4.3 of \cite{KU22}. It is striking that the Vlasov beams we construct in the proof of \cref{thm:third-law-Vlasov} do not even touch the subextremal apparent horizon, which should be compared with the hypothetical situation depicted in Fig.~1 of \cite{Israel-third-law}. As with \cref{thm:main}, \cref{thm:third-law-Vlasov} is proved by desingularizing suitable bouncing charged null dust spacetimes which we construct in \cref{sec:third-law-dust}.

It is now very natural to ask if some critical behavior can be seen in the examples from \cref{thm:third-law-Vlasov}. They are clearly not candidates for critical collapse because they contain trapped surfaces. Nevertheless, by tuning the final charge to mass ratio of the outermost beam in \cref{thm:third-law-Vlasov} (subextremal to superextremal as in \cref{thm:main}), we construct one-parameter families of solutions satisfying the following 

\begin{thm}\label{thm:it-jumps} There exist smooth one-parameter families of smooth, spherically symmetric, one-ended asymptotically flat Cauchy data $\{\Psi_\lambda\}_{\lambda\in[-1,1]}$ for the Einstein--Maxwell--Vlasov system for either massive or massless particles with the following properties. Let $\mathcal D_\lambda$ be a choice of maximal globally hyperbolic development\footnote{Typically, one refers to ``the'' maximal globally hyperbolic development \cite{CBG69}, which is an equivalence class of isometric developments. In this statement, however, it is crucial that the development comes equipped with a fixed coordinate system.} of $\Psi_\lambda$ for which the double null gauge $(u,v)$ is continuously synchronized as a function of $\lambda$. (See already \cref{def:sync} and \cref{def:sync-fam} for the definition of continuous synchronization.) Then the following holds:
    \begin{enumerate}
         \item For $\lambda\ne 0$, $\mathcal D_\lambda$ contains a black hole whose domain of outer communication is isometric to that of a subextremal Reissner--Nordstr\"om black hole with mass $M_\lambda$ and charge $|e_\lambda|<M_\lambda$ for sufficiently large advanced times. 
         \item $\mathcal D_0$ contains a black hole whose domain of outer communication is isometric to that of an extremal Reissner--Nordstr\"om black hole with mass $M_0$ and charge $|e_0|=M_0$ for sufficiently large advanced times.
         \item The location of the event horizon is \ul{discontinuous} as a function of $\lambda$: Let $u_{\lambda,\mathcal H^+}$ denote the retarded time coordinate of the event horizon $\mathcal H^+_\lambda$ of $\mathcal D_\lambda$ with respect to the continuously synchronized gauge $(u,v)$. Then $\lambda\mapsto u_{\lambda,\mathcal H^+}$ is \ul{continuous from the left} but  \ul{discontinuous from the right}, and 
         \begin{equation}
             \lim_{\lambda\to 0^+}u_{\lambda,\mathcal H^+}> \lim_{\lambda\to 0^-}u_{\lambda,\mathcal H^+}.\label{eq:u-jump-intro}
         \end{equation}
         \item  The functions $\lambda\mapsto M_{\lambda}$ and $\lambda\mapsto e_{\lambda}$ are \ul{continuous from the left} but  \ul{discontinuous from the right}, and 
         \begin{equation}
             \lim_{\lambda\to 0^+}M_\lambda < \lim_{\lambda\to 0^-}M_\lambda, \quad \lim_{\lambda\to 0^+}|e_\lambda| < \lim_{\lambda\to 0^-}|e_\lambda|,
         \end{equation}
         \begin{equation}
       \lim_{\lambda\to 0^+}\frac{|e_{\lambda}|}{M_\lambda}< \lim_{\lambda\to 0^-}\frac{|e_{\lambda}|}{M_\lambda}=1,\quad  \lim_{\lambda\to 0^+}r_{\lambda,\mathcal H^+}< \lim_{\lambda\to 0^-}r_{\lambda,\mathcal H^+},
         \end{equation}
         where $r_{\lambda,\mathcal H^+}\doteq M_\lambda + \sqrt{M_\lambda^2-e_\lambda^2}$.
    \end{enumerate}
    In addition, for every $\lambda\in [-1,1]$, $\mathcal D_\lambda$ is past causally geodesically complete, possesses complete null infinities $\mathcal I^+$ and $\mathcal I^-$, and is isometric to Minkowski space near the center $\{r=0\}$ for all time. 
\end{thm}

From the perspective of the dynamical extremal black hole $\mathcal D_0$, an arbitrarily small perturbation to $\mathcal D_\lambda$ with $\lambda>0$ causes the event horizon to \emph{jump} by a definite amount in $u$ (i.e., not $o(1)$ in $\lambda$) and the parameter ratio to drop by a definite amount. The proof of \cref{thm:it-jumps} relies crucially on the absence of trapped surfaces in a double null neighborhood of the horizon in the solutions of \cref{thm:third-law-Vlasov}, cf.~\cref{fig:third-law-Vlasov}. In the asymptotically subextremal case, trapped surfaces are expected to asymptote towards future timelike infinity $i^+$. In this case, we prove in \cref{prop:continuity} below that the \emph{location of the event horizon is continuous} as a function of initial data, under very general assumptions in spherical symmetry. Therefore, \eqref{eq:u-jump-intro} is a characteristic feature of extremal black hole formation. 

We expect this ``local critical behavior'' to be \emph{stable} in the sense of \cref{subsec:stability-of-ecc} and to play a key role in the general stability problem for extremal black holes. 

\begin{rk}
By a suitable modification of the characteristic gluing techniques in \cite{KU22}, \cref{thm:it-jumps} can be proved for the spherically symmetric Einstein--Maxwell-charged scalar field model, but past completeness of the solutions does not follow immediately from our methods. It is also natural to conjecture analogs for \cref{thm:it-jumps} in (electro)vacuum; see in particular \cite[Section IV.2]{DHRT}.\footnote{In fact, the statement of \cref{thm:it-jumps} is not actually reliant on the black holes forming in gravitational collapse and can be made sense of in terms of characteristic data as in \cref{conj:massless-stab,conj:massive-stab}. In this case, one can study the local critical behavior of extremal Reissner--Nordstr\"om in electrovacuum since  \cref{rk:axisymmetry} no longer applies. Indeed, this is precisely the context of the discussion in \cite[Section IV.2]{DHRT}.}
\end{rk}

\subsection{Outline of the paper}
\paragraph{\cref{sec:spherical-symmetry}.} We introduce basic definitions and properties of general Einstein-matter systems and electromagnetic fields in spherical symmetry.  

\paragraph{\cref{sec:EMV-system}.} We first recall the Einstein--Maxwell--Vlasov system in \cref{app:A} and derive its spherically symmetric formulation in double null coordinates in \cref{sec:EMV-SS}. We also prove local well-posedness in \cref{sec:local-WP} (the main iteration argument being deferred to \cref{app:B}) and the generalized extension principle in \cref{sec:ext}. Finally, we define time-symmetric seed data sets and their developments in \cref{sec:general-data} which will play an important role in the paper.

\paragraph{\cref{sec:dust-intro}.} Before turning to the proof of our main result \cref{thm:main} in \cref{sec:Vlasov-beams}, we show in \cref{sec:dust-intro}  that a singular toy model---Ori's bouncing charged null dust model---exhibits extremal critical collapse. We first recall the definition of the model in \cref{sec:ori-bouncing}.  We then introduce a radial parametrization of bouncing charged null dust spacetimes in \cref{sec:radial} in which we teleologically prescribe a regular, spacelike, totally geodesic bounce hypersurface. These spacetimes consist of an explicit ingoing charged Vaidya metric pasted along the radially parametrized bounce hypersurface to an outgoing charged Vaidya metric through a physically motivated surgery procedure. In \cref{sec:Ori-ECC,sec:third-law-dust}, we use the radial parametrization to construct new examples of bouncing charged null dust spacetimes. In \cref{sec:Ori-ECC}, we show that Ori's model exhibits extremal critical collapse (\cref{thm:ECC-Ori}) and in \cref{sec:third-law-dust}, we show that the third law of black hole thermodynamics is false in Ori's model (\cref{thm:third-law-dust}). We then discuss the fundamental flaws of Ori's model in \cref{sec:shortcomings}: the ill-posedness across the bounce hypersurface,  the singular nature of the solutions, and the ill-posedness near the center. In \cref{sec:outgoing-dust-equations}, we conclude \cref{sec:dust-intro} with the formal radial charged null dust system in double null gauge which will be important for the setup of our initial data in \cref{sec:Vlasov-beams}.

\paragraph{\cref{sec:Vlasov-beams}.}
This section is devoted to the proof of our main result, \cref{thm:main}. The proof relies crucially on a very specific teleological choice of Cauchy data which aims at \emph{desingularizing} the dust examples of extremal critical collapse in \cref{thm:ECC-Ori}, globally in time. In \cref{sec:proof-outline}, we give a detailed guide to the proof of \cref{thm:main}.  In \cref{sec:beam-data}, we define the hierarchy of beam parameters and state the key ingredient in the proof of \cref{thm:main}: the existence and global structure of outgoing charged Vlasov beams (\cref{prop:approx-1}). These beams arise from data posed on a Cauchy hypersurface which is analogous to the ``bounce hypersurface'' associated to Ori's model in \cref{sec:dust-intro}. In \cref{sec:data-estimates}, we solve the constraint equations and prove estimates for the solution along the initial data hypersurface. \cref{sec:near-main} is devoted to estimates for the ``near region'' establishing the bouncing character of our Vlasov beams. To overcome certain difficulties associated with low momenta, our construction features an ``auxiliary beam'' which is treated in \cref{sec:near-aux}.  \cref{sec:far-massive} is concerned with the ``far region'' and in \cref{sec:massive-dispersion} we prove the sharp dispersive estimates in the case of massive particles. In \cref{sec:proof-of-prop} we conclude the proof of \cref{prop:approx-1}. This proposition is then used to prove \cref{thm:main} in \cref{sec:patching-ingoing-outgoing}. Finally, in \cref{sec:weak-proof} we show that in a certain hydrodynamic limit of our parameters, the family of solutions constructed in \cref{thm:main} converge in a weak* sense to the family constructed in the charged null dust model in \cref{sec:dust-intro}. This result rigorously justifies Ori's bouncing charged null dust construction from \cite{Ori91}.

\paragraph{\cref{sec:third-law-and-h-jumping}.} In this final section, we disprove the third law of black hole thermodynamics for the Einstein--Maxwell--Vlasov model (\cref{thm:third-law-Vlasov}) in \cref{sec:Vlasov-third-law}. \cref{sec:horizon-jumping} is concerned with the phenomenon of event horizon jumping at extremality. We first show in \cref{prop:continuity} that for a general class of (so-called \emph{weakly tame}) spherically symmetric Einstein-matter systems the retarded time coordinate of the event horizon is lower semicontinuous as a function of initial data.
Secondly, we show by example (\cref{thm:it-jumps}) that event horizon jumping can occur in the Einstein--Maxwell--Vlasov system for extremal horizons, which proves the sharpness of semicontinuity in \cref{prop:continuity}.

\subsection*{Acknowledgments}
The authors would like to express their gratitude to Mihalis Dafermos for many stimulating discussions. They would also like to thank Carsten Gundlach, Istv\'an K\'ad\'ar, Georgios Moschidis, Amos Ori, Frans Pretorius, Harvey Reall, Igor Rodnianski, Andrew Strominger, and Robert Wald for their interest, helpful discussions, and comments. C.K.~acknowledges support by Dr.~Max R\"ossler, the Walter Haefner Foundation, and the ETH Z\"urich Foundation.  R.U.~thanks the University of Cambridge and ETH Z\"urich for hospitality as this work was being carried out.

\section{Spherically symmetric charged spacetimes}\label{sec:spherical-symmetry}

\subsection{Double null gauge}\label{sec:double-null-gauge}

Let $(\mathcal M,g)$ be a smooth, connected, time-oriented, four-dimensional Lorentzian manifold. We say that $(\mathcal M,g)$ is \emph{spherically symmetric} with (possibly empty) \emph{center} of symmetry $\Gamma\subset\mathcal M$ if $\mathcal M \setminus \Gamma $ splits diffeomorphically as $\mathring{\mathcal Q}\times S^2$ with metric 
\begin{align*}
    g = g_\mathcal{Q} + r^2 \gamma,
\end{align*}
where $(\mathcal Q,g_\mathcal Q)$, $\mathcal  Q = \mathring{\mathcal Q} \cup \Gamma$, is a (1+1)-dimensional Lorentzian spacetime with boundary $\Gamma$ (possibly empty), $\gamma\doteq d\vartheta^2+\sin^2\vartheta\,d\varphi^2$ is the round metric on the unit sphere, and $r$ is a nonnegative function on $\mathcal Q$ which can be geometrically interpreted as the area-radius of the orbits of the isometric $\mathrm{SO}(3)$ action on  $(\mathcal M,g)$. In a mild abuse of notation, we denote by $\Gamma$ both the center of symmetry in $\mathcal M$ and its projection to $\mathcal Q$.  Moreover, if  $\Gamma$ is non-empty, we assume that the  $\mathrm{SO}(3)$ action fixes $\Gamma$ and that $\Gamma$ consists of one timelike geodesic along which $r=0$. We further assume that $(\mathcal Q,g_\mathcal Q)$ admits a \emph{global double-null coordinate system} $(u,v)$ such that the metric $g$ takes the form
\begin{equation}
g = -\Omega^2\,dudv + r^2  \gamma\label{eq:dn}
\end{equation}
for a positive function $\Omega^2\doteq -2g_\mathcal{Q}(\partial_u,\partial_v)$ on $\mathcal Q$ and such that $\partial_u$ and $\partial_v$ are future-directed. Our conventions are so that $u=t-r$ and $v=t+r$ give a double null coordinate system on $(3+1)$-dimensional Minkowski space, with $r=\frac 12(v-u)$ and  $\Omega^2\equiv 1$. We will also use the notation $\slashed g\doteq r^2\gamma$. The constant $u$ and $v$ curves are null in $(\mathcal Q,g_\mathcal Q)$ and correspond to null hypersurfaces ``upstairs'' in the full spacetime $(\mathcal M,g)$. We further assume that along the center $\Gamma$, the coordinate $v$ is outgoing and $u$ is ingoing, i.e.,  $\partial_v r \vert_\Gamma >0$, $\partial_u r\vert_\Gamma <0$. We will often refer interchangeably to  $(\mathcal M,g)$ and the reduced spacetime $(\mathcal Q,r,\Omega^2)$.

Recall the \emph{Hawking mass} $m : \mathcal M  \to \mathbb R$, which is defined by \begin{equation*}
    m \doteq \frac r 2 (1 - g (\nabla r, \nabla r))
\end{equation*}
and can be viewed as a function on $\mathcal Q$ according to 
\begin{equation}\label{eq:Hawking-mass}
    m = \frac r2 \left( 1+  \frac{4\partial_ur\partial_vr}{\Omega^2}\right).
\end{equation}
We will frequently use the formula
\begin{equation}
    \Omega^2= \frac{4(-\partial_ur)\partial_vr}{1-\frac{2m}{r}}\label{eq:Omega-formula}
\end{equation}
to estimate $\Omega^2$ when $1-\frac{2m}{r}>0$.

The isometric action of $\mathrm{SO}(3)$ on $(\mathcal M,g)$ extends to the tangent bundle $T\mathcal M$ as follows: Let $\varrho:\mathrm{SO}(3)\to\mathrm{Diff}(\mathcal M)$ be the representation of $\mathrm{SO}(3)$ given by the spherically symmetric ansatz, so that the group action is given by 
\[R\cdot x\doteq \varrho(R)(x)\]
for $R\in\mathrm{SO}(3)$ and $x\in\mathcal M$. For $(x,p)\in T\mathcal M$, we define
\begin{equation}
    R\cdot(x,p)\doteq (\varrho(R)(x),\varrho(R)_*p),\label{eq:TM-action}
\end{equation}
where of course $\varrho(R)_*p$ lies in $T_{\varrho(R)(x)}\mathcal M$.

Finally, we note that the double null coordinates $(u,v)$ above are not uniquely defined and for any strictly increasing smooth functions $U,V\colon \R \to \R$, we obtain new global double null coordinates $(\tilde u,\tilde v) =(U(u),V(v))$ such that $g= -\tilde \Omega^2 \,d \tilde u\, d \tilde v + \slashed g$, where $\tilde \Omega^2(\tilde u , \tilde v) = (U' V')^{-1} \Omega^2(U^{-1}(\tilde u), V^{-1}(\tilde v)) $ and $r (\tilde u, \tilde v) = r(U^{-1} (\tilde u), V^{-1} (\tilde v))$. 

\subsection{Canonical coordinates on the tangent bundle}\label{sec:angular-momentum}

Given local coordinates $(\vartheta^1,\vartheta^2)$ on a (proper open subset of) $S^2$, the quadruple $(u,v,\vartheta^1,\vartheta^2)$ defines a local coordinate system on the spherically symmetric spacetime $(\mathcal M,g)$. Given $p\in T_x\mathcal M$, we may expand 
\[p=p^u\partial_u|_x+p^v\partial_v|_x+p^1\partial_{\vartheta^1}|_x+p^2\partial_{\vartheta^2}|_x.\]
The octuple $(u,v,\vartheta^1,\vartheta^2,p^u,p^v,p^1,p^2)$ defines a local coordinate system on $T\mathcal M$, and is called a \emph{canonical coordinate system} on $T\mathcal M$ dual to $(u,v,\vartheta^1,\vartheta^2)$. One is to think of $p$ as the ``momentum coordinate'' and $x$ as the ``position coordinate.'' The tangent bundle of $\mathcal Q$ trivializes globally as $\mathcal Q\times\Bbb R^2$, with coordinates $p^u$ and $p^v$ on the second factor. We let
\begin{equation*}
    \pi:T\mathcal Q\to \mathcal Q
\end{equation*}
denote the canonical projection.

On a spherically symmetric spacetime, we define the \emph{angular momentum} function by 
\begin{align*}
    \ell:T\mathcal M&\to [0,\infty)\\
    (x,p)&\mapsto \sqrt{r^2\slashed g{}_{AB}p^Ap^B},
\end{align*}
where summation over $A,B\in\{1,2\}$ is implied. This function is independent of the angular coordinate system chosen and is itself spherically symmetric as a function on $T\mathcal M$. 

Given a double null gauge $(u,v)$, it will be convenient to define a ``coordinate time'' function 
\begin{equation*}
    \tau\doteq\tfrac 12(v+u).
\end{equation*}
Associated to this time function is the $\tau$-momentum coordinate 
\begin{equation*}
    p^\tau\doteq\tfrac 12 (p^v+p^u).
\end{equation*}

\subsection{The Einstein equations and helpful identities in double null gauge}\label{sec:formulas}

A tensor field on a spherically symmetric spacetime is said to be spherically symmetric if it is itself invariant under the $\mathrm{SO}(3)$-action of the spacetime. If $(\mathcal M,g)$ is a spherically symmetric spacetime satisfying the Einstein equations \eqref{eq:EFE}, then the energy-momentum tensor $\mathbf T$ is a spherically symmetric, symmetric $(0,2)$-tensor field. We may decompose
\begin{equation*}
    \T= \T_{uu}\,du^2+\T_{uv}(du\otimes dv+dv\otimes du)+\T_{vv}\,dv^2+\mathbf S\slashed g,
\end{equation*}
where
\begin{equation*}
       \mathbf S \doteq \frac 12 \slashed g{}^{AB}\T_{AB} = \frac 12 \tr_g \T + \frac{2}{\Omega^2}\T_{uv}.
\end{equation*}
It will be convenient to work with the contravariant energy momentum tensor, which takes the form
\begin{equation*}
    \T^{\sharp\sharp}= \T^{uu}\partial_u\otimes \partial_u + \T^{uv}(\partial_u\otimes\partial_v+\partial_v\otimes\partial_u)+ \T^{vv}\partial_v\otimes \partial_v + \mathbf S\slashed g^{-1},
\end{equation*}
where
\begin{equation*}
    \T_{uu}=\tfrac 14\Omega^4 \T^{vv},\quad \T_{uv}=\tfrac 14\Omega^4 \T^{uv},\quad  \T_{vv}=\tfrac 14 \Omega^4 \T^{uu}.
\end{equation*}

The Christoffel symbols involving null coordinates are given by
\begin{align*}
    \Gamma^u_{uu}&= \partial_ u\!\log\Omega^2,  & \Gamma^v_{vv}&= \partial_ v\!\log\Omega^2, \\
     \Gamma^u_{AB}&= \frac{2\partial_v r}{\Omega^2 r}\slashed g{}_{AB},& \Gamma^v_{AB}&= \frac{2\partial_u r}{\Omega^2r}\slashed g{}_{AB},\\
     \Gamma^A_{Bu}&= \frac{\partial_u r}{r} \delta^A_B,& \Gamma^A_{Bv}&= \frac{\partial_v r}{r} \delta^A_B,
\end{align*}
and the totally spatial Christoffel symbols $\Gamma^A_{BC}$ are the same as for $\gamma$ in the coordinates $(\vartheta^1,\vartheta^2)$.

For a spherically symmetric metric $g$ written in double null gauge \eqref{eq:dn}, the Einstein equations \eqref{eq:EFE} separate into the wave equations for the geometry, 
\begin{align}
 \label{eq:EE-wave-r}   \partial_u\partial_v r&= -\frac{\Omega^2}{2r^2}m+\tfrac 14r\Omega^4\T^{uv},\\
  \label{eq:EE-wave-Omega}  \partial_u \partial_v\!\log \Omega^2&=\frac{\Omega^2m}{r^3}-\tfrac 12\Omega^4 \T^{uv}-\Omega^2\mathbf S,
\end{align}
and Raychaudhuri's equations 
\begin{align}
\label{eq:EE-Ray-u}\partial_u\left(\frac{\partial_u r}{\Omega^2}\right) &= -\tfrac 14 r\Omega^2  \T^{vv},\\
\label{eq:EE-Ray-v}\partial_v\left(\frac{\partial_v r}{\Omega^2}\right) &= -\tfrac 14 r\Omega^2 \T^{uu}.
\end{align}
The Hawking mass \eqref{eq:Hawking-mass} satisfies the equations 
\begin{align}
   \label{eq:mass-1-1} \partial_u m&=\tfrac 12 r^2\Omega^2( \T^{uv}\partial_u r -\T^{vv}\partial_vr),\\
 \label{eq:mass-2-1}   \partial_v m&=\tfrac 12 r^2\Omega^2(\T^{uv}\partial_v r - \T^{uu}\partial_ur).
\end{align}

If $X$ is a spherically symmetric vector field, then 
\begin{equation*}
    X=X^u\partial_u+X^v\partial_v
\end{equation*}
and $X$ satisfies $\Div_g X=0$ if and only if
\begin{equation}
  \label{eq:Div-SS}  \partial_u(r^2\Omega^2X^u)+\partial_v(r^2\Omega^2 X^v)=0.
\end{equation}
The contracted Bianchi identity,
\begin{equation*}
    \Div_g \T=0,
\end{equation*}
which follows from the Einstein equations \eqref{eq:EFE}, implies the following pair of identities:
\begin{align}
 \label{eq:Bianchi-general-1}\partial_u(r^2\Omega^4 \T^{uu})+\partial_v(r^2\Omega^4 \T^{uv})   &= \partial_v{\log\Omega^2}\,r^2\Omega^4 \T^{uv}-4r\partial_vr \Omega^2\mathbf S,\\
\label{eq:Bianchi-general-2}\partial_v(r^2\Omega^4 \T^{vv})+\partial_u(r^2\Omega^4 \T^{uv})   &= \partial_u{\log\Omega^2}\,r^2\Omega^4 \T^{uv}-4r\partial_ur \Omega^2\mathbf S.
\end{align}
If $\alpha$ is a spherically symmetric two-form which annihilates $TS^2$, then it may be written as
\[\alpha=-\frac{\Omega^2 f}{2r^2}\,du\wedge dv,\]
where $f:\mathcal Q\to\Bbb R$ is a smooth function. We then have 
\begin{equation}\label{eq:2-form-div}
    \nabla^\mu \alpha_{u\mu}=\frac{\partial_u f}{r^2},\quad \nabla^\mu \alpha_{v\mu}=-\frac{\partial_v f}{r^2}.
\end{equation}

\subsection{Spherically symmetric electromagnetic fields}\label{def:charged-spacetime}

We will additionally assume that our spherically symmetric spacetime $(\mathcal M,g)$ carries a spherically symmetric electromagnetic field with no magnetic charge. The electromagnetic field is represented by a closed two-form $F$, which takes the Coulomb form 
\begin{equation}
    F =-\frac{\Omega^2Q}{2r^2}\, du \wedge dv, \label{eq:Q-defn}
\end{equation} for a function $Q:\mathcal Q\to \Bbb R$. The number $Q(u,v)$ is the total electric charge enclosed in the 
$(u,v)$-symmetry sphere $S_{u,v}\subset\mathcal M$, which can be seen from the gauge-invariant formula
\begin{equation}
    Q(u,v)=\frac{1}{4\pi}\int_{S_{u,v}}\star F,\label{eq:Q-hodge-star}
\end{equation}
where $\star$ is the Hodge star operator and we orient $\mathcal M$ by $du\wedge dv\wedge d\vartheta\wedge d\varphi$.

The electromagnetic energy momentum tensor is defined by 
\begin{equation}T^\mathrm{EM}_{\mu\nu}\doteq  F_\mu{}^\alpha F_{\nu\alpha}-\tfrac 14 g_{\mu\nu }F_{\alpha\beta} F^{\alpha\beta}\label{eq:EMEM},\end{equation}
and relative to a double null gauge is given by 
\begin{equation*}
    T^\mathrm{EM}=\frac{\Omega^2 Q^2}{4r^4}(du\otimes dv+dv\otimes du) + \frac{Q^2}{2r^4}\slashed g
\end{equation*} in spherical symmetry. 
If $F$ satisfies Maxwell's equation 
\begin{equation*}
    \nabla^\alpha F_{\mu\alpha}=J_\mu
\end{equation*}
for a charge current $J$, then the divergence of the electromagnetic energy momentum tensor satisfies 
\begin{equation}
    \nabla^\mu T_{\mu\nu}^\mathrm{EM}= - F_{\nu\alpha}J^\alpha.\label{eq:EMEM-div}
\end{equation}
In spherical symmetry, Maxwell's equations read (see \eqref{eq:2-form-div})
\begin{equation*}
    \partial_u Q=- \tfrac 12r^2\Omega^2 J^v,\quad \partial_v Q=\tfrac 12r^2\Omega^2J^u.
\end{equation*}

Finally, we will utilize the \emph{renormalized Hawking mass}
\begin{equation}
    \varpi\doteq m+\frac{Q^2}{2r}\label{eq:varpi-definition}
\end{equation}
to account for the contribution of the electromagnetic field to the Hawking mass $m$. 

\subsection{The Lorentz force} 

We next briefly recall the \emph{Lorentz force law} for the motion of a charged particle. Let $(\mathcal M,g,F)$ be a charged spacetime, where $F$ is a closed 2-form representing the electromagnetic field.  If $\gamma$ is the worldline of a particle of mass $\mathfrak m>0$ and charge $\mathfrak e\in\Bbb R$, then it satisfies the Lorentz force equation
\begin{equation*}
    \mathfrak m \frac{Du^\mu}{d\tau} =\mathfrak e F^\mu{}_\nu u^\nu,
\end{equation*}
where $\tau$ is proper time and $u\doteq d\gamma/d\tau$ (so that $g(u,u)=-1$). Defining the momentum of $\gamma$ by $p\doteq \mathfrak m u$ and rescaling proper time to $s=\mathfrak m^{-1}\tau$ (so that $p=d\gamma/ds$), we can rewrite the Lorentz force equation as
\begin{equation}
    \frac{Dp^\mu}{ds}=\mathfrak eF^\mu{}_\nu p^\nu.\label{eq:Lor-intro}
\end{equation}
This equation, which we call the \emph{electromagnetic geodesic equation}, makes sense for null curves as well, and can be taken as the equation of motion for all charged particles \cite{Ori91}.

 \begin{rk}
     Because the Lorentz force equation \eqref{eq:Lor-intro} is not quadratic in $p$, $s$ is \emph{not} an affine parameter, which has fundamental repercussions for the dynamics of the electromagnetic geodesic flow. In particular, trajectories of the Lorentz force with parallel, but not equal, initial velocities will in general be distinct, even up to reparametrization.
 \end{rk}

The electromagnetic geodesic flow has two fundamental conserved quantities which will feature prominently in this work.

\begin{lem}\label{lem:mass-constant} Let $\gamma:I\to \mathcal M$ be a causal electromagnetic geodesic in a spherically symmetric charged spacetime, where $I\subset\Bbb R$ is an interval. Let $p=d\gamma/ds$. Then the rest mass
  \begin{equation*}
      \mathfrak m[\gamma]\doteq \sqrt{-g_\gamma(p,p)}
  \end{equation*}
  and the angular momentum
  \begin{equation*}
      \ell[\gamma]\doteq \ell(\gamma,p)
  \end{equation*}
  are conserved quantities along $\gamma$.
\end{lem}
\begin{proof}
   To show that $\mathfrak m$ is constant, we compute
    \begin{equation*}
      \frac{d}{ds} g(p,p)= 2g\left(\frac{Dp}{ds},p\right)= 2\mathfrak e F(p,p)=0,
    \end{equation*}
    where the final equality follows from the antisymmetry of $F$.

     Since $\ell$ is independent of the coordinates chosen on $S^2$, we may assume that $(\vartheta^1,\vartheta^2)$ are normal coordinates on $S^2$ at the point $(\gamma^1(s_0),\gamma^2(s_0))$. To show that $\ell$ is constant, we then compute 
    \begin{equation*}
        \left.\frac{d}{ds}r^4 \gamma_{AB}p^Ap^B\right|_{s=s_0} = 4r^3 \frac{dr}{ds} \gamma_{AB}\dot\gamma^A\dot\gamma^B - 2r^4 \gamma_{AB}\left(-2\Gamma^A_{Cu}\dot\gamma^u-2\Gamma^A_{Cv}\dot\gamma^v\right)\dot\gamma^B\dot\gamma^C=0,
    \end{equation*}
    where we used the formulas for the Christoffel symbols in spherical symmetry from \cref{sec:formulas}.
\end{proof}

For an electromagnetic geodesic $\gamma(s)$ with angular momentum $\ell=\ell[\gamma]$ and mass $\mathfrak m=\mathfrak m[\gamma]$, the Lorentz force law can be written as
\begin{align}
  \label{eq:SS-Lor-1}  \frac{d}{ds}p^u&=-\partial_u\!\log\Omega^2 (p^u)^2 - \frac{2\partial_vr}{r\Omega^2 }  \frac{\ell^2}{r^2} - \mathfrak e  \frac{Q}{r^2} p^u,\\
  \label{eq:SS-Lor-2}     \frac{d}{ds}p^v&=-\partial_v\!\log\Omega^2(p^v)^2 - \frac{2\partial_ur}{r \Omega^2 }  \frac{\ell^2}{r^2} + \mathfrak e  \frac{Q}{r^2}p^v,\\
     \label{eq:SS-Lor-3}   \Omega^2 p^up^v&= \frac{\ell^2}{r^2}+\mathfrak m^2,
\end{align} where $p^u\doteq d\gamma^u/ds$, $p^v\doteq d\gamma^v/ds$, and the third equation, known as the \emph{mass shell relation}, is directly equivalent to the definition of mass and angular momentum. In this work, it will not be necessary to follow the angular position of the electromagnetic geodesics in the $(3+1)$-dimensional spacetime. It is very convenient to rewrite \eqref{eq:SS-Lor-1} and \eqref{eq:SS-Lor-2} as
\begin{align}
   \label{eq:SS-Lor-4}  \frac{d}{ds} (\Omega^2 p^u) &=  \left(\partial_v\!\log\Omega^2 - \frac{2\partial_vr}{r}  \right)\frac{\ell^2}{r^2} - \mathfrak e  \frac{Q}{r^2} (\Omega^2 p^u), \\
  \label{eq:SS-Lor-5}   \frac{d}{ds} (\Omega^2 p^v) &=  \left(\partial_u\!\log\Omega^2 - \frac{2\partial_ur}{r}  \right)\frac{\ell^2}{r^2} + \mathfrak e  \frac{Q}{r^2} (\Omega^2 p^v).
\end{align}

\section{The Einstein--Maxwell--Vlasov system}
\label{sec:EMV-system}

In this section, we introduce the main matter model considered in this paper, the \emph{Einstein--Maxwell--Vlasov} system. This model describes an ensemble of collisionless self-gravitating charged particles which are either massive or massless. We begin by defining the general system outside of symmetry in \cref{app:A} and then specialize to the spherically symmetric case in \cref{sec:EMV-SS}. In \cref{sec:local-WP,sec:ext}, we formulate the fundamental local theory for this model, local existence and a robust continuation criterion known as the \emph{generalized extension principle}. Finally, in \cref{sec:general-data}, we define a procedure for solving the constraint equations for the spherically symmetric Einstein--Maxwell--Vlasov system.

\subsection{The general system}\label{app:A}

\subsubsection{The volume form on the mass shell}\label{sec:kinematics}

Let $(\mathcal M^4,g)$ be a spacetime. For $x\in \mathcal M$ and $\mathfrak m\ge 0$, we define the (future) \emph{mass shell} at $x$ by
\begin{equation*}
    P_x^\mathfrak m\doteq \{p\in T_x\mathcal M:p\text{ is future directed and } g_x(p,p)=-\mathfrak m^2, p\ne 0\}
\end{equation*}
and the associated smooth fiber bundle
\begin{equation*}
    P^\mathfrak m\doteq\bigcup_{x\in \mathcal M}P_x^\mathfrak m,
\end{equation*}
with projection maps $\pi_\mathfrak m:P^\mathfrak m\to \mathcal M$. 

Fix $x\in \mathcal M$ and let $(x^\mu)$ be normal coordinates at $x$ so that 
\[g_x=-(dx^0)^2+(dx^1)^2+(dx^2)^2+(dx^3)^2.\]
Let $p^\mu$ be dual coordinates to $x^\mu$ on $T\mathcal M$, defined by $p^\mu=dx^\mu(p)$ for $p\in T_x\mathcal M$. For $\mathfrak m\ge 0$, the restrictions of $p^1,p^2,p^3$ to $P^\mathfrak m_x$, denoted by $\bar p^1,\bar p^2,\bar p^3$, define a global coordinate system on $P^\mathfrak m_x$, with $p^0$ determined by
\begin{equation}
    p^0=\sqrt{\mathfrak m^2 +|\bar p^1|^2+|\bar p^2|^2 + |\bar p^3|^2}.\label{eq:p0}
\end{equation}
\begin{defn}
  Let $\mathfrak m\ge 0$ and $x\in \mathcal M$.  The \emph{canonical volume form} $\mu_x^\mathfrak m\in \Omega^3(P_x^\mathfrak m)$ is defined by 
  \begin{equation*}
      \mu_x^\mathfrak m = (p^0)^{-1}\,d\bar p^1\wedge d\bar p^2\wedge d\bar p^3,
  \end{equation*} in normal coordinates at $x$, where $p^0$ is given by \eqref{eq:p0}.
\end{defn}
One can show that this form is independent of the choice of normal coordinates. When $\mathfrak m>0$, $P^\mathfrak m_x$ is a spacelike hypersurface in $T_x\mathcal M$ if it is equipped with the Lorentzian metric $g_x$. In this case, $\mu^\mathfrak m_x=\mathfrak m^{-1}\omega^\mathfrak m_x$, where $\omega^\mathfrak m_x$ is the induced Riemannian volume form on $P^\mathfrak m_x$. The factor of $\mathfrak m^{-1}$ is needed to account for the degeneration of $\omega^\mathfrak m_x$ as $\mathfrak m\to 0$, since $P^0_x$ is a null hypersurface. For more information about the volume form on the mass shell, see \cite{Sachs-Wu,Ringstrom-topology,sarbach2014tangent}. The canonical volume form is uniquely characterized by the following property, which can be found in \cite[Corollary 5.6.2]{Sachs-Wu}.

\begin{lem}\label{lem:SW}
    The form $\mu_x^\mathfrak m$ defined above does not depend on the choice of local coordinates on $\mathcal M$. Moreover, it is uniquely characterized by the following property. Let $H(p)\doteq\tfrac 12 g_x(p,p)$. If $\alpha$ is a 3-form in $T_x\mathcal M$ defined along $P_x^\mathfrak m$ such that 
    \begin{equation*}
        dH\wedge\alpha =\sqrt{-\det g(x)}\,dp^0\wedge dp^1\wedge dp^2\wedge dp^3,
    \end{equation*}
    then $i_\mathfrak m^*\alpha=\mu_x^\mathfrak m$, where $i_\mathfrak m:P^\mathfrak m_x\to T_x\mathcal M$ denotes the inclusion map.
\end{lem}

We denote the integration measure associated to $\mu_x^\mathfrak m$ by $d\mu_x^\mathfrak m$. A \emph{distribution function} is a nonnegative function $f\in C^\infty(P^\mathfrak m)$ which decays sufficiently quickly on the fibers so that the relevant integrals are well-defined and are smooth functions of $x$.  Given a distribution function $f$ we may now define the \emph{number current} $N$ and the \emph{energy momentum tensor} $T$ of $f$ by
\begin{equation}
    N^\mu(x)\doteq \int_{P_x^\mathfrak m}p^\mu f(x,p)\,d\mu_x^\mathfrak m,\quad 
    T^{\mu\nu}(x)\doteq \int_{P_x^\mathfrak m}p^\mu p^\nu f(x,p)\,d\mu_x^\mathfrak m.\label{eq:Vlasov-EM}
\end{equation}
These are readily verified to be tensor fields on $\mathcal M$. Taking divergences, we have \cite[Appendix D]{Ringstrom-topology}
\begin{equation}
     \nabla_\mu N^\mu = \int_{P_x^\mathfrak m} X_0(f)\,d\mu_x^\mathfrak m,\quad
      \nabla_\mu T^{\mu\nu} = \int_{P_x^\mathfrak m} p^\nu X_0(f)\,d\mu_x^\mathfrak m, \label{eq:T-N-div}
\end{equation}
where $X_0\doteq p^\mu\partial_{x^\mu}-\Gamma^\mu_{\nu\rho}p^\nu p^\rho \partial_{p^\mu}\in \Gamma(TT\mathcal M)$ is the \emph{geodesic spray} vector field. 

\subsubsection{The equations}

\begin{defn}\label{defn:general-Vlasov}
The \emph{Einstein--Maxwell--Vlasov} system for particles of \emph{mass} $\mathfrak m\in\Bbb R_{\ge 0}$ and \emph{fundamental charge} $\mathfrak e\in\Bbb R\setminus\{0\}$ consists of a charged spacetime $(\mathcal M,g,F)$ and a distribution function $f:P^\mathfrak m\to [0,\infty)$ satisfying the following equations:
\begin{align}
 \label{eq:EMmV-1}   R_{\mu\nu}-\tfrac 12 Rg_{\mu\nu}&= 2\left(T^\mathrm{EM}_{\mu\nu}+T_{\mu\nu}\right),\\
 \label{eq:EMmV-2}     \nabla^\alpha F_{\mu\alpha}&=\mathfrak e N_\mu,\\
 \label{eq:EMmV-3}     Xf&=0,
\end{align}
where $T^\mathrm{EM}$ is the electromagnetic energy momentum tensor defined in \eqref{eq:EMEM}, $T_{\mu\nu}$ and $N_\mu$ are the Vlasov energy-momentum tensor and number current, respectively, defined in \eqref{eq:Vlasov-EM}, and
\begin{equation}\label{eq:spray-1}
    X\doteq p^\mu \frac{\partial}{\partial x^\mu} - \left(\Gamma^\mu_{\alpha\beta}p^\alpha p^\beta -\mathfrak e F^\mu{}_\alpha p^\alpha\right)\frac{\partial}{\partial p^\mu}\in\Gamma(TT\mathcal M)
\end{equation}
is the \emph{electromagnetic geodesic spray} vector field.
\end{defn}

The vector field $X$ is easily seen to be tangent to the mass shell $P^\mathfrak m$, which means that the Vlasov equation \eqref{eq:EMmV-3} is indeed a transport equation on $P^\mathfrak m$. The vector field $F^\mu{}_\alpha p^\alpha\partial_{p^\mu}$ is itself tangent to $P^\mathfrak m$ and we have the integration by parts formulas
\begin{equation*}
    \int_{P^\mathfrak m_x}F^\mu{}_\alpha p^\alpha\partial_{p^\mu}f\,d\mu^\mathfrak m_x=0,\quad  \int_{P^\mathfrak m_x}F^\mu{}_\alpha p^\nu p^\alpha\partial_{p^\mu}f\,d\mu^\mathfrak m_x=-F^\nu{}_\alpha N^\alpha,
\end{equation*}
which are easily verified in normal coordinates. Combined with \eqref{eq:T-N-div} and the transport equation \eqref{eq:EMmV-3}, we obtain the fundamental conservation laws 
\begin{align}
  \label{eq:N-div}  \nabla_\mu N^\mu&=0,\\
   \label{eq:T-div} \nabla^\mu T_{\mu\nu}&=\mathfrak eN^\alpha F_{\nu\alpha}. 
\end{align}
We now see that the Einstein--Maxwell--Vlasov system is consistent: \eqref{eq:N-div} implies that Maxwell's equation \eqref{eq:EMmV-2} is consistent with antisymmetry of $F$ and \eqref{eq:T-div} implies (using also \eqref{eq:EMEM-div}) the contracted Bianchi identity 
\begin{equation}
    \label{eq:EMV-4-Bianchi}\nabla^\mu\left(T^\mathrm{EM}_{\mu\nu}+T_{\mu\nu}\right)=0.
\end{equation}
for the total energy-momentum tensor of the system. 

\subsubsection{Relation to the relativistic Maxwell--Vlasov system}

The system \eqref{eq:intro-EMV}--\eqref{eq:intro-MV} includes gravity and thus generalizes the special relativistic \emph{Maxwell--Vlasov system} which is typically written in the form\footnote{The subscript K stands for ``kinetic theory literature.''} (cf. \cite{glassey1996cauchy})
\begin{align*}
    &\partial_t f_\mathrm{K}+ \hat v\cdot \partial_x f_\mathrm{K}+\mathfrak e(E+\hat v\times B)\cdot\partial_v f_\mathrm{K}=0,\\
    &\partial_t E-\nabla\times B=-j_\mathrm{K}, \quad
    \partial_tB+\nabla\times E=0,\\
    &\nabla\cdot E =\rho_\mathrm{K}, \quad \nabla\cdot B=0,
\end{align*}
where $f_\mathrm{K}(t,x,v)\ge 0$ is the distribution function, $(t,x,v)\in \Bbb R\times\Bbb R^3\times\Bbb R^3$, $E$ is the electric field, $B$ is the magnetic field, $ \hat v\doteq (\mathfrak m^2+|v|^2)^{-1/2}v$
is the ``relativistic velocity'' and has modulus smaller than unity, and 
\begin{equation*}
    \rho_\mathrm K(t,x)\doteq \mathfrak e\int_{\Bbb R^3}f(t,x,v)\,dv,\quad j_\mathrm K(t,x)\doteq \mathfrak e\int_{\Bbb R^3}\hat v f(t,x,v)\,dv.
\end{equation*}
This system is equivalent to the covariant equations \eqref{eq:EMmV-2} and \eqref{eq:EMmV-3} in Minkowski space under the identifications $(\sqrt{\mathfrak m^2+|v|^2},v)=p$, $f_\mathrm K(t,x,v)=f(t,x,\sqrt{\mathfrak m^2+|v|^2},v)$, $E_i=F_{i0}$, $B_i=\frac 12 \varepsilon_i{}^{jk}F_{jk}$, $\rho_\mathrm K=\mathfrak eN^0$, and $(j_\mathrm K)^i= \mathfrak eN^i$.

\subsection{Spherically symmetric definitions and equations}\label{sec:EMV-SS}

 Let $(\mathcal Q,r,\Omega^2)$ be the $(1+1)$-dimensional reduced spacetime associated to a spherically symmetric spacetime. For $\mathfrak m\ge 0$, we define the \emph{reduced mass shell} by
\begin{equation}
P_\mathrm{red}^\mathfrak m\doteq \left\{(u,v,p^u,p^v)\in T\mathcal Q:\Omega^2(u,v)p^up^v\ge\mathfrak m^2, p^\tau> 0\right\},\label{def:Pred}
\end{equation}
where the second condition in the definition forbids $\Omega^2p^up^v=0$ in the $\mathfrak m=0$ case. The angular momentum function from \cref{sec:angular-momentum} can be defined on the reduced mass shell by
\begin{align*}
   \ell:P^\mathfrak m_\mathrm{red} &\to \Bbb R_{\ge 0} \\
   (u,v,p^u,p^v) & \mapsto r\sqrt{\Omega^2 p^up^v-\mathfrak m^2}.
\end{align*}
Note that $\ell>0$ on $P^0_\mathrm{red}$. The definition of $\ell$ can be rewritten as the fundamental \emph{mass shell relation}
\begin{equation}
    \Omega^2p^up^v=\frac{\ell^2}{r^2}+\mathfrak m^2.\label{eq:mass-shell}
\end{equation}

\begin{defn}\label{def:f}
    A \emph{spherically symmetric distribution function} of \emph{massive} (if $\mathfrak m>0$) or \emph{massless} (if $\mathfrak m=0$) \emph{particles} is a smooth function \begin{equation*}f:P_\mathrm{red}^\mathfrak m\to \Bbb R_{\ge 0}.\end{equation*}   
     We say that $f$ has \emph{locally compact support in $p$} if for every compact set $K\subset\mathcal Q$ there exists a compact set $K'\subset \Bbb R^2$ such that $\spt(f)\cap P^\mathfrak m_\mathrm{red}|_K\subset K\times K'$. We say that $f$ has \emph{locally positive angular momentum} if for every $K\subset\mathcal Q$ compact there exists a constant $c_K>0$ such that $\ell\ge c_K$ on $\spt(f)\cap P^\mathfrak m_\mathrm{red}|_K$.
\end{defn}

In order to define appropriate moments of a distribution function $f$ on a spherically symmetric spacetime $(\mathcal Q,r,\Omega^2)$, we require that $f$ decays in the momentum variables $p^u$ and $p^v$. For $\sigma>0$, $k\ge 0$ an integer, and $K\subset\mathcal Q$ compact, we define the norm
\begin{equation}
    \|f\|_{C^k_\sigma(P^\mathfrak m_\mathrm{red}|_K)}\doteq \sum_{0\le i_1+i_2\le k}\sup_{P^\mathrm m_\mathrm{red}|_K} {\langle p^\tau\rangle^{\sigma+i_2}} |\partial^{i_1}_x\partial^{i_2}_p f|,\label{eq:f-norm-defn}
\end{equation}
where $\partial^{i_1}_x\partial^{i_2}_p f$ ranges over all expressions involving $i_1$ derivatives in the $(u,v)$-variables and $i_2$ derivatives in the $(p^u,p^v)$-variables. If the norm \eqref{eq:f-norm-defn} is finite for all compact sets $K\subset\mathcal Q$, we say that $f\in C^\infty_{\sigma,\loc}(P^\mathfrak m_\mathrm{red})$. If $f$ has locally compact support in $p$, then it lies in  $C^\infty_{\sigma,\loc}(P^\mathfrak m_\mathrm{red})$.

\begin{rk}
    Our well posedness theory for the Einstein--Maxwell--Vlasov system requires that $p$-derivatives of $f$ decay faster, which is the reason for the $i_2$ weight in \eqref{eq:f-norm-defn}.
\end{rk}

Given a spherically symmetric spacetime $(\mathcal Q,r,\Omega^2)$ with distribution function $f$, we define the \emph{Vlasov number current} by
\begin{align}
  \label{eq:Nu} N^u(u,v) &\doteq \pi\Omega^2\int_{\Omega^2p^up^v\ge \mathfrak m^2} p^uf(u,v,p^u,p^v)\,dp^udp^v,\\
 \label{eq:Nv}  N^v(u,v) &\doteq\pi\Omega^2\int_{\Omega^2p^up^v\ge \mathfrak m^2} p^vf(u,v,p^u,p^v)\,dp^udp^v
\end{align}
and the \emph{Vlasov energy momentum tensor} by
\begin{align}
\label{eq:Tuu}    T^{uu}(u,v)&\doteq\pi\Omega^2\int_{\Omega^2p^up^v\ge \mathfrak m^2} (p^u)^2f(u,v,p^u,p^v)\,dp^udp^v,\\
\label{eq:Tuv}   T^{uv} (u,v)&\doteq\pi\Omega^2\int_{\Omega^2p^up^v\ge \mathfrak m^2}p^up^vf(u,v,p^u,p^v)\,dp^udp^v,\\
\label{eq:Tvv}  T^{vv}  (u,v)&\doteq\pi\Omega^2\int_{\Omega^2p^up^v\ge \mathfrak m^2}(p^v)^2f(u,v,p^u,p^v)\,dp^udp^v,\\
\label{eq:S}  S(u,v)&\doteq\frac{\pi}{2}\Omega^2\int_{\Omega^2p^up^v\ge \mathfrak m^2}(\Omega^2p^up^v-\mathfrak m^2) f(u,v,p^u,p^v)\,dp^udp^v\le \frac{\Omega^2}{2}T^{uv}(u,v).
\end{align}
If $f\in C^\infty_{\sigma,\loc}(P^\mathfrak m_\mathrm{red})$ with $\sigma>4$, then these moments are well defined smooth functions of $u$ and $v$.

\begin{defn}\label{def:sph-sym-EMV}
The \emph{spherically symmetric Einstein--Maxwell--Vlasov} model for \emph{particles of mass} $\mathfrak m\in\Bbb R_{\ge0}$ and \emph{fundamental charge} $\mathfrak e\in\Bbb R\setminus\{0\}$ consists of a smooth spherically symmetric charged spacetime $(\mathcal Q,r,\Omega^2,Q)$ and a smooth distribution function $f\in C^\infty_{\sigma,\loc}(P^\mathfrak m_\mathrm{red})$ for a decay rate $\sigma>4$ fixed. When $\mathfrak m=0$, we require that $f$ also has locally positive angular momentum. To emphasize that the distribution functions we consider in this paper have these regularity properties in $p$, we say that such a solution has \emph{admissible momentum}.

The system satisfies the wave equations
    \begin{align}
 \label{eq:r-wave} \partial_u\partial_v r&= - \frac{\Omega^2}{2r^2}\left(m-\frac{Q^2}{2r}\right) +\tfrac 14 r\Omega^4 T^{uv}, \\
\label{eq:Omega-wave}    \partial_u\partial_v \!\log\Omega^2&=\frac{\Omega^2m}{r^3}-\frac{\Omega^2Q^2}{r^4}-\tfrac 12\Omega^4T^{uv}-\Omega^2 S,
\end{align}
the Raychaudhuri equations 
\begin{align}
\label{eq:Ray-u}  \partial_u\left(\frac{\partial_ur}{\Omega^2}\right)  &=-\tfrac 14r\Omega^2T^{vv},\\
\label{eq:Ray-v}  \partial_v\left(\frac{\partial_vr}{\Omega^2}\right)  &=-\tfrac 14r\Omega^2T^{uu},
\end{align}
and the Maxwell equations
\begin{align}
\label{eq:Max-u}\partial_u Q&=-\tfrac 12\mathfrak e r^2\Omega^2N^v,\\
\label{eq:Max-v}\partial_v Q&=+\tfrac 12\mathfrak e r^2\Omega^2 N^u,
\end{align}
where $N^u,N^v,T^{uu},T^{uv}$, $T^{vv}$, and $S$ are defined by equations \eqref{eq:Nu}--\eqref{eq:S}. Finally, $f$ satisfies the \emph{spherically symmetric Vlasov equation}
\begin{equation}
    Xf=0 ,\label{eq:SS-MV}
\end{equation}
where $X\in \Gamma (TP^\mathfrak m_\mathrm{red})$ is the \emph{reduced electromagnetic geodesic spray}
\begin{multline}\label{eq:spray-2}
     X\doteq  p^u\partial_u +p^v\partial_v-\left(\partial_u{\log\Omega^2}(p^u)^2+\frac{2\partial_vr}{r\Omega^2}(\Omega^2 p^up^v-\mathfrak m^2)+\mathfrak e\frac{Q}{r^2}p^u\right)\partial_{p^u}\\-\left(\partial_v{\log\Omega^2}(p^v)^2+\frac{2\partial_ur}{r\Omega^2}(\Omega^2 p^up^v-\mathfrak m^2)-\mathfrak e\frac{Q}{r^2}p^v\right)\partial_{p^v}.
\end{multline}
\end{defn}

\begin{rk}
Since $f\ge 0$ for a solution of the Einstein--Maxwell--Vlasov system, the components $N^u$ and $N^v$ of the number current are nonnegative. It follows from the Maxwell equations \eqref{eq:Max-u} and \eqref{eq:Max-v} that $\mathfrak eQ$ is decreasing in $u$ and increasing in $v$, unconditionally. This monotonicity property is a fundamental feature of the spherically symmetric Einstein--Maxwell--Vlasov system and will be exploited several times in this paper.
\end{rk}

\begin{rk}\label{rk:DEC}
     Both the electromagnetic energy-momentum tensor $T^\mathrm{EM}$ and the Vlasov energy-momentum tensor $T$ of the Einstein--Maxwell--Vlasov system satisfy the dominant energy condition.
\end{rk}

\begin{rk}
    As an abuse of notation, we have denoted the spray \eqref{eq:spray-1} on $T\mathcal M$ and the spray \eqref{eq:spray-2} on $T\mathcal Q$ by the same letter $X$. It will always be clear from the context which vector field we are referring to. They are related by the pushforward along the natural projection map $P^\mathfrak m\to P^\mathfrak m_\mathrm{red}$.
\end{rk}

For a solution of the Einstein--Maxwell--Vlasov system, the Hawking mass $m$ satisfies the equations
\begin{align}
  \label{eq:Hawking-u}  \partial_u m&=\tfrac 12 r^2\Omega^2( T^{uv}\partial_u r - T^{vv}\partial_vr)+\frac{Q^2}{2r^2}\partial_u r,\\
   \label{eq:Hawking-v} \partial_v m&=\tfrac 12 r^2\Omega^2( T^{uv}\partial_v r - T^{uu}\partial_ur)+\frac{Q^2}{2r^2}\partial_v r,
\end{align} which can be derived from \eqref{eq:mass-1-1} and \eqref{eq:mass-2-1}. The modified Hawking mass $\varpi$ can then be seen to satisfy 
\begin{align}
  \label{eq:mod-Hawking-u}  \partial_u\varpi&=\tfrac 12 r^2\Omega^2( T^{uv}\partial_u r - T^{vv}\partial_vr)-\tfrac 12\mathfrak er\Omega^2QN^v,\\
 \label{eq:mod-Hawking-v}  \partial_v\varpi &=\tfrac 12 r^2\Omega^2( T^{uv}\partial_v r - T^{uu}\partial_ur)+\tfrac 12\mathfrak er\Omega^2 QN^u.
\end{align}
The particle current $N$ satisfies the conservation law
\begin{equation} \partial_u(r^2\Omega^2N^u)+\partial_v(r^2\Omega^2N^v)=0\label{eq:SS-N-div}
\end{equation}
by \eqref{eq:Div-SS} and \eqref{eq:N-div}. Alternatively, it can be directly derived from the spherically symmetric Vlasov equation, which we will do the proof of \cref{prop:char-IVP-Vlasov} in \cref{app:B}.
Finally, for the Einstein--Maxwell--Vlasov system, the Bianchi identities \eqref{eq:Bianchi-general-1} and \eqref{eq:Bianchi-general-2} read
\begin{align}
 \label{eq:Bianchi-1}  \partial_u(r^2\Omega^4 T^{uu})+\partial_v(r^2\Omega^4 T^{uv})   &= \partial_v{\log\Omega^2}\,r^2\Omega^4 T^{uv}-4r\partial_vr \Omega^2 S-\mathfrak e\Omega^4 QN^u,\\
\label{eq:Bianchi-2}\partial_v(r^2\Omega^4 T^{vv})+\partial_u(r^2\Omega^4 T^{uv})   &= \partial_u{\log\Omega^2}\,r^2\Omega^4 T^{uv}-4r\partial_ur \Omega^2S+\mathfrak e\Omega^4 QN^v. 
\end{align}
Again, this follows either from \eqref{eq:EMV-4-Bianchi} or directly from the spherically symmetric equations.

\subsubsection{The spherically symmetric reduction}

\begin{prop}\label{prop:SS-equiv}
    Let $(\mathcal Q,r,\Omega^2,Q,f_\mathrm{sph})$ be a solution of the spherically symmetric Einstein--Maxwell--Vlasov system as defined by \cref{def:sph-sym-EMV}. Then $(\mathcal M,g,F,f)$ solves the Einstein--Maxwell--Vlasov system if we lift the solution according to
    \begin{align}
\nonumber        \mathcal M&\doteq \mathcal Q\times S^2,\\
        g&\doteq -\frac{\Omega^2}{2}(du\otimes dv+dv\otimes du)+r^2\gamma,\\
        F&\doteq-\frac{Q}{2r^2}du\wedge dv,\\
        f(u,v,\vartheta^1,\vartheta^2,p^u,p^v,p^1,p^2)&\doteq f_\mathrm{sph}(u,v,p^u,p^v),\label{eq:f-app-B}
    \end{align}
    where $(\vartheta^1,\vartheta^2)$ is a local coordinate system on $S^2$.
\end{prop}
Note that $f$ in \eqref{eq:f-app-B} is $\mathrm{SO}(3)$-invariant as a function on $T\mathcal M$ as defined in \cref{sec:double-null-gauge}.
\begin{proof}
    As the equations \eqref{eq:r-wave}--\eqref{eq:Max-v} are equivalent to the Einstein equations and Maxwell equations, we must only check that $f$, defined by \eqref{eq:f-app-B}, satisfies $Xf=0$, where $X$ is given by \eqref{eq:spray-1}, and that the spherically symmetric formulas \eqref{eq:Nu}--\eqref{eq:S} appropriately reconstruct the $(3+1)$-dimensional number current and energy-momentum tensor. 

    Let $\gamma(s)$ be an electromagnetic geodesic. Then $Xf=0$ at $(\gamma(s_0),\dot\gamma(s_0))\in P^\mathfrak m$ is equivalent to 
    \begin{equation}
       \left. \frac{d}{ds}\right|_{s=s_0}f(\gamma(s),\dot\gamma(s))=0.\label{eq:app-B-1}
    \end{equation}
    Using the chain rule
    \begin{align}\nonumber
       \left. \frac{d}{ds}\right|_{s=s_0}f_\mathrm{sph}(\gamma^u(s),\gamma^v(s),p^u(s),p^v(s))&= p^u\partial_u f_\mathrm{sph}+p^v\partial_vf_\mathrm{sph}+\frac{dp^u}{ds}\partial_{p^u}f_\mathrm{sph}+\frac{dp^v}{ds}\partial_{p^v}f_\mathrm{sph},
    \end{align}
    the spherically symmetric Lorentz force equations \eqref{eq:SS-Lor-1} and \eqref{eq:SS-Lor-2}, the mass shell relation \eqref{eq:SS-Lor-3}, and the spherically symmetric Vlasov equation \eqref{eq:SS-MV}, we see that \eqref{eq:app-B-1} holds. 

     To compute the energy-momentum tensor in spherical symmetry, we use \cref{lem:SW}. Let $x\in \mathcal M$ and take $(\vartheta^1,\vartheta^2)$ to be local coordinates on $S^2$ which are normal at the spherical component of $x$, so that $\gamma_{AB}=\delta_{AB}$. We have 
    \begin{equation*}
       dH=-\frac{\Omega^2}{2}p^vdp^u-\frac{\Omega^2}{2}p^udp^v+r^2p^1dp^1+r^2p^2dp^2.
    \end{equation*}
Therefore, if we define
\begin{equation*}
    \alpha\doteq - r^2(p^v)^{-1}dp^v\wedge dp^1\wedge dp^2,  
\end{equation*} then
\begin{equation*}
    dH\wedge\alpha = \tfrac 12\Omega^2r^2\, dp^u\wedge dp^v\wedge dp^1\wedge dp^2.
\end{equation*}
Therefore, by \cref{lem:SW}, 
\begin{equation*}
    d\mu^\mathfrak m_x=r^2(p^v)^{-1}\,dp^vdp^1dp^2
\end{equation*}
as measures on $(0,\infty)\times \Bbb R^2_{(p^1,p^2)}$. The remaining momentum variable $p^u$ is obtained from $p^v$, $p^1$, and $p^2$ via the mass shell relation \eqref{eq:mass-shell}. If we set $\tan\beta=p^2/p^1$ and use that $\ell^2=r^4((p^1)^2+(p^2)^2)$, then 
\begin{equation*}
    d\mu^\mathfrak m_x= r^{-2}(p^v)^{-1}dp^v \,\ell\,d\ell \,d\beta.
\end{equation*}
For any weight function $w=w(p^u,p^v)$, we therefore have 
\begin{equation}
    \int_{(0,\infty)\times \Bbb R^2} w f\,d\mu_x^\mathfrak m =r^{-2}\int_0^{2\pi} \int_0^\infty \int_0^\infty w\left(\frac{\ell^2+\mathfrak m^2r^2}{r^2\Omega^2p^v},p^v\right)f_\mathrm{sph}\left(u,v,\frac{\ell^2+\mathfrak m^2r^2}{r^2\Omega^2p^v},p^v\right)\frac{dp^v}{p^v} \,\ell\,d\ell\,d\beta.\label{eq:ingoing-1}
\end{equation}
Integrating out $\beta$ and applying the coordinate transformation $(p^v,\ell)\mapsto (p^u,p^v)$ reproduces the formulas \eqref{eq:Nu}--\eqref{eq:S} for $N$ and $T$.
\end{proof}

\begin{rk}
  Other works on the spherically symmetric Einstein--Vlasov system in double null gauge, such as \cite{DR16,Mosx-Vlasov-WP, Mosx-AdS-Vlasov}, represent the distribution function $f$ differently, opting to (at least implicitly) eliminate either $p^u$ or $p^v$ in terms of $\ell$ using the mass shell relation \eqref{eq:mass-shell}. This leads to different formulas for $N^\mu$ and $T^{\mu\nu}$, as these will then involve an integral over $\ell$, as in \eqref{eq:ingoing-1}. To make this precise, we can define the \emph{outgoing representation}\footnote{Of course, we may also define an \emph{ingoing representation} $f_\nwarrow(u,v,p^u,\ell)$ by interchanging $p^u$ and $p^v$.} of the spherically symmetric $f$ by
\begin{equation}\label{eq:f-o}
    f_{\nearrow}(u,v,p^v,\ell)\doteq f\left(u,v,\frac{\ell^2+r^2\mathfrak m^2}{r^2\Omega^2p^v},p^v\right),
\end{equation}
and \eqref{eq:ingoing-1} implies, for instance,
\begin{equation*}
  T^{uv}=\frac{2\pi}{r^4\Omega^2}\int_0^\infty\int_0^\infty \frac{\ell^2+r^2\mathfrak m^2}{p^v} f_\nearrow(u,v,p^v,\ell)\,dp^v\,\ell\,d\ell.
\end{equation*}
The outgoing representation may be taken as an alternative \emph{definition} of the spherically symmetric Vlasov system. We have chosen the formulation here in terms of $p^u$ and $p^v$ because of its explicit symmetry, which is key for constructing time-symmetric initial data in the proof of \cref{thm:main}. We have also chosen to always write $N^\mu$ and $T^{\mu\nu}$ in contravariant form, so that $N^u$ is associated with $p^u$, etc. This causes extra factors of $g_{uv}=-\frac 12\Omega^2$ to appear in various formulas, compared to  \cite{DR16,Mosx-Vlasov-WP, Mosx-AdS-Vlasov}.
\end{rk}

\subsubsection{Previous work on the Einstein--Maxwell--Vlasov system}

Besides the general local existence result of \cite{Blancel}, the Einstein--Maxwell--Vlasov model does not seem to have been extensively studied. Dispersion for small data solutions of the Einstein--Maxwell--massive Vlasov model (stability of Minkowski space) in spherical symmetry was proved by Noundjeu \cite{EMV-stability}. See also \cite{EMV-local-existence} for local well-posedness in Schwarzschild coordinates and \cite{noundjeu2004existence} for the existence of nontrivial solutions of the constraints. Many static solutions are known to exist for the massive system, first studied numerically by Andr\'easson--Eklund--Rein \cite{andreasson2009numerical} and proved rigorously by Thaller \cite{Thaller-thin-shell} in spherical symmetry. Thaller has also shown that stationary and axisymmetric (but not spherically symmetric) solutions exist \cite{Thaller-rotating}.

\subsection{Local well-posedness in spherical symmetry}
\label{sec:local-WP}

Electromagnetic geodesics, in contrast to ordinary geodesics, can have limit points in $\mathcal M$. By standard ODE theory, this can only happen if $p(s)\to 0$ as $s\to \pm\infty$.\footnote{This does not happen for ordinary geodesics because of the following homogeneity property: If $s\mapsto\gamma(s)$ is a geodesic, then so is $s\mapsto \gamma(as)$ for any $a>0$. See \cite[Lemma 5.8]{Oneill}, where homogeneity is used to identify radial geodesics emanating from the same point with parallel velocity.} On a fixed spherically symmetric background, one can show that an electromagnetic geodesic $\gamma$ cannot have a limit point if either $\mathfrak m[\gamma]>0$ or $\ell[\gamma]>0$. However, even in the massless case, an electromagnetic geodesic with initially positive momentum will still have positive momentum for a short (coordinate) time. Therefore, one can show that local well-posedness in double null gauge holds in the case of massive particles or massless particles with momentum initially supported away from zero. 

\begin{rk}
    Bigorgne has shown that the relativistic Maxwell--massless Vlasov system is classically ill-posed if the initial data are allowed to be supported near zero momentum \cite{Bigorgne}. We expect a similar result to hold for the Einstein--Maxwell--massless Vlasov system. 
\end{rk}

We now state our fundamental local well-posedness result for the spherically symmetric Einstein--Maxwell--Vlasov system. We formulate this in terms of the \emph{characteristic initial value problem}, though the techniques used apply to the Cauchy problem as well. Note that we work in function spaces that allow for noncompact support in the momentum variables, although this is not needed for the applications in this paper (but is useful in the context of cosmic censorship \cite{DR16}). The proof of local existence is deferred to \cref{app:B}.

Given $U_0<U_1$ and $V_0<V_1$, let
\begin{align*}
  \mathcal C(U_0,U_1,V_0,V_1)  &\doteq (\{U_0\}\times[V_0,V_1])\cup([U_0,U_1]\times\{V_0\}), \\
   \mathcal R(U_0,U_1,V_0,V_1) &\doteq [U_0,U_1]\times [V_0,V_1].
\end{align*} We will consistently omit $(U_0,U_1,V_0,V_1)$ from the notation for these sets when the meaning is clear. A function $\phi:\mathcal C\to \Bbb R$ is said to be \emph{smooth} if it is continuous and $\phi|_{\{U_0\}\times [V_0,V_1]}$ and $\phi|_{[U_0,U_1]\times\{V_0\}}$ are $C^\infty$ single-variable functions. This definition extends naturally to functions $f:P^\mathfrak m_\mathrm{red}|_\mathcal C\to \Bbb R_{\ge 0}$.

\begin{defn}
    A smooth \emph{(bifurcate) characteristic initial data set} for the spherically symmetric Einstein--Maxwell--Vlasov system with parameters $\mathfrak m, \mathfrak e$, and $\sigma$ consists of smooth functions $\mathring r,\mathring \Omega^2,\mathring Q:\mathcal C\to \Bbb R$ with $\mathring r$ and $\mathring\Omega^2$ positive, and a smooth function $\mathring f:P^\mathfrak m_\mathrm{red}|_\mathcal C\to\Bbb R_{\ge 0}$, where $P^\mathfrak m_\mathrm{red}|_\mathcal C$ is defined using $\mathring\Omega^2$. Moreover, we assume that the norms 
    \begin{equation*}
  \|\mathring f\|_{C^k_\sigma(P|_\mathcal C)}\doteq \sum_{0\le i_1+i_2\le k}  \left(\sup_{P^\kappa|_{\{U_0\}\times[V_0,V_1]}}\langle p^\tau\rangle^{\sigma+i_2}|\partial_v^{i_1}\partial_{p}^{i_2}\mathring f|+\sup_{P^\kappa|_{[U_0,U_1]\times\{V_0\}}}\langle p^\tau\rangle^{\sigma+i_2} |\partial_u^{i_1}\partial_{p}^{i_2}\mathring f|\right)
\end{equation*}
are finite for every $k\ge 0$. In the case $\mathfrak m=0$, we also assume that $\mathring f$ has locally positive angular momentum. Finally, we assume that Raychaudhuri's equations \eqref{eq:Ray-u} and \eqref{eq:Ray-v}, together with Maxwell's equations \eqref{eq:Max-u} and \eqref{eq:Max-v} are satisfied to all orders in directions tangent to $\mathcal C$. 
\end{defn}

\begin{prop}\label{prop:char-IVP-Vlasov}
For any $\mathfrak m\ge 0$, $\mathfrak e\in\Bbb R$, $\sigma>4$, $B>0$, and $c_\ell>0$, there exists a constant $\ve_\loc>0$ with the following property. Let $(\mathring r,\mathring\Omega^2,\mathring Q,\mathring f)$ be a characteristic initial data set for the spherically symmetric Einstein--Maxwell--Vlasov system on $\mathcal C(U_0,U_1,V_0,V_1)$. If $U_1-U_0<\ve_\loc$, $V_1-V_0<\ve_\loc$, 
\begin{equation*}
    \|{\log\mathring r}\|_{C^2(\mathcal C)}+ \|{\log\mathring \Omega^2}\|_{C^2(\mathcal C)}+\|\mathring Q\|_{C^1(\mathcal C)}+\|\mathring f\|_{C^1_\sigma(P^\mathfrak m_\mathrm{red}|_\mathcal C)}\le B,
\end{equation*}
and in the case $\mathfrak m=0$, $\ell\ge c_\ell$ on $\spt(\mathring f)$, then there exists a unique smooth solution $(r,\Omega^2,Q,f)$ of the spherically symmetric Einstein--Maxwell--Vlasov system on $\mathcal R(U_0,U_1,V_0,V_1)$ which extends the initial data. If $\mathring f$ has locally compact support in $p$, then so does $f$. Moreover, the norms 
    \begin{equation*}
        \|{\log r}\|_{C^k(\mathcal R)}, \|{\log\Omega^2}\|_{C^k(\mathcal R)}, \|Q\|_{C^k(\mathcal R)}, \|f\|_{C^k_\sigma(P^\mathfrak m_\mathrm{red}|_\mathcal R)}
    \end{equation*} are finite for any $k$ and 
    can be bounded in terms of appropriate higher order initial data norms.
\end{prop}

The proof of the proposition is given in \cref{sec:proof-of-LWP}.

\subsection{The generalized extension principle}\label{sec:ext}

Recall that a spherically symmetric Einstein-matter model is said to satisfy the \emph{generalized extension principle} if any ``first singularity'' either emanates from a point on the spacetime boundary with $r=0$, or its causal past has infinite spacetime volume. This property has been shown to hold for the Einstein-massless scalar field system by Christodoulou \cite{TheBVPaper}, for the Einstein--massive Vlasov system by Dafermos and Rendall \cite{DR16}, and for the Einstein--Maxwell--charged Klein--Gordon system by Kommemi \cite{Kommemi13}. In this paper, we extend the generalized extension principle of Dafermos--Rendall to the Einstein--Maxwell--Vlasov system: 

\begin{prop}[The generalized extension principle]\label{prop:ext}
Let $(\mathcal Q,r,\Omega^2,Q,f)$ be a smooth solution of the spherically symmetric Einstein--Maxwell--Vlasov system with admissible momentum, which is defined on an open set $\mathcal Q\subset \Bbb R^2_{u,v}$. If $\mathcal Q$ contains the set $\mathcal R'\doteq \mathcal R(U_0,U_1,V_0,V_1)\setminus\{(U_1,V_1)\}$ and the following two conditions are satisfied:
\begin{enumerate}
    \item $\mathcal R'$ has finite Lorentzian volume, i.e.,   \begin{equation}
     \label{eq:ext-hyp-1} \iint_{\mathcal R'}\Omega^2\,dudv <\infty,
    \end{equation}
    \item the area-radius is bounded above and below, i.e.,  \begin{equation}
     \label{eq:ext-hyp-2}   \sup_{\mathcal R'}|{\log r}|<\infty,
    \end{equation}
\end{enumerate}
then the solution extends smoothly, with admissible momentum, to a neighborhood of $(U_1,V_1)$.
\end{prop}

Therefore, since this system satisfies the dominant energy condition (\cref{rk:DEC}), the Einstein--Maxwell--Vlasov system is \emph{strongly tame} in Kommemi's terminology \cite{Kommemi13}, under the admissible momentum assumption. This is an important ``validation'' of the Einstein--Maxwell--Vlasov model over the charged null dust model and means the model enjoys Kommemi's a priori boundary characterization \cite{Kommemi13}, which will be used in \cref{sec:general-continuity} below. 

  \cref{prop:ext} is also used crucially in the proof of \cref{thm:main} because it provides a continuation criterion at zeroth order. This allows us to avoid commutation when treating the singular ``main beam'' in the construction of bouncing charged Vlasov beams.

We now give the proof of \cref{prop:ext}, assuming the ``fundamental local spacetime estimate'' to be stated and proved in \cref{sec:FLSE} below. The proof of the local estimate is based on a streamlining of the ideas already present in \cite{DR16} together with the monotonicity of charge inherent to the Einstein--Maxwell--Vlasov system and a quantitative lower bound on the ``coordinate time momentum'' $p^u+p^v$ obtained from the mass shell relation. 

\begin{proof}[Proof of \cref{prop:ext}]
    By \cref{lem:fundamental-local-estimate} below, \eqref{eq:ext-hyp-1} and \eqref{eq:ext-hyp-2} imply that 
      \begin{equation*}
        B\doteq \|{\log r}\|_{C^2(\mathcal R')}+ \|{\log \Omega^2}\|_{C^2(\mathcal R')}+\|Q\|_{C^1(\mathcal R')}+\| f\|_{C^1_\sigma(P^\mathfrak m_\mathrm{red}|_{\mathcal R'})}<\infty.
    \end{equation*}
    Let $U_1'>U_1$ and $V_1'>V_1$ be such that the segments $[U_1,U_1']\times\{V_0\}$ and $\{U_0\}\times[V_1,V_1']$ lie inside of $\mathcal Q$ and let $c_\ell>0$ be a lower bound for $\ell$ on $\spt(f)\cap P^\mathfrak m_\mathrm{red}|_{\mathcal C(U_0,U_1',V_0,V_1')}$ if $\mathfrak m=0$. Let $\ve_\loc>0$ be the local existence time for the spherically symmetric Einstein--Maxwell--Vlasov system with parameters $(\mathfrak m,\mathfrak e,\sigma,2B, c_\ell)$  given by \cref{prop:char-IVP-Vlasov}. Fix $(\tilde U,\tilde V)\in \mathcal R'$ with $U_1-\tilde U<\ve_\loc$ and $V_1-\tilde V<\ve_\loc$. 
    
    Observe that if $U_2>U_1$ is sufficiently close to $U_1$ and $V_2>V_1$ is sufficiently close to $V_1$, then 
     \begin{equation*}
        B\doteq \|{\log r}\|_{C^2(\tilde{\mathcal C})}+ \|{\log \Omega^2}\|_{C^2(\tilde{\mathcal C})}+\|Q\|_{C^1(\tilde{\mathcal C})}+\| f\|_{C^1_\sigma(P^\mathfrak m_\mathrm{red}|_{\tilde{\mathcal C}})}\le 2B
    \end{equation*} and $\ell \ge  c_\ell$ on $\spt(f)\cap P^\mathfrak m_\mathrm{red}|_{\tilde{\mathcal C}}$ if $\mathfrak m=0$, where $\tilde{\mathcal C}\doteq \mathcal C(\tilde U,U_2,\tilde V,V_2)$. Indeed, this is clear for $\log r$, $\log\Omega^2$, and $Q$ by smoothness of these functions on $\mathcal Q$. For $f$, we can also easily show this using the mean value theorem and the finiteness of $\|f\|_{C^2_\sigma(P^\mathfrak m_\mathrm{red}|_K)}$ on compact sets $K\subset \mathcal Q$. For $\ell$, this follows immediately from conservation of angular momentum and the domain of dependence property if $U_2\le U_1'$ and $V_2\le V_1'$.

    Therefore, by \cref{prop:char-IVP-Vlasov}, the solution extends smoothly, with admissible momentum, to the rectangle $\mathcal R(\tilde U,U_2,\tilde V,V_2)$, which contains $(U_1,V_1)$. This completes the proof.
\end{proof}

\subsubsection{Horizontal lifts and the commuted Vlasov equation}\label{sec:commuted}

Local well-posedness for $(r,\Omega^2,Q,f)$ takes place in the space $C^2\times C^2\times C^1\times C^1$, since the Christoffel symbols and electromagnetic field need to be Lipschitz regular to obtain a unique classical solution of the Vlasov equation \eqref{eq:SS-MV}. In order to estimate $\partial_u^2\Omega^2$ and $\partial_v^2\Omega^2$, one has to commute the wave equation for $\Omega^2$, \eqref{eq:Omega-wave}, with $\partial_u$ and $\partial_v$. This commuted equation contains terms such as $\partial_u T^{uv}$, which can only be estimated by first estimating $\partial_u f$ and $\partial_vf$. On the other hand, naively commuting the spherically symmetric Vlasov equation, \eqref{eq:SS-MV}, with spatial derivatives introduces highest order nonlinear error terms such as $\partial_u^2\Omega^2\,\partial_{p^u}f$.\footnote{This is clearly not an issue for local well-posedness since the ``time interval'' of the solution can be taken sufficiently small to absorb (the time integral of) this term.} Therefore, it would appear that the system does not close at this level of regularity. 

However, as was observed by Dafermos and Rendall in  \cite{DR-Vlasov-extension} in the case of Einstein--Vlasov (see also the erratum of \cite{RR92}), the \emph{horizontal lifts} 
\begin{align*}
    \hat\partial_u f&\doteq \partial_uf -p^u\partial_u{\log\Omega^2}\partial_{p^u}f,\\
    \hat\partial_v f&\doteq \partial_vf -p^v\partial_v{\log\Omega^2}\partial_{p^v}f
\end{align*}
 of $\partial_uf$ and $\partial_vf$ with respect to the Levi--Civita connection satisfy a better system of equations without these highest order errors. In the case of Einstein--Maxwell--Vlasov, we directly commute \eqref{eq:SS-MV} with $\{\hat\partial_u,\hat\partial_v,\partial_u,\partial_v\}$ to obtain 
 \begin{align}
    X(\hat\partial_uf)&= p^u\partial_u{\log\Omega^2}\hat\partial_u f+\left(p^u\partial_u{\log\Omega^2}\partial_{p^u}\zeta^u-\partial_u\zeta^u-\partial_u{\log\Omega^2}\zeta^u\right)\partial_{p^u}f \nonumber\\
    &\quad+\left(\partial_u\partial_v{\log\Omega^2}(p^u)^2-\partial_u\zeta^v+p^u\partial_u{\log\Omega^2}\partial_{p^u}\zeta^v\right)\partial_{p^v}f,\label{eq:commuted-1}\\
    X(\hat\partial_vf)&= p^v\partial_v{\log\Omega^2}\hat\partial_v f+\left(\partial_u\partial_v{\log\Omega^2}(p^v)^2-\partial_v\zeta^u+p^v\partial_v{\log\Omega^2}\partial_{p^v}\zeta^u\right)\partial_{p^u}f \nonumber\\
    &\quad+\left(p^v\partial_v{\log\Omega^2}\partial_{p^v}\zeta^v-\partial_v\zeta^v-\partial_v{\log\Omega^2}\zeta^v\right)\partial_{p^v}f,\\
     X(\partial_{p^u}f)&=-\hat\partial_uf+(3p^u\partial_u{\log\Omega^2}-\partial_{p^u}\zeta^u)\partial_{p^u}f-\partial_{p^u}\zeta^v\partial_{p^v}f,\\
     X(\partial_{p^v}f)&=-\hat\partial_vf-\partial_{p^v}\zeta^u\partial_{p^u}f+(3p^v\partial_v{\log\Omega^2}-\partial_{p^v}\zeta^v)\partial_{p^v}f,\label{eq:commuted-2}
\end{align}
where
\begin{equation*}
      \zeta^u\doteq\frac{2\partial_vr}{r\Omega^2}(\Omega^2p^up^v-\mathfrak m^2)+\mathfrak e\frac{Q}{r^2}p^u,\quad 
    \zeta^v\doteq\frac{2\partial_ur}{r\Omega^2}(\Omega^2p^up^v-\mathfrak m^2)-\mathfrak e\frac{Q}{r^2}p^v.
\end{equation*}
Upon using the wave equation \eqref{eq:Omega-wave}, we see that the right-hand sides of \eqref{eq:commuted-1}--\eqref{eq:commuted-2} do not contain second derivatives of $\Omega^2$.

\subsubsection{The fundamental local spacetime estimate}\label{sec:FLSE}

\begin{lem}\label{lem:fundamental-local-estimate}
    For any $\mathfrak m\ge 0$, $\mathfrak e\in\Bbb R$, $\sigma>4$, and $C_0>0$, there exists a constant $C_*<\infty$ with the following property. Let $(r,\Omega^2,Q,f)$ be a solution of the spherically symmetric Einstein--Maxwell--Vlasov system with admissible momentum for particles of charge $\mathfrak e$, mass $\mathfrak m$, and momentum decay rate $\sigma$ defined on  $\mathcal R'\doteq \mathcal R(U_0,U_1,V_0,V_1)\setminus\{(U,V)\}$. Assume $U_1-U_0\le C_0$, $V_1-V_0\le C_0$, the initial data estimates
        \begin{equation}
    \|{\log r}\|_{C^2(\mathcal C)}+ \|{\log \Omega^2}\|_{C^2(\mathcal C)}+\|Q\|_{C^1(\mathcal C)}+\| f\|_{C^1_\sigma(P^\mathfrak m_\mathrm{red}|_\mathcal C)}\le C_0,\label{eq:ext-1}
    \end{equation}
     the global estimates
    \begin{equation}
        \iint_{\mathcal R'}\Omega^2\,dudv\le C_0,\label{eq:ext-5}
    \end{equation}
    \begin{equation}
       \sup_{\mathcal R'} |{\log r}|\le C_0,\label{eq:ext-4}
    \end{equation}
    and in the case $\mathfrak m=0$, assume also that
    \begin{equation}
        \inf_{\spt(f)\cap P^\mathfrak m_\mathrm{red}|_\mathcal C}\ell \ge C_0^{-1}.\label{eq:ext-6}
    \end{equation}
    Then we have the estimate
    \begin{equation*}
         \|{\log r}\|_{C^2(\mathcal R')}+ \|{\log \Omega^2}\|_{C^2(\mathcal R')}+\|Q\|_{C^1(\mathcal R')}+\| f\|_{C^1_\sigma(P^\mathfrak m_\mathrm{red}|_{\mathcal R'})}\le C_*.
    \end{equation*}
\end{lem}

\begin{proof} In this proof, we use the notation $A\les 1$ to mean that $A\le C$, where $C$ is a constant depending only on $\mathfrak m$, $\mathfrak e$, and $\sigma$, and $C_0$. When writing area integrals, we will also make no distinction between $\mathcal R$ and $\mathcal R'$, though the integrands are strictly speaking only assumed to be defined on $\mathcal R'$.

From \eqref{eq:ext-1} and the monotonicity properties of Maxwell's equations \eqref{eq:Max-u} and \eqref{eq:Max-v}, it follows that
\begin{equation}
    \sup_{\mathcal R'} |Q|\les 1.\label{eq:ext-8}
\end{equation}
Rewriting \eqref{eq:r-wave}, we obtain
    \begin{equation}
         r^2\Omega^4 T^{uv}=2\partial_u\partial_v r^2 +\Omega^2\left(1-\frac{Q^2}{r^2}\right).\label{eq:ext-9}
    \end{equation}
Integrating over $\mathcal R'$ and using \eqref{eq:ext-1}, \eqref{eq:ext-5}, and \eqref{eq:ext-4} yields
\begin{equation}\label{eq:T-area-est}
    \iint_{\mathcal R}\Omega^4 T^{uv}\,dudv \les \iint_{\mathcal R} r^2\Omega^4 T^{uv}\,dudv \les \iint_{\mathcal R}\partial_u\partial_v r^2\,dudv+ \iint_{\mathcal R} \Omega^2\,dudv\les 1.
\end{equation}
Rewriting \eqref{eq:ext-9} slightly, we obtain
    \begin{equation}
        \partial_u(r\partial_vr)= -\frac{\Omega^2}{4}\left(1-\frac{Q^2}{r^2}\right)+\tfrac 14 r^2\Omega^4 T^{uv}.\label{eq:ext-11}
    \end{equation}
        Integrating this in $u$ and using \eqref{eq:ext-1}, \eqref{eq:ext-4}, and \eqref{eq:ext-8}, we have 
        \begin{equation*}
            \sup_{[U_0,U_1]\times\{v\}}|r\partial_vr|\les 1+\int_{U_0}^{U_1}\Omega^2(u,v)\,du+\int_{U_0}^{U_1}\Omega^4 T^{uv}(u,v)\,du
        \end{equation*}
        for any $v\in[0,V]$.
        Integrating this estimate in $v$ and using \eqref{eq:ext-4}, \eqref{eq:ext-5}, and \eqref{eq:T-area-est} yields 
        \begin{equation}
            \int_{V_0}^{V_1}  \sup_{[U_0,U_1]\times\{v\}}|\partial_vr|\,dv\les   \int_{V_0}^{V_1}  \sup_{[U_0,U_1]\times\{v\}}|r\partial_vr|\,dv\les 1.\label{eq:ext-12}
        \end{equation}
        By Raychaudhuri's equation \eqref{eq:Ray-u}, $\partial_u r$ changes signs at most once on each ingoing null cone. Therefore, by the fundamental theorem of calculus and \eqref{eq:ext-4},
        \begin{equation}
               \sup_{v\in[V_0,V_1]} \int_{U_0}^{U_1}|\partial_ur|(u,v)\,du \le 2\left(\sup_{\mathcal R'}r-\inf_{\mathcal R'}r\right)\les 1.\label{eq:ext-13}
        \end{equation}
        Combining \eqref{eq:ext-12} and \eqref{eq:ext-13} yields 
        \begin{equation}
            \iint_\mathcal R|\partial_ur\partial_vr|\,dudv\le \left(\int_{V_0}^{V_1}  \sup_{[U_0,U_1]\times\{v\}}|\partial_vr|\,dv\right)\left(  \sup_{v\in[V_0,V_1]} \int_{U_0}^{U_1}|\partial_ur|(u,v)\,du\right)\les 1.\label{eq:ext-17}
        \end{equation}
        Using the definition of the Hawking mass \eqref{eq:Hawking-mass}, \eqref{eq:ext-4}, and \eqref{eq:ext-5}, we readily infer 
        \begin{equation}
           \iint_{\mathcal R}\Omega^2 |m|\,dudv\les 1\label{eq:ext-14}
        \end{equation}

By the wave equation \eqref{eq:Omega-wave}, the fundamental theorem of calculus, and the estimates \eqref{eq:S}, \eqref{eq:ext-5}, \eqref{eq:ext-4}, \eqref{eq:ext-8}, \eqref{eq:T-area-est}, and \eqref{eq:ext-14}, we have  
\begin{align*}
    \sup_{\mathcal R'}|\!\log\Omega^2|&\les 1+\left|\iint_\mathcal R\partial_u\partial_v\!\log\Omega^2\,dudv\right|\\
   & \les 1+\iint_\mathcal R\left(\Omega^2+\Omega^2|m|+\Omega^4 T^{uv}\right)dudv\\
   &\les 1.
\end{align*}

We now prove estimates for the electromagnetic geodesic flow. Let $\gamma:[0,S]\to \mathcal R'$ be an electromagnetic geodesic such that $(\gamma(0),p(0))\in \spt(f)\cap P^\mathfrak m_\mathrm{red}|_\mathcal C$. We aim to prove that
    \begin{equation}\label{eq:ext-p-est-1}
\left|\log\left(\frac{\Omega^2p^u(s)}{\Omega^2p^u(0)}\right)\right|+\left|\log\left(\frac{\Omega^2p^v(s)}{\Omega^2p^v(0)}\right)\right|\les 1
    \end{equation} for $s\in[0,S]$,
    where the implied constant does not depend on $\gamma$. 

   It suffices to prove this estimate for $p^u$, as the proof of the estimate for $p^v$ differs only in notation. Following \cite{Mosx-Vlasov-WP} (see also \cite{Dafermos-Vlasov}), we write an integral formula for $\log(\Omega^2p^u)$, which can then be estimated using our previous area estimates. Rewriting the mass shell relation \eqref{eq:mass-shell} as 
\begin{equation*}
    \frac{\ell^2}{r^2}=\left(\frac{\ell^2}{\ell^2+\mathfrak m^2 r^2}\right)\Omega^2p^up^v,
\end{equation*}
we deduce from the Lorentz force equation \eqref{eq:SS-Lor-4} that
\begin{equation*}
    \frac{d}{ds}\log(\Omega^2p^u) = \left(\partial_v\!\log\Omega^2-\frac{2\partial_vr}{r}\right)\left(\frac{\ell^2}{\ell^2+\mathfrak m^2 r^2}\right)p^v-\mathfrak e\frac{Q}{r^2}.
\end{equation*}
Integrating in $s$ and changing variables yields 
\begin{equation*}
    \log\left(\frac{\Omega^2p^u(s)}{\Omega^2p^u(0)}\right)=\int_{\gamma([0,s])}\left(\partial_v{\log\Omega^2}-\frac{2\partial_vr}{r}\right)\left(\frac{\ell^2}{\ell^2+\mathfrak m^2 r^2}\right)dv-\int_0^s\left.\mathfrak e\frac{Q}{r^2}\right|_{\gamma(s')}ds'.
\end{equation*}
We use the fundamental theorem of calculus on the first integral to obtain 
\begin{align*}
    \int_{\gamma([0,s])}\left(\partial_v\!\log\Omega^2-\frac{2\partial_vr}{r}\right)\left(\frac{\ell^2}{\ell^2+\mathfrak m^2 r^2}\right)dv&=\int_{\gamma^v(0)}^{\gamma^v(s)}\int_0^{\gamma^u(s_v)}\partial_u\left[\left(\partial_v{\log\Omega^2}-\frac{2\partial_vr}{r}\right)\left(\frac{\ell^2}{\ell^2+\mathfrak m^2 r^2}\right)\right]dudv\\
    &\quad + \int_{\{0\}\times[0,\gamma^v(s)]}\left(\partial_v{\log\Omega^2}-\frac{2\partial_vr}{r}\right)\left(\frac{\ell^2}{\ell^2+\mathfrak m^2 r^2}\right)dv,
\end{align*}
where $s_v\in[0,S]$ is defined by $\gamma^v(s_v)=v$. Using now the wave equations \eqref{eq:r-wave} and \eqref{eq:Omega-wave}, we arrive at 
\begin{align}
        \log\left(\frac{\Omega^2p^u(s)}{\Omega^2p^u(0)}\right)&=\int_{\gamma^v(0)}^{\gamma^v(s)}\int^{\gamma^u(s_v)}_{0} \left(\frac{3\Omega^2m}{r^3}-\frac{3\Omega^2Q^2}{2r^4}-\frac{\Omega^2}{2r^2}-\Omega^4 T^{uv}-\Omega^2S\right)\left(\frac{\ell^2}{\ell^2+\mathfrak m^2r^2}\right) dudv\nonumber\\
        &\quad + \int_{\gamma^v(0)}^{\gamma^v(s)}\int^{\gamma^u(s_v)}_{0} \left(2\partial_ur\partial_vr-\partial_ur\partial_v\!\log\Omega^2\right)\frac{2\mathfrak m^2\ell^2}{(\ell^2+\mathfrak m^2r^2)^2}\,dudv\nonumber\\
        &\quad +\int_{\{0\}\times[0,\gamma^v(s)]}\left(\partial_v{\log\Omega^2}-\frac{2\partial_vr}{r}\right) dv -\int_0^s\left. \mathfrak e\frac{Q}{r^2}\right|_{\gamma(s')}ds'.\label{eq:ext-15}
    \end{align}

    We now bound each of the four terms in \eqref{eq:ext-15}. Using  \eqref{eq:ext-5}, \eqref{eq:ext-4}, \eqref{eq:ext-8}, \eqref{eq:T-area-est}, and \eqref{eq:ext-14}, the first double integral in \eqref{eq:ext-15} is readily seen to be uniformly bounded. To estimate the second double integral, we integrate the wave equation \eqref{eq:Omega-wave} in $u$ to obtain
    \begin{equation*}
        \sup_{[0,U]\times\{v\}}|\partial_v\!\log\Omega^2|\les 1+\int_0^U\left(\Omega^2|m|+\Omega^2+\Omega^4T^{uv}\right)(u,v)\,du
    \end{equation*}
    for any $v\in[0,V]$. Integrating this estimate in $v$ and using the previous area estimates yields 
    \begin{equation*}
    \int_0^V\sup_{[0,U]\times\{v\}}|\partial_v\!\log\Omega^2|\,dv \les 1
    \end{equation*}
   which when combined with \eqref{eq:ext-13} gives
    \begin{equation}
        \iint_\mathcal R |\partial_ur||\partial_v\!\log\Omega^2|\,dudv\les 1.\label{eq:ext-16}
    \end{equation}
    Combined with \eqref{eq:ext-17}, we now readily see that the second double integral in \eqref{eq:ext-15} is uniformly bounded. The integral along initial data is clearly also bounded by assumption. 

    Using \eqref{eq:ext-4} and \eqref{eq:ext-8}, we estimate
    \begin{equation}
        \int_0^s\mathfrak e\frac{|Q|}{r^2}\,ds\les S.\label{eq:ext-18}
    \end{equation}
    To estimate $S$, define the function $\tau(s)\doteq \tau|_{\gamma(s)}$, which is strictly increasing and satisfies $0\le \tau(s)\les 1$ for every $s\in[0,S]$. Using the mass shell relation \eqref{eq:mass-shell}, we have
    \begin{equation}\label{eq:ptau-AM-GM}
        \sqrt{\frac{4}{\Omega^2}\left(\frac{\ell^2}{r^2}+\mathfrak m^2\right)}\le p^\tau= \frac{d\tau}{ds}.
    \end{equation}
    Since either $\mathfrak m>0$ or \eqref{eq:ext-6} holds, it follows that $d\tau/ds$ is uniformly bounded away from zero, which implies $S\les 1$ for any electromagnetic geodesic in the support of $f$. Combined with \eqref{eq:ext-18}, this uniformly bounds the final term in \eqref{eq:ext-15} and completes the proof of \eqref{eq:ext-p-est-1}.

    Since $f$ is constant along $(\gamma(s),p(s))$, we therefore have $\langle p^\tau(s)\rangle^\sigma f(\gamma(s),p(s))\les  \langle p^\tau(0)\rangle^\sigma f(\gamma(0),p(0))$, which implies
    \begin{equation*}
        \|f\|_{C^0_\sigma(P^\mathfrak m_\mathrm{red}|_{\mathcal R'})}+\sup_{\mathcal R'}\left(N^u+N^v+T^{uu}+T^{uv}+T^{vv}+S\right)\les 1.
    \end{equation*} By \eqref{eq:Max-u} and \eqref{eq:Max-v},
    \begin{equation*}
        \sup_{\mathcal R'}\left(|\partial_uQ|+|\partial_v Q|\right)\les 1.
    \end{equation*}
    By integrating \eqref{eq:ext-11} and also using that   $\partial_u(r\partial_vr)=\partial_v(r\partial_ur)$, we now readily estimate
    \begin{equation*}
        \sup_{\mathcal R'}\left(|\partial_ur|+|\partial_vr|\right)\les 1.
    \end{equation*}
    As this bounds the Hawking mass pointwise, we can now estimate
\begin{equation*}
     \sup_{\mathcal R'}\left(|\partial_u{\log\Omega^2}|+|\partial_v{\log\Omega^2}|+|\partial_u\partial_v{\log\Omega^2}|\right)\les 1
\end{equation*}
    using \eqref{eq:Omega-wave}.  It then follows immediately from \eqref{eq:r-wave}, \eqref{eq:Ray-u}, and \eqref{eq:Ray-v} that
    \begin{equation*}
        \sup_{\mathcal R'}\left(|\partial_u^2r|+|\partial_u\partial_vr|+|\partial_v^2r|\right)\les 1.
    \end{equation*}

Along an electromagnetic geodesic $\gamma$ lying in the support of $f$, we have that
\begin{align*}
    |X(\hat\partial_uf)|&\les p^\tau |\hat\partial_uf|+(p^\tau)^2   |X(\partial_{p^u}f)|+(p^\tau)^2    |X(\partial_{p^v}f)|,\\
    |X(\hat\partial_vf)|&\les p^\tau |\hat\partial_vf|+(p^\tau)^2    |X(\partial_{p^u}f)|+(p^\tau)^2|X(\partial_{p^v}f)| ,  \\
    |X(\partial_{p^u}f)|&\les  |\hat\partial_uf|+p^\tau   |X(\partial_{p^u}f)|+p^\tau    |X(\partial_{p^v}f)|,\\
    |X(\partial_{p^v}f)|&\les  |\hat\partial_vf|+p^\tau   |X(\partial_{p^u}f)|+p^\tau    |X(\partial_{p^v}f)|
\end{align*} by \eqref{eq:commuted-1}--\eqref{eq:commuted-2} and all of the estimates obtained so far. It follows that, defining 
\begin{equation*}
    \mathcal F(s)\doteq \left((p^\tau)^{\sigma}\hat\partial_uf,(p^\tau)^{\sigma}\hat\partial_vf,(p^\tau)^{\sigma+1}\partial_{p^u}f,(p^\tau)^{\sigma+1}\partial_{p^v}f\right)(s)
\end{equation*}
along $\gamma(s)$, we have
\begin{equation*}
    \left|\frac{d}{ds}\mathcal F\right|\les p^\tau |\mathcal F|.
\end{equation*}
By Gr\"onwall's inequality and \eqref{eq:ptau-AM-GM}, it follows that $|\mathcal F|\les 1$ along $\gamma$. Recovering $\partial_uf$ and $\partial_vf$ from $\hat\partial_uf$ and $\hat\partial_vf$, we now readily bound 
\begin{equation*}
    \|f\|_{C^1_\sigma(P^\mathfrak m_\mathrm{red}|_{\mathcal R'})}+\sup_{\mathcal R'}\left(|\partial_uT^{uv}|+|\partial_vT^{uv}|+|\partial_uS|+|\partial_vS|\right)\les 1.
\end{equation*}
Commuting the wave equation \eqref{eq:Omega-wave} with $\partial_u$ and $\partial_v$, we obtain the final estimates
\begin{equation*}
 \sup_{\mathcal R'}\left(|\partial_u^2{\log\Omega^2}|+|\partial_v^2{\log\Omega^2}|\right)\les 1   ,
\end{equation*}
which completes the proof.\end{proof}

\subsubsection{Local existence in characteristic slabs}

The spacetime estimate \cref{lem:fundamental-local-estimate} can also be used to improve \cref{prop:char-IVP-Vlasov} to local existence in a full double null neighborhood of a bifurcate characteristic hypersurface of arbitrary length:

\begin{prop}\label{prop:strip}
    For any $\mathfrak m\ge 0$, $\mathfrak e\in\Bbb R$, $\sigma>4$, $B>0$, and $c_\ell>0$, there exists a constant $\ve_\mathrm{slab}>0$ with the following property. Let $(\mathring r,\mathring\Omega^2,\mathring Q,\mathring f)$ be a characteristic initial data set for the spherically symmetric Einstein--Maxwell--Vlasov system on $\mathcal C(U_0,U_1,V_0,V_1)$ with admissible momentum. If \begin{equation*}
    \|{\log\mathring r}\|_{C^2(\mathcal C)}+ \|{\log\mathring \Omega^2}\|_{C^2(\mathcal C)}+\|\mathring Q\|_{C^1(\mathcal C)}+\|\mathring f\|_{C^1_\sigma(P^\mathfrak m_\mathrm{red}|_\mathcal C)}\le B
\end{equation*} and either $\mathfrak m>0$ or $\mathfrak m=0$ and $\ell\ge c_\ell$ on $\spt(\mathring f)$, then there exists a unique smooth solution $(r,\Omega^2,Q,f)$ of the spherically symmetric Einstein--Maxwell--Vlasov system with admissible momentum on   \begin{equation*}
    \mathcal R(U_0,U_0+\min\{\ve_\mathrm{slab},U_1-U_0\},V_0,V_1)\cup  \mathcal R(U_0,U_1,V_0,V_0+\min\{\ve_\mathrm{slab},V_1-V_0\}) 
\end{equation*} which extends the initial data.
\end{prop}
\begin{proof}
    We prove existence in the slab which is thin in the $u$-direction, the proof in the other slab being identical. Let $C_0=10B$ and let $C_*$ be the constant obtained from the fundamental local spacetime estimate \cref{lem:fundamental-local-estimate} with this choice. Let $\mathcal A\subset[V_0,V_1]$ denote the set of $\tilde V$ such that the solution exists on $\mathcal R(U_0,U_1',V_0,\tilde V)$,
    where $U_1'\doteq U_0+\min\{\ve_\mathrm{slab},U_1-U_0\}$ and satisfies the estimates
    \begin{equation}
        \sup_{\mathcal R(U_0,U_1',V_0,\tilde V)}|{\log r}|+\sup_{\mathcal R(U_0,U_1',V_0,\tilde V)}|{\log\Omega^2}|\le C_0.
    \end{equation}
    We will show that if $\ve_\mathrm{slab}=\min\{\ve_\mathrm{loc}(C_*),B(U_1-U_0)^{-1}C_*^{-1}\}$, then $\mathcal A$ is nonempty, closed, and open. 

    Nonemptyness follows from \cref{prop:char-IVP-Vlasov} and closedness by continuity of the bootstrap assumptions. Let now $\tilde V\in\mathcal A$. To improve the bootstrap assumptions, we note that $|\partial_u\partial_v{\log r}|\le C_*$ on $\mathcal R(U_0,U_1',V_0,\tilde V)$ by \cref{lem:fundamental-local-estimate}, whence
    $|{\log r}|\le \ve_\mathrm{slab}(U_1-U_0)C_*+3B\le \frac 12 C_0$ by the fundamental theorem of calculus. A similar argument applies for $\log\Omega^2$. Therefore, by applying \cref{prop:char-IVP-Vlasov} again, a simple continuity argument shows that $\tilde V+\eta\in \mathcal A$ for $\eta>0$ sufficiently small. 
\end{proof}

\subsection{Time-symmetric seed data and their normalized developments}\label{sec:general-data}

In the proof of \cref{thm:main}, we will pose data for the Einstein--Maxwell--Vlasov system on a mixed spacelike-null hypersurface, with the Vlasov field $f$ supported initially on the spacelike hypersurface and away from the center. The initial data is given by a compactly supported distribution function $\mathring f$ on the spacelike hypersurface, a numerical parameter that fixes the location of the initial outgoing null cone, together with the mass and charge of the particles. As we will only consider \emph{time-symmetric} initial configurations, these data are sufficient to uniquely determine the solution.

\begin{defn}\label{def:seed}
     A \emph{time-symmetric seed data set} $\mathcal S\doteq (\mathring f,r_2,\mathfrak m,\mathfrak e)$ for the spherically symmetric Einstein--Maxwell--Vlasov system consists of a real numbers $r_2\in\Bbb R_{>0}$, $\mathfrak m\in\Bbb R_{\ge 0}$, and $\mathfrak e\in\Bbb R$, together with a compactly supported nonnegative function $\mathring f \in C^\infty\big((0,\infty)_r\times (0,\infty)_{p^u}\times (0,\infty)_{p^v}\big)$ which is symmetric in the second and third variables,
$     \mathring f(\,\cdot\,,p^u,p^v)=\mathring  f(\,\cdot\,,p^v,p^u),$ and satisfies $ \spt \big(\mathring f(\,\cdot\,,p^u,p^v)\big)\subset(0,r_2]$ for every $p^u,p^v\in (0,\infty)$.
\end{defn}

Given a seed data set $\mathcal S=(\mathring f,r_2,\mathfrak m,\mathfrak e)$ and $r\in [0,r_2]$, we define
\begin{align}
\mathring\varrho(r)&\doteq \pi \int_0^\infty\int_0^\infty \mathring f(v,p^u,p^v)\,dp^udp^v,\\
\label{eq:seed-1}  \mathring{\mathcal N}{}^u(r) \doteq \mathring{\mathcal N}{}^v(r) &\doteq \pi \int_0^\infty \int_0^\infty p^u\mathring f(r,p^u,p^v)\,dp^udp^v, \\
\label{eq:seed-2}   \mathring{\mathcal T}{}^{uu}(r)\doteq \mathring{\mathcal T}{}^{vv}(r) &\doteq  \pi \int_0^\infty \int_0^\infty (p^u)^2\mathring f(r,p^u,p^v)\,dp^udp^v,\\
 \label{eq:seed-3}  \mathring{\mathcal T}{}^{uv}(r) &\doteq \pi \int_0^\infty \int_0^\infty p^up^v\mathring f(r,p^u,p^v)\,dp^udp^v.
\end{align}

\begin{rk}
    These formulas are missing a factor of $\Omega^2$ compared to \eqref{eq:Nu}--\eqref{eq:S}. This is because $\Omega^2$ is not explicitly known on the initial data hypersurface and is accounted for by extra factors of $\Omega^2$ in the constraint system \eqref{eq:constraint-1}--\eqref{eq:constraint-2} below.
\end{rk}

Let the functions $\mathring m=\mathring m(r)$ and $\mathring Q=\mathring Q(r)$ be the unique solutions of the first order system
\begin{align}
 \label{eq:constraint-1}   \frac{d}{dr}\mathring m&= \frac{r^2}{4}\left(1-\frac{2\mathring m}{r}\right)^{-2}\left(\mathring{\mathcal T}{}^{uu}+2\mathring{\mathcal T}{}^{uv}+\mathring{\mathcal T}{}^{vv}\right)+\frac{\mathring Q^2}{2r^2},\\
\label{eq:constraint-2}    \frac{d}{dr}\mathring Q&= \frac 12\mathfrak er^2\left(1-\frac{2\mathring m}{r}\right)^{-2}\left(\mathring{\mathcal N}{}^u+\mathring{\mathcal N}{}^v\right)
\end{align}
with initial conditions $\mathring m(0)=0$ and $\mathring Q(0)=0$. If \begin{equation*}
    \sup_{r\in [0,r_2]} \frac{2\mathring m}{r}<1,
\end{equation*}
then $\mathring m$ and $\mathring Q$ exist on the entire interval $[0,r_2]$ and we say that $\mathcal S$ is \emph{untrapped}.  We also define
\begin{equation*}
    \mathring \Omega{}^2\doteq\left(1-\frac{2\mathring m}{r}\right)^{-1},\quad\mathring \varpi\doteq\mathring m+\frac{\mathring Q{}^2}{2r}.
\end{equation*}
Finally, we say that $\mathcal S$ is \emph{consistent with particles of mass $\mathfrak m$} if $\mathring\Omega^2(r)p^up^v\ge \mathfrak m^2$ for every $(r,p^u,p^v)\in \spt \mathring f$.

\begin{rk}
    We have not attempted to formulate the most general notion of seed data for the spherically symmetric Einstein--Maxwell--Vlasov Cauchy problem here as it is not needed for our purposes.
\end{rk}

 \begin{figure}
\centering{
\def\svgwidth{16pc}
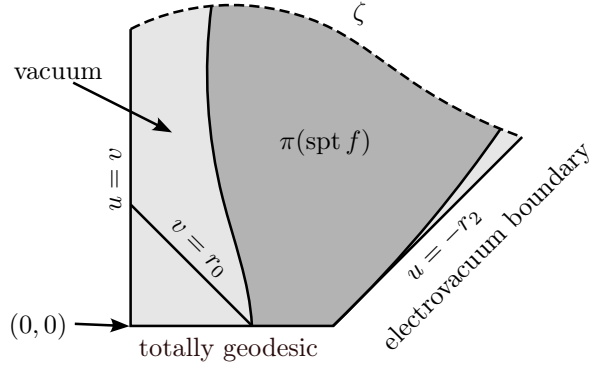}
\caption{Penrose diagram of a normalized development $\mathcal U$ of a time symmetric seed $\mathcal S$. The spacelike hypersurface $\{\tau=0\}$ is totally geodesic, i.e., time symmetric, and the outgoing cone $\{u=-r_2\}$ has $f=0$. To the left of the support of $f$, the spacetime is vacuum: both the Hawking mass $m$ and charge $Q$ vanish identically.
 For the significance of the cone $\{v=r_0\}$, see already \cref{rk:Mink-corner}.}\label{fig:development}
\end{figure}

Associated with time-symmetric seed data as in \cref{def:seed}, we will introduce normalized developments of such data in the following. 
For $r_2>0$, let
\begin{equation*}
    \mathcal C_{r_2}\doteq \{\tau\ge 0\}\cap\{v\ge u\}\cap\{u\ge -r_2\}\subset\Bbb R^2_{u,v}
\end{equation*}
and let $\mathfrak U_{r_2}$ denote the collection of connected relatively open subsets $\mathcal U\subset\{v\ge u\}\subset\Bbb R^2_{u,v}$ for which there exists a (possibly empty) achronal curve  $\zeta\subset\mathcal C_{r_2}$, extending from the center $\{u=v\}$ and reaching the cone $\{u= -r_2\}$, such that  $\mathcal U=\mathcal C_{r_2}\cap\{u+v <\zeta^u+\zeta^v\}$  and $\{\tau=0\}\cap\{0\le v\le r_2\}\subset\mathcal U$. See \cref{fig:development}.

We also define the cones
\begin{equation*}
    C_{u_0}\doteq \mathcal U \cap \{u=u_0\},\quad \underline C_{v_0}\doteq \mathcal U \cap \{v=v_0\}.
\end{equation*}

\begin{defn}\label{def:normalized-defn}
    Let $\mathcal S=(\mathring f,r_2,\mathfrak m,\mathfrak e)$ be an untrapped time-symmetric seed data set which is consistent with particles of mass $\mathfrak m$. A \emph{normalized development} of $\mathcal S$ consists of a domain $\mathcal U\in\mathfrak U_{r_2}$ and a spherically symmetric solution $(r,\Omega^2,Q,f)$ of the Einstein--Maxwell--Vlasov system for particles of mass $\mathfrak m$ and fundamental charge $\mathfrak e$ defined over $\mathcal U\setminus\{u=v\}$ such that the following holds.
\begin{enumerate}
    \item For every $(v,p^u,p^v)\in (0,r_2]\times(0,\infty)\times(0,\infty)$,
    \begin{align}
\label{eq:development-1}  r(-v,v)&= v,\\
  \partial_vr(-v,v)&=\tfrac 12,\\
\label{eq:time-symmetry-r}  \partial_ur(-v,v)&=-\partial_vr(-v,v),\\
    \Omega^2(-v,v)&=\mathring\Omega{}^2(v),\\
    \partial_v\Omega^2(-v,v)&=\frac 12 \left(\frac{d}{dr}\mathring\Omega^2\right)(v),\label{eq:Omega-derivative-data}\\
    \label{eq:time-symmetry-Omega}    \partial_u{\log\Omega^2}(-v,v)&=-\partial_v{\log\Omega^2}(-v,v),\\
f(-v,v,p^u,p^v)&=\mathring f(v,p^u,p^v).\label{eq:f-data-defn}
\end{align}
\item Along the initial outgoing null cone $C_{-r_2}$, 
\begin{align}
\label{eq:development-3} r(-r_2,v)&=\tfrac 12 r_2+\tfrac 12 v,\\
    \Omega^2&=\mathring \Omega{}^2(r_2),\\
 \label{eq:development-4}   Q&=\mathring Q(r_2),\\
  \label{eq:development-2}  f&=0.
\end{align}
\item The functions $r$, $\Omega^2$, $Q$, and $m$ extend smoothly to the center $\Gamma\doteq \mathcal U\cap\{u=v\}$ and satisfy there the boundary conditions
\begin{equation}
    r=m=Q=0,\label{eq:boundary-1}
\end{equation}
\begin{equation}
    \partial_ur<0,\quad\partial_vr>0.
\end{equation}
\item Let $\gamma:[0,S)\to \mathcal U\setminus\Gamma$ be a future-directed electromagnetic geodesic such that $r(\gamma(s))\to 0$ as $s\to S$.\footnote{Such a curve necessarily has $\ell=0$.} Then $(\gamma^u(s),\gamma^v(s),p^u(s),p^v(s))$ attains a limit on $\Gamma$, say $(u_*,v_*,p^u_*,p^v_*)$, and there exists a unique electromagnetic geodesic $\gamma':(S,S+\ve)\to \mathcal U\setminus\Gamma$ for some $\ve>0$ such that $(\gamma'^u(s),\gamma'^v(s),p'^u(s),p'^v(s))\to (u_*,v_*,p^u_*,p^v_*)$ as $s\to S$. We then require that 
\begin{equation*}
  \lim_{s\nearrow S}  f(\gamma^u(s),\gamma^v(s),p^u(s),p^v(s))=\lim_{s\searrow S}f(\gamma'^u(s),\gamma'^v(s),p'^u(s),p'^v(s)).
\end{equation*}
\end{enumerate}
 We use the adjective ``normalized'' to emphasize the choice of a development with double null gauge anchored to the data as in points 1. and 2.~above.
\end{defn}

\begin{rk}
    The ``time-symmetric'' aspect of the development is captured by the first equalities in \eqref{eq:seed-2} and \eqref{eq:seed-3}, and the equations \eqref{eq:time-symmetry-r} and \eqref{eq:time-symmetry-Omega}. One can moreover easily verify, using \eqref{eq:time-symmetry-r}, \eqref{eq:time-symmetry-Omega}, and the formulas for the Christoffel symbols in \cref{sec:formulas}, that $\{\tau=0\}\cap\mathcal U$ is a totally geodesic spacelike hypersurface with respect to the $(3+1)$-dimensional metric \eqref{eq:dn}.
\end{rk}

 For a normalized development of seed data, we clearly have
\begin{align*}
  N^u  &=\mathring\Omega{}^2\mathring{\mathcal N}{}^u, & N^v&=\mathring \Omega{}^{2}\mathring{\mathcal N}{}^v,\\
  T^{uu}   &=\mathring\Omega{}^2 \mathring{\mathcal T}{}^{uu}, &   T^{vv}&=\mathring\Omega{}^2\mathring{\mathcal T}{}^{vv},\\
    T^{uv}&=\mathring\Omega{}^{2}\mathring{\mathcal T}{}^{uv}, & S&=\frac{\mathring\Omega{}^4}{2}\mathring{\mathcal T}^{uv}-\mathring\Omega^2\mathfrak m^2\mathring \varrho,\\
    Q&= \mathring Q, & \varpi &=\mathring\varpi
\end{align*}
along $\{\tau=0\}\cap\mathcal U$. 

\begin{prop}\label{prop:development-1}
    Let $\mathcal S$ be an untrapped time-symmetric seed data set which is consistent with particles of mass $\mathfrak m$. Then there exists a $\delta>0$ and a unique normalized development $(r,\Omega^2,Q,f)$ of $\mathcal S$ defined on $\{0\le \tau<\delta\}\cap\mathcal C_{r_2}$. 
\end{prop}
\begin{proof}
    Using essentially the same methods as the proof of \cref{prop:char-IVP-Vlasov} in \cref{app:B}, we obtain a unique local smooth solution $(r,\Omega^2,Q,f)$ to the system of equations \eqref{eq:r-wave}, \eqref{eq:Omega-wave}, \eqref{eq:Max-u}, and \eqref{eq:SS-MV}, with initial data given by \eqref{eq:development-1}--\eqref{eq:development-2}. It remains to show that the constraints \eqref{eq:Ray-u}, \eqref{eq:Ray-v}, and \eqref{eq:Max-v} hold.
    
    By the same calculation as in the proof of \cref{prop:char-IVP-Vlasov}, equation \eqref{eq:SS-MV} implies the conservation law \eqref{eq:SS-N-div} for $N$. Let $v\in (0,r_2)$. By integration of \eqref{eq:Max-u}, 
    \begin{equation*}
        Q(u,v)=\mathring Q(v)-\int_{-v}^u \tfrac 12 \mathfrak er^2\Omega^2N^v\,du'
    \end{equation*}
    for $u\ge -v$. Differentiating in $v$, using \eqref{eq:constraint-2}, \eqref{eq:SS-N-div}, and the fundamental theorem of calculus yields \eqref{eq:Max-v} at $(u,v)$.

    To prove that \eqref{eq:Ray-u} and \eqref{eq:Ray-v} hold, we argue as in the proof of \cref{prop:char-IVP-Vlasov}. Therefore, it suffices to show that \eqref{eq:Ray-v} holds on initial data (the corresponding argument for \eqref{eq:Ray-u} being the same). By \eqref{eq:development-3} and \eqref{eq:development-2}, \eqref{eq:Ray-v} clearly holds on the initial outgoing cone. By taking the absolute $v$-derivative of $\partial_vr(-v,v)=\frac 12$, we obtain $\partial_v^2r(-v,v)=\partial_u\partial_vr(-v,v)$. Therefore, using \eqref{eq:r-wave}, \eqref{eq:constraint-1}, and \eqref{eq:Omega-derivative-data}, we readily compute
    \begin{align*}
        \partial_v^2r-\partial_vr\partial_v{\log\Omega^2}+\frac 14 r\Omega^4T^{vv}&=\partial_u\partial_vr- \frac{1}{4\Omega^2}\frac{d}{dr}\left(1-\frac{2\mathring m}{r}\right)^{-1}+\frac 14r\Omega^4 T^{vv}\\
        &=-\frac{\Omega^2m}{2r^2}+\frac{\Omega^2Q^2}{4r^3}+\frac 14r\Omega^4T^{uv}\\&\quad-\frac{1}{4\Omega^2}\left(-\frac{2\Omega^4\mathring m}{r^2}+\frac{r\Omega^6}{2}\left(T^{uu}+2T^{uv}+T^{vv}\right)+\frac{\Omega^4 Q^2}{r^3}\right)+\frac 14r\Omega^4 T^{vv}\\
        &=0,
    \end{align*}
    where every function is  being evaluated at $(-v,v)$. This is equivalent to \eqref{eq:Ray-v} and completes the proof.
\end{proof}

\begin{rk}\label{rk:Mink-corner}
    Let $r_0\in (0,r_2)$ be such that $\mathring f(r,p^u,p^v)=0$ if $r\in (0,r_0]$. Since $\mathring f$ is assumed to be compactly supported, such an $r_0$ necessarily exists. Then if $(\mathcal U,r,\Omega^2,Q,f)$ is a normalized development of $\mathcal S$, the portion of the triangle $\{v\le r_0\}$ inside of $\mathcal U$ is identically Minkowskian in the sense that
    \begin{gather*}
        r=\tfrac 12(v-u),\\
        \Omega^2=1,\\
        Q=f=0
    \end{gather*}
    on $\mathcal U\cap \{v\le r_0\}$. In fact, we may therefore assume that any normalized development of $\mathcal S$ contains the full corner $\mathcal C_{r_2}\cap\{v\le r_0\}$.
\end{rk}

\begin{rk}
 One can verify that a normalized development as in \cref{def:normalized-defn} defines a solution of the constraint equations associated to the $(3+1)$-dimensional Einstein--Maxwell--Vlasov system after applying the correspondence of \cref{prop:SS-equiv}. In particular, the lift of $\tau=0$ will be totally geodesic in the $(3+1)$-dimensional spacetime. 
\end{rk}

\begin{rk}
One can ``maximalize'' \cref{prop:development-1} to show the existence of a maximal globally hyperbolic development of $\mathcal S$, but this requires treating the local existence and uniqueness problem for the spherically symmetric Einstein--Maxwell--Vlasov system at the center of symmetry, which we do not address in this paper.\footnote{In the case $\mathfrak m>0$ one could directly appeal to \cite{Blancel} to get local well-posedness near $r=0$.} Indeed, since our charged Vlasov beams spacetimes will always be vacuum near $r=0$, existence and uniqueness near the center will be completely trivial in our specific construction and is established in the following lemma.
\end{rk}

\begin{lem}\label{lem:Minkowski}
    Let $u_0\le v_0<v_1$, $r_0\ge 0$, $\lambda_0>0$, and $\alpha:[u_0,v_0]\to\Bbb R_{>0}$ and $\beta:[v_0,v_1]:\Bbb R_{>0}$ be smooth functions satisfying the relations
\begin{equation*}
    \alpha(u_0)=\beta(v_0),\quad \alpha(v_0)=4\lambda_0^2,\quad r_0=\frac{1}{4\lambda_0}\int_{u_0}^{v_0}\alpha(u')\,du'.
\end{equation*}
    Then there exists a unique smooth solution $(r,\Omega^2,Q,f)$ of the spherically symmetric Einstein--Maxwell--Vlasov system on 
    \begin{equation*}
        [u_0,v_0]\times [v_0,v_1]\cup\left(\{v\ge u\}\cap\{u\ge v_0\}\cap\{v\le v_1\}\right)
    \end{equation*} with $Q$ and $f$ identically vanishing, satisfying the boundary conditions of \cref{def:normalized-defn} along $\{u=v\}$, together with
    \begin{equation*}
         r(u_0,v_0)=r_0,\quad\partial_vr(u_0,v_0)=\lambda_0,\quad\Omega^2|_{[u_0,v_0]\times\{v_0\}}=\alpha,\quad \Omega^2|_{\{u_0\}\times [v_0,v_1]}=\beta.
    \end{equation*}
  The solution is given by the explicit formulas
  \begin{equation*}
      r(u,v)=r_0+\frac{\lambda_0}{\beta(v_0)}\int_{v_0}^v\beta(v')\,dv'-\frac{1}{4\lambda_0}\int_{u_0}^u\alpha(u')\,du',\quad
      \Omega^2(u,v)=\frac{\alpha(u)\beta(v)}{\beta(v_0)}
  \end{equation*}
  for $(u,v)\in [u_0,v_0]\times[v_0,v_1]$ and
\begin{equation*}
    r(u,v)=\frac{\lambda_0}{\beta(v_0)}\int_{u}^v\beta(v')\,dv',\quad
      \Omega^2(u,v)= \frac{4\lambda_0}{\beta(v_0)^2}\beta(u)\beta(v)
\end{equation*}  
for $(u,v)\in \{v\ge u\}\cap\{u\ge v_0\}\cap\{v\le v_1\}$. 
\end{lem}

\begin{rk}
    From the last formula, it follows that 
\begin{equation}
    \partial_u\Omega^2(u,u)=\partial_v\Omega^2(u,u).\label{eq:Omega-center}
\end{equation}
\end{rk}

\section{A singular toy model: bouncing charged null dust}\label{sec:dust-intro}

In this section, we introduce Ori's \emph{bouncing charged null dust model} \cite{Ori91}. We then show that Ori's model exhibits extremal critical collapse and can be used to construct counterexamples to the third law of black hole thermodynamics. Later, in \cref{sec:Vlasov-beams}, we will then show that these dust solutions can be (in an appropriate sense) globally desingularized by passing to smooth solutions of the Einstein--Maxwell--Vlasov system. The constructions in this section are crucial to motivate the choice of initial data in the proof of \cref{thm:main} in \cref{sec:Vlasov-beams}.

    \subsection{Ori's bouncing charged null dust model}
\label{sec:ori-bouncing}
We begin by recalling the general notion of charged null dust from \cite{Ori91}:

\begin{defn}\label{def:charged-null-dust}
     The \emph{Einstein--Maxwell--charged null dust model} for particles of \emph{fundamental charge} $\mathfrak e\in\Bbb R\setminus\{0\}$ consists of a charged spacetime $(\mathcal M,g,F)$, a future-directed null vector field $k$ representing the momentum of the dust particles, and a nonnegative function $\rho$ which describes the energy density of the dust. The equations of motion are
    \begin{align}
    \label{eq:intro-dust-eq-1}    R_{\mu\nu}-\tfrac 12Rg_{\mu\nu} &=2\left(T_{\mu\nu}^\mathrm{EM}+T_{\mu\nu}\right),\\
        \label{eq:intro-dust-eq-2}     \nabla^\alpha F_{\mu\alpha}&=\mathfrak e\rho k_\mu,\\
   \label{eq:intro-dust-eq-3}      k^\nu\nabla_\nu k^\mu&=\mathfrak eF^\mu{}_\nu k^\nu,\\
   \label{eq:intro-dust-eq-4}     \nabla_\mu(\rho k^\mu)&=0,
    \end{align}
    where $T_{\mu\nu}^\mathrm{EM}$ was defined in \eqref{eq:EMEM} and $T_{\mu\nu}\doteq \rho k_\mu k_\nu$ is the energy-momentum tensor of a pressureless perfect fluid. By the forced Euler equation \eqref{eq:intro-dust-eq-3}, the integral curves of $k$ are electromagnetic null geodesics.
\end{defn}

    Any two functions $\varpi_\mathrm{in},Q_\mathrm{in}\in C^\infty(\Bbb R)$ determine a spherically symmetric solution to the system \eqref{eq:intro-dust-eq-1}--\eqref{eq:intro-dust-eq-4} by the formulas 
    \begin{gather}
          g_\mathrm{in}[\varpi_\mathrm{in}, Q_\mathrm{in}]\doteq -D(V,r)\,dV^2+2\,dVdr+r^2\gamma,\label{eq:g-in}\\
         \label{eq:F-Vaidya} F\doteq -\frac{Q_\mathrm{in}}{r^2}\,dV\wedge dr,\\
        k\doteq \frac{\mathfrak e}{\dot Q_\mathrm{in}}\left(\dot \varpi_\mathrm{in}-\frac{ Q_\mathrm{in}\dot Q_\mathrm{in}}{r}\right)(-\partial_r),\quad  \rho \doteq \frac{(\dot Q_\mathrm{in})^2}{\mathfrak e^2r^2}\left( \dot\varpi_\mathrm{in}-\frac{Q_\mathrm{in}\dot Q_\mathrm{in}}{r}\right)^{-1},\label{eq:Ori-1}
    \end{gather}
   where $\cdot$ denotes differentiation with respect to $V$ and 
    \begin{equation*}
    D(V,r)\doteq 1-\frac{2\varpi_\mathrm{in}(V)}{r}+\frac{Q_\mathrm{in}(V)^2}{r^2}.
\end{equation*} 
The metric \eqref{eq:g-in} is known as the \emph{ingoing charged Vaidya metric} \cite{PS68,BVai70} and describes a ``time dependent'' Reissner--Nordstr\"om spacetime in ingoing Eddington--Finkelstein-type coordinates $(V,r,\vartheta,\varphi)$. The spacetime is time oriented by $-\partial_r$. The metric \eqref{eq:g-in} and Maxwell field \eqref{eq:F-Vaidya} are spherically symmetric and may therefore be considered as a spherically symmetric charged spacetime in the framework of \cref{sec:spherical-symmetry}. One easily sees that $D=1-\frac{2m}{r}$, $Q=Q_\mathrm{in}$, and $\varpi=\varpi_\mathrm{in}$.

We will always make the assumption $\dot\varpi_\mathrm{in}\ge 0$ so that $T^{\mu\nu}=\rho k^\mu k^\nu$ satisfies the weak energy condition for $r$ sufficiently large. We also assume that $\mathfrak e> 0$ and impose the condition  $ \dot Q_\mathrm{in}\ge 0$ on the seed function $Q_\mathrm{in}$, which just means that positively charged particles increase the charge of the spacetime. (If $\mathfrak e<0$, we would instead assume $\dot Q_\mathrm{in}\le 0$ and the discussion would otherwise remain unchanged.) 

We define a function $r_\mathrm{b}=r_\mathrm{b}(V)$, called the \emph{bounce radius}, by 
\begin{equation*}
    r_\mathrm{b}\doteq \frac{Q_\mathrm{in}\dot Q_\mathrm{in}}{\dot\varpi_\mathrm{in}}
\end{equation*}
whenever $\dot \varpi_\mathrm{in}(V)>0$. The reason for this terminology will become clear shortly. By inspection of \eqref{eq:Ori-1}, we observe the following: For $r>r_\b(V)$, $(g_\mathrm{in},F,k,\rho)$ defines a solution of the Einstein--Maxwell-charged null dust system, $k$ is future-directed null, and $\rho \ge 0$. If $r_\b(V)>0$ and $r\searrow r_\b(V)$, then $k$ and $T^\mathrm{dust}$ vanish. If also $\dot Q_\mathrm{in}(V)>0$, then $\rho$ blows up at $r=r_\b(V)$, but $\rho k$ is nonzero and bounded. Finally, for $r<r_\b(V)$, $k$ is \emph{past-directed} null and $\rho<0$, so $T^\mathrm{dust}$ violates the weak energy condition.

Physically, the ingoing Vaidya metric and \eqref{eq:Ori-1} describe an ingoing congruence of radial charged massless dust particles which interact with the electromagnetic field that they generate. One can interpret the vanishing of $k$ as the dust being ``stopped'' by the resulting repulsive Lorentz force. Integral curves of $k$ are ingoing radial electromagnetic null geodesics $\gamma(s)$ with limit points on the \emph{bounce hypersurface} $\Sigma_\mathrm{b}\doteq\{r=r_\mathrm{b}\}$ as $s\to \infty$. The charged null dust system is actually \emph{ill-posed} across $\Sigma_\mathrm{b}$ since the transport equation \eqref{eq:intro-dust-eq-3} breaks down there. Because of this, Ori argued in \cite{Ori91} that the ingoing charged Vaidya metric \eqref{eq:g-in} (and the associated formulas in \eqref{eq:Ori-1}) should only be viewed as physical to the \emph{past} of $\Sigma_\b$ and must be modified if we wish to continue the solution beyond $\Sigma_\b$.

  \begin{rk}\label{rk:Mink-blowup}
      The divergence of $\rho$ along $\Sigma_\b$ does not seem to have been explicitly mentioned by Ori, but it is one of the fundamentally singular features of charged null dust. One can also see that $\rho$ can blow up if $\dot Q_\mathrm{in}/Q_\mathrm{in}$ blows up as a function of $V$, which occurs if the dust is injected into Minkowski space. 
  \end{rk}

  \begin{rk}
      Before Ori's paper \cite{Ori91}, the ``standard interpretation'' \cite{Sullivan-Israel,lake1991structure} of the ingoing Vaidya metric \eqref{eq:g-in} did not actually involve Maxwell's equation and the fluid equation was simply taken to be the standard geodesic equation. The set $\{r<r_\b\}$ was included in the ingoing solution and the dust was thought to violate the weak energy condition in this region. We refer to \cite{Ori91} for discussion. 
  \end{rk}

In order to continue the dust solution across $\Sigma_\b$, we must make some further (nontrivial!) assumptions on the seed functions $\varpi_\mathrm{in}$ and $Q_\mathrm{in}$. In order to not trivially violate causality, we must demand that $\Sigma_\mathrm{b}$ is \emph{spacelike}, so that the ``other side'' $\{r<r_\b\}$ of $\Sigma_\mathrm{b}$ does not intersect the past of $\Sigma_\b$. This is equivalent to
\begin{equation}\label{eq:spacelike-bounce}
    D-2\dot r_{\mathrm b}<0\quad\text{on }\Sigma_\b.
\end{equation}
We further assume that $\Sigma_\mathrm b$ does not contain trapped symmetry spheres, which is equivalent to 
\begin{equation}\label{eq:regular-bounce}
 D >0\quad\text{on }\Sigma_\b.
    \end{equation}

    By examining the behavior of almost-radial electromagnetic null geodesics in Reissner--Nordstr\"om, Ori proposed the following \emph{bouncing} continuation of the solution through $\Sigma_\mathrm b$: it should be as an \emph{outgoing} charged Vaidya metric. This metric takes the form 
\begin{equation}
    g_\mathrm{out}[\varpi_\mathrm{out},Q_\mathrm{out}] \doteq -\underline D(U,r)dU^2-2\,dUdr+r^2\gamma,\label{eq:Vaidya-out}
\end{equation}
where
\begin{equation*}
    \underline D(U,r)\doteq 1-\frac{2\varpi_\mathrm{out}(U)}{r}+\frac{Q_\mathrm{out}(U)^2}{r^2},
\end{equation*}
for free functions $\varpi_\mathrm{out}$ and $Q_\mathrm{out}$. The coordinates $(U,r,\vartheta,\varphi)$ are now \emph{outgoing} Eddington--Finkelstein-like. Ori defined a procedure for gluing an outgoing Vaidya metric to the ingoing Vaidya metric along $\Sigma_\b$ by demanding continuity of the second fundamental form of $\Sigma_\mathrm b$ from both sides. One sets
\begin{equation*}
    (\varpi_\mathrm{out},Q_\mathrm{out})(U)= (\varpi_\mathrm{in},Q_\mathrm{in})\circ\mathcal G^{-1}(U),
\end{equation*}
where the \emph{gluing map} $\mathcal G=\mathcal G(V)$ is determined by
\begin{equation}\label{eq:Vaidya-gluing}
    \frac{d\mathcal G}{dV}= \frac{D(V,r_\mathrm b(V))-2\dot r_\mathrm b(V)}{D(V,r_\mathrm b(V))},
\end{equation} 
up to specification of the (unimportant) initial condition. 
Notice that $\mathcal G$ is strictly monotone decreasing on account of \eqref{eq:spacelike-bounce} and \eqref{eq:regular-bounce}. It turns out that this continuation \emph{preserves the weak energy condition} through $\Sigma_\mathrm b$. We formalize this choice of extension of the Vaidya metric with the following

\begin{defn}\label{def:Oris-model} Let $\varpi_\mathrm{in}$ and $Q_\mathrm{in}$ be nondecreasing charged Vaidya seed functions such that $\spt(\dot\varpi_\mathrm{in})=\spt(\dot Q_\mathrm{in})=[V_1,V_2]$ and $r_\b$ is well-defined and positive on $[V_1,V_2]$. Assume also the conditions \eqref{eq:spacelike-bounce} and \eqref{eq:regular-bounce}. \emph{Ori's bouncing charged null dust model} consists of the ingoing charged Vaidya metric $g_\mathrm{in}[\varpi_\mathrm{in},Q_\mathrm{in}]$ on $\mathcal M_{\mathrm{in}}\doteq \{ V\in \spt (\dot \varpi_\mathrm{in} ), r \geq r_\b(V) \}\times S^2$ with spacelike, untrapped bounce hypersurface $\Sigma_{\mathrm b}^{\mathrm{in}}\doteq\{  V\in \spt (\dot \varpi_\mathrm{in} ), r = r_\b(V) \}\times S^2$ glued to the outgoing charged Vaidya metric  $ g_\mathrm{out}[ \varpi_\mathrm{in}\circ\mathcal G^{-1}, Q_\mathrm{in}\circ\mathcal G^{-1}] $ on $\mathcal M_{\mathrm{out}}\doteq \{ U\in \spt (\dot \varpi_\mathrm{out} ), r \geq r_\b\circ \mathcal G^{-1}(U) \}\times S^2$ with spacelike, untrapped bounce hypersurface $\Sigma^{\mathrm{out}}_{\mathrm b}  \doteq\{  U\in \spt(\dot \varpi_\mathrm{out} ), r = r_\b( \mathcal G^{-1}(U)) \}\times S^2$ along the map $\mathcal G\times\id_r : \Sigma_{\mathrm b}^{\mathrm{in}} \to \Sigma^{\mathrm{out}}_{\mathrm b}$ defined by \eqref{eq:Vaidya-gluing}. Outside the support of the dust, Ori's bouncing charged null dust model extends by attaching two Reissner--Nordström solutions with parameters $(\varpi_1,Q_1)\doteq(\varpi_\mathrm{in},Q_\mathrm{in})(V_1)$ and $(\varpi_2,Q_2)\doteq (\varpi_\mathrm{in},Q_\mathrm{in})(V_2)$ as depicted in \cref{fig:Ori-dust}.
\end{defn}

The model can be generalized to allow for multiple beams of dust by iterating the above definition in the obvious manner. 

 \begin{figure}
\centering{
\def\svgwidth{16pc}
\begingroup%
  \makeatletter%
  \providecommand\color[2][]{%
    \errmessage{(Inkscape) Color is used for the text in Inkscape, but the package 'color.sty' is not loaded}%
    \renewcommand\color[2][]{}%
  }%
  \providecommand\transparent[1]{%
    \errmessage{(Inkscape) Transparency is used (non-zero) for the text in Inkscape, but the package 'transparent.sty' is not loaded}%
    \renewcommand\transparent[1]{}%
  }%
  \providecommand\rotatebox[2]{#2}%
  \newcommand*\fsize{\dimexpr\f@size pt\relax}%
  \newcommand*\lineheight[1]{\fontsize{\fsize}{#1\fsize}\selectfont}%
  \ifx\svgwidth\undefined%
    \setlength{\unitlength}{150.30424704bp}%
    \ifx\svgscale\undefined%
      \relax%
    \else%
      \setlength{\unitlength}{\unitlength * \real{\svgscale}}%
    \fi%
  \else%
    \setlength{\unitlength}{\svgwidth}%
  \fi%
  \global\let\svgwidth\undefined%
  \global\let\svgscale\undefined%
  \makeatother%
  \begin{picture}(1,0.87322879)%
    \lineheight{1}%
    \setlength\tabcolsep{0pt}%
    \put(0,0){\includegraphics[width=\unitlength,page=1]{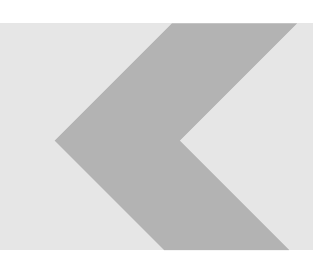}}%
    \put(0.10022879,0.61470247){\color[rgb]{0,0,0}\makebox(0,0)[lt]{\lineheight{1.25}\smash{\begin{tabular}[t]{l}$\Sigma_\mathrm b$\end{tabular}}}}%
    \put(0.0390614,0.16174226){\color[rgb]{0,0,0}\makebox(0,0)[lt]{\lineheight{1.25}\smash{\begin{tabular}[t]{l}RN $(\varpi_1,Q_1)$\end{tabular}}}}%
    \put(0.70271537,0.45115523){\color[rgb]{0,0,0}\makebox(0,0)[lt]{\lineheight{1.25}\smash{\begin{tabular}[t]{l}RN $(\varpi_2,Q_2)$\end{tabular}}}}%
    \put(0,0){\includegraphics[width=\unitlength,page=2]{Ori-2.pdf}}%
    \put(0.61874526,0.01496965){\color[rgb]{0,0,0}\makebox(0,0)[lt]{\lineheight{1.25}\smash{\begin{tabular}[t]{l}$g_\mathrm{in}[\varpi_\mathrm{in},Q_\mathrm{in}]$\end{tabular}}}}%
    \put(0,0){\includegraphics[width=\unitlength,page=3]{Ori-2.pdf}}%
    \put(0.5788259,0.82831962){\color[rgb]{0,0,0}\makebox(0,0)[lt]{\lineheight{1.25}\smash{\begin{tabular}[t]{l}$g_\mathrm{out}[\varpi_\mathrm{out},Q_\mathrm{out}]$\end{tabular}}}}%
  \end{picture}%
\endgroup%
}
\caption{Penrose diagram of Ori's bouncing charged null dust model. The geometry of the beams is described by the ingoing and outgoing Vaidya metrics, $g_\mathrm{in}$ and $g_\mathrm{out}$, which are related by the gluing map $\mathcal G$. The spacetime to the left and right of the bouncing beam is described by the Reissner--Nordstr\"om solution, with parameters $(\varpi_1,Q_1)$ and $(\varpi_2,Q_2)$ with $\varpi_1<\varpi_2$ and $Q_1<Q_2$. The endpoints of $\Sigma_\mathrm b$ correspond to symmetry spheres in these Reissner--Nordstr\"om spacetimes with radii $r_1<r_2$. In this diagram, the $V$ coordinate is normalized according to the ingoing solution. We have depicted here the case of a \emph{totally geodesic} bounce hypersurface $\Sigma_\b$ and the outgoing beam is exactly the time-reflection of the ingoing beam.}
\label{fig:Ori-dust}
\end{figure}

\subsection{The radial parametrization of bouncing charged null dust spacetimes}\label{sec:radial}

It is not immediately clear that interesting seed functions $\varpi_\mathrm{in}$ and $Q_\mathrm{in}$ satisfying the requirements of \cref{def:Oris-model} exist. Therefore, it is helpful to \emph{directly prescribe} the geometry of $\Sigma_\mathrm b$ and the dust along it in the following sense. Given a spacetime as in \cref{fig:Ori-dust}, we can parametrize $\Sigma_\mathrm b$ by the area-radius function $r$. Then the renormalized Hawking mass $\varpi$ and charge $Q$, which are gauge-invariant quantities, can be viewed as functions of $r$ on $\Sigma_\mathrm b$, and we wish to prescribe these functions. We will also prescribe $\Sigma_\b$ to be \emph{totally geodesic}. While not essential, this condition greatly simplifies \cref{prop:radial-parametrization} below and will later play a key role in our Vlasov construction in \cref{sec:Vlasov-beams}.

\begin{defn} \label{defn:radial-parametrization}
Let $\mathcal P$ denote the set of points $(r_1,r_2,\varpi_1,\varpi_2,Q_1,Q_2)\in \Bbb R^6_{\ge 0}$
 subject to the conditions
 \begin{gather}
     \label{eq:P-1-intro}0<r_{1}< r_{2}, \quad 
    Q_1 < Q_2,\quad  \varpi_1 < \varpi_2\\
    \label{eq:P-2-intro}   2r_1\varpi_1\ge Q_1^2,\\
    \label{eq:P-3-intro}     2 r_1 (\varpi_2 - \varpi_1) < Q_2^2 - Q_1^2
         < 2 r_2 (\varpi_2 - \varpi_1), 
         \\
     \label{eq:P-4-intro}   \min_{r\in[r_1,r_2]} (r_1 r^2 - 2 \varpi_1  r_1 r + Q_1^2 r - Q_2^2 r + Q_2^2 r_1) >0.
 \end{gather}
Elements of $\mathcal P$ will typically be denoted by the letter $\alpha$ and are called \emph{admissible parameters}. Let $\mathfrak V$ denote the set of triples $(\alpha,\check\varpi,\check Q)\in \mathcal P\times C^\infty\big([0,\infty)\big)\times C^\infty\big([0,\infty)\big)$ such that the functions $\check\varpi=\check\varpi(r)$ and $\check Q=\check Q(r)$ are monotone increasing and satisfy
\begin{gather}
  \label{eq:V-1-intro}  \operatorname{spt}(\check\varpi') = \operatorname{spt}(\check Q') = [r_1,r_2],\\
 \label{eq:V-2-intro} \frac{d}{dr}\check Q^2(r)=2r\frac{d}{dr}\check\varpi(r),\\
    \label{eq:V-3-intro}    \check \varpi(r_1) = \varpi_1,\quad\check Q(r_1) = Q_1, \quad  \check\varpi(r_2) = \varpi_2,\quad\check Q(r_2)=Q_2,
     \end{gather} where $'$ denotes differentiation with respect to $r$. 
\end{defn}

\begin{rk}
    In the proof of \cref{thm:main} we will employ the \emph{regular center} parameter space $\mathcal P_\Gamma$, consisting of those $\alpha\in \mathcal P$ with $\varpi_1=Q_1=0$.
\end{rk}

\begin{prop}[Radial parametrization of bouncing charged null dust]\label{prop:radial-parametrization}
    Let $(\alpha,\check \varpi,\check Q)\in \mathfrak V$, and define strictly monotone functions $\mathcal V, \mathcal U:[r_1,r_2]\to\Bbb R$ by 
\begin{equation*}
      \mathcal V(r)=-\mathcal U(r)=\int_{r_1}^r D(r')^{-1}\, dr',
\end{equation*} where $D(r)\doteq 1-\frac{2\check\varpi(r)}{r}+\frac{\check Q^2(r)}{r^2}$. Then:
\begin{enumerate}
        \item The seed functions $(\varpi_\mathrm{in},Q_\mathrm{in})\doteq (\check \varpi,\check Q)\circ\mathcal V^{-1}$ and $(\varpi_\mathrm{out},Q_\mathrm{out})\doteq (\check \varpi,\check Q)\circ \mathcal U^{-1}$ define a bouncing charged null dust spacetime as in \cref{def:Oris-model} with  gluing map $\mathcal G(V)=-V$ and bounce radius $r_\b(V) = \mathcal V^{-1}(V)$. 
        \item The bounce hypersurface $\Sigma_\mathrm b$ is spacelike and untrapped. With the setup as in \cref{fig:Ori-dust}, the left edge of $\Sigma_\b$ has area-radius $r_1$ and Reissner--Nordstr\"om parameters $(\varpi_1,Q_1)$ and the right edge has area-radius $r_2$ and Reissner--Nordstr\"om parameters $(\varpi_2,Q_2)$. The Hawking mass $m$ is nonnegative on $\Sigma_\mathrm b$.

\item The bounce hypersurface $\Sigma_\b$ is totally geodesic with respect to $g_\mathrm{in}$ and $g_\mathrm{out}$. 
      \end{enumerate} 
    \end{prop}
    \begin{proof}
        We must check that $\varpi_\mathrm{in}\doteq \check \varpi \circ\mathcal V^{-1}$ and $Q_\mathrm{in}\doteq \check Q\circ\mathcal V^{-1}$ satisfy the assumptions of \cref{def:Oris-model}. Using the chain rule and \eqref{eq:V-2-intro}, we compute
        \begin{equation*}
            r_\b(V)=\frac{\check Q(\mathcal V^{-1}(V))\check Q'(\mathcal V^{-1}(V))}{\check\varpi'(\mathcal V^{-1}(V))}= \mathcal V^{-1}(V).
        \end{equation*}
        Differentiating, we obtain
        \begin{equation}
            \dot r_\b(V)= D(\mathcal V^{-1}(V))=D(V,r_\b(V)),\label{eq:r-b-thm}
        \end{equation}
        which implies that $d\mathcal G/dV=-1$. To prove \eqref{eq:regular-bounce}, we show that $D(r)>0$ for $r\in[r_1,r_2]$. Integrating \eqref{eq:V-2-intro} in $r$ and integrating by parts yields
        \begin{equation}
            \check\varpi(r) = \varpi_1+\frac 12\int_{r_1}^r\frac{1}{r'}\frac{d}{dr'}\check Q^2(r')\,dr'=\varpi_1+\frac{\check Q^2(r)}{2r}-\frac{Q_1^2}{2r_1}+\frac{1}{2}\int_{r_1}^r \frac{\check Q^2(r')}{r'^2}\,dr'.\label{eq:defn-varpi-proof}
        \end{equation}
        Using condition \eqref{eq:P-4-intro} and $\check Q\le Q_2$, we then find
        \begin{align}
        \nonumber   D(r)=1 - \frac{2\check\varpi(r)}{r} + \frac{\check Q^2(r)}{r^2}   &=  1 - \frac{2\varpi_1}{r} + \frac{Q_1^2}{r_1 r} - \frac 1r \int_{r_1}^{r} \frac{\check Q^2(r')}{r'^2} \,dr'\\
           &> \frac{1}{r_1r^2} \left( r_1 r^2 - 2 \varpi_1 r r_1 + Q_1^2 r - Q_2^2 r + Q_2^2 r_1\right) >0\label{eq:dust-untrapped}
        \end{align}
        for $r\in [r_1,r_2]$. This proves \eqref{eq:regular-bounce} and since $D-2\dot r_\b=-D$, also proves \eqref{eq:spacelike-bounce}. Condition \eqref{eq:P-2-intro} implies that the Hawking mass is nonnegative at $r_1$. Finally, that $\Sigma_\b$ is a totally geodesic hypersurface is shown by directly computing its second fundamental form and using \eqref{eq:r-b-thm}. 
    \end{proof}

     The definition of $\mathfrak V$ involves many more conditions than just \eqref{eq:spacelike-bounce} and \eqref{eq:regular-bounce} alone, but it turns out that these are relatively easy to satisfy. In particular, we have:
       
     \begin{prop}\label{prop:KU23c-existence}
               The natural projection map $\mathfrak V\to \mathcal P$ admits a smooth section $\varsigma:\mathcal P\to \mathfrak V$. In other words, given any smooth family of parameters in $\mathcal P$ we may associate a smooth family of bouncing charged null dust spacetimes attaining those parameters, with totally geodesic bounce hypersurfaces. 
     \end{prop}
\begin{rk}\label{rk:varsigma}
In the remainder of the paper, we fix the choice of section to be the one constructed in the proof below.
\end{rk}
     
 \begin{proof}
Define a smooth, surjective function $\psi:\Bbb R^2\to (0,1)$ by
\begin{align*}
    \psi(x, \xi)=\frac{1}{1 + \exp\left[ -   (x-\xi )  e^{(x-\xi)^2}  \right]}.
\end{align*}
Note that for each fixed $\xi\in\Bbb R$ the function $x\mapsto \psi(x,\xi)$ is strictly monotone increasing and surjective. Moreover, for  $x\in \mathbb R$ we have $\psi(x,\xi) \to 0$ as $\xi \to \infty$ and $\psi(x,\xi) \to 1 $ as $\xi \to -\infty$. 

We now define the   function    $ \check Q: (r_1,r_2)\times \mathbb R \times \mathcal P \to \mathbb R$ by
\begin{equation}\label{eq:seed-data-lemma}
    \check Q(r,\xi,\alpha) \doteq Q_1+ (Q_2-Q_1) \, \psi \left( \log \left(\frac{r-r_1}{r_2-r}\right), \xi\right).
    \end{equation}  
    By construction of $\psi$, the function $  \check Q$  extends smoothly to $[0,\infty)\times \mathbb R \times \mathcal P$   by setting  $  \check Q(r,\xi,\alpha) =Q_1$ for $ 0 \leq r\leq r_1$ and $ \check Q(r,\xi,\alpha) =Q_2$ for $r\geq r_2$. 

With our family of candidate $\check Q$'s at hand, we aim to satisfy the constraint $\check \varpi (r_2) = \varpi_2$, where $\check \varpi(r)$ is defined by \eqref{eq:defn-varpi-proof}. Consider the smooth map $\Pi : (\xi, \alpha)  \in \mathbb R  \times \mathcal P\to \mathbb R$ defined by  
\begin{equation*}
    \Pi (\xi,\alpha ) \doteq \varpi_1+
        \frac{ Q_{2}^2}{2r_2} -\frac{Q_1^2}{2r_1} +  \int_{r_1}^{r_2} \frac{\check Q^2(r',\xi,\alpha)}{2r'^2}\,dr'   .   
\end{equation*}
Since $\psi$ satisfies $\frac{\partial \psi}{\partial \xi} <0$ on $\mathbb R^2$,  we have that $\frac{\partial \Pi}{\partial \xi} < 0$ on $\mathbb R  \times  \mathcal P$.
Moreover, using the pointwise limits of $\psi$, a direct computation gives \begin{equation*}
\lim_{\xi \to \infty} \Pi(\alpha, \xi) =\varpi_1+\frac{Q_2^2}{2 r_2} -\frac{Q_1^2}{2r_2}   ,\quad \lim_{\xi \to -\infty} \Pi(\alpha, \xi) =\varpi_1+\frac{Q_2^2}{2 r_1} -\frac{Q_1^2}{2r_1}. \end{equation*}
By condition \eqref{eq:P-3-intro}, this implies that 
   \begin{equation*}\lim_{\xi \to \infty} \Pi(\alpha, \xi) <  \varpi_2 < \lim_{\xi \to -\infty} \Pi(\alpha, \xi).\end{equation*} Thus, the intermediate value theorem and the fact that $\frac{\partial \Pi}{\partial \xi} < 0$ show that there exists a unique $\xi(\alpha)$ such that $\Pi(\alpha,\xi(\alpha))= \varpi_2$. 
Moreover, a direct consequence of the implicit function theorem is that the assignment $\mathcal P \ni \alpha \mapsto \xi(\alpha) \in \mathbb R$ is smooth. The above construction shows that the functions $\check Q(r,\xi(\alpha),\alpha)$ and $\check \varpi$ satisfy all required properties.
 \end{proof}

The set $\mathcal P$ is defined by simple polynomial relations and includes many interesting examples as we will see in the next two sections.

\subsection{Extremal critical collapse in Ori's model}\label{sec:Ori-ECC}

The first application of \cref{prop:radial-parametrization,prop:KU23c-existence} is the construction of one-parameter families of bouncing charged null dust spacetimes exhibiting \emph{extremal critical collapse}. We first show that the regular center parameter space $\mathcal P_\Gamma$ contains elements with arbitrary final Reissner--Nordstr\"om parameters:

\begin{lem}\label{lem:existence-P}
    Let $\varpi_2,Q_2>0$. Then there exist $0<r_1<r_2$ such that $(r_1,r_2,0,\varpi_2,0,Q_2)\in\mathcal P_\Gamma$. If $\varpi_2\ge Q_2$, then $r_2$ can moreover be chosen so that $r_2<\varpi_2-\sqrt{\varpi_2^2-Q_2^2}$.
\end{lem}
\begin{proof}
    Let 
    \begin{equation*}
        r_1\doteq Q_2\left(\frac{Q_2}{2\varpi_2}-\varepsilon\right),\quad r_2\doteq Q_2\left(\frac{Q_2}{2\varpi_2}+\varepsilon\right),
    \end{equation*}
    where $\epsilon>0$ is a small parameter to be determined. With this choice, \eqref{eq:P-3-intro} is clearly satisfied. Let $p(r)\doteq r_1r^2-Q_2^2r+Q_2^2r_1$ and observe that 
    \begin{equation*}
        \lim_{\varepsilon\to 0}p\left(\frac{Q_2^2}{2\varpi_2}\right)= \frac{Q_2^6}{8\varpi_2^3}>0.
    \end{equation*}
    It follows that \eqref{eq:P-4-intro} is satisfied for $\varepsilon$ sufficiently small. If $x\ge 1$, then $(2x)^{-1}<x-\sqrt{x^2-1}$,
    so taking $\varepsilon$ perhaps smaller ensures that $r_2<\varpi_2-\sqrt{\varpi_2^2-Q_2^2}$.
\end{proof}

Using this, we can show that Ori's model exhibits extremal critical collapse. Compare the following theorem with \cref{thm:main} and refer to \cref{fig:bounce-3} for Penrose diagrams.

\begin{figure}
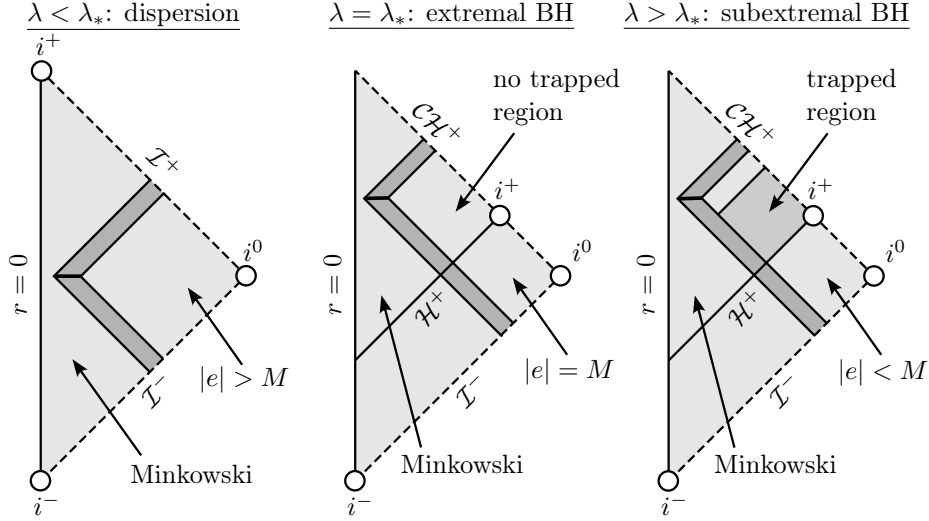

\centering{
\def\svgwidth{30pc}
}
\caption{Penrose diagrams of extremal critical collapse in Ori's bouncing charged null dust model. Compare with \cref{fig:bounce-2}. In \cref{thm:ECC-Ori}, $\lambda_*=1$.}\label{fig:bounce-3}
\end{figure}

\begin{thm}
\label{thm:ECC-Ori}
For any $M>0$ and fundamental charge $\mathfrak e\in\Bbb R\setminus\{0\}$, there exist a small parameter $
\delta  >0$ and a smooth two-parameter family of regular center parameters $\{\alpha_{\lambda,M'}\}\subset\mathcal P_\Gamma$ for $\lambda\in(0,2]$ and $M'\in[M-\delta,M+\delta]$
such that the two-parameter family of bouncing charged null dust spacetimes $\{ \mathcal D_{\lambda,M'}\}$, obtained by applying \cref{prop:radial-parametrization} to $\varsigma(\alpha_{\lambda,M'})$, has the following properties:

\begin{enumerate}
    
\item For $0<\lambda<1$, $\mathcal D_{\lambda,M'}$ is isometric to Minkowski space for all sufficiently late retarded times $u$ and hence future causally geodesically complete. In particular, it \ul{does not contain a black hole or naked singularity}, and for $\lambda<1$ sufficiently close to $1$, sufficiently large advanced times $v\ge v_0$ and sufficiently small retarded times $u \leq u_0$, the spacetime is isometric to an appropriate causal diamond in a \underline{superextremal} Reissner--Nordstr\"om solution. Moreover,  $\mathcal D_{\lambda,M'}$ converges smoothly to Minkowski space as $\lambda \to 0$.

     \item $\lambda=1$ is critical: $\mathcal D_{1,M'}$ contains a nonempty black hole region $\mathcal{BH}$ and for sufficiently large advanced times $v\ge v_0$, the domain of outer communication, including the event horizon $\mathcal H^+$, is isometric to that of an \underline{extremal} Reissner--Nordstr\"om solution of mass $M'$. The spacetime contains \ul{no trapped surfaces}. 
     
\item For $1<\lambda\leq 2$, $\mathcal D_{\lambda,M'}$ contains a nonempty black hole region $\mathcal{BH}$ and for sufficiently large advanced times $v\ge v_0$, the domain of outer communication, including the event horizon $\mathcal H^+$, is isometric to that of a  \underline{subextremal} Reissner--Nordstr\"om solution. The spacetime contains an \ul{open set of trapped surfaces}. 
   \end{enumerate}
In addition, for every $\lambda\in [0,2]$, $\mathcal D_{\lambda,M'}$ is isometric to Minkowski space for sufficiently early advanced time and near the center $\{r=0\}$ for all time, and possesses complete null infinities $\mathcal I^+$ and $\mathcal I^-$.
\end{thm}

\begin{proof} Using \cref{lem:existence-P}, choose $0<r_1<r_2<r_-(4M,2M)$ such that $(r_1,r_2,0,4M,0,2M)\in\mathcal P_\Gamma$. We  consider  
\begin{equation}\alpha_{\lambda,M'}\doteq (r_1,r_2,0,\lambda^2 M',0,\lambda M')\label{eq:alpha-lambda-M'}\end{equation}
and note that $\alpha_{\lambda,M'}$ lies in $\mathcal P_\Gamma$ for $|\lambda -2|$ sufficiently small and $|M-M'|\leq \delta $ sufficiently small by the openness of the conditions defining $\mathcal P_\Gamma$.  Moreover, from  the scaling properties of \eqref{eq:P-3-intro} and the monotonicity of \eqref{eq:P-4-intro}, we observe that $\alpha_{\lambda,M'}\in \mathcal P_{\Gamma}$ for all $0<\lambda \leq2$ and $|M-M'|\leq \delta $. 

After applying \cref{prop:radial-parametrization} to $\varsigma(\alpha_{\lambda,M'})$ for $\lambda>0$, it remains only to show that $\mathcal D_{\lambda, M'}$ extends smoothly to Minkowski space as $\lambda \to 0$. Indeed, a direct inspection of the proof of \cref{prop:KU23c-existence} shows that $\xi(\alpha_{\lambda,M'})$ is independent of $\lambda$, so that the function $r\mapsto \check Q(r,\xi(\alpha_{\lambda,M'}),\alpha_{\lambda,M'})$ defined in \eqref{eq:seed-data-lemma} converges smoothly to the function $\check Q\equiv 0$ as $\lambda\to 0$. Therefore, $\check\varpi$ also converges smoothly to the zero function and hence $\mathcal D_{\lambda, M'}$ converges smoothly to Minkowski space as $\lambda\to 0$. \end{proof}

The construction in the proof of \cref{thm:main} can be thought of as a global-in-time desingularization of this family of dust solutions. In fact, we will make essential use of the one-parameter family $\{\varsigma(\alpha_{\lambda,M'})\}$ when constructing initial data for the Einstein--Maxwell--Vlasov system.

\subsection{A counterexample to the third law of black hole thermodynamics in Ori's model}\label{sec:third-law-dust}

 \begin{figure}
\centering{
\def\svgwidth{14pc}
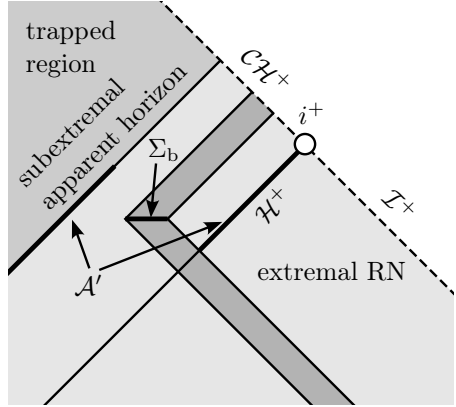}
\caption{Penrose diagram of a counterexample to the third law of black hole thermodynamics in Ori's charged null dust model, \cref{thm:third-law-dust}. Note that the bounce $\Sigma_\mathrm{b}$ lies behind the extremal event horizon $\mathcal H^+$ since $r_2<\varpi_2$ and $\varpi_2$ is the area-radius of the extremal horizon. The broken curve $\mathcal A'$ is the outermost apparent horizon of the spacetime. The disconnectedness of $\mathcal A'$ is necessary in third law violating spacetimes---see the discussion in Section 1.4.3 of \cite{KU22}. A crucial feature of this counterexample is that $\Sigma_\mathrm{b}$ lies strictly between the (initially) subextremal apparent horizon and the (eventually) extremal event horizon. Compare with \cref{fig:third-law-Vlasov}.}
\label{fig:third-law-dust}
\end{figure}

Using the radial parametrization, we can now give a very simple disproof of the third law:

\begin{thm}
\label{thm:third-law-dust}
There exist bouncing charged null dust spacetimes that violate the third law of black hole thermodynamics: a subextremal Reissner--Nordstr\"om apparent horizon can evolve into an extremal Reissner--Nordstr\"om event horizon in finite advanced time due to the incidence of charged null dust. 
\end{thm}
\begin{proof}
 Apply \cref{prop:radial-parametrization,prop:KU23c-existence} to $(r_1,r_2,\varpi_1,\varpi_2,Q_1,Q_2)\in\mathcal P$ satisfying $r_2<\varpi_2$, $Q_1<\varpi_1$, and $Q_2=\varpi_2$. For example, one may take $(0.85,0.88,0.56,1,0.5,1)\in\mathcal P$. See \cref{fig:third-law-dust}.
\end{proof}

\begin{rk} Since the energy-momentum tensor remains bounded in Ori's model and the weak energy condition is satisfied, this is indeed a counterexample to Israel's formulation of the third law \cite{Israel-third-law}.
\end{rk}

The counterexample in \cref{thm:third-law-dust} explicitly displays the disconnectedness of the outermost apparent horizon which is also present in our charged scalar field counterexamples to the third law \cite{KU22}. Note that the bouncing dust beam does not cross the subextremal apparent horizon, as is required by \eqref{eq:regular-bounce}.

\begin{rk}
    In the example depicted in \cref{fig:third-law-dust}, the parameters $\check \varpi$ and $\check Q$ satisfy $\check\varpi(r)<\check Q(r)$ for $r\in (r_2-\ve,r_2)$ and some $\ve>0$. Indeed, the ODE \eqref{eq:V-2-intro} implies 
    \begin{equation*}
       \check Q'(r)=\frac{r}{\check Q(r)}\check \varpi'(r)<\check \varpi'(r)
    \end{equation*}
    near $r_2$, where we have used $r_2<Q_2$. The possibility (in fact, apparent inevitability) of the Vaidya parameters being superextremal right before extremality is reached seems to have been overlooked in the literature \cite{Sullivan-Israel, fairoos2017massless}.\footnote{The paper \cite{fairoos2017massless} reexamines the third law in light of Ori's paper \cite{Ori91}, but always makes the assumption that the parameters satisfy $Q(V)<\varpi(V)$ right before extremality. Therefore, they seemingly reaffirm the third law!}
\end{rk}

\begin{rk}
    If one applies the old ``standard interpretation'' of the ingoing Vaidya metric from \cite{Sullivan-Israel,lake1991structure} to the seed functions $\varpi_\mathrm{in}(V)$ and $Q_\mathrm{in}(V)$ constructed in the proof of \cref{thm:third-law-dust}, one sees that the beam will hit the subextremal apparent horizon with a negative energy density, which is consistent with \cite{Sullivan-Israel}.
\end{rk}

    Using methods from the proof of \cref{thm:main}, the dust spacetimes in \cref{thm:third-law-dust} can be ``desingularized'' to smooth Einstein--Maxwell--Vlasov solutions. The desingularized solutions can also be chosen to have the property that the matter remains strictly between the subextremal apparent horizon and the event horizon and we refer back to \cref{sec:EMV-third-law}.

\subsection{Issues with the bouncing charged null dust model}\label{sec:shortcomings}

While \cref{prop:radial-parametrization} allows us to construct these interesting examples, the bouncing charged null dust model is unsatisfactory and we should seek to replace it for several reasons:
\begin{enumerate}
    \item The model does not arise as a well-posed initial value problem for a system of PDEs. Pasting the ingoing and outgoing Vaidya solutions together is a deliberate surgery procedure that only works for seed functions $\varpi_\mathrm{in}$ and $Q_\mathrm{in}$ satisfying several nontrivial and nongeneric conditions. 
    \item The solutions are generally not smooth along $\Sigma_\mathrm b$, nor along any cone where $Q=0$ (recall \cref{rk:Mink-blowup}). The fluid density $\rho$ is unbounded along $\Sigma_\b$ and the number current $N=\rho k$ is discontinuous across $\Sigma_\mathrm b$.
      \item Null dust is ill-posed once the dust reaches the center of symmetry \cite{Moschidis2017-lx}.
\end{enumerate}

 Nevertheless, we will show in the course of this paper that the bouncing charged null dust model can be well-approximated (near $\Sigma_\mathrm b$), in a precise manner, by smooth solutions of the Einstein--Maxwell--Vlasov system. See already \cref{sec:weak-proof}.

\subsection{The formal radial charged null dust system in double null gauge}\label{sec:outgoing-dust-equations}

In order to precisely phrase the manner in which Einstein--Maxwell--Vlasov approximates bouncing charged null dust, as well as to motivate the choice of Vlasov initial data, we now reformulate Ori's  model in double null gauge. Following Moschidis \cite{Moschidis2017-lx,moschidis2020proof}, we reformulate the system by treating $N$ and $T$ as the fundamental variables. By eliminating the fluid variables $k$ and $\rho$, we can view the ingoing and outgoing phases as two separate well-posed initial value problems, with data posed along the bounce hypersurface. This helpfully suppresses the issue of blowup of $\rho$ on $\Sigma_\b$.

\begin{defn}\label{def:outgoing-dust-equations}
The \emph{spherically symmetric formal outgoing charged null dust} model for particles of fundamental charge $\mathfrak e\in\Bbb R\setminus\{0\}$ consists of a smooth spherically symmetric charged spacetime $(\mathcal Q,r,\Omega^2,Q)$ and two nonnegative smooth functions $N^v$ and $T^{vv}$ on $\mathcal Q.$

The system satisfies the wave equations
    \begin{align}
 \label{eq:r-wave-dust} \partial_u\partial_v r&= - \frac{\Omega^2}{2r^2}\left(m-\frac{Q^2}{2r}\right), \\
\label{eq:Omega-wave-dust}    \partial_u\partial_v \!\log\Omega^2&=\frac{\Omega^2m}{r^3}-\frac{\Omega^2Q^2}{r^4},
\end{align}
the Raychaudhuri equations 
\begin{align}
\label{eq:Ray-u-dust}  \partial_u\left(\frac{\partial_ur}{\Omega^2}\right)  &=-\tfrac 14r\Omega^2T^{vv},\\
\label{eq:Ray-v-dust}  \partial_v\left(\frac{\partial_vr}{\Omega^2}\right)  &=0,
\end{align}
and the Maxwell equations
\begin{align}
\label{eq:Max-u-dust}\partial_u Q&=-\tfrac 12\mathfrak e r^2\Omega^2N^v,\\
\label{eq:Max-v-dust}\partial_v Q&=0.
\end{align}
The number current satisfies the conservation law
\begin{equation}
    \label{eq:dust-conservation-1} \partial_v(r^2\Omega^2N^v)=0
\end{equation}
and the energy-momentum tensor satisfies the Bianchi equation
\begin{equation}
    \label{eq:dust-conservation-2} \partial_v(r^2\Omega^4T^{vv})= +\mathfrak e\Omega^4 QN^v.
\end{equation}
\end{defn}

In the outgoing model, we may think of $N^u$, $T^{uu}$, and $T^{uv}$ to just be defined as identically zero. From \eqref{eq:r-wave-dust} and \eqref{eq:Ray-u-dust}--\eqref{eq:Max-v-dust} one easily derives 
\begin{equation}
    \partial_um=-\tfrac 12 r^2\Omega^2 T^{vv} \partial_vr+\frac{Q^2}{2r^2}\partial_ur,\quad \partial_vm=\frac{Q^2}{2r^2}\partial_vr,
\end{equation}
\begin{equation}
    \partial_u\varpi = -\tfrac 12 r^2\Omega^2T^{vv}\partial_vr-\tfrac 12\mathfrak e r\Omega^2QN^v,\quad\partial_v\varpi=0.
\end{equation}
Furthermore, if we set $k^v\doteq T^{vv}/N^v$, then
\begin{equation}
   k^v \partial_v k^v +\partial_v{\log\Omega^2}(k^v)^2 = +\mathfrak e\frac{Q}{r^2}k^v,\label{eq:Euler-2}
\end{equation}
which is the spherically symmetric version of \eqref{eq:intro-dust-eq-3}  for the vector field $k\doteq k^v\partial_v$. The energy density of the fluid is defined by $\rho\doteq (N^v)^2/T^{vv}$ whenever the denominator is nonvanishing. 

\begin{defn}
    The \emph{spherically symmetric formal ingoing charged null dust} model for particles of fundamental charge $\mathfrak e\in\Bbb R\setminus\{0\}$ consists of a smooth spherically symmetric charged spacetime $(\mathcal Q,r,\Omega^2,Q)$ and two nonnegative smooth functions $N^u$ and $T^{uu}$ on $\mathcal Q$. The system satisfies the same equations as the ingoing system with $u\leftrightarrow v$ and the opposite sign in front of $N^u$.
\end{defn}

In the ingoing case, $k^u\doteq T^{uu}/N^u$ and $\rho\doteq (N^u)^2/T^{uu}$.

\begin{rk}
   By \eqref{eq:Euler-2}, these formal systems define solutions of the Einstein--Maxwell--charged null dust system (see \cref{def:charged-null-dust}) whenever $k$ and $\rho$ are well-defined. 
\end{rk}

\begin{rk}
    Inspection of \eqref{eq:dust-conservation-2} reveals that $T^{vv}$ can reach zero in finite backwards time. If one were to continue the solution further, $T^{vv}$ could become negative, which shows that the formal system actually reproduces the old ``standard interpretation'' of the charged Vaidya metric discussed in \cite{Ori91}. As we will see, because the dominant energy condition holds in the Einstein--Maxwell--Vlasov model, only dust solutions with $T^{uu},T^{vv}\ge 0$ will arise as limiting spacetimes, confirming Ori's heuristic picture discussed in \cite{Ori91}.
\end{rk}

\subsubsection{The Cauchy problem for outgoing formal charged null dust}

Mirroring the treatment of time-symmetric\footnote{In the Vlasov case, time symmetry referred to both the geometry of the spacelike part of the initial data hypersurface and the matter configuration. Since purely outgoing dust is clearly not time symmetric, it refers here only to the geometry of the spacelike part of the initial data hypersurface.} seed data for the Einstein--Maxwell--Vlasov system in \cref{sec:general-data}, we make the following definition:
\begin{defn}
    A \emph{time-symmetric seed data set} $\mathcal S_\d\doteq (\mathring{\mathcal N}^v,\mathring{\mathcal T}^{vv},r_2,\mathfrak e)$ for the spherically symmetric formal outgoing charged null dust system consists of real numbers $r_2\in \Bbb R_{>0}$ and $\mathfrak e\in\Bbb R\setminus\{0\}$, together with nonnegative compactly supported smooth functions $\mathring{\mathcal N}^v$ and $\mathring{\mathcal T}^{vv}$ with support contained in $(0,r_2]$.
\end{defn}

In the dust case, we define $\mathring m$ and $\mathring Q$ on $[0,r_2]$ with $\mathring m (0) = \mathring Q (0) =0$ by solving 
\begin{align}
 \label{eq:dust-constraint-1}   \frac{d}{dr}\mathring m&=\frac{r^2}{4}\left(1-\frac{2\mathring m}{r}\right)^{-2}\mathring{\mathcal T}^{vv}+\frac{\mathring Q^2}{2r^2},\\
  \frac{d}{dr}\mathring Q  &=\frac 12\mathfrak er^2\left(1-\frac{2\mathring m}{r}\right)^{-2}\mathring{\mathcal N}^v,\label{eq:dust-constraint-2}
\end{align}
provided $2\mathring m<r$ on $[0,r_2]$. The remaining definitions from the Vlasov case, in particular \cref{def:normalized-defn}, can be carried over to dust with the obvious modification that $N^v=T^{vv}=0$ along $\Gamma$.\footnote{Since the dust here is purely outgoing, we do not have to be concerned about dust going into $\Gamma$.}

\begin{prop}\label{prop:outgoing-dust-existence}
    Let $\mathcal S_\d$ be an untrapped time-symmetric seed data set for outgoing dust. Then there exists a unique global smooth solution of the formal outgoing charged null dust system $(r,\Omega^2,Q,N^v,T^{vv})$ on $\mathcal C_{r_2}$ attaining the seed data. 
\end{prop}
\begin{rk}
    Let $r_1\doteq \inf(\spt \mathring{\mathcal N}^{v}\cup \spt \mathring{\mathcal T}^{vv})$. Then $(r,\Omega^2,Q)$ is isometric to Minkowski space for $u\ge -r_1$.
\end{rk}
\begin{proof}
This can be proved by applying a suitable coordinate transformation to a suitable outgoing charged Vaidya metric. However, it is instructive to give a direct proof using the evolution equations.

 We pose initial data
\begin{equation*}
    \mathring r(r)=r, \quad \mathring\Omega^2(r)=\left(1-\frac{2\mathring m}{r}\right)^{-1},\quad  \mathring Q(r)=\int_0^{r}\frac 12\mathfrak er^2\left(1-\frac{2\mathring m}{r}\right)^{-2}\mathring{\mathcal N}^v\,dr',
\end{equation*}
and for derivatives according to \cref{def:normalized-defn}, for the equations  \eqref{eq:r-wave-dust}, \eqref{eq:Omega-wave-dust}, and \eqref{eq:Max-v-dust}. By a standard iteration argument, this determines the functions $(r,\Omega^2,Q)$ uniquely. The existence of a \emph{global} development is strictly easier than the corresponding proof in \cref{prop:approx-1} once the rest of the system has been derived and is omitted.  We now \emph{define}
\begin{equation*}
      N^v\doteq -\frac{2}{\mathfrak e r^2\Omega^2}\partial_uQ ,\quad
    T^{vv}\doteq -\frac{4}{r\Omega^2}\partial_u\left(\frac{\partial_ur}{\Omega^2}\right)
\end{equation*}
and aim to show that the rest of the equations in \cref{def:outgoing-dust-equations} are satisfied. 

To prove \eqref{eq:dust-conservation-1}, simply rearrange the definition of $N^v$ and use \eqref{eq:Max-v-dust}. Note that the definition of $N^v$ is consistent with $\mathring{\mathcal N}^v=\mathring\Omega^{-2}\mathring N^v$ by \eqref{eq:dust-constraint-2}.

Using \eqref{eq:r-wave-dust}, \eqref{eq:Omega-wave-dust}, and \eqref{eq:Max-v-dust}, a tedious calculation yields
\begin{equation}
    \partial_u(r\partial_v^2r-r\partial_vr\partial_v{\log\Omega^2})=0.\label{eq:5-1}
\end{equation}
Arguing as in \cref{prop:development-1}, we see that \eqref{eq:Ray-v-dust} holds on initial data and is therefore propagated by \eqref{eq:5-1}. This proves the evolution equation $\partial_vm= \partial_vrQ^2/(2r^2)$ and by using \eqref{eq:r-wave-dust} once more, we see that
\begin{equation*}
    \partial_um= -2r\partial_vr \partial_u\left(\frac{\partial_ur}{\Omega^2}\right)+\frac{Q^2}{2r^2}\partial_ur.
\end{equation*}
Comparing this with \eqref{eq:dust-constraint-1} and the definition of $T^{vv}$ yields $\mathring{\mathcal T}^{vv}=\mathring\Omega^{-2}\mathring T^{vv}$, as desired. Finally, \eqref{eq:dust-conservation-2} is proved by directly differentiating the definition of $T^{vv}$ and using \eqref{eq:r-wave-dust}, \eqref{eq:Omega-wave-dust}, and \eqref{eq:Max-u-dust}. 
\end{proof}

\subsubsection{Outgoing charged Vaidya as formal outgoing dust}\label{sec:outgoing-Cauchy}

\begin{figure}
\centering{
\def\svgwidth{10pc}
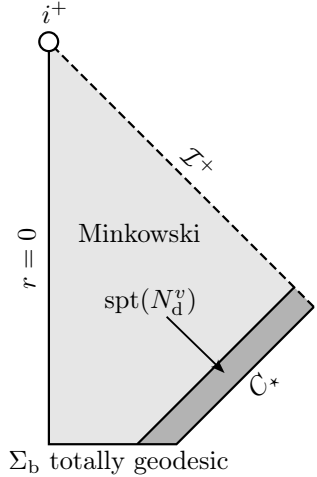}
\caption{An outgoing charged null dust beam obtained by applying \cref{prop:outgoing-dust-existence} to the seed $\mathcal S_{\d,\alpha}$ for parameters $\alpha=(r_1,r_2,0,\varpi_2,0,Q_2)$. The electrovacuum boundary $C_\star$ can be attached to a Reissner--Nordstr\"om spacetime with parameters $\varpi_2$ and $Q_2$.}\label{fig:outgoing-dust}
\end{figure}

We now want to represent the outgoing portion of a regular center bouncing charged null dust beam given by \cref{prop:radial-parametrization} in terms of the outgoing formal system. Let $\alpha\in\mathcal P_\Gamma$, $\varsigma(\alpha)=(\alpha,\check\varpi,\check Q)$ be given by \cref{prop:KU23c-existence}, and consider the time-symmetric dust seed data $\mathcal S_{\mathrm d,\alpha}\doteq (\mathring{\mathcal N}^v_\d,0,r_2,\mathfrak e)$, where\begin{equation}\label{eq:weak-7}
  \mathring{\mathcal N}^v_\d\doteq  \frac{2}{\mathfrak er^2}\left(1-\frac{2\check\varpi}{r}+\frac{\check Q^2}{r^2}\right)^{2}\check Q'.
\end{equation}
For this choice of seed, the constraints \eqref{eq:dust-constraint-1}--\eqref{eq:dust-constraint-2} read
\begin{align}
  \label{eq:dust-constraint-3}  \frac{d}{dr}\mathring m&=\frac{\mathring Q^2}{2r^2},\\
  \label{eq:dust-constraint-4}    \frac{d}{dr}\mathring Q&=\left(1-\frac{2\mathring m}{r}\right)^{-2}\left(1-\frac{2\check m}{r}\right)^2\check Q'.
\end{align}
Therefore, by \eqref{eq:V-2-intro}, $\mathring m = \check m$ and $\mathring Q=\check Q$, where $\check m\doteq \check\varpi - \check Q^2/(2r)$. It follows that the outgoing formal dust solution $(r_\d,\Omega^2_\d,Q_\d,N^v_\d,T^{vv}_\d)$ provided by \cref{prop:outgoing-dust-existence} on $\mathcal C_{r_2}$ is indeed the same as the outgoing charged Vaidya metric provided by the radial parametrization method, \cref{prop:radial-parametrization}. See \cref{fig:outgoing-dust}.

\section{The construction of bouncing charged Vlasov beams and the proof of the main theorem}
\label{sec:Vlasov-beams}

In this section, we prove \cref{thm:main} by constructing \emph{bouncing charged Vlasov beams} as in \cref{fig:bounce-1} and \cref{fig:bounce-2} with prescribed final parameters. This is achieved by a very specific choice of time-symmetric Vlasov seed data and global estimates for the resulting developments. We give a detailed outline of the proof in \cref{sec:proof-outline} and the proof itself occupies \cref{sec:beam-data,sec:data-estimates,sec:near-main,sec:near-aux,sec:far-massive,sec:massive-dispersion,sec:proof-of-prop,sec:patching-ingoing-outgoing}. Finally, in \cref{sec:weak-proof}, we show that these bouncing charged Vlasov beams weak* converge to the bouncing charged null dust spacetimes of \cref{prop:radial-parametrization} in a hydrodynamic limit of the beam parameters.

\subsection{A guide to the proof of \texorpdfstring{\cref{thm:main}}{Theorem 1}}\label{sec:proof-outline}

\subsubsection{The heuristic picture}

The essential idea in the proof of \cref{thm:main} is to ``approximate'' the bouncing radial charged null dust solutions from \cref{thm:ECC-Ori} and \cref{fig:bounce-3} by smooth families of smooth Einstein--Maxwell--Vlasov solutions. Indeed, at least formally, dust can be viewed as Vlasov matter $f(x,p)$ concentrated on a single momentum $p=k(x)$ at each spacetime point $x$. One is faced with having to perform a \emph{global-in-time} desingularization of families of dust solutions which are singular in both the space and momentum variables.

Assuming that this can be done, the heuristic picture is that of a focusing beam of Vlasov matter coming in from infinity with particles of mass $\mathfrak m=0$ or $\mathfrak m\ll 1$ (so that the particles look almost massless for very large time scales) and very small angular momentum $0<\ell\ll 1$, which are decelerated by the electromagnetic field that they generate. Then, along some ``approximate bounce hypersurface,'' the congruence smoothly ``turns around'' and becomes outgoing, escaping to infinity if a black hole has not yet formed. Along the way, the particles do not hit the center of symmetry. By appropriately varying the beam parameters, we can construct families of spacetimes as depicted in \cref{fig:bounce-1} or \cref{fig:bounce-2}. 

As should be apparent from the treatment of the Cauchy problem for the Einstein--Maxwell--Vlasov system in \cref{sec:general-data} and for charged null dust in \cref{sec:outgoing-dust-equations}, we want to pose Cauchy data on (what will be) the approximate bounce hypersurface for the desingularized Vlasov solutions. We will choose the initial data for $f$ to be supported on small angular momenta $\ell\sim\ve$ and so that the charge $\mathring Q$ and Hawking mass $\mathring m$ profiles closely approximate the initial data for dust as described in \cref{sec:outgoing-Cauchy}. The Vlasov beam which is intended to approximate charged null dust is called the \emph{main beam}.

As we will see, desingularizing bouncing charged null dust requires an ansatz for $\mathring f$ which necessarily degenerates in $\ve$. Closing estimates in the region of spacetime where $Q\les \ve$ is then a fundamental issue because the repulsive effect of the electromagnetic field is relatively weak there. We overcome this issue by adding an \emph{auxiliary beam} to the construction, which stabilizes the main beam by adding a small amount of charge on the order of $\eta\gg \ve$. This beam is not dust-like, consists of particles with angular momentum $\sim 1$, and is repelled away from the center by the centrifugal force. 

The goal will be to construct a smooth family of Vlasov seeds $\lambda\mapsto \mathcal S_\lambda$ for $\lambda\in [-1,2]$ such that $\mathcal S_{-1}$ is trivial (i.e., evolves into Minkowski),  $\mathcal S_2$ forms a subextremal Reissner--Nordstr\"om black hole with charge to mass ratio $\approx \tfrac 12$, and $\lambda_*\approx 1$ is the critical parameter for which an extremal Reissner--Nordstr\"om black hole with mass $M$ forms. For $\lambda\in [0,2]$, the Vlasov development $\mathcal D_\lambda$ closely approximates the dust developments from \cref{thm:ECC-Ori} (in a sense to be made precise in \cref{sec:weak-proof} below) and $\lambda\in [-1,0]$ smoothly ``turns on'' the auxiliary beam. At the very end of the proof, $\lambda$ is simply rescaled to have range $[0,1]$.

In fact, our methods allow us to desingularize any bouncing charged null dust beam given by \cref{prop:KU23c-existence}. Adding dependence on $\lambda$ is then essentially only a notational hurdle. We now highlight specific aspects of the construction in more detail. 

\subsubsection{Time symmetry and reduction to the outgoing case}\label{sec:intro-outline-2}

\begin{figure}[ht]
\centering{
\def\svgwidth{33pc}
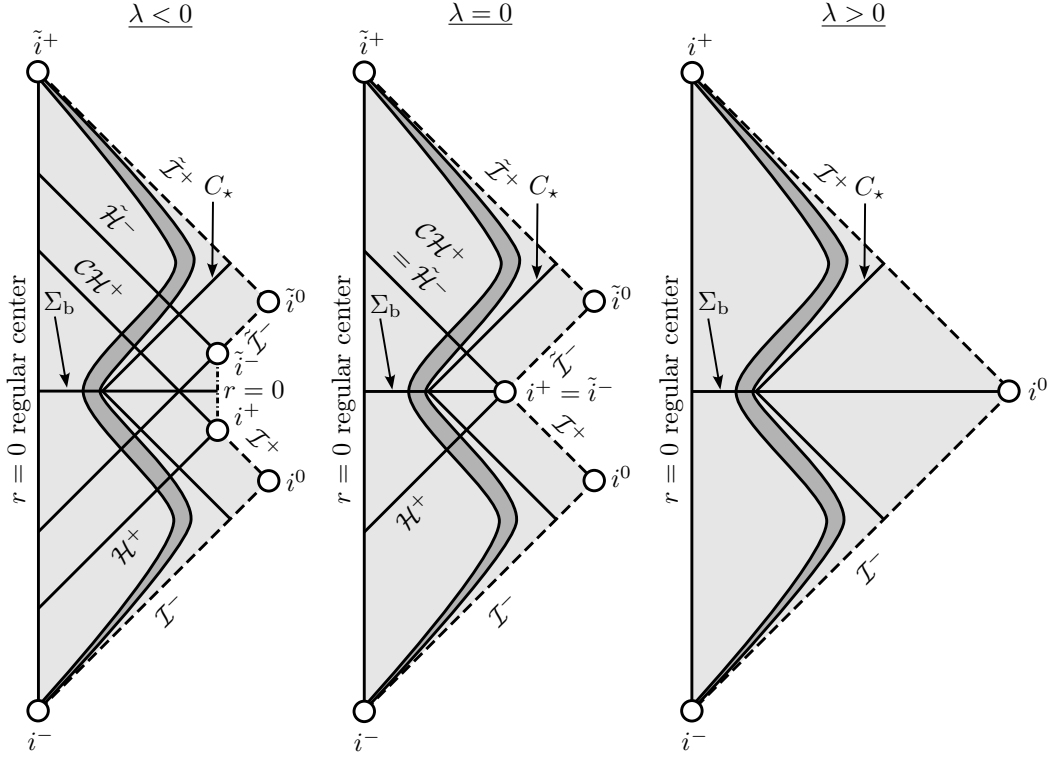}
\caption{Penrose diagrams of the ``maximal time-symmetric doubled spacetimes'' used in the proof of \cref{thm:main} when $\mathfrak m>0$. When $\lambda\le 0$, these spacetimes are evidently not globally hyperbolic, but one can easily observe that the globally hyperbolic spacetimes depicted in \cref{fig:bounce-1} when $\lambda\le 0$ are simply the above spacetimes restricted to the past of $\mathcal{CH^+}\cup\{i^+\}\cup \mathcal I^+$. The exterior region is isometric to a subset of the maximally extended Reissner--Nordstr\"om solution with parameters depending on $\lambda$. }\label{fig:time-symmetry}
\end{figure}

The starting point of the construction of bouncing charged Vlasov beams is the prescription of Cauchy data on an approximate bounce hypersurface $\Sigma_\b$, using the radial parametrization of bouncing charged null dust as a guide. We can now see the utility of the time-symmetric ansatz in \cref{sec:general-data}: it reduces the problem to constructing an \emph{outgoing} beam, which is then reflected and glued to maximally extended Reissner--Nordstr\"om to construct a \emph{time symmetric} spacetime as depicted in \cref{fig:time-symmetry} below. These ``maximal time-symmetric spacetimes'' are constructed in \cref{sec:doubled}. The globally hyperbolic developments in \cref{thm:main} are obtained by taking appropriate subsets and identifying suitable Cauchy hypersurfaces.

The problem now reduces to constructing the region bounded to the past by $C_\star$, $\Sigma_\mathrm{b}$, and the center in \cref{fig:time-symmetry}. In this region, \emph{the solution is always dispersive}. Therefore, we can actually treat the subextremal, extremal, and superextremal cases at once. Detection of whether a black hole forms in the doubled spacetime takes place \emph{on the level of initial data} and we heavily exploit the global structure of the Reissner--Nordstr\"om family itself in this process. Note that while we prescribe data in the black hole interior when $\lambda\le 0$, there clearly exist Cauchy surfaces lying entirely in the domain of outer communication. In fact, the solutions are always past complete and disperse to the past. See already \cref{prop:global-structure}.

\subsubsection{The choice of seed data}

We now describe our desingularization procedure for bouncing charged null dust on the level of initial data. Consider the outgoing portion of a charged null dust beam $(r_\d,\Omega^2_\d,Q_\d,N^v_\d,T^{vv}_\d)$ as in \cref{sec:outgoing-Cauchy}, with Cauchy data posed along the bounce hypersurface $\Sigma_\b$. The geometry of the outgoing dust beam is entirely driven by the choice of renormalized number current $\mathring{\mathcal N}^v_\d$ in \eqref{eq:weak-7}. Importantly, the energy-momentum tensor of dust vanishes identically along $\Sigma_\b$.

Since radial charged null dust has $\ell=0$, we wish to approximate dust with Vlasov matter consisting of particles with angular momentum $\ell\sim \ve$, where $0<\ve\ll1$ is a small parameter to be chosen. We want to choose the initial distribution function $\mathring f$ so that 
\begin{equation}
\mathring{\mathcal N}^u+\mathring{\mathcal N}^v= \mathring{\mathcal N}^v_\mathrm d = \frac{1}{\mathfrak er^2}\left(1-\frac{2\check\varpi}{r}+\frac{\check Q^2}{r^2}\right)^2\frac{d\check Q}{dr}, \label{eq:inro-N-f-approx} 
 \end{equation}
 \begin{equation}
     \mathring Q\approx\check Q,\quad    \mathring \varpi\approx\check\varpi,\quad 
     \mathring {\mathcal T}^{uu},\mathring {\mathcal T}^{uv},\mathring {\mathcal T}^{vv}\approx 0
 \end{equation}
on $\Sigma_\b$, as $\ve\to 0$. These conditions are satisfied if we choose
\begin{equation}
    \mathring f_\mathrm{main}^{\,\alpha,\ve}(r,p^u,p^v)\doteq \frac{c}{\mathfrak er^2\ve}\left(1-\frac{2\check\varpi}{r}+\frac{\check Q^2}{r^2}\right)^2 \frac{d\check Q}{dr} \delta_\ve(p^u)\delta_\ve(p^v)\label{eq:intro-f-main}
\end{equation}
for $r\in [r_1,r_2]$, where $\delta_\ve$ are approximations of the identity with support $[\ve,2\ve]$ and $c$ is a normalization constant that depends on the precise choice of the family $\delta_\ve$. In order for the mass shell inequality $\Omega^2p^up^v\ge \mathfrak m^2$ to hold on the support of $ \mathring f_\mathrm{main}^{\,\alpha,\ve}$, \eqref{eq:intro-f-main} forces us to choose $\mathfrak m\in[0,\mathfrak m_0]$ with $0<\mathfrak m_0\ll \ve$. 

\begin{rk}
    In the full bouncing null dust model, $N$ is discontinuous across $\Sigma_\b$. Indeed, to the past of $\Sigma_\mathrm b$, $N$ points in the $u$-direction and has a nonzero limit along $\Sigma_\mathrm b$, but to the future points in the $v$-direction and also has a nonzero limit. In the Vlasov case, time symmetry demands $N$ be smooth across, and orthogonal to, $\Sigma_\b$. By comparing \eqref{eq:constraint-2} with \eqref{eq:dust-constraint-2}, we see that $\mathring{\mathcal N}^u+\mathring{\mathcal N}^v$ in Vlasov takes the role of $\mathring{\mathcal N}^v$ in dust.
\end{rk}

Observe directly from \eqref{eq:intro-f-main} that $\mathring f_\mathrm{main}^{\,\alpha,\ve}$ behaves pointwise like $\ve^{-3}$ and therefore pointwise estimates for $N$ and $T$ in evolution must utilize precise estimates of the electromagnetic flow to cancel factors of $\ve$. Closing estimates independently of $\ve$ is the main challenge of this scheme and we directly exploit the null structure of the spherically symmetric Einstein--Maxwell--Vlasov system in the proof. The main mechanisms ensuring boundedness of $N$ and $T$ in the main beam are:
\begin{enumerate}
    \item The angular momentum $\ell$ is conserved, so that $\ell\sim\ve$ throughout the main beam. 
    \item If $\gamma$ is an electromagnetic geodesic arising from the support of $ \mathring f_\mathrm{main}^{\,\alpha,\ve}$, then $p^v$ should rapidly increase due to electromagnetic repulsion. Dually, $p^u$ should rapidly decrease, which ought to suppress the ingoing moments $N^u$, $T^{uu}$, and $T^{uv}$. This should be compared with the vanishing of $N^u$, $T^{uu}$, and $T^{uv}$ in outgoing null dust. We say that the main beam \emph{bounces due to electromagnetic repulsion}.
\end{enumerate}

\begin{figure}
\centering{
\def\svgwidth{27pc}
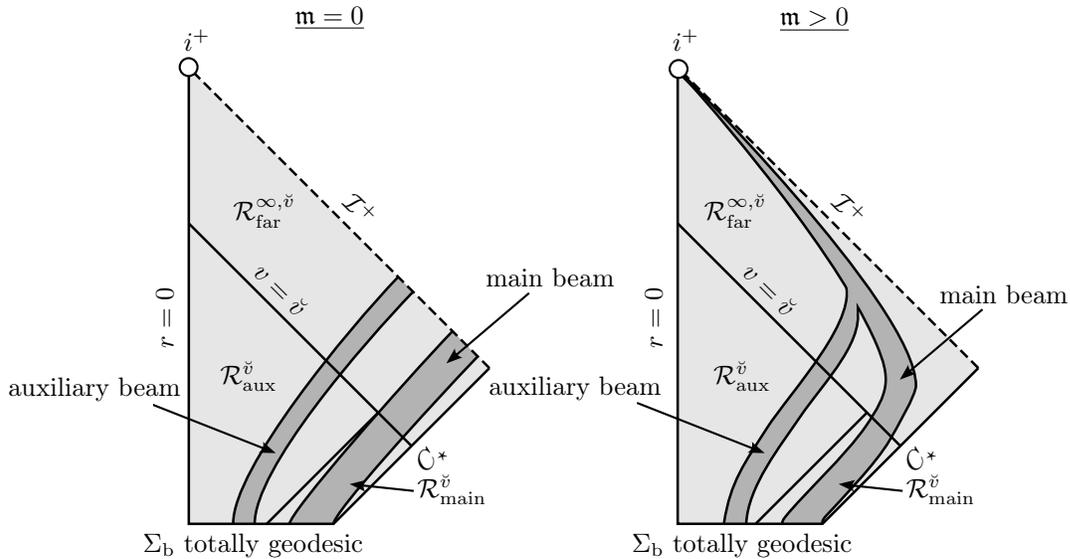}
\caption{Penrose diagram of outgoing charged Vlasov beams (evolution of the seed data $\mathcal S_{\alpha,\eta,\ve}$). Note that the beams do not intersect when $\mathfrak m=0$. When $\mathfrak m>0$, one can show that they do, but it is not necessary to do so for our purposes here.}
\label{fig:zoomed-1}
\end{figure}

As is apparent from \eqref{eq:SS-Lor-2}, the magnitude of the repulsive effect is proportional to $Q$. If we were to evolve the seed $\mathring f_\mathrm{main}^{\,\alpha,\ve}$ on its own, the inner edge of the beam would experience less electromagnetic repulsion since $Q$ is potentially quite small in the inner region. 

In order to reinforce the repulsive effect of the electric field in the main beam and get a consistent hierarchy of powers of $\ve$, we introduce an \emph{auxiliary beam} on the inside of the main beam which \emph{bounces due to the centrifugal force} associated to electromagnetic geodesics with large angular momentum. The initial data for the auxiliary beam is chosen to be 
\begin{equation}
    \mathring f_\mathrm{aux}^{\,r_1,\eta}(r,p^u,p^v)\doteq \eta\, \varphi (r,p^u,p^v),\label{eq:intro-f-aux}
\end{equation}
where $\eta\gg \ve$ is a constant determining the amplitude, $\varphi$ is a cutoff function supported on the set  $[\tfrac 13r_1,\tfrac 23r_1]\times[\Lambda-1,\Lambda+1]\times[\Lambda-1,\Lambda+1]$, and $\Lambda$ is a fixed large constant that determines the strength of the centrifugal force felt by the auxiliary beam. The auxiliary beam ensures that the main beam always interacts with an electric field of amplitude $\gtrsim \eta$, which acts as a crucial stabilizing mechanism.

\subsubsection{The near and far regions and the hierarchy of scales}

 The total seed for an outgoing  Vlasov beam is taken to be $\mathcal S_{\alpha,\eta,\ve}\doteq(\mathring f_\mathrm{tot}^{\,\alpha,\eta,\ve},r_2,\mathfrak m,\mathfrak e)$, where
\begin{equation}
    \mathring f_\mathrm{tot}^{\,\alpha,\eta,\ve}\doteq \mathring f_\mathrm{aux}^{\,r_1,\eta}+ \mathring f_\mathrm{main}^{\,\alpha,\ve},\label{eq:intro-f-tot}
\end{equation} the fundamental charge $\mathfrak e>0$ is fixed, the mass $\mathfrak m$ lies in the interval $[0,\mathfrak m_0]$, and $\eta,\ve$, and $\mathfrak m_0$ need to be chosen appropriately small.

To study the evolution of $\mathcal S_{\alpha,\eta,\ve}$, depicted in \cref{fig:zoomed-1}, we distinguish between the \emph{near region} $\{v\le \breve v\}$ and the \emph{far region} $\mathcal R_\mathrm{far}^{\breve v,\infty}=\{v\ge \breve v\}$, where $\breve v$ is a large advanced time to be determined. Roughly, the ingoing cone $\{v=\breve v\}$ is chosen so that the geometry is very close to Minkowskian and the Vlasov field is ``strongly outgoing'' and supported far away from the center, i.e., 
\begin{equation}
   \frac{p^u}{p^v}\les r^{-2}\ll 1\label{eq:intro-outgoingness}
\end{equation}
for every $p^u$ and $p^v$ such that $f(\,\cdot\,,\breve v,p^u,p^v)\ne 0$. The near region is further divided into the \emph{main} and \emph{auxiliary} regions, corresponding to the physical space support of the main and auxiliary beams and denoted by $\mathcal R^{\breve v}_\mathrm{main}$ and $\mathcal R^{\breve v}_\mathrm{aux}$, respectively.\footnote{For reasons of convenience, $\mathcal R_\mathrm{main}$ is defined slightly differently in the actual proof than the region depicted in \cref{fig:zoomed-1}, but this is inconsequential at this level of discussion.}

The beam parameters $\eta$, $\ve, \mathfrak m_0$ and the auxiliary parameter $\breve v$ satisfy the hierarchy
\begin{equation}
    0<\mathfrak m_0\ll \ve\ll\eta\ll\breve v^{-1}\ll 1.\label{eq:intro-hierarchy}
\end{equation}
 To prove the sharp rate of dispersion when $\mathfrak m>0$, we augment this hierarchy with 
\begin{equation*}
    0<v_\#^{-1}\ll \mathfrak m,
\end{equation*}
where $v_\#$ is a very large time after which the additional dispersion associated to massive particles kicks in. 

The proof of \cref{thm:main} proceeds by showing that if \eqref{eq:intro-hierarchy} holds, then the solution exists, with certain properties, in the regions $\mathcal R_\mathrm{main}^{\breve v}$, $\mathcal R^{\breve v}_\mathrm{aux}$, and $\mathcal R^{\breve v,\infty}_\mathrm{far}$, in that order. The sharp decay rates of $N$ and $T$ are then shown a posteriori by re-analyzing the electromagnetic geodesic flow in the far region.

\begin{rk}
   The final Reissner--Nordstr\"om parameters of the total Vlasov seed \eqref{eq:intro-f-tot} depend on the approximation parameters $\eta$ and $\ve$, but are $O(\eta)$-close to $(\varpi_2,Q_2)$. Therefore, in order to reach any fixed set of parameters, the background dust seed has to be appropriately modulated. See already \cref{sec:proof-main}.
\end{rk}

\subsubsection{Outline of the main estimates}\label{sec:intro-near} 

\ul{The main beam in the near region:} In this region, the main goal is proving smallness (in $\ve$) of $N^u$, $T^{uu}$, and $T^{uv}$, which are identically zero in the background dust solution. We define the \emph{phase space volume function} $\mathscr V:\mathcal Q\to \Bbb R_{\ge 0}$ by
\begin{equation}\label{eq:phase-space-volume}
    \mathscr V(u,v)\doteq \Omega^2(u,v)|\{(p^u,p^v):f(u,v,p^u,p^v)\ne 0\}|,
\end{equation}
where $|\,\cdot\,|$ is the Lebesgue measure on $\Bbb R^2_{p^u,p^v}$. The function $\mathscr V$ is invariant under gauge transformations of $u$ and $v$. Using the mass shell relation \eqref{eq:mass-shell} and the change of variables formula, we find
\begin{align}
  \label{eq:V-v}  \mathscr V(u,v)= \frac{2}{r^2}\int_0^\infty\int_{\{p^v:f(u,v,p^u,p^v)\ne 0\}} \frac{dp^v}{p^v}\,\ell\,d\ell,
\end{align}
where we view $p^u$ as a function of $p^v$ and $\ell$. Because of the addition of $\mathring f^{\,r_1,\eta}_\mathrm{aux}$ to the seed data and the good monotonicity properties of \eqref{eq:Max-u} and \eqref{eq:Max-v}, it holds that $Q\gtrsim \eta$ in $\mathcal R^{\breve v}_\mathrm{main}$. Under relatively mild bootstrap assumptions, any electromagnetic geodesic $\gamma$ in the main beam is accelerated outwards at a rate $\gtrsim \eta$, i.e., 
\begin{equation*}
     p^v\gtrsim \ve+\eta\min\{\tau,1\},\quad p^u\les\frac{\ve^2}{r^2(\ve+\eta\min\{\tau,1\})},
\end{equation*}
where $\tau\doteq \frac 12(u+v)$ is a ``coordinate time.'' We also show that if $\gamma_1$ and $\gamma_2$ are two electromagnetic geodesics in the main beam which reach the same point $(u,v)\in\mathcal R_\mathrm{main}^{\breve v}$, then
\begin{equation*}
    |p^v_1-p^v_2|\les \frac{\ve}{\eta^2}
\end{equation*}
at $(u,v)$. Using these estimates, conservation of angular momentum, and the hierarchy \eqref{eq:intro-hierarchy}, we show that 
\begin{equation*}
    \mathscr V(u,v)\les_\eta \frac{\ve^3}{\ve+\eta\min\{\tau,1\}},\label{eq:volume-est-intro}
\end{equation*}
where the notation $A\les_\eta B$ means $A\le CB$, where $C$ is a constant depending on $\eta$. Then, simply using the transport nature of the Vlasov equation, we obtain the estimates
\begin{equation*}\label{eq:intro-main-est-1}
    T^{uu}(u,v)\les \ve^{1/2},\quad T^{uv}(u,v)\les \ve^{1/2},\quad \int_{-u}^v N^u(u,v')\,dv'\les \ve^{1/2},
\end{equation*}
which capture the fundamental characteristic of outgoing null dust. These estimates allow us to control the geometry at $C^1$ order, which is more than enough to use the generalized extension principle, \cref{prop:ext}, to extend the solution. For details, see \cref{sec:near-main}. When $\lambda\in[-1,0]$ and the main beam has not yet been turned on, constructing the solution in this region is trivial since the solution is electrovacuum.

\ul{The auxiliary beam in the near region:} Since the auxiliary beam is genuinely weak ($\mathring f^{\,r_1,\eta}_\mathrm{aux}\les \eta$ pointwise), the bootstrap argument in $\mathcal R^{\breve v}_\mathrm{aux}$ is a standard Cauchy stability argument, perturbing off of Minkowski space. We use explicit knowledge of the impact parameter and asymptotics of \emph{null geodesics} with angular momentum $\sim\Lambda$ on Minkowski space and treat the charge as an error term in this region. For details, see \cref{sec:near-aux}.

\ul{Existence in the far region:} The argument in this region is a refinement of Dafermos' proof of the stability of Minkowski space for the spherically symmetric Einstein--massless Vlasov system \cite{Dafermos-Vlasov} (see also \cite[Chapter 4]{TaylorPhD}). Because of the singular nature of $f_\mathrm{main}$ in powers of $\ve$, it seems difficult to obtain uniform in $\ve$ pointwise decay estimates for $T^{uv}$ by the usual method of estimating decay of the phase space volume of the support of $f$ at late times. Fortunately, we are able to exploit the a priori energy estimates 
\begin{equation}
     \int r^2\Omega^2 T^{uv}\partial_ur \,du'  \les 1,\quad \int r^2\Omega^2 T^{uv}\partial_vr  \,dv'\les 1\label{eq:energy-est-intro}
\end{equation}
coming from the monotonicity of the Hawking mass when $\partial_vr>0$ and $\partial_ur<0$ (see \cite{dafermos-trapped-surface}). It is important to note that these energy estimates are independent of initial data and are a fundamental feature of the spherically symmetric Einstein equations. Under the bootstrap assumption that the electromagnetic geodesics making up the support of $f$ are ``outgoing'' for $v\ge \breve v$, the energy estimates \eqref{eq:energy-est-intro} imply decay for the unweighted fluxes of $T^{uv}$. This shows that the geometry remains close to Minkowski in $C^1$ and recovering the bootstrap assumption on the support of $f$ follows from good monotonicity properties of the electromagnetic geodesic flow when close to Minkowski. We also note that this approach using energy estimates allows us to treat the cases $\mathfrak m=0$ and $\mathfrak m>0$ simultaneously. For details, see \cref{sec:far-massive}.  

\ul{Dispersion in the far region:} Once the solution has been shown to exist globally, we prove sharp (in coordinate time $\tau$) pointwise decay statements for $N$ and $T$ (see \cite{RR92,EMV-stability,TaylorPhD}). As the decay rates differ when $\mathfrak m=0$ or $\mathfrak m>0$, these two cases are treated separately. 

\emph{The massless case.} It follows immediately from the mass shell relation \eqref{eq:mass-shell} that $p^u\les r^{-2}$ in the far region. Since this is integrable, the beams are confined to null slabs and can even be shown to be disjoint as depicted in \cref{fig:zoomed-1}. Since each $p^u$ contributes a factor of $r^{-2}$ and our solutions have bounded angular momentum, we obtain the sharp dispersive hierarchy 
\begin{equation*}
    N^v+T^{vv}\le C(1+\tau)^{-2},\quad N^u+T^{uv}\le C(1+\tau)^{-4},\quad T^{uu}\le C(1+\tau)^{-6},
\end{equation*}
where the constant $C$ depends on $\alpha,\eta$, and $\ve$. For details, see \cref{sec:massive-dispersion}.

\emph{The massive case.} When $\mathfrak m>0$, $p^u$ does not decay asymptotically. After a very late time $v_\#\gg \mathfrak m^{-1}$, $p^u\sim_\eta \mathfrak m^2$, which drives additional decay of the phase space volume. We prove this by a change of variables argument, turning volume in $p^u$ at later times $v\ge v_\#$ into physical space volume of the support of $f$ at time $v=v_\#$. This leads to the sharp isotropic decay rate
\begin{equation*}
    \mathcal M\le C(1+\tau)^{-3}
\end{equation*}
for any moment $\mathcal M$ of $f$, where $C$ depends on $\alpha,\eta,\ve$, and a lower bound for $\mathfrak m$. For details, see \cref{sec:proof-of-prop}.

\subsection{Outgoing charged Vlasov beams}\label{sec:beam-data}

\subsubsection{The beam parameters, fixed constants, and conventions}\label{sec:parameters}

First, we fix once and for all the fundamental charge $\mathfrak e\in\Bbb R\setminus\{0\}$. Without loss of generality, we may take $\mathfrak e>0$, as all of the arguments and definitions in the remainder of \cref{sec:Vlasov-beams} require only minor cosmetic modifications to handle the case $\mathfrak e<0$. Next, we fix an even function $\varphi\in C^\infty_c(\Bbb R)$ satisfying $\spt\varphi=[-1,1]$, $\varphi\ge 0$, and 
\begin{equation*}
    \int_{-1}^1\varphi\,dx=1.
\end{equation*}
 Let $\theta\in C^\infty(\Bbb R)$ be a nondecreasing function such that $\theta(\lambda)=0$ for $\lambda \le -1$ and $\theta(\lambda)=1$ for $\lambda\ge 0$. Let $\zeta\in C^\infty(\Bbb R)$ be a nondecreasing function such that $\zeta(\lambda)=0$ for $\lambda\le 0$, $\zeta(\lambda)=\lambda$ for $\lambda\ge \frac12$, and $\zeta'(\lambda)>0$ for $\lambda\in(0,\tfrac 12]$. Finally, we fix a large\footnote{Large relative to the other beam parameters. For instance, $\Lambda=20$ suffices.} number $\Lambda\ge 1$, such that 
\begin{equation}\label{eq:Lambda}
    \min_{p_1,p_2\in[\Lambda-1,\Lambda+1]}\frac{p_1p_2}{(p_1+p_2)^2}\ge\tfrac{81}{400}.
\end{equation}
We emphasize that:
\begin{center}
    \emph{The quintuple $(\mathfrak e,\varphi,\theta,\zeta,\Lambda)$ is fixed for the remainder of the paper.}
\end{center}

Recall the set $\mathcal P_\Gamma$ of regular center admissible parameters of the form $\alpha=(r_1,r_2,0,\varpi_2,0,Q_2)$ which was defined in \cref{sec:radial}. Let $\eta,\ve$, and $\mathfrak m_0$ be positive real numbers. In the course of the proofs below, the particle mass $\mathfrak m$ will be restricted to satisfy $0\le \mathfrak m\le \mathfrak m_0$.

\begin{defn}\label{def:S}
    Let $\alpha=(r_1,r_2,0,\varpi_2,0,Q_2)\in\mathcal P_\Gamma$, $\eta>0$, and $\ve>0$. The \emph{time-symmetric outgoing charged Vlasov beam seed} $\mathcal S_{\alpha,\eta,\ve}$ is given by $(\mathring f_\mathrm{aux}^{\,r_1,\eta}+\mathring f^{\,\alpha,\ve}_\mathrm{main},r_2,\mathfrak m,\mathfrak e)$, where 
\begin{equation}
    \mathring f_{\mathrm{aux}}^{\,r_1,\eta}(r,p^u,p^v)\doteq\eta\,\varphi\!\left(\frac{6}{r_1}r-3\right)\varphi(p^u-\Lambda)\varphi(p^v-\Lambda),\label{eq:aux-def}
\end{equation}
and 
    \begin{equation}\label{eq:f-main}
        \mathring f_\mathrm{main}^{\,\alpha,\ve}(r,p^u,p^v)\doteq \frac{8}{3\pi\mathfrak e\ve^3r^2} \left(1 - \frac{2\check\varpi(r)}{r} + \frac{\check Q^2(r)}{r^2}\right)^2 \check Q'(r)\,\varphi\!\left(\frac{2p^u}{\ve}-3\right)\varphi\!\left(\frac{2p^v}{\ve}-3\right),
    \end{equation}
    where $\check\varpi$ and $\check Q$ are taken from $\varsigma(\alpha)$, where the map $\varsigma$ was defined in \cref{prop:KU23c-existence} (cf. \cref{rk:varsigma}).
\end{defn}

\begin{defn}\label{def:S-interp}
    Let $M>0$, let $0<r_1<r_2$, and let $\{\alpha_{\lambda,M'}\}$ (with $|M-M'|\le \delta$) be as in \eqref{eq:alpha-lambda-M'} in the proof of \cref{thm:ECC-Ori}. For $\lambda\in [-1,2]$, $\eta>0$, $\ve>0$, $\mathfrak m_0$, we define
    \begin{equation}
        \mathcal S_{\lambda,M',\eta,\ve}\doteq\mathcal S_{\alpha_{\zeta(\lambda),M'},\theta(\lambda)\eta,\ve}
    \end{equation}
    for particles of mass $0\le \mathfrak m\le \mathfrak m_0$. For $\lambda \le 0$, the $\check\varpi$ and $\check Q$ components of  $\varsigma(\alpha_{\zeta(\lambda),M'})$ are interpreted as identically zero, in correspondence with the proof of \cref{thm:ECC-Ori}.
\end{defn}

Throughout \cref{sec:Vlasov-beams}, the notation $A\les B$ means that there exists a constant $C>0$, which only depends on  $\mathfrak e$, $\varphi$, $\theta$, $\zeta$, $\Lambda$, $M$, $r_1$, and $r_2$ such that $A\le CB$. The notation $A\gtrsim B$ is defined similarly and $A\sim B$ means $A\les B$ and $A\gtrsim B$. Moreover, we use the convention that all small (large) constants in ``sufficiently small (large)'' may depend on $\mathfrak e$, $\varphi$, $\theta$, $\zeta$, $\Lambda$, $M$, $r_1$, and $r_2$. In \cref{sec:massive-dispersion}, we will also use the notation $A\les_{\eta} B$ (resp., $A\les_{\eta,\ve}B$), in which we allow the constants to also depend on $\eta$ (resp., $\eta$ and $\ve$). The relations $A\sim_\eta B$ and $A\sim_{\eta,\ve}B$ are defined in the obvious way.

For the evolution problem, we will introduce a large parameter $\breve v$ to separate $\mathcal C_{r_2}$ into the ``near'' and ``far'' regions. We will always assume that the parameter hierarchy
\begin{equation}
    0<\mathfrak m_0\ll \ve\ll \eta\ll\breve v^{-1}\ll 1\label{eq:parameter-smallness}
\end{equation}
holds, by which we mean that any given statement holds for $\breve v$ sufficiently large, $\eta$ sufficiently small depending on $\breve v$, $\ve$ sufficiently small depending on $\breve v$ and $\eta$, and $\mathfrak m_0$ sufficiently small depending on $\breve v$, $\eta$, and $\ve$. To prove dispersion in the massive case, we introduce an even larger parameter $v_\#$ satisfying
\begin{equation*}
    0<v^{-1}_\#\ll \mathfrak m \le \mathfrak m_0,
\end{equation*}
so that $v_\#$ is chosen sufficiently large depending on $\breve v, \eta,\ve$, and $\mathfrak m$.

\subsubsection{The global structure of outgoing charged Vlasov beams}\label{sec:global-structure}

\begin{prop}\label{prop:approx-1}
    Fix a fundamental charge $\mathfrak e>0$, cutoff functions $\varphi,\theta$, and $\zeta$ as in \cref{sec:parameters}, a number $\Lambda\ge 1$ satisfying \eqref{eq:Lambda}, $M>0$, and $0<r_1<r_2<r_-(4M,2M)$. Let $\delta>0$ be as in the statement of \cref{thm:ECC-Ori} and define $\mathcal S_{\lambda,M',\eta,\ve}$ as in \cref{def:S-interp} for $\lambda\in[-1,2]$, $|M'-M|\le\delta$, $\eta>0$, $\ve>0$, and for particles of mass $0\le\mathfrak m\le \mathfrak m_0$, where $\mathfrak m_0>0$.

    If $\eta$ is sufficiently small, $\ve$ is sufficiently small depending on $\eta$, and $\mathfrak m_0$ is sufficiently small depending on $\eta$ and $\ve$, then for any $\lambda\in[-1,2]$ and $|M'-M|\le \delta$, the seed $\mathcal S_{\lambda,M',\eta,\ve}$ is untrapped and consistent with particles of mass $\mathfrak m$. There exists a unique maximal normalized development $(\mathcal U,r,\Omega^2,Q,f)$ of $\mathcal S_{\lambda,M',\eta,\ve}$ for particles of charge $\mathfrak e$ and mass $\mathfrak m$ with the following properties.\footnote{Here, uniqueness is in the class of normalized developments as in  \cref{def:normalized-defn}. We have not shown an unconditional existence and uniqueness statement for maximal developments for the Einstein--Maxwell--Vlasov model in this paper (although this can be done) and will therefore infer uniqueness directly in the course of the construction.} If $\mathfrak m>0$:

     \begin{enumerate}     
    \item The development is global in the normalized double null gauge, i.e., $\mathcal U=\mathcal C_{r_2}$. 
        \item The $(3+1)$-dimensional spacetime obtained by lifting $(\mathcal U,r,\Omega^2)$ is future causally geodesically complete and satisfies globally the estimates
\begin{equation}
  \partial_vr\sim 1,\quad   |\partial_ur|\les 1,\quad   \Omega^2\sim 1,\quad (1+u^2)|\partial_u\Omega^2|+(1+v^2)|\partial_v\Omega^2|\les 1.   \label{eq:prop-approx-1} 
\end{equation}

        \item   Define the \emph{final Reissner--Nordstr\"om parameters} $\tilde M$ and $\tilde e$ of $\mathcal S_{\lambda,M',\eta,\ve}$ to be the constant values of $\varpi$ and $Q$, respectively, on the cone $C_{-r_2}$. Then $\tilde M$ and $\tilde e$ are smooth functions of $(\lambda,M',\eta,\ve,\mathfrak m)$, satisfy the estimate \begin{equation}\label{eq:final-parameter-estimate}
            |\tilde M-\zeta(\lambda)^2M'|+|\tilde e-\zeta(\lambda)M'|\les \eta,
        \end{equation} and extend smoothly to $\eta=\ve=0$, where they equal $\zeta(\lambda)^2M'$ and $\zeta(\lambda)M'$, respectively. The spacetime $(\mathcal U,r,\Omega^2)$ contains antitrapped surfaces (symmetry spheres where $\partial_ur\ge 0$) if and only if $\tilde e\le \tilde M$ and $r_2<r_-$, where $r_\pm\doteq\tilde  M\pm \sqrt{\tilde M^2-\tilde e^2}$.
         In this case, we nevertheless have $ \partial_ur\sim -1$
 for $v$ sufficiently large and the antitrapped surfaces are restricted to lie in the slab $\{2r_--r_2\le v\le 2r_+-r_2\}$.
        \item The Vlasov distribution function $f$ is quantitatively supported away from the center,
        \begin{equation}
            \inf_{\pi(\spt f)}r \ge \tfrac 16r_1,\label{eq:r-min-statement}
        \end{equation}
        and the beam asymptotes to future timelike infinity $i^+$ in the sense that
 \begin{equation}
           \pi( \spt f)\subset \{C_1 v\le u\le C_2 v\},\label{eq:spt-f-statement}
        \end{equation}
        where $C_1$ and $C_2$ are positive constants that may additionally depend on $\eta,\ve,\mathfrak m$, and $\lambda$.  The connected component of $\mathcal U\setminus \pi(\spt f)$ containing the center is isometric to Minkowski space. The connected component of $\mathcal U\setminus\pi(\spt f)$ containing future null infinity $\mathcal I^+$ is isometric to an appropriate neighborhood of future null infinity in the Reissner--Nordstr\"om solution with parameters $M$ and $e$. 
        \item  The Vlasov matter disperses in the sense that the macroscopic observables decay pointwise:
        \begin{equation}
           \mathcal M\le C(1+\tau)^{-3},\label{eq:massive-decay}
        \end{equation} 
        where $\mathcal M\in\{N^u,N^v,T^{uu},T^{uv},T^{vv},S\}$ and the constant $C$ may additionally depend on $\eta,\ve,\mathfrak m$, and $\lambda$. 
    \end{enumerate}
    The same conclusions hold if $\mathfrak m=0$, but points 4.~and 5.~are improved to:
    \begin{itemize}
        \item [4.$'$]  The estimate \eqref{eq:r-min-statement} still holds but \eqref{eq:spt-f-statement} is improved to 
        \begin{equation}
            \pi(\spt f)\subset\{-r_2\le u\le \tfrac 13r_1\},\label{eq:spt-f-massless}
        \end{equation}
        i.e., the beam is confined to a null slab. The spacetime is isometric to Minkowski space for $u\ge r_1$.  
  
        \item [5.$'$] The Vlasov matter disperses in the sense that the macroscopic observables decay pointwise:
        \begin{align}
      \label{eq:massless-decay-1}  N^v+T^{vv}&\le C (1+\tau)^{-2},\\
            N^u+T^{uv}&\le C(1+\tau)^{-4} ,\\
        \label{eq:massless-decay-3}     T^{uu}&\le C (1+\tau)^{-6},
        \end{align}
        where the constant $C$ may additionally depend on $\eta,\ve$, and $\lambda$.
    \end{itemize}
  \end{prop} 

\begin{rk}\label{rk:limiting-case}
An analogous version of  \cref{prop:approx-1} may be proved for any $\alpha=(r_1,r_2,0,\varpi_2,0,Q_2)\in\mathcal P_\Gamma$ or $\alpha=(r_1,r_2,0,0,0,0)$ by evolving the seed data $\mathcal S_{\alpha,\eta,\ve}$ given by \cref{def:S}. In that case, \eqref{eq:final-parameter-estimate} becomes
\begin{equation}
    |\tilde M-\varpi_2|+|\tilde e-Q_2|\les \eta.
\end{equation}
\end{rk}

  \begin{rk}
      The decay rate $\tau^{-3}$ in \eqref{eq:massive-decay} is sharp for massive particles \cite{RR92,EMV-stability}. The hierarchy of decay rates in \eqref{eq:massless-decay-1}--\eqref{eq:massless-decay-3} is sharp for massless particles \cite{TaylorPhD}.
  \end{rk}

\subsection{Estimates on the initial data}\label{sec:data-estimates}

For the remainder of \cref{sec:Vlasov-beams}, we assume the notation and hypotheses of \cref{prop:approx-1}.

\begin{lem}\label{lem:est-dat-1}
    For $\eta$, $\ve$, and $\mathfrak m_0$ sufficiently small and any $\lambda\in[-1,2]$ and $|M'-M|\le \delta$, $\mathcal S_{\lambda,M',\eta,\ve}$ is untrapped and consistent with particles of mass $0\le \mathfrak m\le \mathfrak m_0$. Let $(\mathcal U,r,\Omega^2,Q,f)$ be a development of $\mathcal S_{\lambda,M',\eta,\ve}$, such as the one obtained from \cref{prop:development-1}. Then the following holds:

  \begin{enumerate}
      \item The set $\mathcal U$ can be assumed to contain the corner $\mathcal C_{r_2}\cap\{v\le \frac 13r_1\}$. The solution is equal to Minkowski space in this region in the sense that
    \begin{gather*}
        r=\tfrac 12(v-u),\\
        \Omega^2=1,\\
        Q=f=0
    \end{gather*}
    on $\mathcal C_{r_2}\cap\{v\le \frac 13r_1\}$.

      \item Estimates on the initial data for the auxiliary beam:         
       \begin{equation}
        \mathring f_\mathrm{aux}^{\,r_1,\theta(\lambda)\eta}\les \theta(\lambda)\eta\mathbf 1_{[\frac 13 r_1,\frac 23r_1]\times[\Lambda-1,\Lambda+1]\times[\Lambda-1,\Lambda+1]}\label{eq:aux-1},
    \end{equation}
    \begin{equation}
        \sup_{v\in[\frac 13r_1,r_1]} |(\Omega^2-1,\partial_u{\log\Omega^2},\partial_v{\log\Omega^2},Q,m)(-v,v)|\les\theta(\lambda)\eta,\label{eq:aux-boundedness}
    \end{equation}
    \begin{equation}
        Q(-\tfrac 23r_1,\tfrac 23r_1)\gtrsim \theta(\lambda)\eta,\label{eq:charge-eta}
    \end{equation}
    \begin{equation}
        \ell(-v,v,p^u,p^v)\approx 1 \quad{\text{for every }(v,p^u,p^v)\in \spt(\mathring f_\mathrm{aux}^{\,r_1,\theta(\lambda)\eta})}.
        \label{eq:ell-aux}
    \end{equation}

      \item Estimates on the initial data for the main beam:
      \begin{equation}
  \label{eq:f-main-est}  \mathring f_\mathrm{main}^{\,\alpha_{\zeta(\lambda),M'},\theta(\lambda)\eta,\ve}\les \ve^{-3}\mathbf 1_{[r_1,r_2]\times[\ve,2\ve]\times[\ve,2\ve]}\cdot \begin{cases}
      1 & \text{if }\lambda >0 \\
      0 & \text{if }\lambda \le 0
  \end{cases},
\end{equation}
\begin{equation}
  \label{eq:main-boundedness-data}   \sup_{v\in[\frac 23 r_1,r_2]}  |(\Omega^2,\partial_u{\log\Omega^2},\partial_v{\log\Omega^2}, Q,m)(-v,v)-(\check\Omega^2,\check{\underline\omega},\check\omega,\check Q,\check m)(v)|\les \theta(\lambda)\eta,
\end{equation}
\begin{equation}
    \inf_{v\in [\frac 23r_1]}Q(-v,v)\gtrsim \theta(\lambda)\eta,\label{eq:Q-est-data-main}
\end{equation}
\begin{equation}
    \label{eq:ell-estimate-data} \ell(-v,v,p^u,p^v)\approx \ve\quad\text{for every }(v,p^u,p^v)\in \spt(\mathring f_\mathrm{main}^{\,\alpha_{\zeta(\lambda),M'},\theta(\lambda)\eta,\ve}),
\end{equation}
where 
\begin{equation}\label{eq:check-m-defn}
    \check m(v)\doteq \check \varpi(v)-\frac{\check Q^2(v)}{2v},\quad \check\Omega^2(v)\doteq \left(1-\frac{2\check m(v) }{v}\right)^{-1},\quad \check\omega\doteq- \check{\underline\omega}\doteq \frac{1}{2}\frac{d}{dv}\log \check\Omega^2(v).
\end{equation}
\item Estimates on the initial outgoing cone $C_{-r_2}$:
\begin{gather}
    \label{eq:cone-1}     \varpi(-r_2,v)=\varpi(-r_2,r_2),\\
\label{eq:cone-2}      Q(-r_2,v)=Q(-r_2,r_2),\\
      \label{eq:cone-3} 0\le    m(-r_2,v)\le 10M,\\
        \label{eq:cone-4}  r(-r_2,v)=\tfrac 12 v+\tfrac 12 r_2
\end{gather}
for $v\ge r_2$.
  \end{enumerate}
\end{lem}

\begin{proof} Consistency with particles of mass $\mathfrak m\le\mathfrak m_0$ follows immediately from \cref{def:S} and the estimates \eqref{eq:aux-boundedness} and \eqref{eq:main-boundedness-data} by taking $\mathfrak m_0$ sufficiently small. We therefore focus on proving the estimates and as a byproduct infer the untrapped property of $\mathcal S_{\lambda,M',\eta,\ve}$.

\ul{Part 1.} This is a restatement of \cref{rk:Mink-corner}.

 \ul{Part 2.}   The estimate \eqref{eq:aux-1} follows immediately from the definition \eqref{eq:aux-def}. Inserting the ansatz \eqref{eq:aux-1} into \eqref{eq:seed-1}--\eqref{eq:seed-3}, we find
\begin{align*}
   \mathring{\mathcal N}^u(r)=\mathring{\mathcal N}^v(r)&=\pi \theta(\lambda)\eta\Lambda \varphi\!\left(\frac{6}{r_1}r-3\right) ,\\
   \mathring {\mathcal T}^{uu}(r)=\mathring {\mathcal{T}}^{vv}(r)&=\pi \theta(\lambda)\eta\left(\Lambda^2+\int_{-1}^1 x^2\varphi(x)\,dx\right) \varphi\!\left(\frac{6}{r_1}r-3\right),\\
   \mathring{\mathcal T}^{uv}(r)&=\pi \theta(\lambda)\eta\Lambda^2 \varphi\!\left(\frac{6}{r_1}r-3\right)
\end{align*}
for $r\in[\frac 13r_1,\frac 23r_1]$. For $\eta$ sufficiently small, it then follows readily from the system \eqref{eq:constraint-1} and \eqref{eq:constraint-2} that $\mathring Q$ and $\mathring m$ are nonnegative, nondecreasing functions, and 
    \begin{equation*}
  0\le   \mathring Q(r)+\mathring m(r)\les \theta(\lambda) \eta
    \end{equation*}
    for $r\in [\frac 13r_1,\frac 23r_1]$ and
    \begin{equation*}
    \mathring Q(\tfrac 23 r_1)\gtrsim \theta(\lambda) \eta.
    \end{equation*}
    Using the definition of $\mathring\Omega{}^2$, we infer $|\mathring\Omega{}^2-1|\les \theta(\lambda)\eta$, and to estimate $|\partial_v{\log\Omega^2}(-v,v)|$, we observe that 
\begin{align*}
     | \partial_v{\log\Omega^2}(-v,v)|&= \left|\frac{d}{dv}\log\mathring\Omega{}^2(v)\right| \\
    &=\left| \mathring\Omega{}^{-2}\left(\frac{2\mathring m(v)}{v^2}-\frac{2}{v}\frac{d}{dv}\mathring m(v)\right)\right|\\
&=\left|\frac{2\mathring m}{\mathring \Omega{}^2v^2}-\frac{2}{\mathring\Omega{}^2 v}\left(\frac{\mathring\Omega{}^4 v^2}{4}\left(\mathring T{}^{uu}+2\mathring T{}^{uv}+\mathring T{}^{vv}\right)+\frac{\mathring Q{}^2}{2v^2}\right)\right|\\
&\les \theta(\lambda)\eta.
\end{align*}
The estimate for $\partial_u{\log\Omega^2}(-v,v)$ follows from \eqref{eq:time-symmetry-Omega}. This establishes \eqref{eq:aux-boundedness} and \eqref{eq:charge-eta}. Finally, \eqref{eq:ell-aux} follows from the mass shell relation and \eqref{eq:aux-boundedness}, provided $\mathfrak m_0$ is chosen sufficiently small.

\ul{Part 3.} The estimate \eqref{eq:f-main-est} follows immediately from the definition \eqref{eq:f-main}. Inserting the ansatz \eqref{eq:f-main} into \eqref{eq:seed-1}--\eqref{eq:seed-3}, we find 
\begin{align}
   \label{eq:N-main-est-data}   \mathring{\mathcal N}{}^u(r)=\mathring{\mathcal N}{}^v(r)&=\frac{1}{\mathfrak e r^2}\left(1 - \frac{2\check\varpi(r)}{r} + \frac{\check Q^2(r)}{r^2}\right)^2 \check Q'(r),\\
  \label{eq:Tuu-main-est-data}   \mathring{\mathcal T}^{uu}(r)=\mathring {\mathcal T}^{vv}(r)&=\frac{\ve}{6\mathfrak er^2}\left(9+\int_{-1}^1x^2\varphi(x)\,dx\right)\left(1 - \frac{2\check\varpi(r)}{r} + \frac{\check Q^2(r)}{r^2}\right)^2 \check Q'(r),\\
  \label{eq:Tuv-main-est-data}   \mathring{\mathcal T}^{uv}(r)&= \frac{3\ve}{4\mathfrak er^2}\left(1 - \frac{2\check\varpi(r)}{r} + \frac{\check Q^2(r)}{r^2}\right)^2 \check Q'(r),
\end{align} where $\check\varpi$ and $\check Q$ are obtained from $\varsigma(\alpha_{\lambda,M'})$.
 Inserting \eqref{eq:N-main-est-data}--\eqref{eq:Tuv-main-est-data} into \eqref{eq:constraint-1} and \eqref{eq:constraint-2} yields
    \begin{align*}
        \frac{d}{dr}\mathring m&=\frac{\mathring Q{}^2}{2r^2}+\mathrm{Err},\\
        \frac{d}{dr}\mathring Q&=\left(1-\frac{2\mathring m}{r}\right)^{-2}\left(1-\frac{2\check m}{r}\right)^2 \check Q',
    \end{align*}
where \[|\mathrm{Err}|\les \left(1-\frac{2\mathring m}{r}\right)^{-2}\ve\]
    and $\mathring m(r_1),\mathring Q(r_1)\les \theta(\lambda)\eta$.
    Therefore, by \eqref{eq:V-2-intro}, \eqref{eq:dust-untrapped}, and a simple Gr\"onwall and bootstrap argument, $\mathring m$ and $\mathring Q$ exist on $[r_1,r_2]$ and satisfy
    \begin{equation*}
      \sup_{r\in[r_1,r_2]}  |(\mathring m-\check m,\mathring Q-\check Q)(r)|\les \theta(\lambda)\eta.
    \end{equation*}
      This implies the same estimate for $|\mathring\Omega^2-\check\Omega^2|$ by definition. To estimate the other quantities, we may now argue as in the proof of Part 2.

      \ul{Part 4.} Equations \eqref{eq:cone-1}, \eqref{eq:cone-2}, and \eqref{eq:cone-4} follow immediately from the definitions. Inequality \eqref{eq:cone-3} follows from \eqref{eq:main-boundedness-data} provided $\eta$ is chosen sufficiently small.
\end{proof}

\subsection{The main beam in the near region}\label{sec:near-main}

For $\tau_0>0$, let
  \begin{equation}
      \mathcal R_\mathrm{main}^{\tau_0}\doteq\{0\le \tau\le \tau_0\}\cap\{-\tfrac 23r_1\le u\le -r_2\}\subset\mathcal C_{r_2}.\label{eq:def-R-main}
  \end{equation}

 \begin{figure}
\centering{
\def\svgwidth{19pc}
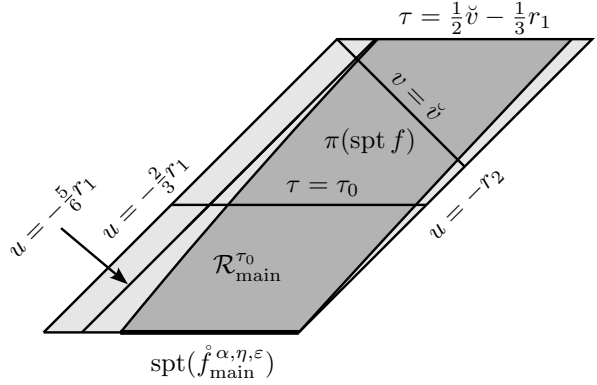}
\caption{Penrose diagram of the bootstrap region $\mathcal R^{\tau_0}_\mathrm{main}$ used in the proof of \cref{lem:approx-1}.}\label{fig:R-main}
\end{figure}

\begin{lem}\label{lem:approx-1}   For any $\breve v$, $\eta$, $\ve$, and $\mathfrak m_0$ satisfying \eqref{eq:parameter-smallness}, the following holds. Any normalized development $(\mathcal U,r,\Omega^2,Q,f)$ of $\mathcal S_{\lambda,M',\eta,\ve}$, such as the one obtained from \cref{prop:development-1}, can be uniquely extended to $\mathcal R_\mathrm{main}^{\frac 12 \breve v-\frac 13r_1}$. Moreover, the solution satisfies the estimates
\begin{gather}
  \label{eq:lem-1-1}  0\le m\le 10M,\quad  0\le Q\le 6M,\\
  \label{eq:lem-1-2}   r\sim v,\quad \Omega^2\sim 1,\\
    \label{eq:lem-1-3} \partial_vr \sim1,\quad   |\partial_ur|\les 1,\\
  |\partial_u\Omega^2|\les 1,\quad  \label{eq:lem-1-4}|\partial_v\Omega^2|\les v^{-3}
\end{gather}
on $\mathcal R_\mathrm{main}^{\frac 12 \breve v-\frac 13r_1}$ and 
\begin{equation}
 \label{eq:lem-1-5}    \tfrac 13\le \partial_vr \le \tfrac 23,\quad \partial_ur \sim -1
\end{equation}
on $\mathcal R_\mathrm{main}^{\frac 12 \breve v-\frac 13r_1}\cap \underline C_{\breve v}$. Finally, the support of the distribution function satisfies 
\begin{equation}
   \label{eq:lem-1-6}  \pi(\spt f)\cap \mathcal R^{\frac 12\breve v-\frac 13r_1}_\mathrm{main}\subset\{-\tfrac 56r_1\le u\le -r_2\}
\end{equation}
and if $u\in[-r_2,-\frac 56r_1]$, $p^u$, and $p^v$ are such that $f(u,\breve v,p^u,p^v)\ne 0$, then \begin{equation}
    \frac{p^u}{p^v}\les \breve v^{-2},\quad \frac{\ell^2}{p^v}\les 1\label{eq:main-final-pv}. 
\end{equation} When $\lambda\le 0$, \eqref{eq:lem-1-1} reads instead
\begin{equation*}
    0\le m\les \theta(\lambda)\eta,\quad 0\le Q\les \theta(\lambda)\eta.
\end{equation*}
 \end{lem}

The proof of \cref{lem:approx-1} will be given on \cpageref{pf-of-lemma}. We will make use of a bootstrap argument in the regions $\mathcal R^{\tau_0}_\mathrm{main}$, where $\tau_0$ ranges over $[0,\tfrac 12\breve v-\tfrac 13 r_1]$. For the basic geometric setup of the lemma and its proof, refer to \cref{fig:R-main}. As the proof is much simpler when $\lambda\le 0$ (the main beam is absent), we focus only on the case $\lambda >0$, in which case $\theta(\lambda)\eta=\eta$. 

We first make some definitions that will be used to define the bootstrap assumptions. Let $C_1>0$ be a constant such that 
\begin{equation*}
    C_1^{-1}\le \left(1-\frac{2\check m(v)}{v}\right)^{-1}\le C_1
\end{equation*}
for $v\in [\frac 23 r_1,r_2]$, where $\check m$ is given by \eqref{eq:check-m-defn}. (Recall that $\check m(v)=0$ for $v\le r_1$.) We then define 
\begin{gather*}
         C_2 \doteq 8C_1\left(\frac{3r_2}{2r_1}-1\right)\left(5M+\frac{18M^2}{r_1}\right), \\
   C_3 \doteq 2\max_{v\in[\frac 23r_1,r_2]}|\check{\underline\omega}(v)| + 100C_1e^{C_2}\left(5M+\frac{27M^2}{r_1}\right)\int_{\frac 23r_1}^\infty \frac{dv}{(\frac 23(1-\frac 16e^{-C_2})r_1+\frac 16 e^{-C_2}v)^3},
\end{gather*}
The constants $C_1,C_2$, and $C_3$ do not depend on $\eta$, $\ve$, or $\mathfrak m_0$. 

The quantitative bootstrap assumptions for the proof of \cref{lem:approx-1} are
\begin{gather}
 \label{eq:boot-1}   \tfrac{1}{6}e^{-C_2}\le \partial_vr\le \tfrac{3}{2}e^{C_2},\\
 \label{eq:boot-2} \tfrac 18 C_1^{-1} \le \Omega^{-2}\partial_vr\le  C_1,\\
 \label{eq:boot-3}   |\partial_u{\log\Omega^2}|\le C_3,\\
\label{eq:boot-5}    \varpi\le 5M,\\
\label{eq:boot-6}  N^v\le Ae^{B\tau},
\end{gather} on $\mathcal R_\mathrm{main}^{\tau_0}$
where $A\ge 1$ and $B\ge 1$ are constants to be determined which may depend on $\breve v$ and $\eta$, but not on $\ve$. We now derive some consequences of the bootstrap assumptions for the geometry of the solution.

\begin{lem}\label{lem:geometry-near-main}
If \eqref{eq:parameter-smallness} holds, $\tau_0\in[0,\frac 12 \breve v-\frac 13 r_1]$, $\mathcal R^{\tau_0}_\mathrm{main}\subset\mathcal U$, and the bootstrap assumptions \eqref{eq:boot-1}--\eqref{eq:boot-5} hold on $\mathcal R^{\tau_0}_\mathrm{main}$, then 
\begin{equation}
   \eta\les Q\le 6M,\label{eq:boot-9}
\end{equation}
\begin{equation}
   \tfrac 23(1-\tfrac 16 e^{-C_2}) r_1+\tfrac 16 e^{-C_2}v\le r\le r_2+\tfrac 32 e^{C_2}v,\label{eq:boot-10}
\end{equation}
\begin{equation}
   \Omega^2\approx 1,\label{eq:boot-11}
\end{equation}
\begin{equation}
    |\partial_ur|\les 1\label{eq:boot-12}
\end{equation}
on $\mathcal R^{\tau_0}_\mathrm{main}$.
\end{lem}
We will frequently use that \eqref{eq:boot-10} implies
\begin{equation*}
    r\sim v
\end{equation*} on $\mathcal R^{\tau_0}_\mathrm{main}$ without further comment.
\begin{proof}
For $\eta$ and $\ve$ sufficiently small, $\eta\les Q\le 6M$ on $\{\tau=0\}\cap\{-r_2\le u\le -\tfrac 23r_1\}$ and $\{\tau\ge 0\}\cap\{v=r_2\}$ by \cref{lem:est-dat-1}. Since $N^v\ge 0$ by definition, Maxwell's equation \eqref{eq:Max-u} implies the upper bound in \eqref{eq:boot-9}. The lower bound also follows from Maxwell's equation \eqref{eq:Max-v} and $N^u\ge 0$. The inequality \eqref{eq:boot-10} follows from integrating the bootstrap assumption \eqref{eq:boot-1}. The inequality \eqref{eq:boot-11} follows directly by multiplying the bootstrap assumptions \eqref{eq:boot-1} and \eqref{eq:boot-2}. To estimate $\partial_u r$, we rewrite the definition of the Hawking mass \eqref{eq:Hawking-mass} and the renormalized Hawking mass \eqref{eq:varpi-definition} as
\begin{equation}
  \label{eq:solved-for-nu}  \partial_ur=-\frac{1}{4}\left(1-\frac{2\varpi}{r}+\frac{Q^2}{r^2}\right)\frac{\Omega^2}{\partial_vr}.
\end{equation}
Now \eqref{eq:boot-12} follows immediately from \eqref{eq:boot-2}, \eqref{eq:boot-5}, and \eqref{eq:boot-10}.
\end{proof}

 We now use the basic geometric control obtained in \cref{lem:geometry-near-main} to obtain crucial control of the electromagnetic geodesic flow. It is convenient to first introduce some notation. Let $\Gamma_f$ denote the set of maximally extended electromagnetic geodesics $\gamma:I\to \mathcal U$, where $I$ is an interval, such that $(\gamma,p)(I)\subset\spt f$, where $p=d\gamma/ds$. If $\gamma$ passes through the point $(u,v)$, we denote by $s_{u,v}$ the parameter value such that $\gamma(s_{u,v})=(u,v)$. Let $\Gamma_f(u,v)$ denote the subset of $\Gamma_f$ consisting of curves passing through $(u,v)$. Note that every curve in $\Gamma_f$ intersects $\mathcal C_{r_2}\cap\{\tau=0\}$.

\begin{lem}\label{lem:pv-spread}
    If \eqref{eq:parameter-smallness} holds, $A$ is sufficiently large depending only on $\alpha$, $\tau_0\in[0,\frac 12 \breve v-\frac 13 r_1]$, $\mathcal R^{\tau_0}_\mathrm{main}\subset\mathcal U$, the bootstrap assumptions \eqref{eq:boot-1}--\eqref{eq:boot-6} hold on $\mathcal R^{\tau_0}_\mathrm{main}$, and $(u,v)\in\mathcal R^{\tau_0}_\mathrm{main}$, then
    \begin{equation}
        \mathscr V(u,v)\les \frac{\ve^3}{\eta(\ve+\eta\min\{\tau,1\})}\left(1+\frac{A}{B}e^{B\tau}\right),\label{eq:volume-main}
    \end{equation}
    where $\mathscr V$ is the phase space volume function defined by \eqref{eq:V-v}. Furthermore, if $\gamma\in\Gamma_f(u,v)$, then
    \begin{equation}
         0<u-u_0\les \eta^{-1}\ve,\label{eq:u-spread-1}
    \end{equation}
    \begin{equation}
       \ve+\eta\min\{\tau,1\} \les p^v(s_{u,v})\les \ve+\min\{\tau,1\}.\label{eq:p-v-main-1}
    \end{equation}
    where $u_0$ is the retarded time coordinate of the intersection of $\gamma$ with $\{\tau=0\}$.
\end{lem}
\begin{proof}
    Let $\gamma\in \Gamma_f(u,v)$. We will use the Lorentz force written in the form of equation \eqref{eq:SS-Lor-5} to estimate $p^v$. Since $(u,v)\in\mathcal R_\mathrm{main}^{\tau_0}$, $\gamma$ intersects $\{\tau=0\}$ in $\spt(\mathring f_\mathrm{main}^{\,\alpha,\ve})$ and therefore has angular momentum $\ell\sim \ve$. By the bootstrap assumptions and \cref{lem:geometry-near-main}, it holds that 
\begin{equation}
\left|\left(\partial_u{\log\Omega^2} - \frac{2\partial_ur}{r} \right)\frac{\ell^2}{r^2} \right|\les \frac{\ve^2}{v^2}\label{eq:2a1}
\end{equation} along the entire length of $\gamma$. Let $(u_0,v_0)$ be the coordinate of the intersection of $\gamma$ with $\{\tau=0\}$. Using \eqref{eq:boot-9}--\eqref{eq:boot-11} and the fact that $p^v(s_{u_0,v_0})\in[\ve,2\ve]$, we have
\begin{equation*}
 \left.   \mathfrak e \frac{Q}{r^2}(\Omega^2p^v)\right|_{s=0}\gtrsim  \eta\ve.
\end{equation*}
If $\ve$ is sufficiently small (independent of $\gamma$, but depending on $\eta$), these estimates show that 
\begin{equation*}
    \frac{d}{ds}(\Omega^2p^v)\gtrsim \eta\ve>0
\end{equation*}
along $\gamma$. Using \eqref{eq:boot-11}, we see that 
\begin{equation}
    p^v\gtrsim \ve\label{eq:2a2}
\end{equation}
along  $\gamma$. 

It is now convenient to parametrize $\gamma$ by the advanced time coordinate $v$ of the spacetime. The Lorentz force equation then becomes 
 \begin{align}
 \label{eq:gamma-u}\frac{d\gamma^u}{dv}&=\frac{p^u}{p^v}=\frac{\ell^2r^{-2}+\mathfrak m^2}{\Omega^2(p^v)^2},\\
  \label{eq:p-u-v} \frac{d}{dv} (\Omega^2 p^v) &=  \frac{1}{p^v}\left(\partial_u{\log\Omega^2} - \frac{2\partial_ur}{r}  \right)\frac{\ell^2}{r^2} + \mathfrak e  \frac{Q}{r^2}\Omega^2,
\end{align} where however $p^v$ is still given by $d\gamma^v/ds$ and we have used the mass shell relation in \eqref{eq:gamma-u}.
By \eqref{eq:2a1} and \eqref{eq:2a2}, 
\begin{equation*}
   \left|\frac{1}{p^v} \left(\partial_u{\log\Omega^2} - \frac{2\partial_ur}{r}  \right)\frac{\ell^2}{r^2}\right|\les \frac{\ve}{v^2}
\end{equation*}
along $\gamma$. This implies, using $v_0\ge r_1$ and hence $\int_{v_0}^\infty v'^{-2}\,dv'\les 1$, that
\begin{equation}\label{eq:pv-formula}
\left|   \left.\Omega^2p^v\right|_{(\gamma^u(v),v)} -\int_{v_0}^v \left. \mathfrak e \frac{Q}{r^2}\Omega^2\right|_{(\gamma^u(v'),v')} dv'\right|\les \ve.
\end{equation}
Using \cref{lem:geometry-near-main}, we readily deduce that 
\begin{equation}
    \int_{v_0}^v \left. \mathfrak e \frac{Q}{r^2}\Omega^2\right|_{(\gamma^u(v'),v')}dv'\les 1\label{eq:pv-upper-bound-proof}
\end{equation}
and 
\begin{equation*}
    \int_{v_0}^v \left. \mathfrak e \frac{Q}{r^2}\Omega^2\right|_{(\gamma^u(v'),v')}dv'\gtrsim \eta\left(\frac{1}{v_0}-\frac{1}{v}\right)\gtrsim \eta\min\{1,v-v_0\}.
\end{equation*}
Combining this with \eqref{eq:2a2} and \eqref{eq:pv-formula}, we deduce
\begin{equation}
\ve+\eta\min\{1,v-v_0\}  \les  p^v(v)\les 1\label{eq:p-1}
\end{equation}
along $\gamma$.
    
We are now able to prove \eqref{eq:u-spread-1}. Since $r\les v\les\breve v$, $r^2\mathfrak m^2\les \ve^2\les \ell^2$ for $\mathfrak m_0$ sufficiently small while respecting the hierarchy \eqref{eq:parameter-smallness}. Therefore, using also \cref{lem:geometry-near-main} and \eqref{eq:p-1}, we find
\begin{equation}\label{eq:gamma-u-2}
    \frac{d\gamma^u}{dv} \les \frac{\ve^2}{v^2(\ve+\eta\min\{1,v-v_0\})^2}.
\end{equation}
If $v\in [v_0,v_0+1]$, we compute
\begin{equation*}
    \int_{v_0}^v \frac{\ve^2}{v'^2(\ve+\eta(v'-v_0))^2} \,dv'\les \frac{\ve(v-v_0)}{\ve+\eta(v-v_0)}\les \eta^{-1}\ve
\end{equation*}
and if $v\in[v_0+1,\infty)$, we compute 
\begin{equation*}
    \int_{v_0+1}^v \frac{\ve^2}{v'^2(\ve+\eta)^2} \,dv'\les \eta^{-2}\ve^2\les \eta^{-1}\ve,
\end{equation*}
for $\ve\le \eta$. Therefore, integrating \eqref{eq:gamma-u-2}, we find 
\begin{equation}
   u-u_0= \gamma^u(v)-\gamma^u(v_0)=\int_{v_0}^v \frac{d\gamma^u}{dv}\,dv' \les \eta^{-1}\ve,\label{eq:d1}
\end{equation} which proves \eqref{eq:u-spread-1}. 

Now $u_0=-v_0$, so \eqref{eq:d1} implies $u + v_0\les \eta^{-1}\ve$.
Therefore, we have
\begin{equation*}
   2\tau= v+u = ( v-v_0)+(u + v_0)\les \eta^{-1}\ve +v-v_0
\end{equation*}
and 
\begin{equation*}
    \ve+\eta\tau \les  \ve+\eta(v-v_0) \les p^v(v)
\end{equation*}
for $\tau\les 1$. This, together with \eqref{eq:pv-formula} and \eqref{eq:pv-upper-bound-proof}, proves \eqref{eq:p-v-main-1}.

To prove \eqref{eq:volume-main}, we use the approximate representation formula \eqref{eq:pv-formula} for $p^v$ and the change of variables formula \eqref{eq:V-v}. Using the bootstrap assumptions, \cref{lem:geometry-near-main}, the mean value theorem, Maxwell's equation \eqref{eq:Max-u}, and the estimate \eqref{eq:u-spread-1},  we have 
    \begin{equation*}
      \left|\left.\frac{Q}{r^2}\Omega^2\right|_{(\gamma^u(v'),v')}-\left.\frac{Q}{r^2}\Omega^2\right|_{(u,v')}\right|\les \left(1+\sup_{[\gamma^u(v'),u]\times \{v'\}}|\partial_uQ|\right)\left(\gamma^u(v')-u\right)\les A\eta^{-1}\ve e^{B(u+v')/2}
    \end{equation*}
    for every $v'\in [v_0,v]$ and $A$ sufficiently large depending only on $\alpha$. Using this and \eqref{eq:pv-formula}, we find
\begin{align}
\nonumber  \left|   \left. \Omega^2 p^v\right|_{(\gamma^u(v),v)}-\int_{v_0}^v\left.\mathfrak e\frac{Q}{r^2}\Omega^2\right|_{(u,v')} \,dv'\right|   &\les \left|   \left.\Omega^2p^v\right|_{(\gamma^u(v),v)} -\int_{v_0}^v \left.\mathfrak e\frac{Q}{r^2}\Omega^2\right|_{(\gamma^u(v'),v')} \,dv'\right|\\
   \nonumber &\quad\quad +\int_{v_0}^v\left|\left.\frac{Q}{r^2}\Omega^2\right|_{(\gamma^u(v'),v')}-\left.\frac{Q}{r^2}\Omega^2\right|_{(u,v')}\right|dv'\\
\nonumber    &\les \ve+ \int_{v_0}^vA\eta^{-1}\ve e^{B(u+v')/2}\,dv'\\
    &\le\ve\left(1+\frac{A}{B\eta} e^{B(u+v)/2}\right).\label{eq:pv-est-2}
\end{align}
 Next, we estimate
    \begin{equation}
  0\le  \int_{-u}^{v_0} \left.\mathfrak e\frac{Q}{r^2}\Omega^2\right|_{(u,v')}\,dv'\les v_0+u=u-u_0\les \eta^{-1}\ve.\label{eq:pv-est-3}
    \end{equation}
    Combining \eqref{eq:pv-est-2} and \eqref{eq:pv-est-3} yields
    \begin{equation}
      \left|   \left. \Omega^2 p^v\right|_{(u,v)}-\int_{-u}^v \left.\mathfrak e\frac{Q}{r^2}\Omega^2\right|_{(u,v')} \,dv'\right|\les \frac\ve\eta\left(1+\frac{A}{B} e^{B(u+v)/2}\right).\label{eq:pv-est-4}
    \end{equation}
Therefore, if $\gamma_1,\gamma_2\in \Gamma_f(u,v)$ and we parametrize both by advanced time, and denote the $v$-momentum of $\gamma_i$ by $p^v_i$ for $i=1,2$, we find
\begin{align*}
|p^v_1(v)-p^v_2(v)|&\le  \left|   p^v_1(v)-\frac{1}{\Omega^{2}(u,v)}\int_{-u}^v \left.\mathfrak e\frac{Q}{r^2}\Omega^2\right|_{(u,v')} \,dv'\right|+ \left|   p^v_2(v)-\frac{1}{\Omega^{2}(u,v)}\int_{-u}^v \left.\mathfrak e\frac{Q}{r^2}\Omega^2\right|_{(u,v')} \,dv'\right|\\
&\les\frac\ve\eta\left(1+\frac{A}{B} e^{B(u+v)/2}\right).
\end{align*}
Inserting this estimate, \eqref{eq:p-v-main-1}, and $\ell\sim \ve$ in \eqref{eq:V-v} yields \eqref{eq:volume-main}, as desired.\end{proof}

\begin{proof}[Proof of \cref{lem:approx-1}] \label{pf-of-lemma}
    The proof is a bootstrap argument based on the bootstrap assumptions \eqref{eq:boot-1}--\eqref{eq:boot-6} and continuation criterion given by the extension principle \cref{prop:ext}. Let 
   \begin{equation*}
        \mathcal A\doteq \{\tau_0\in[0,\tfrac 12 \breve v-\tfrac 13 r_1]:\text{the solution extends uniquely to $\mathcal R^{\tau_0}_\mathrm{main}$ and \eqref{eq:boot-1}--\eqref{eq:boot-6} hold on $\mathcal R^{\tau_0}_\mathrm{main}$}\}. 
    \end{equation*}
    The set $\mathcal A$ is nonempty  by \cref{prop:development-1} if $A$ is chosen sufficiently large and $\eta$, $\ve$, and $\mathfrak m_0$ are sufficiently small. It is also manifestly connected and closed by continuity of the bootstrap assumptions. We now show that if $A$ and $B$ are sufficiently large depending on $\eta$ and \eqref{eq:parameter-smallness} holds, then $\mathcal A$ is also open. 
    
Let $\tau_0\in\mathcal A$.    First, we use \cref{lem:geometry-near-main,lem:pv-spread} to estimate $N^v$ and improve \eqref{eq:boot-6}. Since $f$ is transported along electromagnetic geodesics, we have $f(u,v,p^u,p^v)\les\ve^{-3}$ for $(u,v)\in \mathcal R^{\tau_0}_\mathrm{main}$. Using \eqref{eq:volume-main} and \eqref{eq:p-v-main-1}, we infer directly from the definition of $N^v$ that
\begin{equation*}
    N^v(u,v)\les (\ve+\min\{\tau,1\}) \ve^{-3}\mathscr V(u,v)\les \frac{\ve+\min\{\tau,1\}}{\eta(\ve+\eta\min\{\tau,1\})}\left(1+\frac{A}{B}\right)e^{B\tau}\les \eta^{-2}\left(1+\frac{A}{B}\right)e^{B\tau}.
\end{equation*}
Letting $C_*= C_*(\mathfrak e,\varphi, \Lambda,\alpha)$ denote the implicit constant in this inequality, and choosing $A= 4C_*\eta^{-2}$ and $B\ge A$, we see that 
\begin{equation*}
    N^v(u,v)\le \tfrac 12 Ae^{B\tau},
\end{equation*}
which improves \eqref{eq:boot-6}. 

    To continue, we now estimate $N^u$, $T^{uu}$, and $T^{uv}$ in the same fashion, making use now of the strong decay of $p^u$. If $\gamma\in \Gamma_f(u,v)$, then
    \begin{equation*}
        p^u(s_{u,v})=\frac{\ell^2r^{-2}+\mathfrak m^2}{\Omega^2p^v(s_{u,v})}\les \frac{\ve^2}{\ve+\eta\min\{\tau,1\}}
    \end{equation*}
    by \eqref{eq:p-v-main-1}. Using \eqref{eq:volume-main}, we therefore find 
    \begin{align}
    \label{eq:Nu-est} N^u&\les \frac{\ve^2}{\eta}\frac{1}{(\ve+\eta\min\{1,\tau\})^2}Ae^{B\tau},\\
    \label{eq:Tuv-est}      T^{uv}  &\les \frac{\ve^2}{\eta}\frac{\ve+\min\{\tau,1\}}{\ve+\eta\min\{1,\tau\}}Ae^{B\tau},\\
  \label{eq:Tuu-est} T^{uu} &\les \frac{\ve^4}{\eta}\frac{1}{(\ve+\eta\min\{1,\tau\})^3}Ae^{B\tau}.
    \end{align}
    Using the hierarchy \eqref{eq:parameter-smallness}, these estimates imply  
    \begin{equation}\label{eq:lem-1-total}
      r^2T^{uv}+r^2T^{uu}+  \int^v_{-u}r^2N^u(u,v') \,dv'\les \ve^{1/2}
    \end{equation} for any $(u,v)\in \mathcal R^{\tau_0}_\mathrm{main}$. With these final estimates in hand, we may begin to improve the remaining bootstrap assumptions \eqref{eq:boot-1}--\eqref{eq:boot-5}. We will then carry out the rest of the continuity argument and prove all of the stated conclusions of the lemma.

 Improving \eqref{eq:boot-1}: The wave equation \eqref{eq:r-wave} can be rewritten as
 \begin{equation*}
     \partial_u\partial_vr = -\frac{1}{2r^2}\frac{\Omega^2}{\partial_vr}\left(\varpi-\frac{Q^2}{r}\right)\partial_vr+\frac{r}{4}\Omega^4 T^{uv}.
 \end{equation*}
 Using an integrating factor, we find
 \begin{equation}
\partial_u\left[\exp\left(\int_{-r_2}^u\frac{1}{2r}\frac{\Omega^2}{\partial_vr}\left(\varpi-\frac{Q^2}{r}\right)du'\right)\partial_vr\right]= \frac r4\Omega^4 T^{uv}\exp\left(\int_{-r_2}^u\frac{1}{2r}\frac{\Omega^2}{\partial_vr}\left(\varpi-\frac{Q^2}{r}\right)du'\right),\label{eq:step-1-proof-1}
 \end{equation}where the integral is taken over fixed $v$. The bootstrap assumptions imply
 \begin{equation*}
     \int_{-r_2}^{-\frac 23r_1}\frac{1}{2r}\frac{\Omega^2}{\partial_vr}\left|\varpi-\frac{Q^2}{r}\right|du'\le C_2.
 \end{equation*}
 For $\ve$ sufficiently small, the right-hand side of \eqref{eq:step-1-proof-1} is pointwise $\le \frac {1}{10}$ on $\mathcal R^{\tau_0}_\mathrm{main}$, so integrating this equation yields 
 \begin{equation*}
     \tfrac{2}{5}e^{-C_2}\le \partial_vr(u,v)\le \tfrac 35 e^{C_2}
 \end{equation*}
 for any $(u,v)\in\mathcal R^{\tau_0}_\mathrm{main}$, which improves \eqref{eq:boot-1}.

Improving \eqref{eq:boot-2}: Raychaudhuri's equation \eqref{eq:Ray-v} can be rewritten as
\begin{equation}
    \partial_v\log\left(\frac{\partial_vr}{\Omega^2}\right)=-\frac{r}{4}\partial_vr T^{uu}. \label{eq:step-1-proof-2}
\end{equation}
To improve the upper bound in \eqref{eq:boot-2}, note that the right-hand side of \eqref{eq:step-1-proof-2} is nonpositive by \eqref{eq:boot-1} and hence 
\begin{equation*}
    \frac{\partial_vr}{\Omega^2}(u,v)\le \frac{\partial_vr}{\Omega^2}(-v,v)=\frac 12\left(1-\frac{2\check m(v)}{v}\right) \le \frac{C_1}{2}.
\end{equation*}
To improve the lower bound, note that the right-hand side of \eqref{eq:step-1-proof-2} is bounded by $\log 2$ in absolute value for $\ve$ sufficiently small and hence 
\begin{equation*}
    \frac{\partial_vr}{\Omega^2}(u,v)\ge \frac 12 \frac{\partial_vr}{\Omega^2}(-v,v)\ge \frac{1}{4C_1},
\end{equation*}
which improves \eqref{eq:boot-2}. 

Improving \eqref{eq:boot-3}: The wave equation \eqref{eq:Omega-wave}  can be rewritten as
\begin{equation}
    \partial_v\partial_u{\log\Omega^2}= \frac{\Omega^2}{r^3}\left(\varpi-\frac{3Q^2}{2r}\right) -\frac 12 \Omega^4T^{uv}-\Omega^2 S.\label{eq:step-1-proof-3}
\end{equation}
Integrating this equation in $v$ and using the bootstrap assumptions, \eqref{eq:main-boundedness-data},  and \eqref{eq:lem-1-total} yields 
\begin{equation*}
    |\partial_u{\log\Omega}^2|\le \frac{3}{4}C_3
\end{equation*}
for $\eta$ and $\ve$ sufficiently small, which improves \eqref{eq:boot-3}.

Improving \eqref{eq:boot-5}: Integrating the evolution equation for the renormalized Hawking mass \eqref{eq:mod-Hawking-v} and using \eqref{eq:lem-1-total}, we have
\begin{equation*}
      |\varpi(u,v)-\varpi(-v,v)|\les \ve^{1/2},
\end{equation*}
which improves \eqref{eq:boot-5}.

We have thus improved the constants in all of the bootstrap assumptions \eqref{eq:boot-1}--\eqref{eq:boot-6}. Using the local existence theory \cref{prop:char-IVP-Vlasov} and generalized extension principle \cref{prop:ext}, there exists a $\tau_0'>\tau_0$ such that $\mathcal U\subset\mathcal R^{\tau_0'}_\mathrm{main}$. Choosing $\tau_0'>\tau_0$ perhaps smaller, the bootstrap assumptions \eqref{eq:boot-1}--\eqref{eq:boot-6} extend to $\mathcal R^{\tau_0'}_\mathrm{main}$ by continuity. Therefore, $\mathcal A$ is open and the bootstrap argument is complete. 

We now prove the remaining conclusions of the lemma. First, $m(u,v)\ge 0$ for every $(u,v)\in\mathcal R_\mathrm{main}^{\frac 12\breve v-\frac 13r_1}$ because either $\partial_ur(u,v)<0$ (and by Raychaudhuri \eqref{eq:Ray-u} also for any $u'> u$) or $\partial_ur(u,v)\ge 0$ and then $m(u,v)\ge 0$ directly from the definition \eqref{eq:Hawking-mass}. In the first case, the evolution equation \eqref{eq:Hawking-v} implies $m$ is nondecreasing along the outgoing cone terminating at $(u,v)$. Since this cone either intersects $\{\tau=0\}$, where $m\ge 0$, or a sphere where $\partial_ur\le 0$ (and hence $m\ge 0$), we conclude $m(u,v)\ge 0$. Now integrating the wave equation \eqref{eq:Omega-wave} in $u$ and using \eqref{eq:lem-1-total}, we see that $|\partial_v{\log\Omega^2}|\les v^{-3}$. Together with the bootstrap assumptions and \cref{lem:geometry-near-main}, this proves \eqref{eq:lem-1-1}--\eqref{eq:lem-1-4}. Next, \eqref{eq:lem-1-5} follows from integrating the wave equation \eqref{eq:r-wave} in $u$ along $\underline C_{\breve v}$ and taking $\breve v\sim r$ sufficiently large and similarly in \eqref{eq:solved-for-nu}. The inclusion \eqref{eq:lem-1-6} follows immediately from the $u$-deflection estimate \eqref{eq:u-spread-1} for electromagnetic geodesics in $\Gamma_f$ and the hierarchy \eqref{eq:parameter-smallness}. Finally, to prove \eqref{eq:main-final-pv} we use the mass shell relation, \eqref{eq:p-v-main-1}, and the parameter hierarchy to estimate
\begin{align*}
  r^2\frac{p^u}{p^v} &= \frac{\ell^2+r^2\mathfrak m^2}{\Omega^2(p^v)^2}\les \frac{\ve^2}{\eta^2}\les 1,\\
    \frac{\ell^2}{p^v}&\les \frac{\ve^2}{\eta}\les 1,
\end{align*} which completes the proof. \end{proof}

\subsection{The auxiliary beam in the near region}\label{sec:near-aux}

For $v_0>0$, let 
\begin{align}
     \mathcal R_\mathrm{aux}^{v_0}&\doteq\{v\ge u\} \cap\{\tau\ge 0\}\cap\{ -\tfrac 23r_1\le u\le \tfrac 13 r_1\}\cap\{\tfrac 13r_1\le v\le v_0\},\label{eq:def-R-aux}\\
     \tilde{\mathcal R}_\mathrm{aux}^{v_0}&\doteq \{v\ge u\}\cap \{\tau\ge 0\}\cap\{u\ge -\tfrac 23r_1\}\cap\{v\le v_0\}.
\end{align}

  \begin{lem}\label{lem:approx-2}
      For any $\breve v$, $\eta$, $\ve$, and $\mathfrak m_0$ satisfying \eqref{eq:parameter-smallness}, the following holds. The development of $\mathcal S_{\lambda,M',\eta,\ve}$ obtained in \cref{lem:approx-1} can be uniquely extended to $\tilde{\mathcal R}^{\breve v}_\mathrm{aux}$. The spacetime is vacuum for $u\ge \frac 13r_1$ and $v\le \breve v$. Moreover, the solution satisfies the estimates 
    \begin{gather*}
        0\le m\les \theta(\lambda)\eta,\quad 0\le Q\les\theta(\lambda)\eta,\\
        \Omega^2\sim 1,\quad
        \partial_vr\sim -\partial_ur \sim 1,\\
        (1+u^2)|\partial_u\Omega^2|+(1+v^2)|\partial_v\Omega|^2\les 1
    \end{gather*}
    on $\tilde{\mathcal R}^{\breve v}_\mathrm{aux}$ and 
    \begin{equation*}
   \tfrac 14   \le  \partial_vr\le \tfrac 34
    \end{equation*}
on $\tilde{\mathcal R}^{\breve v}_\mathrm{aux}\cap \underline C_{\breve v}$. Finally, the support of the distribution function satisfies 
\begin{equation*}
    \pi(\spt f) \cap \tilde{\mathcal R}^{\breve v}_\mathrm{aux}\subset\{-\tfrac 23r_1\le u\le \tfrac 16r_1\},
\end{equation*}
\begin{equation*}
    \inf_{\spt(f)\cap  \tilde{\mathcal R}^{\breve v}_\mathrm{aux}} r\ge \tfrac 16r_1,
\end{equation*}
and if $u\in[-\tfrac 23r_1,\tfrac 13r_1]$, $p^u$, and $p^v$ are such that $f(u,\breve v,p^u,p^v)\ne 0$, then 
\begin{equation}
    \frac{p^u}{p^v}\les \breve v^{-2},\quad \frac{\ell^2}{p^v}\les 1.\label{eq:aux-final-pv}
\end{equation}  \end{lem}

 \begin{figure}
\centering{
\def\svgwidth{11pc}
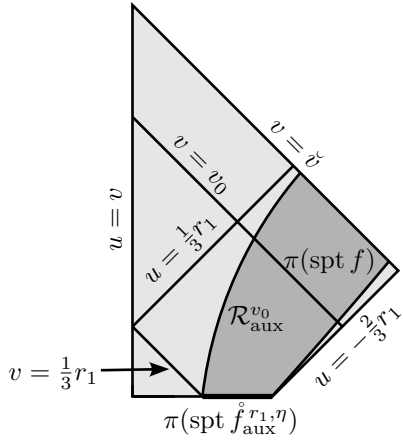}
\caption{Penrose diagram of the bootstrap region $\mathcal R^{v_0}_\mathrm{aux}$ used in the proof of \cref{lem:approx-2}.}\label{fig:R-aux}
\end{figure}

The proof of \cref{lem:approx-2} will be given on \cpageref{approx-2-pf}. We will make use of a bootstrap argument in the regions $\mathcal R^{v_0}_\mathrm{aux}$,
  where $v_0$ ranges over $[\frac 13r_1,\breve v]$. The triangle $\{v\ge u\}\cap\{u\ge \frac 13r_1\}\cap \{v\le \breve v\}$ is Minkowskian and can simply be attached at the very end of the argument, cf.~\cref{lem:Minkowski}. For the basic geometric setup of the lemma and its proof, refer to \cref{fig:R-aux}. 

  For $(u,v)\in \mathcal R_\mathrm{aux}^{\breve v}$, let
  \begin{gather*}
      \hat r(u,v)=\frac{r_1}{6}-\frac{u}{2}+\frac 12 \int_{\frac 13r_1}^v \beta(v')\,dv',\\
       \hat\Omega^2(u,v)=\beta(v),
  \end{gather*}
  where 
\begin{equation*}
    \beta(v)\doteq \begin{cases}
    1 & \text{if }v<\frac 23r_1\\
   \Omega^2(-\tfrac 23r_1,v) & \text{if }v\ge \tfrac 23 r_1  \end{cases}.
\end{equation*}
It is easily verified that $(\hat r,\hat\Omega^2)$ is a solution of the spherically symmetric Einstein vacuum equations and matches smoothly in $v$ with $(r,\Omega^2)$ from $\mathcal R_\mathrm{main}^{\frac 12\breve v-\frac 13r_1}$ along $C_{-\frac 23r_1}$.

The first bootstrap assumption is
\begin{equation}
    |Q|+|\varpi|+|\partial_vr-\partial_v\hat r|+|\partial_ur-\partial_u\hat r|+|\Omega^2-\hat\Omega^2|+|\partial_v\Omega^2-\partial_v\hat\Omega^2|+|\partial_u\Omega^2|\le A\theta(\lambda)\eta e^{B\tau},\label{eq:approx-2-boot-1}
\end{equation}
where $A\ge 1$ and $B\ge 1$ are constants to be determined that may depend on $\breve v$ but not $\eta$. We also make the following assumption on the electromagnetic geodesic flow. For $(v_0',p^u_0,p^v_0)\in\spt( \mathring f_\mathrm{aux}^{\,r_1,\theta(\lambda)\eta})$, let $\gamma$ be an electromagnetic geodesic of mass $\mathfrak m$ for $(r,\Omega^2,Q)$ starting at $(-v_0',v_0',p^u_0,p^v_0)$, and let $\hat\gamma$ be a \emph{null geodesic} for $(\hat r,\hat\Omega^2)$ starting at $(-v_0',v_0',p^u_0,p^v_0)$. Then, assuming both $\gamma$ and $\hat \gamma$ remain within $\mathcal R^{v_0}_\mathrm{aux}$, we assume that
\begin{equation}
    |\Omega^2p^u-\hat\Omega^2\hat p^u|+|\Omega^2p^v-\hat\Omega^2\hat p^v|\le A\theta(\lambda)\eta e^{B(\gamma^v-v_0')}.\label{eq:approx-2-boot-2}
\end{equation}

First, we note the following immediate consequences of the first bootstrap assumption:

\begin{lem}\label{lem:approx-2-1} If \eqref{eq:parameter-smallness} holds, $v_0\in [\tfrac 13r_1,\breve v]$, $\mathcal R^{v_0}_\mathrm{aux}\subset\mathcal U$, the bootstrap assumption \eqref{eq:approx-2-boot-1} holds on $\mathcal R^{v_0}_\mathrm{aux}$, and $\eta$ is sufficiently small depending on $A$ and $B$, then on $\mathcal R_\mathrm{aux}^{v_f}$ it holds that
\begin{gather*}
    0\le\varpi\les \theta(\lambda)\eta, \quad 0\le Q\les \theta(\lambda)\eta,\\
  |{\log\partial_vr}|+|{\log(-\partial_ur)}|+|{\log\Omega^2}|\les 1,\\
    |\partial_v\Omega^2|+|\partial_u\Omega^2|\les 1.
\end{gather*}
\end{lem}

Next, we use the second bootstrap assumption to obtain 

\begin{lem}\label{lem:approx-2-2}
    If \eqref{eq:parameter-smallness} holds, $v_0\in [\tfrac 13r_1,\breve v]$, $\mathcal R^{v_0}_\mathrm{aux}\subset\mathcal U$, the bootstrap assumptions \eqref{eq:approx-2-boot-1} and \eqref{eq:approx-2-boot-2} hold on $\mathcal R^{v_0}_\mathrm{aux}$, $B$ is sufficiently large, and $\eta$ is sufficiently small depending on $A$ and $B$, then the following holds. Let $\gamma:[0,S]\to \mathcal R_\mathrm{aux}^{v_0}$ be a future-directed electromagnetic geodesic starting in $\spt(\mathring f_\mathrm{aux}^{\,r_1,\theta(\lambda)\eta})$, then 
\begin{gather*}
      r \ge\tfrac 16r_1,\\
      u \le \tfrac 13 r_1,\\
     r^2\frac{p^u}{p^v}\les 1,\quad \frac{\ell^2}{p^v}\les 1
\end{gather*}
along $\gamma$.
\end{lem}
\begin{proof} Upon making the coordinate transformation \begin{equation}
    \tilde u=u,\quad \tilde v=\tfrac 13r_1+\int_{\frac 13r_1}^v\beta(v')\,dv',\label{eq:coordinate-trafo}
\end{equation}
the metric $(\hat r,\hat\Omega^2)$ is brought into the standard Minkowski form $(\frac 12(\tilde v-\tilde u),1)$. If $\tilde t\doteq \frac 12(\tilde v+\tilde u)$, $\hat\gamma$ is a null geodesic in $\mathcal R^{\breve v}_\mathrm{aux}$ with respect to $(\hat r,\hat\Omega^2)$ intersecting $\{\tau=0\}$ with momentum $(p_0^u,p^v_0)=(p_0^{\tilde u},p^{\tilde v}_0)$ at an area-radius of $\hat r_0$, then it is easy to check that
\begin{gather}
 \label{eq:Mink-geo-1}   \hat r^2=\left(\tilde t+\sign(p_0^v-p^u_0)\sqrt{\hat r_0^2-\frac{\hat\ell^2}{\hat E^2}}\right)^2+\frac{\hat\ell^2}{\hat E^2},\\
 \label{eq:Mink-geo-2}  p^{\tilde u}=\begin{cases}
\hat E+\sqrt{\hat E^2-\hat \ell^2/\hat r^2} & \text{if }\tilde t < - \sign(p_0^v-p^u_0)\sqrt{\hat r_0^2-\hat \ell^2/\hat E^2} \\
\hat E-\sqrt{\hat E^2-\hat \ell^2/\hat r^2} &\text{if }\tilde t \ge - \sign(p_0^v-p^u_0)\sqrt{\hat r_0^2-\hat \ell^2/\hat E^2}
     \end{cases},\\
  \label{eq:Mink-geo-3}      p^{\tilde v}=\begin{cases}
\hat E-\sqrt{\hat E^2-\hat \ell^2/\hat r^2} & \text{if }\tilde t < - \sign(p_0^v-p^u_0)\sqrt{\hat r_0^2-\hat \ell^2/\hat E^2} \\
\hat E+\sqrt{\hat E^2-\hat \ell^2/\hat r^2} &\text{if }\tilde t \ge - \sign(p_0^v-p^u_0)\sqrt{\hat r_0^2-\hat \ell^2/\hat E^2}
     \end{cases}
\end{gather}
along $\hat\gamma$, where $\hat\ell^2\doteq \hat r^2 p^{\tilde u}p^{\tilde v}$ and $\hat E\doteq \frac 12(p^{\tilde v}+p^{\tilde u})$ are conserved quantities.  If $(p^u_0,p^v_0)\in [\Lambda-1,\Lambda+1]^2$ and $\Lambda$ satisfies \eqref{eq:Lambda}, then 
\begin{equation*}
    \tfrac{81}{100}\hat r_0^2\le \frac{\hat\ell^2}{\hat E^2}\le \hat r_0^2.
\end{equation*}
From \eqref{eq:Mink-geo-1} it is apparent that 
\begin{align}
  \label{eq:aux-lem-1}  \min_{\hat \gamma}r&\ge \frac{\hat\ell}{\hat E}\ge \tfrac{9}{10} \hat r_0,\\
 \label{eq:aux-lem-2}   \sup_{\hat \gamma}\tilde u&=-\sign(p^v_0-p^u_0)\sqrt{\hat r_0^2-\frac{\hat\ell^2}{\hat E^2}}\le \tfrac{45}{100}\hat r_0,
\end{align}
and by inspection of \eqref{eq:Mink-geo-2} and \eqref{eq:Mink-geo-3} that 
\begin{equation}
\label{eq:aux-lem-3}  \hat r^2\frac{p^{\tilde u}}{p^{\tilde v}}\les 1,\quad \frac{\hat \ell^2}{p^{\tilde v}}\les 1
\end{equation}
along $\hat\gamma$. 

Let $\gamma$ and $\hat\gamma$ be as defined before \eqref{eq:approx-2-boot-2}. Parametrize $\gamma$ and $\hat\gamma$ by $v$ as in the proof of \cref{lem:pv-spread}. Then 
\begin{equation*}
    \left|\frac{d}{dv}(\gamma^u-\hat\gamma^u)\right|=\left|\frac{p^u}{p^v}-\frac{\hat p^u}{\hat p^v}\right|\les A\theta(\lambda)\eta e^{B(v-v_0')}
\end{equation*}
by \eqref{eq:approx-2-boot-2} and the observation that $p^v\gtrsim 1$ and $p^{\tilde v}\gtrsim 1$, so that 
\begin{equation}
    |\gamma^u-\hat\gamma^u|\le \theta(\lambda) \eta e^{B(v-v_0')}\label{eq:approx-2-u-deflection}
\end{equation}
for $B$ chosen sufficiently large. Therefore, the conclusions of the lemma follow from the estimates  \eqref{eq:aux-lem-1}--\eqref{eq:aux-lem-3} and the fact that $\hat r_0\in [\tfrac 13r_1,\tfrac 23r_1]$ after undoing the coordinate transformation \eqref{eq:coordinate-trafo} and applying the bootstrap assumptions. 
\end{proof}

\begin{proof}[Proof of \cref{lem:approx-2}]\label{approx-2-pf}
    The proof is a bootstrap argument based on the bootstrap assumptions \eqref{eq:approx-2-boot-1} and \eqref{eq:approx-2-boot-2}, the continuation criterion \cref{prop:ext}, and \cref{lem:Minkowski}. Define the bootstrap set 
    \begin{equation*}
        \mathcal A\doteq \{v_0\in[\tfrac 13r_1,\breve v]:\text{the solution extends uniquely to $\mathcal R_\mathrm{aux}^{v_0}$ and \eqref{eq:approx-2-boot-1}, \eqref{eq:approx-2-boot-2} hold on $\mathcal R_\mathrm{aux}^{v_0}$}\}.
    \end{equation*}
    The set $\mathcal A$ is nonempty by \cref{prop:strip,prop:development-1} if $A$ is chosen sufficiently large and is manifestly closed and connected. We now show that if the parameters satisfy \eqref{eq:parameter-smallness}, then $\mathcal A$ is also open.

Let $v_0\in\mathcal A$.    Taking $\mathfrak m_0$ sufficiently small and using the formula \eqref{eq:V-v}, the initial data estimate \cref{lem:est-dat-1}, and \cref{lem:approx-2-1,lem:approx-2-2}, we immediately find 
    \begin{equation}
     \label{eq:aux-lem-4}     N^u+N^v+T^{uu}+T^{uv}+T^{vv}+S\les r^{-2}\theta(\lambda)\eta\mathbf 1_{\{r\ge\frac 16r_1\}}
    \end{equation}
on $\mathcal R^{v_0}_\mathrm{aux}$. By \eqref{eq:aux-boundedness} and the observation that $f=0$ along $C_{-\frac 23r_1}\cap \mathcal R^{v_0}_\mathrm{aux}$, we have
\begin{equation}
  \label{eq:aux-lem-5}   |Q|+|\varpi|+|\partial_vr-\partial_v\hat r|+|\Omega^{-2}\partial_vr-\hat\Omega^{-2}\partial_v\hat r|+|\Omega^2-\hat\Omega^2|+|\partial_v\Omega^2-\partial_v\hat\Omega^2|+|\partial_u\Omega^2|\les \theta(\lambda)\eta
\end{equation}
along $\{\tau=0\}$ and $C_{-\frac 23r_1}$ in $\mathcal R^{v_0}_\mathrm{aux}$. Then, using \eqref{eq:aux-lem-4} and the Einstein--Maxwell--Vlasov system, we see that  \eqref{eq:aux-lem-5} holds on $\mathcal R^{v_0}_\mathrm{aux}$. This improves \eqref{eq:approx-2-boot-1} for an appropriate choice of $A$ (independent of $\eta$). 

We now improve \eqref{eq:approx-2-boot-2}.  Using the Lorentz force equations \eqref{eq:SS-Lor-4} and \eqref{eq:SS-Lor-5}, we have
\begin{align*}
    \frac{d}{dv}(\Omega^2p^u-\hat\Omega^2\hat p^u)&= \left.\left(\partial_v{\log\Omega^2}-\frac{2\partial_vr}{r}\right)\frac{\ell^2}{r^2p^v}\right|_{\gamma}-\left.\left(\partial_v{\log\hat\Omega^2}-\frac{2\partial_v\hat r}{\hat r}\right)\frac{\hat\ell^2}{\hat r^2\hat p^v}\right|_{\hat\gamma} -\mathfrak e \frac{\Omega^2 Q}{r^2}\frac{p^u}{p^v}, \\
    \frac{d}{dv}(\Omega^2p^v-\hat\Omega^2\hat p^v)&=\left.\left(\partial_u{\log\Omega^2}-\frac{2\partial_ur}{r}\right)\frac{\ell^2}{r^2p^v}\right|_{\gamma}+\left.\frac{2\partial_u\hat r}{\hat r}\frac{\hat\ell^2}{\hat r^2\hat p^v} \right|_{\hat\gamma}+\mathfrak e \frac{\Omega^2 Q}{r^2}.
\end{align*}
Using the parameter hierarchy \eqref{eq:parameter-smallness}, the bootstrap assumptions \eqref{eq:approx-2-boot-1} and \eqref{eq:approx-2-boot-2}, \cref{lem:approx-2-1,lem:approx-2-2}, and the bound \eqref{eq:approx-2-u-deflection}, we can estimate
\begin{equation*}
      \left| \frac{d}{dv}(\Omega^2p^u-\hat\Omega^2\hat p^u)\right|+\left| \frac{d}{dv}(\Omega^2p^u-\hat\Omega^2\hat p^u)\right|\les\theta(\lambda) \eta+ A\theta(\lambda)\eta e^{B(v-v_0)}.
\end{equation*}
Integrating and choosing the constants $A$ and $B$ sufficiently large in terms of the implied constants and $\breve v$ improves \eqref{eq:approx-2-boot-2} and shows that the solution extends to $\mathcal R_\mathrm{aux}^{\breve v}$. 

The solution is at once extended to $\tilde{\mathcal R}_\mathrm{aux}^{\breve v}$ by \cref{lem:Minkowski}. The rest of the conclusions of the lemma follow immediately from \cref{lem:Minkowski,lem:approx-1,lem:approx-2-1,lem:approx-2-2} and \eqref{eq:aux-lem-4}. \end{proof}

\subsection{The far region}\label{sec:far-massive}

 \begin{figure}
\centering{
\def\svgwidth{12pc}
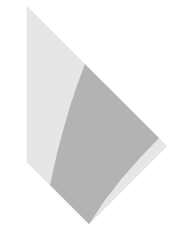}
\caption{Penrose diagram of the bootstrap region $\mathcal R^{\breve v,v_f}_\mathrm{far}$ used in the proof of \cref{lem:beam-far}.}\label{fig:R-far}
\end{figure}

\begin{lem}\label{lem:beam-far}
  For any $\breve v$, $\eta$, $\ve$, and $\mathfrak m_0$ satisfying \eqref{eq:parameter-smallness}, there exists a constant $C_\nu>0$ such that the following holds. The development of $\mathcal S_{\lambda,M',\eta,\ve}$ obtained in \cref{lem:approx-2} can be uniquely extended globally to $\mathcal C_{r_2}$. Moreover, the solution satisfies the estimates
\begin{gather}
  \label{eq:approx-3-1}  0\le m\le 10M,\quad 0\le Q\le 6M,\\
    \label{eq:approx-3-2}  \Omega^2\sim 1,\quad\partial_vr\sim-\partial_ur\sim 1,\\
    \label{eq:approx-3-3} (1+u^2)|\partial_u\Omega^2|\les 1,\quad v^2|\partial_v\Omega^2|\les 1
\end{gather}
on $\mathcal C_{r_2}\cap\{v\ge\breve v\}$, and the distribution function satisfies
    \begin{equation*}
        \pi(\spt f)\cap \{v\ge \breve v\}\subset \{6C_\nu u\le v\}.
    \end{equation*}
\end{lem}

We will make use of a bootstrap argument in the regions 
  \begin{equation*}
      \mathcal R_\mathrm{far}^{\breve v,v_f}\doteq \{v\ge u\}\cap\{\tau\ge 0\}\cap\{u\ge -r_2\}\cap\{\breve v\le v\le v_f\},
  \end{equation*}
where $v_f\ge \breve v$. Refer to \cref{fig:R-far}. The bootstrap assumptions are
\begin{gather}
\label{eq:boot-3-1}    -C_\nu\le \partial_ur \le -C_\nu^{-1},\\
 \label{eq:boot-3-2}   \tfrac 15 \le \partial_vr\le 1,\\
\label{eq:boot-3-3}    \pi(\spt f)\cap \mathcal R_\mathrm{far}^{\breve v,v_f}\subset \mathcal W,
\end{gather}
where 
\begin{equation*}
    \mathcal W \doteq\{6C_\nu u\le v\}\cap\{v\ge \breve v\}
\end{equation*}
and the constant $10\le C_\nu\les 1$ is chosen so that $-\tfrac 12 C_\nu \le \partial_ur \le -2C_\nu^{-1}$  on $\underline C_{\breve v}$. Such a constant exists by \cref{lem:approx-1,lem:approx-2}.

\begin{lem}\label{lem:approx-3-1} If \eqref{eq:parameter-smallness} holds, $v_f\ge \breve v$, $\mathcal R^{\breve v,v_f}_\mathrm{far}\subset\mathcal U$, and the bootstrap assumptions \eqref{eq:boot-3-1}--\eqref{eq:boot-3-3} hold on $\mathcal R^{\breve v,v_f}_\mathrm{far}$, then
\begin{gather}
    \label{eq:3-lem-1}  0\le  m\le 10M,\quad   0\le   Q\le 6M,\\
\label{eq:3-lem-3}    \Omega^2\sim 1,\\
\label{eq:3-lem-4}    \partial_v\!\log\Omega^2\les v^{-2},\\
\label{eq:3-lem-5}  \tfrac 12\le 1- \frac{2m}{r}\le 1 \end{gather}
on $\mathcal R^{\breve v,v_f}_\mathrm{far}$ and 
\begin{gather}
 \label{eq:3-lem-6}   r\sim v,\\
 \label{eq:3-lem-7}   \partial_v{\log\Omega^2 }<\frac{2\partial_vr}{r}
\end{gather}
on $\mathcal W$. Furthermore, 
\begin{equation}
    m=Q=0\label{eq:far-Mink-center}
\end{equation}
on $\mathcal R_\mathrm{far}^{\breve v,v_f}\setminus\mathcal W$. 
\end{lem}
\begin{proof}
Proof of \eqref{eq:3-lem-1} and \eqref{eq:far-Mink-center}:  The bootstrap assumptions \eqref{eq:boot-3-1} and \eqref{eq:boot-3-2} imply $\partial_ur<0$ and $\partial_vr>0$. Therefore, \eqref{eq:3-lem-1}
follows from the monotonicity properties of the Einstein--Maxwell--Vlasov system, \cref{lem:est-dat-1,lem:approx-1,lem:approx-2}, and the boundary condition \eqref{eq:boundary-1}. 

Proof of \eqref{eq:3-lem-6}: By the bootstrap assumption \eqref{eq:boot-3-1}, 
\begin{equation*}
    r(u,v)-r(-r_2,v)\ge -C_\nu(u+r_2) \ge -\tfrac 16 v-C_\nu r_2
\end{equation*}
for $(u,v)\in\mathcal W$. By \eqref{eq:cone-4}, the lower bound in \eqref{eq:3-lem-6} easily follows if $\breve v$ is taken sufficiently large. Since $\partial_ur<0$, $r(u,v)\le r(-r_2,v)\les v$ for $v\ge \breve v$ and $\breve v$ sufficiently large, which proves the upper bound in \eqref{eq:3-lem-6}.

Proof of \eqref{eq:3-lem-5}: This is immediate for $\breve v$ chosen sufficiently large in light of \eqref{eq:3-lem-1} and the fact that
\begin{equation}
    \inf_{\mathcal W} r \gtrsim \breve v,\label{eq:r-W-lower-bound}
\end{equation}
which follows from \eqref{eq:3-lem-6}.

Proof of \eqref{eq:3-lem-3}: This follows from \eqref{eq:Omega-formula} by combining the bootstrap assumptions \eqref{eq:boot-3-1} and \eqref{eq:boot-3-2} with \eqref{eq:3-lem-5}. 

Proof of \eqref{eq:3-lem-4}: Let $(u,v)\in\mathcal W$. We will show that 
\begin{equation}
\int_{-r_2}^u  T^{uv}(u',v)\,du'\les v^{-2},\label{eq:3-lem-8}
\end{equation}
which together with \eqref{eq:Omega-wave}, \eqref{eq:3-lem-1}, \eqref{eq:3-lem-3}, and \eqref{eq:3-lem-6}, readily implies \eqref{eq:3-lem-4}. To prove \eqref{eq:3-lem-8}, we observe that by the bootstrap assumptions \eqref{eq:boot-3-1} and \eqref{eq:boot-3-2} and the evolution equation \eqref{eq:Hawking-u}, $\partial_um\le 0$. Using that $m\ge 0$ in $\mathcal R^{\breve v,v_f}_\mathrm{aux}$, we therefore infer 
\begin{equation*}
    \int_{-r_2}^u \left.\left[\tfrac 12 r^2\Omega^2( T^{uv}(-\partial_u r) + T^{vv}\partial_vr)+\frac{Q^2}{2r^2}(-\partial_u r)\right]\right|_{(u',v)}du'= m(-r_2,v)-m(u,v)\le 2\varpi_2.
\end{equation*}
Since all three integrands are nonnegative, this energy estimate, taken together with \eqref{eq:boot-3-1}, \eqref{eq:3-lem-3} and \eqref{eq:3-lem-6} imply 
\begin{equation*}
    v^2\int_{-r_2}^u  T^{uv}(u',v)\,du'\les \int_{-r_2}^u \left.r^2\Omega^2 T^{uv}\partial_v r\right|_{(u',v)}\,du'\les 1
\end{equation*} for $(u,v)\in\mathcal W$,
which proves \eqref{eq:3-lem-8}.

Proof of \eqref{eq:3-lem-7}: This follows from the fact that $\partial_vr/r\gtrsim v^{-1}$ in $\mathcal W$ by \eqref{eq:boot-3-2} and \eqref{eq:3-lem-6}.
\end{proof}

\begin{proof}[Proof of \cref{lem:beam-far}]
The proof is a bootstrap argument based on the bootstrap assumptions \eqref{eq:boot-3-1}--\eqref{eq:boot-3-3}. Let 
\begin{equation*}
    \mathcal A\doteq \{v_f\in [\breve v,\infty): \mathcal R^{\breve v,v_f}_\mathrm{aux}\subset\mathcal U\text{ and \eqref{eq:boot-3-1}--\eqref{eq:boot-3-3} hold on $\mathcal R^{\breve v,v_0}_\mathrm{aux}$}\}.
\end{equation*}
The set $\mathcal A$ is nonempty by \cref{prop:development-1}, \cref{lem:approx-1}, and \cref{lem:approx-2}. It is also manifestly closed by continuity of the bootstrap assumptions. We now show that if \eqref{eq:parameter-smallness} holds, then $\mathcal A$ is also open. Let $v_f\in\mathcal A$.

Improving \eqref{eq:boot-3-2}: Let $(u,v)\in\mathcal R^{\breve v,v_f}_\mathrm{aux}$. Integrating the wave equation \eqref{eq:r-wave} in $u$ starting at $u'=-r_2$ and using the estimates of \cref{lem:approx-3-1} yields
 \begin{equation*}
     |\partial_vr(u,v)-\tfrac 12|\le \int_{-r_2}^u \left.\left(\frac{\Omega^2}{2r^2}\left(m+\frac{Q^2}{2r}\right)+\tfrac 14 r\Omega^4 T^{uv}\right)\right|_{(u',v)}du'\les v^{-1}\le \breve v^{-1},
 \end{equation*}
 which improves \eqref{eq:boot-3-2} for $\breve v$ sufficiently large.
 
 Improving \eqref{eq:boot-3-1}: Using \eqref{eq:Hawking-v}, $\partial_vm\ge 0$, and \eqref{eq:r-W-lower-bound}, we obtain (similarly to \eqref{eq:3-lem-8})
 \begin{equation}
     \int_{v_1}^{v_2} \left.rT^{uv}\right|_{(u,v')}\,dv'\les \breve v^{-1}\label{eq:3-lem-9}
 \end{equation}
for any $(u,v_1),(u,v_2)\in \mathcal R^{\breve v,v_f}_\mathrm{aux}$. Let $(u,v)\in\mathcal R^{\breve v,v_f}_\mathrm{aux}$. We integrate the wave equation \eqref{eq:r-wave} in $v$ starting at $\underline C_{\breve v}$ if $u\le \breve v$ and at $(u,u)\in\Gamma$ if $u>\breve v$. In the former case, 
 \begin{equation*}
     |\partial_ur(u,v)-\partial_ur(u,\breve v)|\le \int_{\breve v}^v\left.\left(\frac{\Omega^2}{2r^2}\left(m+\frac{Q^2}{2r}\right)+\tfrac 14 r\Omega^2 T^{uv}\right)\right|_{(u,v')}dv'\les \breve v^{-1},
 \end{equation*}
 which improves \eqref{eq:boot-3-1} for $\breve v$ sufficiently large by the definition of $C_\nu$. In the latter case, the boundary condition \eqref{eq:boundary-1} implies $\partial_ur(u,u)=-\partial_vr(u,u)$, so 
\begin{equation*}
     |\partial_ur(u,v)+\partial_vr(u,u)|\le \int_{u}^v\left.\left(\frac{\Omega^2}{2r^2}\left(m+\frac{Q^2}{2r}\right)+\tfrac 14 r\Omega^2 T^{uv}\right)\right|_{(u,v')}dv'\les \breve v^{-1},
 \end{equation*}
 which by \eqref{eq:boot-3-2} improves \eqref{eq:boot-3-1} for $\breve v$ sufficiently large by the definition of $C_\nu$.

 Improving \eqref{eq:boot-3-3}: Let $\gamma:[0,S)\to \mathcal R^{\breve v,v_f}_\mathrm{aux}$ be an electromagnetic geodesic in the support of $f$ starting at $\underline C_{\breve v}$ at $s=0$. By \eqref{eq:SS-Lor-4} and \eqref{eq:3-lem-7},
 \begin{equation*}
     \frac{d}{ds}(\Omega^2p^u)\le 0,
 \end{equation*}
 so by \eqref{eq:3-lem-3},
 \begin{equation}
     p^u(s)\les p^u(0).\label{eq:p^u-final}
 \end{equation}
Using \eqref{eq:SS-Lor-2}, the signs of $\partial_ur$ and $Q$, and parametrizing $\gamma$ by $v$ yields 
\begin{equation}
    \frac{d}{dv}\log p^v\ge -\partial_v{\log\Omega^2}. \label{eq:3-lem-10}
\end{equation}
By \eqref{eq:3-lem-4}, it is easy to see that 
\begin{equation*}
    \int_{\breve v}^v|\partial_v{\log\Omega^2(\gamma^u(v'),v')}|\,dv' \les \breve v^{-1}
\end{equation*}
and therefore 
\begin{equation*}
     \exp\left(-\int_{\breve v}^v\partial_v{\log\Omega^2(\gamma^u(v'),v')}\,dv'\right)\ge \frac 12
\end{equation*}
for $\breve v$ sufficiently large. It follows from \eqref{eq:3-lem-10} that 
\begin{equation}
    p^v(s)\gtrsim p^v(0).\label{eq:p^v-final}
\end{equation}
Combining \eqref{eq:p^u-final} and \eqref{eq:p^v-final} yields 
\begin{equation*}
    \frac{p^u}{p^v}(s)\le \frac{1}{7C_\nu},
\end{equation*}
for $\breve v$ sufficiently large by \eqref{eq:main-final-pv} and \eqref{eq:aux-final-pv}. It is then easy to show that $\gamma(s)$ stays in $\{7C_\nu u\le v\}$ for every $s\in [0,S)$, which quantitatively improves  \eqref{eq:boot-3-3}. The rest of the existence and uniqueness proof now follows a standard continuity argument using \cref{prop:strip} and \cref{lem:Minkowski}.

To estimate $|\partial_u\Omega^2|$, note that $|\partial_u\Omega^2|\les (1+u^2)^{-1}$ along $\underline C_{\breve v}\cup \Gamma$ by \cref{lem:approx-2} and \eqref{eq:Omega-center}. Observe that
\[\inf_{C_u\cap\mathcal W\cap \mathcal R^{\breve v,\infty}_\mathrm{far}} r\sim \inf_{C_u\cap\mathcal W\cap \mathcal R^{\breve v,\infty}_\mathrm{far}} v=\max\{\breve v,6C_\nu u\}\gtrsim 1+u,\]
so the energy estimate \eqref{eq:3-lem-9} can be improved to 
\begin{equation*}
    \int_{v_1}^{v_2}T^{uv}(u,v')\,dv'\les (1+u)^{-2}
\end{equation*}
for any $u\le v_1\le v_2$. Therefore the desired estimate can be propagated to the interior by integrating the wave equation \eqref{eq:Omega-wave} in $v$. Together with the bootstrap assumptions and \cref{lem:approx-3-1}, this completes the proof of the estimates \eqref{eq:approx-3-1}--\eqref{eq:approx-3-3}. \end{proof}

\subsection{The dispersive estimate in the massive case}\label{sec:massive-dispersion}

Let $\mathfrak m>0$ and consider the solution $(r,\Omega^2,Q,f)$ given by \cref{lem:beam-far}, defined globally on $\mathcal C_{r_2}$. We augment the hierarchy \eqref{eq:parameter-smallness} with a large parameter $v_\#$ satisfying
\begin{equation}
    0<v_\#^{-1}\ll \mathfrak m\le \mathfrak m_0\label{eq:parameter-smallness-2}
\end{equation}
and aim to prove the following
\begin{lem}\label{lem:dispersion-main}
For any $\breve v, \eta,\ve,\mathfrak m_0,\mathfrak m$, and $v_\#$ satisfying \eqref{eq:parameter-smallness} and \eqref{eq:parameter-smallness-2}, we have the decay
\begin{equation*}
    \mathcal M\le Cv^{-3}
\end{equation*} for $v\ge v_\#$ and any $\mathcal M\in \{N^u,N^v,T^{uu},T^{uv},T^{vv},S\}$, where $C$ may depend on $\eta,\ve$, $\mathfrak m$, and $v_\#$.
\end{lem}

The proof is based on a bootstrap argument for the dispersion of ingoing momentum $p^u$ along $\spt f$ as $v\to \infty$, which leads to cubic decay of the phase space volume $\mathscr V$, which was defined in \eqref{eq:phase-space-volume}. Using the mass shell relation \eqref{eq:mass-shell} and the change of variables formula, we have
\begin{align}
  \label{eq:V-u}  \mathscr V(u,v)=\frac{2}{r^2}\int_0^\infty\int_{\{p^u:f(u,v,p^u,p^v)\ne 0\}} \frac{dp^u}{p^u}\,\ell\,d\ell,
\end{align}
where we view $p^v$ as a function of $p^u$ and $\ell$. Compare with \eqref{eq:V-v}.

\begin{lem}\label{lem:momentum-very-far} If $\eqref{eq:parameter-smallness}$ and $\eqref{eq:parameter-smallness-2}$ hold and $(u,v,p^u,p^v)\in \spt f$, then
  $p^v\les 1$ if $v\ge\breve v$ and  $p^u\sim_\eta \mathfrak m^2$ if $v\ge v_\#$.
\end{lem}
\begin{rk}
    The estimate $p^v\les 1$ also holds in the massless case. The non-decay of $p^u$ for massive particles drives the decay rate $v^{-3}$, but only at very late times.
\end{rk}
\begin{proof}
    Let $\gamma\in\Gamma_f$ and $v\ge \breve v$. Parametrizing $\gamma$ by $v$ and using \eqref{eq:main-final-pv}, \eqref{eq:p-u-v}, \eqref{eq:aux-final-pv}, and the estimates in \cref{lem:beam-far}, we infer
   \begin{equation*}
       \left|\frac{d}{dv}(\Omega^2p^v)\right|\les v^{-2},
   \end{equation*}
   which is integrable and hence shows that $p^v\les 1$. By \eqref{eq:main-final-pv} and  \eqref{eq:aux-final-pv}, we have $p^v\gtrsim_\eta 1$. We then obtain
   \begin{equation*}
       p^u=\frac{\ell^2r^{-2}+\mathfrak m^2}{\Omega^2p^v}\sim_\eta \frac{\ell^2}{r^2}+\mathfrak m^2\sim \mathfrak m^2
   \end{equation*}
    for $r\gtrsim v_\#$ sufficiently large.
\end{proof}
Using this lemma and \eqref{eq:gamma-u}, we immediately infer:
\begin{lem}\label{lem:dispersion-1} If $\eqref{eq:parameter-smallness}$ and $\eqref{eq:parameter-smallness-2}$ hold, $(u,v,p^u,p^v)\in \spt f$, and $v\ge v_\#$, then $u\sim_\eta \mathfrak m^2v$. If $\gamma_1,\gamma_2\in \Gamma_f(u,v)$, then we have $|\gamma_1^u(s_{v_\#}^1)-\gamma_2^u(s_{v_\#}^2)|(s_{v_\#})\les_\eta \mathfrak m^2(v-v_\#)$, where $s_{v_\#}^i$ is the parameter time for which $\gamma^v_i(s_\#^i)=v_\#$.
\end{lem}

Let $(u_0,v_0)\in\mathcal C_{r_2}$ and let $\gamma\in \Gamma_f(u_0,v_0)$ have ingoing momentum $p^u_0$ and angular momentum $\ell$ at $(u_0,v_0)$. We parametrize $\gamma$ by $v$ going backwards in time and denote this by
\begin{align*}
    \gamma^u(v)&\doteq \gamma^u(v;u_0,v_0,p^u_0,\ell),\\
    p^u(v)&\doteq p^u(v;u_0,v_0,p^u_0,\ell).
\end{align*}
We readily derive the equations
\begin{align}
 \label{eq:v-param-1}  \frac{d}{dv}\gamma^u &=\frac{(\Omega^2p^u)^2}{\Omega^2(\ell^2r^{-2}+\mathfrak m^2)},\\
  \label{eq:v-param-2}   \frac{d}{dv}(\Omega^2p^u)&=\left(\partial_v{\log\Omega^2}- \frac{2\partial_vr}{r}\right)\frac{\ell^2 \Omega^2 p^u}{r^2 \mathfrak m^2 + \ell^2}  - \mathfrak e  \frac{Q}{r^2 \mathfrak m^2 + \ell^2} (\Omega^2 p^u)^2.
\end{align}
Next, we define the variational quantities
\begin{equation*}
       \mathfrak u(v;u_0,v_0,p^u_0,\ell)\doteq \frac{\partial}
{\partial p_0^u}\gamma^u(v;u_0,v_0,p^u_0,\ell),\quad   \mathfrak p(v;u_0,v_0,p^u_0,\ell)\doteq \frac{\partial}
{\partial p_0^u}(\Omega^2p^u)(v;u_0,v_0,p^u_0,\ell),
\end{equation*}
where we emphasize that the derivative in $p^u_0$ is taken with $\ell$ fixed. From \eqref{eq:v-param-1} and \eqref{eq:v-param-2} we obtain
\begin{align}
   \label{eq:mathfrak-var-1}   \frac{d}{dv} \mathfrak u  &= \frac{2 p^u}{\ell^2 r^{-2} + \mathfrak m^2} \mathfrak p - \frac{(p^u)^2}{(\ell^2 r^{-2} + \mathfrak m^2)^2} \left[ \partial_u\Omega^2 \left(\frac{\ell^2}{r^2} + \mathfrak m^2\right) -\frac{2\Omega^2\ell^2\partial_ur}{r^3}\right]  \mathfrak u,\\
       \frac{d}{dv} \mathfrak p  & =\left[\left(\partial_v \log \Omega^2 - \frac{2 \partial_v r}{r} \right) \frac{\ell^2 }{r^2 \mathfrak m^2 + \ell^2 }-  \frac{2\mathfrak eQ \Omega^2p^u}{r^2 \mathfrak m^2 + \ell^2} \right]   \mathfrak p\nonumber \\&  \quad+ \left[\left( \partial_u \partial_v{\log \Omega^2} - \frac{2 \partial_u \partial_v r}{r} + 2 \frac{ \partial_v r \partial_u r }{r^2} \right) \frac{\ell^2  \Omega^2p^u}{r^2 \mathfrak m^2 + \ell^2 } - \left(\partial_v{ \log \Omega^2} - \frac{2 \partial_v r}{r}\right) \frac{2 \mathfrak m^2 r \partial_u r\ell^2 \Omega^2 p^u}{(r^2 \mathfrak m^2 + \ell^2)^2} \right.\nonumber \\ &  \quad\quad\left. - \frac{\mathfrak e(\Omega^2 p^u)^2}{r^2 \mathfrak m^2 + \ell^2}\partial_uQ  + \frac{2\mathfrak eQ\mathfrak m^2 r \partial_u r (\Omega^2 p^u)^2}{(r^2 \mathfrak m^2 + \ell^2)^2}   \right]\mathfrak u.\label{eq:mathfrak-var-2}
\end{align}
Note that $\mathfrak u(v_0)=0$ and $\mathfrak p(v_0)=\Omega^2(u_0,v_0)\sim 1$.
\begin{lem}
    If $\eqref{eq:parameter-smallness}$ and $\eqref{eq:parameter-smallness-2}$ hold, and $v\ge v_\#$, then 
    \begin{equation}
        \mathscr V(u,v)\les_{\eta,\ve}\frac{v_\#}{v^3}.\label{eq:volume-far}
    \end{equation}
\end{lem}
\begin{proof}
    We claim that there exists a constant $C_*$, depending on $\eta$ and $\ve$, such that
    \begin{equation}
        C^{-1}_*\le\frac{\mathfrak u(v)}{v-v_0}\le C_*\label{eq:u-bootstrap}
    \end{equation}
    for any $v_\#\le v\le v_0$. To see how this proves \eqref{eq:volume-far}, let $\Phi_{u_0,v_0,\ell}(p_0^u)\doteq \gamma^u(v_\#;u_0,v_0,p^u_0,\ell)$ and observe that $\Phi_{u_0,v_0,\ell}'(p^u_0)=\mathfrak u(v_\#)<0$. Changing variables in the $p^u$ integral in \eqref{eq:V-u} to $\gamma^u(v_\#)$ and using \cref{lem:dispersion-1} to estimate the $u$-dispersion along $\underline C_{v_\#}$, we find
    \begin{equation}
     \label{eq:volume-far-2}   \mathscr V(u_0,v_0)\les_{\eta,\ve} \frac{1}{r^2\mathfrak m^2}\frac{C_*}{|v_\# -v_0|}\min\{\mathfrak m^2v_\#,\mathfrak m^2(v_0-v_\#)\}\les C_*\frac{v_\#}{v_0^3}.
    \end{equation}
    
    We prove \eqref{eq:u-bootstrap} by a bootstrap argument as follows. Let $v_f\ge v_\#$ and assume \eqref{eq:u-bootstrap} holds for all $(v,u_0,v_0,p^u_0,\ell)$ with $v_\#\le v\le v_0\le v_f$. The assumption is clearly satisfied for some choice of $C_*$ for $v_f$ sufficiently close to $v_\#$ on account of $\mathfrak u(v_0)=0$. 
    
    We now show that for $\mathfrak m$ sufficiently small and $v_\#$ sufficiently large, we can improve the constant in \eqref{eq:u-bootstrap}. Using \eqref{eq:volume-far}, we estimate
    \begin{equation*}
        N^v\les_{\eta,\ve} C_* \frac{v_\#}{v^3}, \quad T^{uv}\les_{\eta,\ve}C_*\frac{\mathfrak m^2v_\#}{v^3}
    \end{equation*}
    for $v\in [v_\#,v_f]$. Using this, as well as \eqref{eq:Omega-wave}, \eqref{eq:Max-u}, \cref{lem:beam-far}, and $p^u\sim_\eta \mathfrak m^2$ in \eqref{eq:mathfrak-var-2} yields
\begin{equation*}
    \left|\frac{d}{dv}\mathfrak p\right|\les_{\eta,\ve}\left(\frac{1}{\mathfrak m^2v^3}+\frac{1}{v^2}\right)|\mathfrak p|+\left(\frac{1}{v^4}+\frac{\mathfrak m^2}{v^3}+C_*\frac{\mathfrak m^2v_\#}{v^3}\right)|\mathfrak u|.
\end{equation*}
We use the bootstrap assumption \eqref{eq:u-bootstrap} and \eqref{eq:parameter-smallness-2} to infer
\begin{equation*}
        \left|\frac{d}{dv}(\mathfrak p-\mathfrak p(v_0))\right|\les_{\eta,\ve}\frac{1}{v^2} |\mathfrak p-\mathfrak p(v_0)|+ \frac{1}{v^2}+C_*^2\frac{\mathfrak m^2v_\#}{v^3}|v-v_0|.
    \end{equation*}
    and then use Gr\"onwall's inequality to obtain
 \begin{equation*}
     |\mathfrak p(v)-\mathfrak p(v_0)|\les_{\eta,\ve}v_\#^{-1}+C_*^2\frac{\mathfrak m^2v_\#}{v^2}|v-v_0|.
 \end{equation*}
 Using $\mathfrak p(v_0)\sim 1$, we therefore have
 \begin{equation}
  \label{eq:dispersion-4}   \int_{v}^{v_0}\mathfrak p(v')\,dv'\sim_{\eta,\ve} v_0-v + O_{\eta,\ve}(C_*^2\mathfrak m^2|v-v_0|)\sim v_0-v
 \end{equation}
 for $\mathfrak m$ sufficiently small. Finally, we use \eqref{eq:approx-3-3} and \cref{lem:dispersion-1} to estimate
 \begin{equation}
\label{eq:dispersion-5}    \int_v^{v_0}|\partial_u\Omega^2|(\gamma^u(v'),v')\left(\frac{\ell^2}{r^2}+\mathfrak m^2\right)\,dv' \les_\eta \int_{\gamma^u(v)}^{\gamma^u(v_0)}\frac{1}{1+u'^2}\,du'\les \frac{1}{\gamma^u(v_\#)}\les \frac{1}{\mathfrak m^2v_\#}.
 \end{equation}
 Integrating \eqref{eq:mathfrak-var-1} and using \eqref{eq:dispersion-4} and \eqref{eq:dispersion-5} improves the constant in \eqref{eq:u-bootstrap} for $\mathfrak m$ sufficiently small and $v_\#$ sufficiently large, which completes the proof. \end{proof}

\begin{proof}[Proof of \cref{lem:dispersion-main}]
    Immediate from \eqref{eq:volume-far}.
\end{proof}

\subsection{\texorpdfstring{Proof of \cref{prop:approx-1}}{Proof of Proposition}}\label{sec:proof-of-prop}

\begin{proof}
   \ul{Part 1.} This follows immediately from combining \cref{lem:approx-1,lem:approx-2,lem:beam-far}.

   \ul{Part 2.} The estimates in \eqref{eq:prop-approx-1} follow directly from the estimates in \cref{lem:approx-1,lem:approx-2,lem:beam-far}.  We will now prove causal geodesic completeness of the spacetime. Let $\gamma$ be a future-directed causal geodesic in the $(3+1)$-dimensional spacetime. Then the projection of $\gamma$ to the reduced spacetime, still denoted $\gamma$, satisfies
       \begin{align}
     \label{eq:geo-1}   \frac{d}{ds}(\Omega^2p^u)&=\left(\partial_v{\log\Omega^2}-\frac{2\partial_vr}{r}\right)\frac{\ell^2}{r^2},\\
      \label{eq:geo-2}  \frac{d}{ds}(\Omega^2p^v)&=\left(\partial_u{\log\Omega^2}-\frac{2\partial_ur}{r}\right)\frac{\ell^2}{r^2},\\
    \label{eq:geo-3}    \Omega^2p^up^v&=\frac{\ell^2}{r^2}+\mathfrak m^2,
    \end{align}
    where $s$ is an affine parameter and $\gamma$ is continued through the center according to Part 4~of \cref{def:normalized-defn}. We will show for any future-directed causal geodesic $\gamma:[0,S)\to \mathcal C_{r_2}$, $p^\tau$ is uniformly bounded in any compact interval of coordinate time $0\le \tau \le \tau_0$ along $\gamma$. This implies that $\gamma$ can be extended to $[0,\infty)$ by the normal neighborhood property of the geodesic flow in  Lorentzian manifolds \cite{Oneill}.

    By \eqref{eq:geo-1}--\eqref{eq:geo-3} and \eqref{eq:prop-approx-1}, 
    \begin{equation}
        \left|\frac{d}{ds}(\Omega^2p^\tau)\right|\les  \left(1+\left|\frac{\partial_vr+\partial_ur}{r}\right|\right)(\Omega^2p^\tau)^2. \label{eq:approx-proof-1}
    \end{equation}
    When $r\ge \tfrac 16r_1$, the term in the absolute value on the right-hand side of this estimate is clearly bounded by \eqref{eq:prop-approx-1}.  When $r\le \tfrac 16r_1$, the spacetime is Minkowski and the formulas in \cref{lem:Minkowski} can be used to show that
    \begin{equation}
        \partial_vr+\partial_ur \les v-u \les r.\label{eq:approx-proof-2}
    \end{equation}
    Therefore, passing to a $\tau$ parametrization of $\gamma$, \eqref{eq:approx-proof-1} implies
    \begin{equation*}
        \left|\frac{d}{d\tau}(\Omega^2p^\tau)\right|\les \Omega^2 p^\tau,
    \end{equation*}
    and the proof is completed by an application of Gr\"onwall's lemma.\footnote{We have based this argument off of a non-gauge-invariant energy $\Omega^2p^\tau$, which is why the estimate \eqref{eq:approx-proof-2} has to be performed even in the Minkowski region of the spacetime. One could alternatively consider the gauge-invariant energy $(\partial_vr) p^u-(\partial_ur)p^v$, which is constant in the Minkowski region, but satisfies a more complicated evolution equation where $f\ne 0$.}

   \ul{Part 3.} The estimate \eqref{eq:final-parameter-estimate} for the final parameters follows from \eqref{eq:main-boundedness-data}. The remaining claims in this part follow from the proof of \cref{lem:est-dat-1}, \eqref{eq:solved-for-nu} on $C_{-\frac 23r_1}$, and \cref{lem:beam-far}.

   \ul{Part 4.} The estimate \eqref{eq:r-min-statement}, the upper bound in \eqref{eq:spt-f-statement}, and the claim about the neighborhood of the center follow from \cref{lem:beam-far}. The lower bound in \eqref{eq:spt-f-statement} follows from \cref{lem:momentum-very-far}, which implies that $\gamma^u$ grows linearly in $v$ at very late times for any electromagnetic geodesic $\gamma$ in the support of $f$. Since a neighborhood of $\mathcal I^+$ is then electrovacuum, it is isometric to an appropriate Reissner--Nordstr\"om solution by Birkhoff's theorem.

   \ul{Part 5.} This follows immediately from \cref{lem:dispersion-main}.

   \ul{Part 4$'$.} We now take $\mathfrak m=0$ and seek to prove the upper bound in \eqref{eq:spt-f-massless}, the rest following immediately from Birkhoff's theorem. Let $\gamma(s)$ be an electromagnetic geodesic lying in the support of $f$ with $\gamma^v(0)=\breve v$. When $\mathfrak m=0$, the mass shell relation, together with \eqref{eq:main-final-pv}, \eqref{eq:aux-final-pv}, and \eqref{eq:p^v-final} gives the estimate
   \begin{equation*}
       \frac{p^u}{p^v}(s) = \frac{r^2(0)}{r^2(s)}\frac{\Omega^2(0)}{\Omega^2(s)}\frac{p^u(0)}{p^v(s)}\frac{p^v(0)}{p^v(s)}\les v^{-2},
   \end{equation*}
   which has integral $\les \breve v^{-1}$. Therefore, the upper bound in \eqref{eq:spt-f-massless} follows from \eqref{eq:gamma-u} and taking $\breve v$ sufficiently large. 

   \ul{Part 5$'$.} When $\mathfrak m=0$, the mass shell relation implies $p^u\les v^{-2}$ for any electromagnetic geodesic lying in the support of $f$. The estimates \eqref{eq:massless-decay-1}--\eqref{eq:massless-decay-3} follow from this and the formula \eqref{eq:V-v}.
\end{proof}

\subsection{Patching together the ingoing and outgoing beams}
\label{sec:patching-ingoing-outgoing}
\subsubsection{The maximal time-symmetric doubled spacetime}\label{sec:doubled}

Let $\alpha\in \mathcal P_\Gamma$ or be of the form $(r_1,r_2,0,0,0,0)$. Let $\eta$ and $\ve$ be beam parameters for which the conclusion of \cref{prop:approx-1} holds, recalling also \cref{rk:limiting-case}. Let $\tilde M$ and $\tilde e$ denote the final Reissner--Nordstr\"om parameters of $\mathcal S=\mathcal S_{\alpha,\eta,\ve}$. We say that $\mathcal S$ is \emph{subextremal} if $\tilde e<\tilde M$, \emph{extremal} if $\tilde e=\tilde M$, and \emph{superextremal} if $\tilde e>\tilde M$.\footnote{Recall that we are always taking $\mathfrak e>0$ and hence $\tilde e> 0$.} If $\mathcal S$ is not superextremal, we may define
\begin{equation*}
    r_\pm =\tilde M\pm \sqrt{\tilde M^2-\tilde e^2},
\end{equation*}
which is the formula for the area radii of the outer and inner horizons in Reissner--Nordstr\"om.

The following lemma is an easy consequence of the well-known structure of the maximally extended Reissner--Nordstr\"om solution:
\begin{lem}\label{lem:RN-piece} For any $\tilde M, \tilde e>0$ and $0<r_2<r_-$ (if $\tilde e\le \tilde M$), there exists a relatively open set
\[\mathcal E_{\tilde M,\tilde e,r_2}\subset\{u\le -r_2\}\cap\{v\ge r_2\}\subset\Bbb R^2_{u,v}\] and an analytic spherically symmetric solution $(r,\Omega^2,Q)$ of the Einstein--Maxwell equations on $\mathcal E_{\tilde M,\tilde e,r_2}$ with the following properties: The renormalized Hawking mass $\varpi=\tilde M$ globally, the charge $Q=\tilde e$ globally, $r(-r_2,r_2)=r_2$, $\partial_vr(-r_2,r_2)=-\partial_ur(-r_2,r_2)=\frac 12$, and $\Omega^2$ is constant along the cones $\{u=-r_2\}\cap\{v\ge r_2\}$ and $\{v=r_2\}\cap\{u\le -r_2\}$. 
Moreover, we may choose $(\mathcal E_{\tilde M,\tilde e,r_2},r,\Omega^2,Q)$ to be maximal with these properties, and it will then be unique.
\end{lem}
\begin{rk} The $(3+1)$-dimensional lift of $(\mathcal E_{\tilde M,\tilde e,r_2},r,\Omega^2,Q)$ is isometric to a subset of the maximally extended Reissner--Nordstr\"om solution with parameters $\tilde M$ and $\tilde e$. The hypersurface $(\{\tau=0\}\cap \mathcal E_{\tilde M,\tilde e,r_2})\times S^2$ is then totally geodesic.
\end{rk}
\begin{rk}
    If $\tilde M<\tilde e$, then $\mathcal E_{\tilde M,\tilde e,r_2}=\{u\le -r_2\}\cap\{v\ge r_2\}$. In the case $\tilde e\le \tilde M$, $\mathcal E_{\tilde M,\tilde e,r_2}$ is a strict subset of $\{u\le -r_2\}\cap\{v\ge r_2\}$ with this choice of gauge. 
\end{rk}

\begin{defn}\label{def:max-time-symmetric}
    Let $(\mathcal C_{r_2},r,\Omega^2,Q,f)$ be the maximal normalized development of $\mathcal S=\mathcal S_{\alpha,\eta,\ve}$ for particles of mass $\mathfrak m$ with final Reissner--Nordstr\"om parameters $(\tilde M,\tilde e)$ for which the conclusion of \cref{prop:approx-1} holds, recalling also \cref{rk:limiting-case}. Let 
    \[\tilde{\mathcal C}_{r_2}\doteq\{\tau\le 0\}\cap\{v\ge u\}\cap\{v\le r_2\}\] be the time-reflection of $\mathcal C_{r_2}$. For $(u,v)\in\tilde{\mathcal C}_{r_2}$, define
    \begin{align*}
        \tilde r(u,v)&=r(-v,-u),\\
        \tilde\Omega^2(u,v)&=\Omega^2(-v,-u),\\
        \tilde Q(u,v)&=Q(-v,-u),\\
        \tilde f(u,v,p^u,p^v)&=f(-v,-u,p^v,p^u).
    \end{align*}
    Let $\mathscr M\doteq {\mathcal C}_{r_2}\cup \tilde{\mathcal C}_{r_2}\cup\mathcal E_{\tilde M,\tilde e,r_2}$ and define $(r,\Omega^2,Q,f)$ on $\mathscr M$ by simply gluing the corresponding functions across the boundaries of the sets ${\mathcal C}_{r_2}$, $\tilde{\mathcal C}_{r_2}$, and $\mathcal E_{\tilde M,\tilde e,r_2}$. (We therefore now drop the tilde notation on the solution in $\mathcal C_{r_2}$, except for in the proof of \cref{lem:smoothness-across} below.) The tuple $(\mathscr M,r,\Omega^2,Q,f)$ is called the \emph{maximal time-symmetric doubled spacetime} associated to $\mathcal S$.
\end{defn}

\begin{lem}\label{lem:smoothness-across}
    $(\mathscr M,r,\Omega^2,Q,f)$ is a smooth solution of the spherically symmetric Einstein--Maxwell--Vlasov system. The hypersurface $\{\tau=0\}\cap\mathscr M$ is a totally geodesic hypersurface once lifted to the $(3+1)$-dimensional spacetime by \cref{prop:SS-equiv}.
\end{lem}
\begin{proof}
    This is immediate except perhaps across $\{\tau=0\}\cap \{v\le r_2\}$. We only have to show that the solution is $C^2\times C^2\times C^1\times C^1$ regular across this interface by the regularity theory of \cref{prop:char-IVP-Vlasov}. By \cref{def:seed} and \cref{def:normalized-defn}, the solution is clearly $C^1\times C^1\times C^0\times C^0$ regular. Using Raychaudhuri's equations \eqref{eq:Ray-u}, \eqref{eq:Ray-v} and Maxwell's equations \eqref{eq:Max-u}, \eqref{eq:Max-v}, it is easy to see that $r$ is $C^2$ and $Q$ is $C^1$ across $\{\tau=0\}$. To check second derivatives of $\Omega^2$, differentiate $(\partial_u+\partial_v)\Omega^2(-v,v)=0$ in $v$ to obtain $\partial_u^2\Omega^2(-v,v)=\partial_v^2\Omega^2(-v,v)$, which together with the wave equation \eqref{eq:Omega-wave} implies $C^2$ matching. The $p^u$ and $p^v$ derivatives of $f$ are also continuous by inspection and continuity of spatial derivatives can be proved as follows: On $\{\tau=0\}$, $\partial_vf$ can be eliminated in terms of $\partial_uf$ and $\frac{d}{dv}\mathring f$ by \eqref{eq:f-data-defn}. Then the Maxwell--Vlasov equation \eqref{eq:SS-MV} can be solved for $\partial_uf$. Performing the same calculation for $\tilde f$ shows that $\partial_uf$ is continuous across $\{\tau=0\}$. The same argument applies for $\partial_vf$ and the proof is complete. 
\end{proof}

\subsubsection{The anchored Cauchy hypersurface}\label{sec:anchored} 

In order to define the family of Cauchy data in \cref{thm:main}, we need to identify an appropriate Cauchy hypersurface in each $\mathscr M$. Let
\begin{equation}\label{eq:v-infty}
    v_\infty\doteq\begin{cases}
2r_--r_2 &\text{if }\tilde e\le \tilde M \\
\infty &\text{if }\tilde e>\tilde M
\end{cases}
\end{equation}
and let $\Sigma^u:[-50\tilde M,v_\infty)\to\Bbb R$ be the smooth function defined by
\begin{equation*}
    \Sigma^u(v)=-100\tilde M-v
\end{equation*} for $v\in [-50\tilde M,r_1]$, by
solving the ODE
\begin{equation}
    \frac{d}{dv}\Sigma^u(v)= \left.\frac{\partial_vr}{\partial_ur}\right|_{(\Sigma^u(v),v)}\label{eq:Sigma-ODE}
\end{equation}
with initial condition $\Sigma^u(r_2)=-100\tilde M-r_2$ for $v\in [r_2,v_\infty)$, and by applying the following easy consequence of Borel's lemma for $v\in[r_1,r_2]$:
\begin{lem}\label{lem:Borel}
    Given $0<r_1<r_2$, $\tilde M>0$, and a sequence of real numbers $a_1,a_2,\dotsc$ with $a_1<0$, there exists a smooth function $\Sigma^u:[r_1,r_2]\to [-100\tilde M-r_2,-100\tilde M-r_1]$ such that $\frac{d}{dv}{\Sigma^u}(v)<0$ for $v\in [r_1,r_2]$, $\Sigma^u$ has Taylor coefficients $(-100\tilde M-r_1,-1,0,0,\dotsc)$ at $r_1$, and Taylor coefficients $(-100\tilde M-r_2,a_1,a_2,\dotsc)$ at $r_2$. Moreover, if $\tilde M$ and each $a_j$ are smooth functions of the parameters $\varpi_2,Q_2,\ve,\eta$, then $f$ can be chosen to depend smoothly on $\varpi_2,Q_2,\ve,\eta$.
\end{lem}

\begin{rk}\label{rk:Kodama}
    For $v\ge r_2$, the curve $\Sigma:v\mapsto(\Sigma^u(v),v)\in \mathcal E_{\tilde M,\tilde e,r_2}$ lies in the domain of outer communication if $\tilde e\le\tilde M$ and is contained in a constant time hypersurface in Schwarzschild coordinates. Indeed, the time-translation Killing vector field in $\mathcal E_{\tilde M,\tilde e,r_2}$ is given by the Kodama vector field
    \begin{equation*}
        K\doteq 2 \Omega^{-2}\partial_vr \,\partial_u-2\Omega^{-2}\partial_ur\,\partial_v,
    \end{equation*}
     which is clearly orthogonal to $\Sigma$.
\end{rk}

\subsubsection{Cauchy data for the Einstein--Maxwell--Vlasov system}

Let $(\mathcal M,g,F,f)$ be a solution of the $(3+1)$-dimensional Einstein--Maxwell--Vlasov system as defined in \cref{app:A}. Let $i:\Bbb R^3\to \tilde\Sigma\subset \mathcal M$ be a spacelike embedding with future-directed unit timelike normal $n$ to $\tilde\Sigma$. As usual, we may consider the induced metric $\bar g\doteq i^*g$ and second fundamental form $k$ (pulled back to $\Bbb R^3$ along $i$) of $\tilde\Sigma$. The electric field is defined by $\bar E\doteq i^*F(\cdot,n)$ and the magnetic field is defined by $\bar B\doteq (\star_{\bar g}i^*F)^\sharp$, where $\sharp$ is taken relative to $\bar g$. Since the domain of $f$ is the spacetime mass shell $P^\mathfrak m$ which is not intrinsic to $\tilde\Sigma$, one first has to define a projection procedure to $T\tilde\Sigma\cong T\Bbb R^3$, after which the restriction of $f$ to $\tilde\Sigma$ can be thought of as a function $\bar f:T\Bbb R^3\to [0,\infty)$. Similarly, the volume forms in the Vlasov energy momentum tensor $T$ and number current $N$ have to be written in terms of $\bar g$, after which $\bar \rho_T \doteq i^\ast T(n,n)$ , $\bar j_T \doteq i^\ast T(n,\cdot)$ and $\bar \rho_{N} \doteq i^* N(n)$ can be evaluated on $\tilde\Sigma$ only in terms of $\bar g$ and $\bar f$. For details of this procedure, we refer to \cite[Section 13.4]{Ringstrom-topology}.

\begin{defn}\label{def:moduli-space-coarse} A \emph{Cauchy data set} for the Einstein--Maxwell--Vlasov system for particles of mass $\mathfrak m$ and fundamental charge $\mathfrak e$ consists of the tuple $\Psi=(\bar g,\bar k,\bar E,\bar B,\bar f)$ on $\Bbb R^3$ satisfying the constraint equations\footnote{The particle mass $\mathfrak m$ is implicitly contained in the formulas for $\bar T$ and $\bar N$.}
    \begin{align*}
      R_{\bar g} - |\bar k|^2_{\bar g} + (\tr_{\bar g}\bar k)^2  &= |\bar E|_{\bar g}^2+|\bar B|_{\bar g}^2+ 2\bar \rho_T [\bar f],\\
      \Div_{\bar g}\bar k -d \tr_{\bar g}\bar k  &=-(\star_{\bar g}(\bar E^\flat\wedge\bar B^\flat))^\sharp+\bar j_T [\bar f] \\
       \Div_{\bar g} \bar E &=\mathfrak e \bar \rho_N [\bar f] ,\\
        \Div_{\bar g}\bar B&=0.
    \end{align*}
    We denote by $\mathfrak M_\infty(\Bbb R^3,\mathfrak m,\mathfrak e)$ the set of solutions $\Psi$ of the Einstein--Maxwell--Vlasov constraint system on $\Bbb R^3$ with the $C^\infty_\mathrm{loc}$ subspace topology. 
\end{defn}

\subsubsection{The globally hyperbolic region}

By \cref{prop:approx-1}, \cref{rk:limiting-case}, their time-reversed ($u\mapsto -v$, $v\mapsto -u$) versions, \cref{rk:Kodama}, and the structure of the Reissner--Nordstr\"om family, we have:

\begin{prop}\label{prop:global-structure}
 Let $\mathcal S$ and $(\mathscr M,r,\Omega^2,Q,f)$ be as in \cref{def:max-time-symmetric} and let $\Sigma^u(v)$ be the function defined in \cref{sec:anchored}. Let
    \begin{equation*}
        \mathcal X\doteq \{v\ge u\}\cap\{v<v_\infty\}\subset\mathscr M
        \end{equation*}
    and let $\Sigma=\{(\Sigma^u(v),v):v\in[-50\tilde M,v_\infty)\}$. Then the following holds:
    \begin{enumerate}
        \item The manifold $\mathcal M\doteq ((\mathcal X\setminus\Gamma)\times S^2)\cup\Gamma$ with metric $g=-\Omega^2\,dudv+r^2\gamma$, electromagnetic field, and Vlasov field lifted according to \cref{prop:SS-equiv} is a globally hyperbolic, asymptotically flat spacetime, free of antitrapped surfaces, with Cauchy hypersurface $\tilde\Sigma\doteq ((\Sigma\setminus\Gamma)\times S^2)\cup (\Sigma \cap \Gamma)$. 
        \item $(\mathcal M,g)$ possesses complete null infinities $\mathcal I^\pm$ and is past causally geodesically complete. 
        \item If $\tilde e> \tilde M$, then $(\mathcal M,g)$ is future causally geodesically complete. 
        \item If $\tilde e\le \tilde M$, there are two options. (If $\mathcal S_{\alpha,\eta,\ve}$ is to be untrapped, then necessarily $r_2\notin [r_-,r_+]$.)
        \begin{enumerate}
            \item If $r_2<r_-$, then $(\mathcal M,g)$ is future causally geodesically incomplete. The spacetime contains a nonempty black hole region, i.e., $\mathcal{BH}\doteq \mathcal M\setminus J^-(\mathcal I^+)\ne\emptyset$. The Cauchy hypersurface $\tilde\Sigma$ is disjoint from $\mathcal{BH}$. The event horizon $\mathcal H^+=\partial(\mathcal{BH})$ is located at $u=r_2-2r_+$.

            \item If $r_2>r_+$,  then $(\mathcal M,g)$ is future causally geodesically complete.
        \end{enumerate}
    \end{enumerate}
\end{prop}

\subsubsection{Proof of the main theorem}\label{sec:proof-main}

\begin{proof}[Proof of \cref{thm:main}] Fix a fundamental charge $\mathfrak e>0$, cutoff functions $\varphi$, $\theta$, and $\zeta$ as in \cref{sec:parameters}, and a number $\Lambda\ge 1$ satisfying \eqref{eq:Lambda}. Fix the extremal black hole target mass $M>0$ and let $0<r_1<r_2$ and $\delta>0$ be as in \cref{thm:ECC-Ori}. Let $\eta_0>0$ be such that \cref{prop:approx-1} applies to the multiparameter family of seed data $\mathcal S_{\lambda,M',\eta,\ve}$ (which was defined in \cref{def:S-interp}) for $\lambda\in[-1,2]$, $|M'-M|\le \delta$, $0<\eta\le \eta_0$, $\ve>0$ sufficiently small depending on $\eta$, and particle mass $0\le \mathfrak m\le \mathfrak m_0$, where $\mathfrak m_0$ is sufficiently small depending on $\eta$ and $\ve$.

Let $\mathcal F:[0,\infty)^2\to [0,\infty)^2$ be defined by $\mathcal F(\lambda,M')=(\lambda^2 M',\lambda M')$, which is easily verified to be smoothly invertible on $(0,\infty)^2$. Define the function $F_{\eta,\ve}(\lambda,M')=(\tilde M,\tilde e)$, the final Reissner--Nordstr\"om parameters of $\mathcal S_{\lambda,M',\eta,\ve}$. By \eqref{eq:final-parameter-estimate}, we find 
\begin{equation}
  \label{eq:parameter-estimate-3}  |F_{\eta,\ve}(\lambda,M')-\mathcal F(\zeta(\lambda),M')|\les \eta
\end{equation}
for $\lambda\in [-1,2]$, $|M'-M|\le \delta$, and $0<\eta\le \eta_0$. There exists a constant $0<\lambda_0\ll 1$ depending on $\eta_0$ such that if $\lambda\in [-1,\lambda_0]$, then $|F_{\eta,\ve}(\lambda,M')|< \frac{1}{2} r_1$. Also, \eqref{eq:parameter-estimate-3} implies
\begin{equation}
   \sup_{\lambda_0\le\lambda\le 2,|M'-M|\le \delta} \left|\frac{\tilde e}{\tilde M}-\frac{1}{\zeta(\lambda)}\right|\les_{\lambda_0} \eta.\label{eq:parameter-ratio-4}
\end{equation}
From smooth convergence of $F_{\eta,\ve}$ to $\mathcal F$ as $\eta\to 0$, we obtain
\begin{equation}
   \sup_{\lambda_0\le\lambda \le 2,|M'-M|\le \delta}\left|\frac{d}{d\lambda}\left(\frac{\tilde e}{\tilde M}\right)+\frac{\zeta'(\lambda)}{\zeta(\lambda)^2}\right|\les_{\lambda_0} \eta.\label{eq:parameter-ratio-5}
\end{equation}
It follows that the charge to mass ratio $\tilde e/\tilde M$ is \emph{strictly decreasing} as a function of $\lambda$, for $\lambda\ge \lambda_0$.

On any fixed neighborhood of $(1,M)\in\Bbb R^2$, $\mathcal F^{-1}\circ F_{\eta,\ve}$ converges uniformly to the identity map by \eqref{eq:parameter-estimate-3}. Therefore, by a simple degree argument\footnote{Specifically, we use the following statement: Let $f:\overline B_1\to\Bbb R^d$ be a continuous map, where $\overline B_1$ is the closed unit ball in $\Bbb R^d$. If $\sup_{\overline B_1}|f-\mathrm{id}|<\tfrac 12$, then the image of $f$ contains $B_{1/2}$.} we find an assignment $(\eta,\ve)\to (\lambda(\eta,\ve),M'(\eta,\ve))$ such that $F_{\eta,\ve}(\lambda(\eta,\ve),M'(\eta,\ve))=(M,M)$. We claim that for $\eta$ sufficiently small, the family of seed data $\lambda\mapsto \mathcal S_\lambda\doteq \mathcal S_{\lambda,M'(\eta,\ve),\eta,\ve}$ gives rise to the desired family of spacetimes and Cauchy data: 

\begin{enumerate}
    \item For $\lambda\in[-1,\lambda_0]$, the final Reissner--Nordstr\"om parameters of $\mathcal S_\lambda$ are $< \frac 12  r_1$ by definition of $\lambda_0$ and hence the globally hyperbolic spacetime $\mathcal D_\lambda$ associated to $\mathcal S_\lambda$ by \cref{prop:global-structure} is future causally geodesically complete and dispersive. At $\lambda=-1$, the seed data is trivial and hence the development is isometric to Minkowski space. 

    \item For $\lambda\in[\lambda_0,\lambda(\eta,\ve))$, the final charge to mass ratio $\tilde e/\tilde M$ is strictly decreasing towards $1$ by \eqref{eq:parameter-ratio-5}. Therefore, $\mathcal D_\lambda$ is future causally geodesically complete and dispersive. A neighborhood of spatial infinity $i^0$ is isometric to a superextremal Reissner--Nordstr\"om solution. 

    \item $\lambda= \lambda(\eta,\ve)$ is, by construction, critical, with parameter ratio $\tilde e/\tilde M=1$. $\mathcal D_\lambda$ contains a nonempty black hole region and for sufficiently large advanced time, the domain of outer communication and event horizon are isometric to an appropriate portion of an extremal Reissner--Nordstr\"om black hole.

    \item For $\lambda\in (\lambda(\eta,\ve),2]$, $\tilde e/\tilde M$ decreases away from $1$ by \eqref{eq:parameter-ratio-5}. By definition of $r_2$, $r_->r_2$ for $\eta$ sufficiently small and hence $\mathcal D_\lambda$ contains a nonempty black hole region and for sufficiently large advanced time, the domain of outer communication and event horizon are isometric to an appropriate portion of a subextremal Reissner--Nordstr\"om black hole. By \eqref{eq:parameter-ratio-4}, the charge to mass ratio at $\lambda=2$ can be made arbitrarily close to $\frac 12$. 
\end{enumerate}

To complete the proof, we assign a smooth family of Cauchy data to $\mathcal D_\lambda$. Let $i_\lambda:[0,\infty)\to\Sigma_\lambda$ be the arc length parametrization (with respect to the metric $g_\lambda$) of the Cauchy surface associated to $\mathcal D_\lambda$ by \cref{prop:global-structure}. Then we may define the embedding $\tilde i_\lambda=i_\lambda\times \mathrm{id}_{S^2}/{\sim}:\Bbb R^3\to \tilde\Sigma_\lambda\subset\mathcal M_\lambda$ (where the central sphere is collapsed to a point). The natural map $\lambda\mapsto \Psi_\lambda$, where $\Psi_\lambda$ is the Cauchy data induced on $\tilde \Sigma_\lambda$ by pullback along $\tilde i_\lambda$, is smooth. This completes the proof of the theorem.\end{proof}

\subsection{Weak* convergence to dust}\label{sec:weak-proof}

In this section, we show that the spacetimes constructed in \cref{prop:global-structure} weak* converge in an appropriate sense to the bouncing charged null dust spacetimes given by \cref{prop:radial-parametrization}. First, we show convergence of an outgoing Vlasov beam to the underlying outgoing formal dust beam:

\begin{prop}\label{prop:weak-1} Let $\alpha\in\mathcal P_\Gamma$, let $\{\eta_i\}$ and $\{\ve_j\}$ be decreasing sequences of positive numbers tending to zero, let $(r,\Omega^2,Q,f)$ be the solution of the Einstein--Maxwell--Vlasov system associated to $(\alpha,\eta_i,\ve_j)$ by \cref{prop:approx-1} for $j\gg i$ and arbitrary allowed mass, and let $(r_\d,\Omega^2_\d,Q_\d,N^v_\d,T^{vv}_\d)$ be the outgoing formal dust solution from \cref{sec:outgoing-Cauchy} on $\mathcal C_{r_2}$. Then the following holds for any relatively compact and relatively open set $U\subset\mathcal C_{r_2}$: 
\begin{enumerate}
    \item We have the following strong convergence:
    \begin{equation}
         \lim_{\substack{i\to\infty \\ j\gg i}}\left(\|r-r_\d\|_{C^1(U)}+\|\Omega^2-\Omega^2_\d\|_{C^1(U)}+\|Q-Q_\d\|_{C^0(U)}+\|T^{uu}\|_{C^0(U)}+\|T^{uv}\|_{C^0(U)}\right)=0,\label{eq:weak-1}
    \end{equation}
    where the limit is to be understood as taking $i\to \infty$ while keeping $j$ sufficiently large for each $i$ such that the conclusion of \cref{prop:approx-1} applies for $\mathcal S_{\alpha,\eta_i,\ve_j}$.

    \item If $U'\subset U$ is disjoint from a neighborhood of $\{\tau=0\}$, then 
    \begin{equation}
        \lim_{\substack{i\to\infty \\ j\gg i}}\|N^u\|_{C^0(U')}=0.\label{eq:weak-3}
    \end{equation}
   
    \item We have the following weak* convergence: For any $\varphi\in C^1_c(\Bbb R^2)$ with $\spt\varphi\cap\mathcal C_{r_2}\subset U$, 
    \begin{equation}
     \lim_{\substack{i\to\infty \\ j\gg i}}   \int_U(N^u,N^v,T^{vv}) \varphi\,dudv= \int_U (0,N^v_\d,T^{vv}_\d)\varphi\,dudv.\label{eq:weak-2}
    \end{equation}
    \item We have the following weak* convergence: For any $\varphi\in L^1(U)$,
    \begin{equation}
        \lim_{i\to\infty}\lim_{j\to\infty}\int_U(N^u,N^v,T^{vv}) \varphi\,dudv= \int_U (0,N^v_\d,T^{vv}_\d)\varphi\,dudv.\label{eq:weak-6}
    \end{equation}
\end{enumerate}
\end{prop}

\begin{proof}
It is clear from the estimates used in the proof of \cref{lem:approx-2} that the Vlasov solutions converge strongly to Minkowski space in the regions $\mathcal R_\mathrm{aux}^{\breve v}$. It therefore suffices to prove the proposition only for the case $U=\mathcal R^{\frac 12\breve v-\frac 13r_1}_\mathrm{main}$, where we use the estimates of \cref{lem:approx-1}. 

   \ul{Part 1.}  Using \cref{lem:est-dat-1}, \eqref{eq:Max-v}, and \eqref{eq:lem-1-total}, we already obtain
     \begin{equation*}
         |Q(u,v)-\check Q(u)|\les \eta \label{eq:weak-4}
     \end{equation*}
     for any $(u,v)\in U\doteq\mathcal R^{\frac 12\breve v-\frac 13r_1}_\mathrm{main}$ and note that $\check Q(u)=Q_\d(u,v)$ by \eqref{eq:Max-v-dust}. Using this estimate, \cref{lem:est-dat-1}, \eqref{eq:lem-1-total}, and Gr\"onwall's inequality on the differences of \eqref{eq:r-wave}, \eqref{eq:Omega-wave} and \eqref{eq:r-wave-dust}, \eqref{eq:Omega-wave-dust}, we readily infer
     \begin{equation*}
         |r-r_\d|+|\Omega^2-\Omega^2_\d|\les \eta
     \end{equation*}
     on $U$, which completes the proof of \eqref{eq:weak-1}.

     \ul{Part 2.} This follows immediately from \eqref{eq:Nu-est} since $\tau$ is bounded below on $U'$.

     \ul{Part 3.} Since $\varphi$ is bounded,
     \begin{equation*}
         \left|\int_{U}N^u\varphi\,dudv\right|\les \int_{-r_2}^{-\frac 23r_1} \int_{-u}^{\breve v}N^u\,dvdu\les \ve^{1/2}
     \end{equation*}
     by \eqref{eq:lem-1-total}. Next, let $\tilde\varphi\doteq \varphi/(-\mathfrak er^2\Omega^2)$, use Maxwell's equation \eqref{eq:Max-u}, and integrate by parts:
     \begin{equation}
         \int_{U} N^v\varphi\,dudv = \int_{U} \partial_uQ\tilde\varphi\,dudv=-\int_{U}Q\partial_u\tilde\varphi\,dudv+\int_{\partial U}Q\tilde\varphi\label{eq:weak-5}
     \end{equation}
     By Part 1, $Q\partial_u\tilde\varphi\to Q_\d\partial_u\tilde\varphi_\d$ and $Q\tilde\varphi\to Q_\d\tilde\varphi_\d$ uniformly as $i\to \infty$ and $j\ll i$, where $\tilde\varphi_\d\doteq\varphi/(-\mathfrak er^2_\d\Omega^2_\d)$. Therefore, passing to the limit, we have
     \begin{equation*}
        \lim_{\substack{i\to\infty \\ j\gg i}}  \int_{U} N^v\varphi\,dudv = -\int_{U}Q_\d\partial_u\tilde\varphi_\d\,dudv+\int_{U}Q_\d\tilde\varphi_\d= \int_{U}\partial_uQ_\d \tilde\varphi_\d= \int_U N^v_\d\,dudv,
     \end{equation*}
     where we have again integrated by parts and used \eqref{eq:Max-u-dust}. A similar argument applies for the convergence of $T^{vv}$ using the Raychaudhuri equations \eqref{eq:Ray-u} and \eqref{eq:Ray-u-dust}. This completes the proof of \eqref{eq:weak-2}.

     \ul{Part 4.} We first fix $i$ and take $j\to \infty$. By \eqref{eq:Nu-est}, $N^u\les_i1$, for every $j\ll i$, where we use the notation $A\les_iB$ to denote $A\le CB$, where $C$ may depend on $i$. By the Banach--Alaoglu theorem, there exists a subsequence $j_n$ and an $L^\infty(U)$ function $h$ such that $N^u\stackrel{\ast}{\rightharpoonup} h$ in $L^\infty(U)$. However, by \eqref{eq:weak-3}, it is clear that $h=0$ almost everywhere. Since the subsequential limit is unique, $N^u\stackrel{\ast}{\rightharpoonup}0$ as $j\to\infty$. This completes the proof of \eqref{eq:weak-6} for $N^u$. 

     Let $\mathring m_{\mathrm V,i}$ and $\mathring Q_{\mathrm V,i}$ be the values of $\mathring m$ and $\mathring Q$ at $r=\tfrac 23r_1$ for the Vlasov seed $\mathcal S_{\alpha,\eta_i,\ve_j}$. Note that these numbers do not depend on $j$. Let $\mathring m_i$ and $\mathring Q_i$ be the solutions of the system \eqref{eq:dust-constraint-1} and \eqref{eq:dust-constraint-2} on $[\tfrac 23r_1,r_2]$ with initial conditions $\mathring m_{\mathrm V,i}$ and $\mathring Q_{\mathrm V,i}$ at $r=\tfrac 23r_1$, current $\mathring{\mathcal N}^v$ given by \eqref{eq:weak-7}, and identically vanishing energy-momentum tensor $\mathring{\mathcal T}^{vv}$. Following the proof of \cref{prop:outgoing-dust-existence}, we obtain a unique smooth solution $(r_{\d,i},\Omega^2_{\d,i},Q_{\d,i},N^v_{\d,i},T^{vv}_{\d,i})$ of the formal outgoing charged null dust system on $U$ which attains the initial data just described. 

     By repeating the proof of Part 1 of the present proposition, we see that
     \begin{equation*}
            \lim_{j\to\infty}\left(\|r-r_{\d,i}\|_{C^1(U)}+\|\Omega^2-\Omega^2_{\d,i}\|_{C^1(U)}+\|Q-Q_{\d,i}\|_{C^0(U)}\right)=0
     \end{equation*}
     for fixed $i$. Arguing as in Part 2, we then see that
     \begin{equation}
         \lim_{j\to\infty}\int_{U}(N^v,T^{vv})\varphi\,dudv=\int_{U}(N^v_{\d,i},T^{vv}_{\d,i})\varphi\,dudv\label{eq:weak-8}
     \end{equation}
     for any $\varphi\in C^1_c(\Bbb R^2)$. Since $N^v,T^{vv}\les_i1$, a standard triangle inequality argument shows that $\varphi$ can be replaced by any $L^1$ function in \eqref{eq:weak-8}. Now it follows by construction that 
     \begin{equation*}
         |N^v_{\d,i}-N^v_\d|+|T^{vv}_{\d,i}-T^{vv}_\d|\les \eta_i
     \end{equation*}
     on $U$, so we can safely take $i\to\infty$ in \eqref{eq:weak-8}, which completes the proof of \eqref{eq:weak-6}.
\end{proof}

In order to globalize this, we must first define bouncing charged null dust in the formal system in double null gauge. So consider again the outgoing solution $(r_\d,\Omega^2_\d,Q_\d,N^v_\d,T^{vv}_\d)$ on $\mathcal C_{r_2}$ with seed data $(\mathring{\mathcal N}^v_\d,0,r_2,\mathfrak e)$ given by \eqref{eq:weak-7}. As in \cref{def:max-time-symmetric}, we extend $r_\d$, $\Omega^2_\d$, and $Q_\d$ to $\tilde{\mathcal C}_{r_2}$ by reflection, and set
\begin{align*}
    N^u_\d(u,v)&=N^v_\d(-v,-u),\\
    T^{uu}_\d(u,v)&=T^{vv}_\d(-v,-u)
\end{align*}
for $(u,v)\in \tilde{\mathcal C}_{r_2}\setminus\{\tau=0\}$. We extend $N^v_\d$ and $T^{vv}_\d$ to zero in $\tilde{\mathcal C}_{r_2}\setminus\{\tau=0\}$ and similarly extend $N^u_\d$ and $T^{uu}_\d$ to zero in $\mathcal C_{r_2}$. Using \cref{lem:RN-piece}, we attach a maximal piece of Reissner--Nordstr\"om with parameters $\varpi_2$ and $Q_2$ to $\mathcal C_{r_2}\cup \tilde{\mathcal C}_{r_2}$. Let $v_\infty$ be as defined in \eqref{eq:v-infty}.

\begin{defn}\label{def:bouncing-double-null}
    The \emph{globally hyperbolic bouncing formal charged null dust spacetime} associated to a set of parameters $\alpha\in\mathcal P_\Gamma$ is the tuple $(\mathcal X_\d,r_\d,\Omega^2_\d,Q_\d,N^u_\d,N^v_\d,T^{uu}_\d,T^{vv}_\d)$, where $\mathcal X_\d\doteq \{v\ge u\}\cap\{v<v_\infty\}$. 
\end{defn}

For $\tau\ge 0$, $(r_d,\Omega^2_\d, N^v_\d,T^{vv}_\d)$ solves the outgoing formal dust system and for $\tau<0$, $(r_d,\Omega^2_\d, N^u_\d,T^{uu}_\d)$ solves the ingoing formal dust system. The functions $r_\d$ and $\Omega^2_\d$ are $C^1$ across $\{\tau=0\}$ and $Q_\d$, $T^{uu}_\d$, and $T^{vv}_\d$ are $C^0$ across $\{\tau=0\}$. The currents $N^u_\d$ and $N^v_\d$ are \emph{discontinuous} across $\{\tau=0\}$ (since we extended by zero), but of course $N^u_\d=N^v_\d$ at $\{\tau=0\}$.

Using this definition and \cref{prop:weak-1}, we immediately obtain the following

\begin{thm}
    Let $\alpha\in\mathcal P_\Gamma$, let $\{\eta_i\}$ and $\{\ve_j\}$ be decreasing sequences of positive numbers tending to zero, let $(\mathcal X,r,\Omega^2,Q,f)$ be the globally hyperbolic solution of the Einstein--Maxwell--Vlasov system associated to $(\alpha,\eta_i,\ve_j)$ by \cref{prop:global-structure} for $j\ll i$ and arbitrary allowed mass, and let $(\mathcal X_\d,r_\d,\Omega^2_\d,Q_\d,N^u_\d,N^v_\d,T^{uu}_\d,T^{vv}_\d)$ be the globally hyperbolic bouncing formal charged null dust spacetime associated to $\alpha$ by \cref{def:bouncing-double-null}. Then the following holds for any relatively compact open set $U\subset\mathcal X_\d$:
    \begin{enumerate}
        \item We have the following strong convergence:
\begin{equation*}
    \lim_{\substack{i\to\infty \\ j\gg i}}\left(\|r-r_\d\|_{C^1(U)}+\|\Omega^2-\Omega^2_\d\|_{C^1(U)}+\|Q-Q_\d\|_{C^0(U)}+\|T^{uv}\|_{C^0(U)}\right)=0.
\end{equation*}  
\item We have the following weak* convergence: For any $\varphi\in C^1_c(U)$,
\begin{equation*}
     \lim_{\substack{i\to\infty \\ j\gg i}}\int_U(N^u,N^v,T^{uu},T^{vv})\varphi\,dudv=\lim_{\substack{i\to\infty \\ j\ll i}}\int_U(N^u_\d,N^v_\d,T^{uu}_\d,T^{vv}_\d)\varphi\,dudv.
\end{equation*}
\item We have the following weak* convergence: For any $\varphi\in L^1(\Bbb R^2)$ with $\spt\varphi\subset U$, \begin{equation*}
\lim_{i\to\infty}\lim_{j\to\infty}\int_U(N^u,N^v,T^{uu},T^{vv})\varphi\,dudv=\lim_{\substack{i\to\infty \\ j\ll i}}\int_U(N^u_\d,N^v_\d,T^{uu}_\d,T^{vv}_\d)\varphi\,dudv.
\end{equation*}
    \end{enumerate}
\end{thm}
\begin{rk}
    The globally hyperbolic region $\mathcal X$ depends on $\eta_i$ and $\ve_j$, but it always holds that $U\subset\mathcal X$ for $i$ sufficiently large and $j\gg i$.
\end{rk}

\section{The third law and event horizon jumping at extremality}
\label{sec:third-law-and-h-jumping}

\begin{figure}
\centering{
\def\svgwidth{23pc}
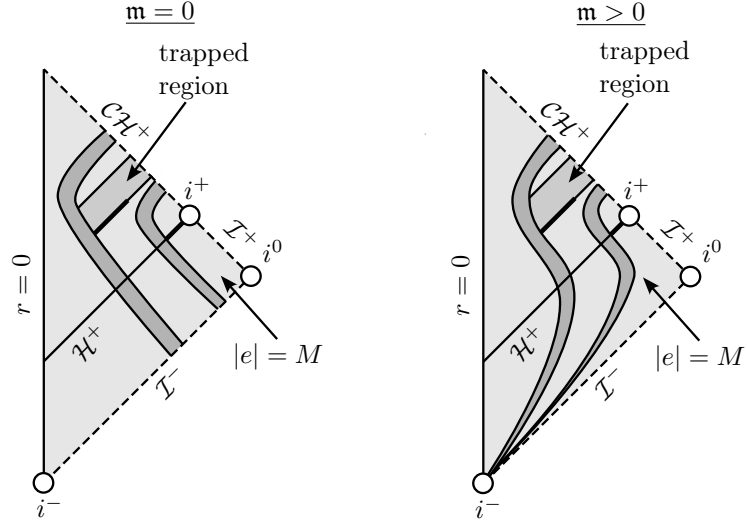}
\caption{Penrose diagrams of counterexamples to the third law of black hole thermodynamics in the Einstein--Maxwell--Vlasov model. The disconnected thick black curve denotes the outermost apparent horizon $\mathcal A'$, which jumps outward as the black hole becomes extremal. This behavior is necessary in third law violating spacetimes which obey the weak energy condition, see \cite[Proposition 1.1]{KU22} and \cite{Israel-third-law}.}
\label{fig:third-law-Vlasov-proof}
\end{figure}

Using the technology of bouncing Vlasov beams developed in the proof of \cref{thm:main}, we are now able to quickly prove \cref{thm:third-law-Vlasov,thm:it-jumps}. As complete proofs would require more lengthy setup, we only sketch the proofs of these results. We refer the reader back to \cref{sec:EMV-third-law} for the theorem statements and discussion.

\subsection{Counterexamples to the third law}\label{sec:Vlasov-third-law}

Refer to \cref{fig:third-law-Vlasov-proof} for global Penrose diagrams.

\begin{proof}[Proof of \cref{thm:third-law-Vlasov}]

Let $\alpha\in\mathcal P$ be third law violating dust parameters as in \cref{thm:third-law-dust}. We desingularize this dust beam as in the proof of \cref{thm:main}, noting that the charge on the inner edge of the beam is bounded below and hence no auxiliary beam is required. In order to achieve extremality we must modulate $\alpha$ slightly as in \cref{thm:main}, but all required inequalities for this construction are strict, so this can be done. By this procedure we obtain the time-symmetric solution $\mathcal D_\mathrm{ext}$ depicted in \cref{fig:third-law-Vlasov-2} below.

\begin{figure}[ht]
\centering{
\def\svgwidth{10pc}
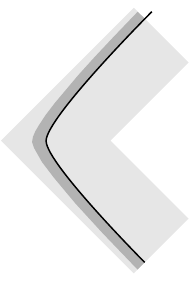}
\caption{Penrose diagram of the time symmetric Einstein--Maxwell--Vlasov solution $\mathcal D_\mathrm{ext}$ interpolating between subextremality and extremality. This diagram is valid for both massive and massless particles.}
\label{fig:third-law-Vlasov-2}
\end{figure}

Let $(\varpi_1,Q_1)$ be the initial Reissner--Nordstr\"om parameters of $\alpha$, which will also be the initial parameters of $\mathcal D_\mathrm{ext}$.  Using the same methods as \cref{thm:main} and \cref{lem:existence-P}, we can construct a solution $\mathcal D_\mathrm{sub}$ of Einstein--Maxwell--Vlasov collapsing to a subextremal Reissner--Nordstr\"om black hole with parameters $\varpi_1$ and $Q_1$. In the case of massless particles, the desired third law violating spacetime is then obtained by deleting an appropriate double null rectangle from $\mathcal D_\mathrm{sub}$ and gluing in an appropriate piece of $\mathcal D_\mathrm{ext}$. In the case of massive particles, the beams from $\mathcal D_\mathrm{ext}$ and $\mathcal D_\mathrm{sub}$ will possibly interact in the early past, but as is clear from the proof of \cref{prop:approx-1}, this happens in the dispersive region and the proof of \cref{lem:beam-far} can be repeated to show global existence and causal geodesic completeness in the past.
\end{proof}

\subsection{(Semi)continuity of the location of the event horizon} 
\label{sec:horizon-jumping}

\subsubsection{General results for spherically symmetric event horizons}\label{sec:general-continuity}

We now consider general \emph{weakly tame} spherically symmetric Einstein-matter systems, i.e., those satisfying the dominant energy condition and the weak extension principle. We refer to \cite{dafermos-trapped-surface,Kommemi13} for the precise definition of the weak extension principle, but note that it is a strictly weaker condition than the generalized extension principle as formulated in \cref{prop:ext} and therefore holds for the Einstein--Maxwell--Klein--Gordon, Einstein--Higgs, and Einstein--Maxwell--Vlasov systems \cite{TheBVPaper,dafermos2005naked,Kommemi13,DR16}.

Let $\Psi=(\Sigma,\bar g,\bar k,\dotsc)$ be a spherically symmetric, asymptotically flat Cauchy data set on $\Bbb R^3$ with a regular center in the given matter model. We will assume that $\Psi$ contains no spherically symmetric antitrapped surfaces, i.e., that $\partial_ur<0$ on $\Sigma$. This assumption is physically motivated by the observation that if the maximal Cauchy development $\mathcal D=(\mathcal Q,r,\Omega^2,\dotsc)$ does not contain a white hole, then there are no antitrapped surfaces in the spacetime. Furthermore, by Raychaudhuri's equation \eqref{eq:EE-Ray-u} and the dominant energy condition, $\partial_ur<0$ is propagated to the future of the Cauchy hypersurface $\Sigma$.

Under these assumptions, a very general a priori characterization of the boundary of $(\mathcal Q,r,\Omega^2)$ is available due to the work of Dafermos \cite{dafermos-trapped-surface} and we refer to Kommemi \cite{Kommemi13} for a detailed account. We will utilize the following two facts:

\vspace{1mm}

\textbf{Fact 1.} If $(\mathcal Q,r,\Omega^2)$ contains a trapped or marginally trapped surface, i.e., $\partial_vr(u_0,v_0)\le 0$ for some $(u_0,v_0)\in \mathcal Q$, then the black hole region is nonempty ($\mathcal{BH}\doteq\mathcal Q\setminus J^-(\mathcal I^+)\ne\emptyset$), future null infinity is complete in the sense of Christodoulou \cite{christodoulou1999global}, and $(u_0,v_0)\in \mathcal{BH}$. 

\vspace{1mm}

\textbf{Fact 2.} The Hawking mass $m$ extends to a (not necessarily continuous) nonincreasing and nonnegative function on future null infinity $\mathcal I^+$, called the \emph{Bondi mass}. If the Hawking mass of $\Sigma$ is bounded,\footnote{This can be regarded as a part of the definition of asymptotic flatness.}  then the Bondi mass is also bounded and the \emph{final Bondi mass} $M_f\doteq \inf_{\mathcal I^+}m$ is finite. Then the ``event horizon Penrose inequality'' $\sup_{\mathcal H^+}r\le 2M_f$ holds and by the no antitrapped surfaces condition we also obtain
\begin{equation}
    \sup_{\mathcal{BH}}r \le 2M_f.\label{eq:BH-r}
\end{equation}

\vspace{1mm}

We wish to consider sequences of initial data and their developments. In order to compare them, we have to ensure that the double null gauges are synchronized in an appropriate sense. We consider only developments $\mathcal D=(\mathcal Q,r,\Omega^2,\dotsc)$ for which the center $\Gamma$ is a subset of $\{u=v\}\subset\Bbb R^2_{u,v}$ and if $i:[0,\infty)\to\Sigma$ denotes the embedding map of the Cauchy hypersurface into $\mathcal Q$, we demand that $i(0)\in\Gamma$. Clearly, these conditions can always be enforced by an appropriate transformation of the double null gauge.

\begin{ass}\label{def:sync}
    Let $\{\Psi_j\}$ be a sequence of spherically symmetric asymptotically flat Cauchy data for a weakly tame Einstein-matter system. Let $\mathcal D_j=(\mathcal Q_j,r_j,\Omega_j^2,\dotsc)$ denote the maximal development of $\Psi_j$ with Cauchy hypersurface $\Sigma_j\subset\mathcal Q_j$ and embedding map $i_j:[0,\infty)\to \Sigma_j$ normalized as above. We assume that $\Psi_j$ converges to another data set $\Psi_\infty$ and the developments converge in the following sense:
    \begin{enumerate}
        \item \ul{Gauge condition:} Let $\mathcal D_\infty$ denote the maximal development of $\Psi_\infty$ with Cauchy hypersurface and embedding map $i_\infty:[0,\infty)\to \Sigma_\infty$. We assume that $\mathcal D_j$ and $\mathcal D_\infty$ have \emph{continuously synchronized gauges} in the sense that $(u,v)\circ i_j:[0,\infty)\to\Bbb R^2$ converges uniformly on compact sets as $j\to\infty$.

        \item \ul{Cauchy stability of the area-radius:} If $U\subset\mathcal Q_\infty$ is a relatively compact open set, then $U\subset\mathcal Q_j$ for $j$ sufficiently large and $r_j\to r_\infty$ in $C^1(U)$.
    \end{enumerate}
\end{ass}

\begin{rk}\label{def:sync-fam}
    The notion of continuous synchronization also makes sense for continuous one-parameter families of Cauchy data $\lambda\mapsto \Psi_\lambda$. In this case we require the maps $(u,v)\circ i_\lambda(x)$ to be jointly continuous in $\lambda$ and $x\in [0,\infty)$.
\end{rk}

\begin{rk}
    The initial data and developments given by \cref{prop:global-structure} are continuously synchronized as functions of the beam parameters $(\alpha,\eta,\ve,\mathfrak m)$.
\end{rk}

\begin{prop}\label{prop:continuity}
    Let $\Psi_j\to \Psi_\infty$ be a convergent sequence of one-ended asymptotically flat Cauchy data for a weakly tame spherically symmetric Einstein-matter system, containing no spherically symmetric antitrapped surfaces, and satisfying \cref{def:sync}. Assume that the sequence has uniformly bounded Bondi mass and that the development $\mathcal D_j$ of $\Psi_j$ contains a black hole for each $j\in\Bbb N$. Let $r_{j,\mathcal H^+}$ denote the limiting area-radius of the event horizon $\mathcal H^+_j$ of $\mathcal D_j$ and $u_{j,\mathcal H^+}$ its retarded time coordinate (also for $j=\infty$ if $\mathcal D_\infty$ contains a black hole). Then the following holds:
    \begin{enumerate}
        \item If future timelike infinity $i^+_\infty$ is a limit point of the center $\Gamma$ in $\mathcal D_\infty$ (in particular, $\mathcal D_\infty$ does not contain a black hole), then
        \begin{equation}\label{eq:cont-1}
            \lim_{j\to\infty} u_{j,\mathcal H^+}=\sup_{\mathcal Q_\infty}u.
        \end{equation}
        \item If $\mathcal D_\infty$ contains a black hole, then
\begin{enumerate}
    \item The retarded time of the event horizon is lower semicontinuous: \begin{align}
       \label{eq:cont-2}     \liminf_{j\to\infty}u_{j,\mathcal H^+}\ge u_{\infty,\mathcal H^+}.
        \end{align}
    \item Assume further that there are trapped surfaces asymptoting to $i_\infty^+$  in the following sense: Let $(u_{\infty,i^+},v_{\infty,i^+})$ denote the coordinates\footnote{$u_{\infty,i^+}$ is necessarily finite and equals $u_{\infty,\mathcal H^+}$, but $v_{\infty,i^+}$ could be $+\infty$.} of $i^+_\infty$ and suppose there exist sequences $u_1>u_2>\cdots \to u_{\infty,i^+}$ and $v_1<v_2<\cdots\to v_{\infty,i^+}$ such that $\partial_vr_\infty(u_i,v_i)<0$ for every $i\ge 1$. Then \eqref{eq:cont-2} is upgraded to \begin{equation}
         \label{eq:cont-4}   \lim_{j\to\infty}u_{j,\mathcal H^+}= u_{\infty,\mathcal H^+}
    \end{equation}
    and it additionally holds that 
    \begin{equation}
        \liminf_{j\to\infty}r_{j,\mathcal H^+}\ge r_{\infty,\mathcal H^+}.\label{eq:cont-5}
    \end{equation}
\end{enumerate}
    \end{enumerate}
\end{prop}
\begin{proof}
\ul{Part 1:} The inequality $\le$ in \eqref{eq:cont-1} follows directly from Cauchy stability, which implies the stronger statement \[\limsup_{j\to\infty}\sup_{\mathcal Q_j}u\le \sup_{\mathcal Q_\infty}u.\] 
We now prove the inequality $\ge$. Let $u_0<\sup_{\mathcal Q_\infty}u$ and let $M$ be an upper bound for the Bondi mass of the sequence $\{\mathcal D_j\}$. Since the cone $C_{u_0}$ in $\mathcal Q_\infty$ reaches future null infinity, we can choose $v_0$ such that $r_\infty(u_0,v_0)>2M$. By Cauchy stability, $r_j(u_0,v_0)>2M$ for $j$ sufficiently large, so $(u_0,v_0)\notin \mathcal{BH}_j$ by \eqref{eq:BH-r}. This implies $u_{j,\mathcal H^+}\ge u_0$ which completes the proof. 

\ul{Part 2 (a):} The argument to prove \eqref{eq:cont-2} is the same as the proof of Part 1, since $C_{u_0}$ reaches future null infinity in $\mathcal Q_\infty$ for any $u_0< u_{\infty,\mathcal H^+}$. 

\ul{Part 2 (b):} By \eqref{eq:cont-2} we must now show that for $u_0> u_{\infty,\mathcal H^+}$, $u_{j,\mathcal H^+}\le u_0$ for $j$ sufficiently large. Let $(u_i,v_i)$ be a trapped sphere for $\mathcal D_\infty$ as in the statement, with $u_0>u_i> u_{\infty,\mathcal H^+}$. By Cauchy stability, $(u_i,v_i)$ is then trapped in $\mathcal D_j$ for $j$ sufficiently large and therefore $u_i> u_{j,\mathcal H^+}$, which completes the proof of \eqref{eq:cont-4}. Using this, we now have by Cauchy stability and monotonicity of $r_j$ along $\mathcal H^+_j$ that
\begin{equation*}
\liminf_{j\to\infty}r_{j,\mathcal H^+}\ge\lim_{j\to\infty}r_j(u_{j,\mathcal H^+},v_0)=r_\infty(u_{\infty,\mathcal H^+},v_0)
\end{equation*} for any $v_0<v_{\infty,i^+}$.
Letting $v_0\to v_{\infty,i^+}$ completes the proof. \end{proof}

\begin{rk}
It is natural to ask if the ``reverse'' of \eqref{eq:cont-5}, i.e.,   
\begin{equation}
     \label{eq:cont-3}   \limsup_{j\to\infty}r_{j,\mathcal H^+}\le r_{\infty,\mathcal H^+},
\end{equation} holds at this level of generality (even assuming trapped surfaces asymptoting towards $i^+_\infty$). It turns out that \eqref{eq:cont-3} is \emph{false} without additional assumptions. On the one hand, assuming additionally a strict inequality in \eqref{eq:cont-2}, a minor modification of the arguments used to show \eqref{eq:cont-2} can be used to show \eqref{eq:cont-3}.\footnote{Note that in this case, there cannot be trapped surfaces asymptoting towards $i^+_\infty$!} On the other hand, without the assumption of a strict inequality in \eqref{eq:cont-2}, using the ingoing (uncharged) Vaidya metric, one can construct a counterexample to \eqref{eq:cont-3} which moreover satisfies the asymptotic trapped surface assumption of Part 2 (b). One is to imagine inflating the event horizon of a Schwarzschild black hole by injecting a fixed dust packet at later and later advanced times $v\sim j$. In the limit $j\to\infty$, the dust disappears, and $\limsup r_{j,\mathcal H^+}>r_{\infty,\mathcal H^+}$. Curiously, since black holes in the Vaidya model always have trapped surfaces behind the horizon, Part 2 (b) of the proposition implies that $u_{j,\mathcal H^+}$ is actually continuous in this process. This is because injecting a fixed amount of matter at later and later times causes the horizon to move outwards less and less (in $u$), causing it to converge back to the original Schwarzschild horizon as $j\to\infty$. Therefore, in order for \eqref{eq:cont-3} to hold, one must assume that $\Psi_j$ converges to $\Psi_\infty$ in a norm that sufficiently respects the asymptotically flat structure. We emphasize at this point that the conclusions of \cref{prop:continuity} hold only under an assumption of \emph{local Cauchy stability}---no asymptotic stability or weighted convergence is required.
\end{rk}

\subsubsection{Proof of event horizon jumping in the Einstein--Maxwell--Vlasov model}

\begin{proof}[Proof of \cref{thm:it-jumps}]
    This is proved by following the proof of \cref{thm:third-law-Vlasov} and varying the final parameters of $\mathcal D_\mathrm{ext}$ as in the proof of \cref{thm:main}.
\end{proof}

This theorem shows that it is not always possible to have equality in \eqref{eq:cont-2} when $\mathcal D_\infty$ is extremal: the event horizon can and does jump as a function of the initial data.

\appendix
\section[The characteristic initial value problem for spherically symmetric nonlinear wave-transport systems]{The characteristic initial value problem for spherically \\ symmetric nonlinear wave-transport systems}\label{app:B}

In this appendix, we prove local well-posedness for the spherically symmetric Einstein--Maxwell--Vlasov system in small characteristic rectangles away from the center. In fact, we consider the general system of equations
\begin{align}
\label{eq:B-1}\partial_u\partial_v\Psi&=F(\Psi,\partial\Psi,Q,M[f],M^{uv}[f]),\\
\label{eq:B-2}\partial_u Q &= K(\Psi,M^v[f]),\\
\label{eq:B-3}X(p^u,p^v,\Psi,\partial\Psi,Q)f&=0,
\end{align}
where $\Psi:\Bbb R^2_{u,v}\to\Bbb R^N $ is a vector-valued function taking the role of the ``wave-type'' variables $r$ and $\Omega^2$, $Q:\Bbb R^2_{u,v}\to\Bbb R$ is the charge, $f:T\Bbb R^2_{u,v}\to \Bbb R_{\ge 0}$ is the distribution function,
\begin{equation*}
     M[f]\doteq \int_0^\infty\int_0^\infty f\,dp^udp^v,\quad
    M^v[f]\doteq \int_0^\infty\int_0^\infty p^vf\,dp^udp^v,\quad
    M^{uv}[f]\doteq \int_0^\infty\int_0^\infty p^up^vf\,dp^udp^v
\end{equation*}
are moments of $f$, $F=(F_1,\dotsc, F_N)$ and $K$ are smooth functions of their variables, and $X$ is a vector field on $\Bbb R^2_{u,v}$ of the form
\begin{equation*}
X(p^u,p^v,\Psi,\partial\Psi,Q)=p^u\partial_u+p^v\partial_v+\xi^u(p^u,p^v,\Psi,\partial\Psi,Q)\partial_{p^u}+\xi^v(p^u,p^v,\Psi,\partial\Psi,Q)\partial_{p^v},
\end{equation*}
where $\xi^u$ and $\xi^v$ are smooth functions of their variables. Letting $a\in\{u,v\}$, we can write $X$ using Einstein notation as
\begin{equation*}
    X=p^a\partial_a+\xi^a\partial_{p^a}.
\end{equation*}

We assume that there exist functions $G_{k}:\Bbb R_{\ge 0}\to \Bbb R_{\ge 0}$ for $k\ge 0$ such that 
\begin{equation}
     \label{eq:B-assumption-on-xi}
|D^{i_1}_{\Psi,\partial\Psi,Q}\partial_p^{i_2}\xi^u(p^u,p^v,\Psi,\partial\Psi,Q)|+|D^{i_1}_{\Psi,\partial\Psi,Q}\partial_p^{i_2}\xi^v(p^u,p^v,\Psi,\partial\Psi,Q)|\le G_{k}(M)\langle p^\tau\rangle^{2-i_2}
\end{equation} if $|\Psi|+|\partial\Psi|+|Q|\le M$ and $i_1+i_2=k$,
where $\langle s\rangle \doteq \sqrt{1+s^2}$ and $p^\tau\doteq \frac 12(p^u+p^v)$. Here $D^{i_1}_{\Psi,\partial\Psi,Q}\partial_p^{i_2}$ denotes any expression involving $i_1$ derivatives in the $(\Psi,\partial\Psi,Q)$-variables and $i_2$ derivatives in the $(p^u,p^v)$-variables. We also assume that there exists a constant $\mathfrak m\ge 0$ such that 
\begin{equation}
    \xi^up^v+\xi^vp^u=0\label{eq:tangency}
\end{equation}
whenever $p^up^v=\mathfrak m^2$. These structural assumptions are verified for a renormalized version of the spherically symmetric Einstein--Maxwell--Vlasov system.

Given $U_0<U_1$ and $V_0<V_1$, let
\begin{align*}
  \mathcal C(U_0,U_1,V_0,V_1)  &\doteq (\{U_0\}\times[V_0,V_1])\cup([U_0,U_1]\times\{V_0\}), \\
   \mathcal R(U_0,U_1,V_0,V_1) &\doteq [U_0,U_1]\times [V_0,V_1].
\end{align*} We will consistently omit $(U_0,U_1,V_0,V_1)$ from the notation for these sets.
We also define
\begin{align*}
    P^\mathfrak m&\doteq \{(u,v,p^u,p^v)\in T\Bbb R^2:p^u\ge0,p^v\ge 0, p^\tau>0, p^up^v\ge \mathfrak m^2\},\\
     H^\kappa&\doteq \{(u,v,p^u,p^v)\in T\Bbb R^2:p^u\ge0,p^v\ge 0,p^\tau> \kappa\}
\end{align*}
 and set $P\doteq P^0$.
A function $\phi:\mathcal C\to \Bbb R$ is said to be \emph{smooth} if it is continuous and $\phi|_{\{U_0\}\times [V_0,V_1]}$ and $\phi|_{[U_0,U_1]\times\{V_0\}}$ are $C^\infty$ single-variable functions. This definition extends naturally to functions $f:P^\mathfrak m|_\mathcal C\to \Bbb R$. A smooth \emph{characteristic initial data set} for the system \eqref{eq:B-1}--\eqref{eq:B-3} consists of a triple $(\mathring\Psi,\mathring Q,\mathring f)$ and numbers $\kappa>0$, $\sigma>4$, where $\mathring \Psi:\mathcal C\to\Bbb R^N$, $\mathring Q:\{U_0\}\times[V_0,V_1]\to\Bbb R$, and $\mathring f:P^\mathfrak m|_\mathcal C\to \Bbb R_{\ge 0}$ are smooth. We additionally assume that $\spt(\mathring f)\subset H^\kappa|_\mathcal C$ for some $\kappa>0$ (which is only an extra assumption when $\mathfrak m=0$) and that
\begin{equation}\label{eq:B-f-data-bounded}
  \|\mathring f\|_{C^k_\sigma(P|_\mathcal C)}\doteq \sum_{0\le i_1+i_2\le k}  \left(\sup_{P^\mathfrak m|_{\{U_0\}\times[V_0,V_1]}}\langle p^\tau\rangle^{\sigma+i_2}|\partial_v^{i_1}\partial_{p}^{i_2}\mathring f|+\sup_{P^\mathfrak m|_{[U_0,U_1]\times\{V_0\}}}\langle p^\tau\rangle^{\sigma+i_2} |\partial_u^{i_1}\partial_{p}^{i_2}\mathring f|\right)
\end{equation}
is finite for every $k\ge 0$. For $f:P^\mathfrak m|_\mathcal R\to \Bbb R_{\ge 0}$ and $k\ge 0$, we define the norms
\begin{equation*}
    \|f\|_{C^k_\sigma(P|_\mathcal R)}\doteq \sum_{0\le i_1+i_2\le k}\sup_{P^\mathfrak m|_\mathcal R}\langle p^\tau\rangle^{\sigma+i_2} |\partial^{i_1}_x\partial^{i_2}_p f|,
\end{equation*}
where $\partial_x^{i_1}$ denotes $i_1$ derivatives in the $(u,v)$-variables. 

\begin{prop}\label{prop:B-main}
    For any $B>0$, $\kappa>0$, and $\sigma>4$ there exists a constant $\ve>0$ (depending also on $F$, $K$, and $X$) with the following property. Let $(\mathring \Psi,\mathring Q,\mathring f)$ be a smooth characteristic initial data set for the system \eqref{eq:B-1}--\eqref{eq:B-3} on $\mathcal C(U_0,U_1,V_0,V_1)$. If $U_1-U_0<\ve$, $V_1-V_0<\ve$, and 
    \begin{equation}\label{B:data-1}
        \|\mathring \Psi\|_{C^2(\mathcal C)}+\|\mathring Q\|_{C^1(\mathcal C)}+\|\mathring f\|_{C^1_\sigma(P|_\mathcal C)}\le B,
    \end{equation}
    then there exists a unique smooth solution $(\Psi,Q,f)$ of \eqref{eq:B-1}--\eqref{eq:B-3} on $\mathcal R(U_0,U_1,V_0,V_1)$ which extends the initial data. Moreover, the distribution function $f$ is supported in $H^{\kappa/2}$ and for any $k\ge 0$, the norm 
    \begin{equation*}
        \|\Psi\|_{C^k(\mathcal R)}+\|Q\|_{C^k(\mathcal R)}+\|f\|_{C^k_\sigma(P|_\mathcal R)}
    \end{equation*} is finite and 
    can be bounded in terms of initial data norms. 
\end{prop}

\begin{rk}
    While we assume that the initial data $(\mathring\Psi,\mathring Q,\mathring f)$ are smooth (and that $\mathring f$ satisfies the nontrivial bound \eqref{eq:B-f-data-bounded} at any order), the existence time $\ve$ in the proposition depends only on the estimate \eqref{B:data-1}.
\end{rk}

\begin{rk}
    If we assume that $\mathring f$ has compact support in the momentum variables, then \eqref{eq:B-f-data-bounded} is automatic by smoothness.
\end{rk}

\subsection{Proof of \texorpdfstring{\cref{prop:B-main}}{Proposition B.1}}

In this section, we assume the hypotheses and setup of \cref{prop:B-main}. We also set $U_0=V_0=0$ and define $\tau\doteq \frac 12(v+u)$. Therefore, $0\le \tau\le \ve$ on $\mathcal R$.

We will construct the solution $(\Psi,Q,f)$ as the limit of an iteration scheme.

\begin{lem}\label{lem:B-iteration} There exist sequences of constants $\{\tilde C_k\}$ and $\{C_k\}$ such that the following holds. For any $\ve$ sufficiently small and every $n\ge 1$, there exist functions $(\Psi_n,Q_n,f_n)\in C^\infty(\mathcal R)\times C^\infty(\mathcal R)\times C^\infty(P^\mathfrak m|_\mathcal R)$ solving the iterative system 
\begin{align}
  \label{iterate-1}  \partial_u\partial_v\Psi_n&= F(\Psi_{n-1},\partial\Psi_{n-1},Q_{n-1},M[f_{n-1}],M^{uv}[f_{n-1}]),\\
 \label{iterate-2}   \partial_uQ_n&= K(\Psi_{n-1},M^v[f_{n-1}]),\\
 \label{iterate-3}   X(p^u,p^v,\Psi_{n-1},\partial\Psi_{n-1},Q_{n-1})f_n&=0,
\end{align}
where we set $(\Psi_0,Q_0,f_0)$ to be identically zero, with initial conditions
\begin{equation}\label{B-initial-conditions-1}
    \Psi_n|_\mathcal C=\mathring \Psi,\quad Q_n|_\mathcal C= \mathring Q,\quad f_n|_{P|_\mathcal C}=\mathring f.
\end{equation}
Moreover, $\spt(f_n)\subset H^{\kappa/2}$ and these functions satisfy the bounds
\begin{align}
\label{eq:B-Psi-main-est}    \|\Psi\|_{C^k(\mathcal R)} &\le\tilde C_ke^{C_k\tau}, \\
 \label{eq:B-Q-main-est}    \|Q_n\|_{C^{k}(\mathcal R)} &\le\tilde C_{k+1}e^{C_{k+1}\tau}, \\
  \label{eq:B-f-main-est}   \|f_n\|_{C^{k}_\sigma(P|_\mathcal R)} &\le \tilde C_{k+1}e^{C_{k+1}\tau}.
\end{align}
\end{lem}

It is convenient to set \begin{align*}
    F_{n}&\doteq F(\Psi_{n},\partial\Psi_{n},Q_{n},M[f_{n}],M^{uv}[f_{n}]),\\
    K_{n}&\doteq K(\Psi_{n},M^v[f_{n}]),\\
    X_{n}&\doteq X(p^u,p^v,\Psi_{n},\partial\Psi_{n},Q_{n})
\end{align*} for $n\ge 0$.
 We first require a preliminary lemma about integral curves of the vector field $X_{n}$.

\begin{lem}\label{lem:B-geodesics}
For $n\ge 0$, let $\Gamma_{n}^{\mathfrak m,\kappa}$ denote the set of maximal integral curves $\tilde\gamma=(\gamma,p):I\to T\mathcal R$ (where $I$ is a closed interval containing $0$) of the vector field $X_{n}$ subject to the condition that $\tilde \gamma(0)\in H^\kappa\cap P^\mathfrak m$. Assume that $\Psi_n$ satisfies \eqref{eq:B-Psi-main-est} for $k=1$ and $Q_n$ satisfies \eqref{eq:B-Q-main-est} for $k=0$. Then for $\ve$ sufficiently small (depending in particular on $\kappa$) and any $\tilde\gamma\in \Gamma_{n}$,  $\gamma$ is a future-directed causal curve in $\mathcal R$ connecting $\mathcal C$ with the future boundary of $\mathcal R$, $\tilde\gamma(s)\in P^\mathfrak m\cap H^{\kappa/2}$ for every $s\in I$, and 
\begin{equation}
    \frac 12 p^\tau(0)\le p^\tau(s)\le 2p^\tau(0)\label{eq:B-basic-p-estimate}
\end{equation}
for every $s\in I$.
\end{lem}
\begin{proof} Let $\tilde\gamma=(\gamma,p)\in \Gamma_n$ and set $\tau_0\doteq \gamma^\tau(0)$. By definition, 
\begin{align*}
    \frac{d\gamma^a}{ds}=p^a,\quad
    \frac{dp^a}{ds} = \xi^a_n
\end{align*}
for $a\in\{u,v\}$, where $\xi^a_n\doteq \xi^a(p^u,p^v,\Psi_n,\partial\Psi_n,Q_n)$. Observe that $X_n$ is tangent to the the boundary of $P^\mathfrak m$, $\{(u,v,p^u,p^v)\in T\mathcal R:p^up^v=\mathfrak m^2\}$, by \eqref{eq:tangency}, so $\tilde\gamma$ remains within $P^\mathfrak m$. Reparametrizing $\tilde\gamma$ by $\tau$ gives 
\begin{equation*}
    \frac{d}{d\tau}(p^\tau)^2 = 2(\xi^u_n+\xi^v_n).
\end{equation*}
 Using \eqref{eq:B-assumption-on-xi}, the assumptions on $\Psi_n$ and $Q_n$, and Gr\"onwall's inequality, we have
    \begin{equation*}
        p^\tau(\tau)^2\le e^{O(\ve)}\left(p^\tau (\tau_0)^2+O(\ve)\right)\le 2p^\tau(\tau_0)^2
    \end{equation*}
    for $\tau$ in the domain of $\tilde\gamma$ and for $\ve$ sufficiently small. This proves the second inequality in \eqref{eq:B-basic-p-estimate}. To prove the first inequality, we observe that by the estimate we have just proved, 
    \begin{equation*}
        |(p^\tau(\tau))^2-(p^\tau(\tau_0))^2|\les \ve\sup_{\tilde\gamma}|\xi^u_n+\xi^v_n|\les \ve p^\tau(\tau_0)^2
    \end{equation*}
    Choosing $\ve$ perhaps even smaller proves \eqref{eq:B-basic-p-estimate} and completes the proof of the lemma.
\end{proof}

\begin{proof}[Proof of \cref{lem:B-iteration}] The proof is by induction on $n$ and induction on $k$ for each fixed $n$. As the existence and estimates for the base case $n=0$ are trivial, we assume the existence of $(\Psi_{n-1},Q_{n-1},f_{n-1})$ satisfying \eqref{eq:B-Psi-main-est}--\eqref{eq:B-f-main-est}, where the constants are still to be determined. We will choose the constants to satisfy $\tilde C_k\le C_k\le \tilde C_{k+1}$, which we use without comment in the sequel. By \eqref{eq:B-Psi-main-est}--\eqref{eq:B-f-main-est} for $(\Psi_{n-1},Q_{n-1},f_{n-1})$ and iterating the chain rule, it is easy to see that 
\begin{align}
\label{eq:B-coeff-1}|\partial^k\mathcal M_{n-1}|&\le C(\tilde C_{k+1})e^{C_{k+1}\tau}\quad (k\ge 0),\\
\label{eq:B-coeff-2}    |F_{n-1}|&\le C(C_1),\\
 \label{eq:B-coeff-3}   |\partial^kF_{n-1}|&\le C(\tilde C_{k+1})e^{C_{k+1}\tau}\quad (k\ge 1),\\
  \label{eq:B-coeff-4}  |K_{n-1}|&\le  C(C_1),\\
  \label{eq:B-coeff-5}  |\partial^k K_{n-1}|& \le C(\tilde C_{k+1})e^{C_{k+1}\tau} \quad (k\ge 1),
\end{align}
where $\mathcal M_{n-1}\in\{M[f_{n-1}],M^v[f_{n-1}],M^{uv}[f_{n-1}]\}$. We also define the number 
\begin{equation*}
        B'\doteq |F(\mathring \Psi,\partial\mathring\Psi,\mathring Q,M[\mathring f],M^{uv}[\mathring f])|_{(0,0)}+|K(\mathring \Psi,M^v[\mathring f])|_{(0,0)}.
\end{equation*}

\textbf{Step 1.} The function $\Psi_n$ is defined by the explicit formula 
    \begin{equation}
        \Psi_n(u,v)\doteq \int_{0}^u\int_{0}^vF_{n-1}(u',v')\,dv'du'+\mathring \Psi(u,0)+\mathring\Psi(0,v)-\mathring\Psi(0,0).\label{eq:Psi-n-defn}
    \end{equation}    
It follows by inspection of this representation formula and \eqref{eq:B-coeff-2} that 
\begin{equation}
    \|\Psi_n\|_{C^1(\mathcal R)}\le 10B\label{eq:B-Psi-low-est}
\end{equation}
if $\ve$ is sufficiently small depending on $C_1$. We estimate $k$-th order derivatives $(k\ge 2)$ of the form $\partial^k_u\Psi_n$, $\partial^{k-2}_x\partial_u\partial_v\Psi_n$, and $\partial_v^k\Psi_n$ separately. For the first type, we simply differentiate the representation formula \eqref{eq:Psi-n-defn} $k$ times to obtain
\begin{equation*}
    |\partial^k_u\Psi_n|\le (\mathrm{data})+\int_{0}^v |\partial_u^{k-1}F_{n-1}|\,dv'\le C+\frac{C(\tilde C_k)}{C_k}e^{C_k\tau}\le \tilde C_k e^{C_k\tau}
\end{equation*}
for appropriate choices of $\tilde C_k$ and $C_k$. For mixed derivatives, we differentiate the wave equation to obtain 
\begin{equation*}
    |\partial^{k-2}_x\partial_u\partial_v\Psi_{n-1}|\le \sup_{\{0\}\times[0,V_0]}|\partial^{k-2}_xF_{n-1}|+\int_{0}^{u}|\partial_u\partial^{k-2}F_{n-1}|\,du'\le C(C_{k-1})+\frac{C(\tilde C_k)}{C_k}e^{C_k\tau}\le \tilde C_ke^{C_k\tau} 
\end{equation*}
for appropriate choices of $\tilde C_k$ and $C_k$. The estimate for $\partial_v^k\Psi_n$ is similar to $\partial_u^k\Psi_n$. For later use, we derive a slightly improved estimate for $\partial_x^2\Psi_n$. Using the mean value theorem and \eqref{eq:B-coeff-3}, we have
\begin{equation*}
    |F_{n-1}-B'|\le C(C_2)\tau.
\end{equation*}
We also have 
\begin{equation*}
        |\partial_u^2\Psi_n-(\mathrm{data})|+ |\partial_v^2\Psi_n-(\mathrm{data})|\le C(C_2)\tau.
\end{equation*}
Therefore, for $\ve$ sufficiently small depending on $C_2$, combined with \eqref{eq:B-Psi-low-est}, we infer
\begin{equation}\label{eq:B-Psi-low-est-2}
    \|\Psi_n\|_{C^2(\mathcal R)}\le 20(B+B').
\end{equation}
This completes Step 1.

\textbf{Step 2.} The function $Q_n$ is defined by the explicit formula
    \begin{equation*}
        Q_n(u,v)=\int_{0}^uK_{n-1}(u',v)\,du'+\mathring Q(0,v).
    \end{equation*}
    It follows by inspection of this representation formula and \eqref{eq:B-coeff-4} that 
    \begin{equation*}
        \|Q\|_{C^0(\mathcal R)}\le 10B 
    \end{equation*}
for $\ve$ sufficiently small depending on $C_1$. For $k\ge 1$, we estimate 
\begin{equation*}
    |\partial_v^kQ_n|\le (\mathrm{data})+ \int_0^u |\partial_v^kK_{n-1}|\,du'\le C+\frac{C(\tilde C_{k+1})}{C_{k+1}}e^{C_{k+1}\tau}\le \tilde C_{k+1}e^{C_{k+1}\tau},
\end{equation*}
\begin{equation*}
    | \partial^{k-1}\partial_uQ_n|\le \sup_{\{0\}\times[0,V_0]}|\partial^{k-1}K_{n-1}|+\int_0^u|\partial^kK_{n-1}|\,du'\le C(C_{k})+\frac{C(\tilde C_{k+1})}{C_{k+1}}e^{C_{k+1}\tau}\le \tilde C_{k+1}e^{C_{k+1}\tau}
\end{equation*}
for appropriate chices of $\tilde C_{k+1}$ and $C_{k+1}$. Arguing as in Step 1, for $\ve$ sufficiently small depending on $C_2$, we also infer
\begin{equation}
    \|Q_n\|_{C^1(\mathcal R)}\le 20(B+B').\label{eq:B-Q-low-est-2}
\end{equation}
This completes Step 2. 

    \textbf{Step 3.} For $f_n$ we do not have an explicit representation formula and must instead infer its existence from general properties of flows of vector fields. A slight technical issue is that the ``initial data hypersurface'' $P|_\mathcal C$ is not smooth because of the corner in $\mathcal C$. Therefore, in order to prove the existence of a smooth $f_n$, we will construct it as a smooth limit of smooth solutions to the $X_{n-1}$ transport equation, corresponding to initial data where we smooth out the corner in $\mathcal C$. To carry out this idea, we first extend $\mathring f$ to $C^\infty_\sigma(P^\mathfrak m|_\mathcal R)$ according to 
\begin{equation*}
    \mathcal E\mathring f(u,v,p^u,p^v)\doteq \mathring f(u,0,p^u,p^v)+\mathring f(0,v,p^u,p^v)- \mathring f(0,0,p^u,p^v)
\end{equation*} 
and set, for $j\ge 1$,
\begin{align*}
    \mathcal S_j&\doteq\{(u,v)\in\mathcal R:uv=2^{-j}\},\\
   \mathcal R_j &\doteq \{(u,v)\in\mathcal R:uv\ge 2^{-j}.\}
\end{align*}
For any $j\ge j_0$ sufficiently large that $\mathcal S_j\ne\emptyset$, there exists a unique function $f_{n,j}\in C^\infty(P^\mathfrak m|_{\mathcal R_j})$ such that 
\begin{equation}
     X_{n-1}f_{n,j}=0\label{eq:B-approx-transport}
\end{equation} in $\mathcal R_j$, with initial data $f_{n,j}=\mathcal E\mathring f$ on $P^\mathfrak m|_{\mathcal S_j}$.
The existence follows immediately from the flowout theorem (see \cite[Theorem 9.20]{Lee-smooth}), the fact that $X_{n-1}$ is transverse to $P^\mathfrak m|_{\mathcal S_j}$, and \cref{lem:B-geodesics}. It also follows from \cref{lem:B-geodesics} that $\spt(f_{n,j})\subset H^{\kappa/2}$. We will use the fact that $p^\tau\ge \kappa/2$ for any $\tilde\gamma\in \Gamma_{n-1}^{\mathfrak m,\kappa}$ often and without further comment in the sequel.

We now claim that we can choose $\tilde C_k$ and $C_k$ such that 
\begin{equation}\label{eq:B-f-n-j-est}
    \|f_{n,j}\|_{C^{k}_\sigma(P|_{\mathcal R_j})} \le \tilde C_{k+1}e^{C_{k+1}\tau}
\end{equation}
for every $n\ge 0$, $j\ge j_0$, and $k\ge 0$. Let $\mathcal F_{n,j,k}$ denote the vector with $\binom{k+3}{k}$ components of the form 
\begin{equation}
    (p^\tau)^{\sigma+i_2}\partial_x^{i_1}\partial_p^{i_2}f_{n,j},\label{eq:mathcalF-component}
\end{equation}
where $i_1+i_2=k$.  We will show inductively that $\tilde C_{k+1}$ and $C_{k+1}$ can be chosen so that
\begin{equation}\label{eq:B-F-est-1}
    \sup_{P^\mathfrak m|_{\mathcal R_j}} |\mathcal F_{n,j,k}|\le \tilde C_{k+1}e^{C_{k+1}\tau},
\end{equation} which would imply \eqref{eq:B-f-n-j-est}.

Orders $k=0$ and $k=1$ are slightly anomalous in our scheme and we handle them first. We require the estimate
\begin{equation}
   \langle p^\tau\rangle^{-2} |\partial_x^{\le 1}\xi^a_{n-1}|+\langle p^\tau\rangle^{-1} |\partial_p\xi^a_{n-1}|\les 1,\label{eq:B-xi-low-est}
\end{equation}
which follows from \eqref{eq:B-assumption-on-xi}, \eqref{eq:B-Psi-low-est-2}, and \eqref{eq:B-Q-low-est-2}. Using \eqref{eq:B-approx-transport}, we compute
\begin{equation*}
    X_{n-1}((p^\tau)^\sigma f_{n,j})= \frac \sigma 2 (p^\tau)^{\sigma-1}(\xi^u_{n-1}+\xi^v_{n-1})f_{n,j}.
\end{equation*}
From \eqref{eq:B-xi-low-est} we then infer
\begin{equation*}
    |X_{n-1}((p^\tau)^\sigma f_{n,j})|\les (p^\tau)^{\sigma+1}f_{n,j}.
\end{equation*}
Let $\tilde\gamma\in\Gamma_{n-1}^{\mathfrak m,\kappa}$ be parametrized by coordinate time $\tau$. Then along $\tilde\gamma$ we have
\begin{equation*}
  \left|  \frac{d}{d\tau}\mathcal F_{n,j,0}\right|\les |\mathcal F_{n,j,0}|, 
\end{equation*}
whence by Gr\"onwall's inequality
\begin{equation*}
    |\mathcal F_{n,j,0}|\les \|\mathcal E\mathring f\|_{C^0_\sigma(P|_{\mathcal R})}\les \|\mathring f\|_{C^0_\sigma(P|_\mathcal{C})}.
\end{equation*}
Next, we compute
\begin{equation*}
    |X_{n-1}((p^\tau)^\sigma\partial_x f_{n,j})|\les (p^\tau)^{\sigma+1}|\partial_x f_{n,j}|+(p^\tau)^\sigma |[X_{n-1},\partial_x]f_{n,j}|.
\end{equation*}
The commutator is estimated using \eqref{eq:B-xi-low-est} by
\begin{equation*}
     |[X_{n-1},\partial_x]f_{n,j}|\les |\partial_x \xi_{n-1}^a||\partial_p f_{n,j}|\les (p^\tau)^2|\partial_p f_{n,j}|.
\end{equation*}
The $p$ derivative of $f_{n,j}$ satisfies
\begin{equation*}
    |X_{n-1}((p^\tau)^{\sigma+1}\partial_p f_{n,j})|\les (p^\tau)^{\sigma+2}|\partial_p f_{n,j}|+(p^\tau)^{\sigma+1} |[X_{n-1},\partial_p]f_{n,j}|,
\end{equation*}
where the commutator is now estimated by
\begin{equation*}
     |[X_{n-1},\partial_p]f_{n,j}|\les |\partial_xf_{n,j}| +|\partial_p \xi_{n-1}^a||\partial_x f_{n,j}|\les |\partial_xf_{n,j}|+ p^\tau|\partial_p f_{n,j}|.
\end{equation*}
Putting these estimates together, we find that 
\begin{equation*}
    \left|\frac{d}{d\tau}\mathcal F_{n,j,1}\right|\les |\mathcal F_{n,j,1}|, 
\end{equation*}
whence again by Gr\"onwall's inequality we conclude a uniform bound
\begin{equation}
    |\mathcal F_{n,j,1}|\les \|\mathring f\|_{C^1_\sigma(P|_\mathcal C)}.\label{eq:B-F-first-order}
\end{equation}

Having now established cases $k=0$ and $k=1$ of \eqref{eq:B-F-est-1}, we now assume \eqref{eq:B-F-est-1} up to order $k-1$. Let $\varphi_{i_1,i_2}$ be a component of $\mathcal F_{n,j,k}$. We adopt the convention that if either $i_1$ or $i_2$ are negative, then $\varphi_{i_1,i_2}$ is interpreted as identically zero. Using \eqref{eq:B-assumption-on-xi} and \eqref{eq:B-xi-low-est}, we estimate
\begin{align}
 \nonumber   \left|X_{n-1}\left(\varphi_{i_1,i_2}\right)\right|&\les (p^\tau)^{\sigma+i_2-1}|\xi^u_{n-1}+\xi^v_{n-1}||\partial_x^{i_1}\partial_p^{i_2}f_{n,j}|+(p^\tau)^{\sigma+i_2} |[X_{n-1},\partial^k]f_{n,j}|\\
    &\les p^\tau |\varphi_{i_1,i_2}|+ (p^\tau)^{\sigma+i_2} |[p^a\partial_a,\partial^{i_1}_x\partial^{i_2}_p]f_{n,j}|+(p^\tau)^{\sigma+i_2} |[\xi^a_{n-1}\partial_{p^a},\partial^{i_1}_x\partial^{i_2}_p]f_{n,j}|.\label{eq:B-intermediate-1}
\end{align}
The first commutator, $[p^a\partial_a,\partial^{i_1}_x\partial^{i_2}_p]f_{n,j}$, vanishes unless $i_1\ge 1$ and therefore consists of terms of the form $\partial_x^{i_1+1}\partial_p^{i_2-1}f_{n,j}$, which implies\begin{equation}
    (p^\tau)^{\sigma+i_2} |[p^a\partial_a,\partial^{i_1}_x\partial^{i_2}_p]f_{n,j}|\les p^\tau |\varphi_{i_1+1,i_2-1}|\les p^\tau |\mathcal F_{n,j,k}|.\label{eq:B-intermediate-2}
\end{equation}
The second commutator can be estimated by 
\begin{equation*}
    |[\xi^a_{n-1}\partial_{p^a},\partial_x^{i_1}\partial_p^{i_2}]f|\les \sum_{\substack{ 1\le j_1+j_2 \\ j_1\le i_1,j_2\le i_2}}|\partial_x^{j_1}\partial_p^{j_2}\xi^a_{n-1}|| \partial^{i_1-j_1}_x\partial_p^{i_2+1-j_2}f_{n,j}|.
\end{equation*}
By inspection, $\partial_x^{j_1}\partial_p^{j_2}\xi^a_{n-1}$ is linear in $\partial_x^{j_1+1}\Psi_{n-1}$, which is the worst behaved term in our inductive hierarchy. Therefore, using again \eqref{eq:B-assumption-on-xi}, we may estimate
\begin{equation*}
    |\partial_x^{j_1}\partial_p^{j_2}\xi^a_{n-1}|\les \left(C(C_k)+\tilde C_{j_1+1}e^{C_{j_1+1}\tau}\right)(p^\tau)^{2-j_2}
\end{equation*}
and therefore infer
\begin{equation*}
    (p^\tau)^{\sigma+i_2} |[\xi^a_{n-1}\partial_{p^a},\partial^{i_1}_x\partial^{i_2}_p]f_{n,j}|\les p^\tau \sum_{\substack{1\le j_1+j_2\\j_1\le i_1,j_2\le i_2}} \left(C(C_k)+\tilde C_{j_1+1}e^{C_{j_1+1}\tau}\right)|\mathcal F_{n,j,k+1-j_1-j_2}|.\label{eq:B-F-intermediate-1}
\end{equation*}
If $j_1<k$, then $\tilde C_{j_1+1}e^{C_{j_1+1}\tau}\le C(C_k)$. If $j_1=k$, then $j_2=0$ and 
\begin{equation}
    \left(C(C_k)+\tilde C_{j_1+1}e^{C_{j_1+1}\tau}\right)|\mathcal F_{n,j,k+1-j_1-j_2}|=\left(C(C_k)+\tilde C_{k+1}e^{C_{k+1}\tau}\right)|\mathcal F_{n,j,1}|\les C(C_k)\tilde C_{k+1}e^{C_{k+1}\tau}, \label{eq:B-intermediate-3}
\end{equation} where we have used \eqref{eq:B-F-first-order}.
Therefore, the sum in \eqref{eq:B-F-intermediate-1} can be estimated by
\begin{equation*}
    \les C(C_k)\left(|\mathcal F_{n,j,k}|+\tilde C_{k+1}e^{C_{k+1}\tau}\right).
\end{equation*}
Putting \eqref{eq:B-intermediate-1}, \eqref{eq:B-intermediate-2}, \eqref{eq:B-F-intermediate-1}, and \eqref{eq:B-intermediate-3} together, we arrive at 
\begin{equation*}
    |X_{n-1}\mathcal F_{n,j,k}|\les C(C_k)p^\tau\left(|\mathcal F_{n,j,k}|+\tilde C_{k+1}e^{C_{k+1}\tau}\right).
\end{equation*}
As before, a simple Gr\"onwall argument now establishes \eqref{eq:B-f-n-j-est} for appropriate choices of $\tilde C_{k+1}$ and $C_{k+1}$. 

Having now established the boundedness of the sequence $f_{n,j}$, we may take the limit $j\to \infty$ (after perhaps passing to a subsequence). This shows the existence of a function $f_n\in C^\infty(P^\mathfrak m|_\mathcal R)$ with $\spt(f_n)\subset H^{\kappa/2}$, satisfying the estimates \eqref{eq:B-f-main-est}, and attaining $\mathring f$ on $P^\mathfrak m|_\mathcal C$. Finally, uniqueness of $f_n$ is immediate since it is constant along the integral curves of $X_{n-1}$.
\end{proof}

\begin{proof}[Proof of \cref{prop:B-main}] We prove first that the sequence iterative sequence $(\Psi_n,Q_n,f_n)$ constructed in \cref{lem:B-iteration} is Cauchy in $C^2\times C^1\times C^1_\sigma$. We claim that if $\ve$ is sufficiently small, then the following estimate holds for every $n\ge 2$:
\vspace{-2mm}\begin{multline}\label{eq:B-Cauchy-main}
       \|\Psi_n-\Psi_{n-1}\|_{C^2(\mathcal R)} +\|Q_n-Q_{n-1}\|_{C^1(\mathcal R)}+  \|f_n-f_{n-1}\|_{C^1_{\sigma}(P|_\mathcal R)}\\ \le \frac 12\left(\|\Psi_{n-1}-\Psi_{n-2}\|_{C^2(\mathcal R)}+\|Q_{n-1}-Q_{n-2}\|_{C^1(\mathcal R)}+\|f_{n-1}-f_{n-2}\|_{C^1_{\sigma}(P|_\mathcal R)}\right).
\end{multline}
    
 Using the mean value theorem and the boundedness of the iterative sequence, we immediately estimate
\begin{equation}
    \|F_{n-1}-F_{n-2}\|_{C^1(\mathcal R)}\les \|\Psi_{n-1}-\Psi_{n-2}\|_{C^2(\mathcal R)}+\|Q_{n-1}-Q_{n-2}\|_{C^1(\mathcal R)}+\|f_{n-1}-f_{n-2}\|_{C^1_\sigma(P|_\mathcal R)}.\label{eq:iterate-Cauchy}
\end{equation}
Using the formula
\begin{equation*}
    (\Psi_n-\Psi_{n-1})(u,v)=\int_0^u\int_0^v (F_{n-1}-F_{n-2})\,dv'du'
\end{equation*}
we readily infer that for $\ve$ sufficiently small, $\Psi_n-\Psi_{n-1}$ and $\partial_a^i(\Psi_n-\Psi_{n-1})$ for $a\in\{u,v\}$ and $i\in\{1,2\}$ are bounded by an arbitrarily small multiple of the right-hand side of \eqref{eq:B-Cauchy-main}. To estimate the mixed second partial derivative, we simply use the fundamental theorem of calculus to bound
\begin{equation*}
    \|\Psi_{n-1}-\Psi_{n-2}\|_{C^1(\mathcal R)}+\|Q_{n-1}-Q_{n-2}\|_{C^0(\mathcal R)}+\|\mathcal M_{n-1}-\mathcal M_{n-2}\|_{C^0(\mathcal R)}
\end{equation*}
by an arbitrarily small multiple of the right-hand side of \eqref{eq:B-Cauchy-main} and then use the mean value theorem to estimate 
\begin{align*}
      |\partial_u\partial_v(\Psi_{n}-\Psi_{n-1})|&\le \|F_{n-1}-F_{n-2}\|_{C^0(\mathcal R)}\\
      &\les  \|\Psi_{n-1}-\Psi_{n-2}\|_{C^1(\mathcal R)}+\|Q_{n-1}-Q_{n-2}\|_{C^0(\mathcal R)}+\|\mathcal M_{n-1}-\mathcal M_{n-2}\|_{C^0(\mathcal R)}.
\end{align*}
The argument for bounding $\|Q_n-Q_{n-1}\|_{C^1(\mathcal R)}$ is essentially the same and is omitted.

To estimate $f_n-f_{n-1}$, we examine the quantity
\begin{equation*}
    \mathcal F_n\doteq (p^\tau)^\sigma \left(  f_n-f_{n-1},\partial_x(f_n-f_{n-1}),p^\tau\partial_p(f_n-f_{n-1})\right)
\end{equation*}
along integral curves of $X_{n-1}$. First, note that $\mathcal F_n\ne 0$ only along curves in $\Gamma^{\mathfrak m,\kappa}_{n-1}\cup\Gamma^{\mathfrak m,\kappa}_{n-2}$. By \cref{lem:B-geodesics} (and choosing $\ve$ perhaps smaller), any value $\tilde\gamma_{n-2}(s)$ for a curve $\tilde\gamma_{n-2}\in\Gamma^{\mathfrak m,\kappa}_{n-2}$ can be realized as an initial value $\tilde\gamma_{n-1}(0)$ for a curve $\tilde\gamma_{n-1}\in\Gamma^{\mathfrak m,\kappa/2}_{n-1}$. Therefore, it suffices to observe that the following estimate holds along any curve $\tilde\gamma_{n-1}\in\Gamma^{\mathfrak m,\kappa/2}_{n-1}$:
\begin{equation*}
    \left|\frac{d}{d\tau}\mathcal F_n\right|\les |\mathcal F_n|+\|\Psi_{n-1}-\Psi_{n-2}\|_{C^2(\mathcal R)}+\|Q_{n-1}-Q_{n-2}\|_{C^1(\mathcal R)}+\|f_{n-1}-f_{n-2}\|_{C^1_{\sigma}(P|_\mathcal R)},
\end{equation*}
which is obtained by simply differentiating $\mathcal F_n$ and using the estimates proved in \cref{lem:B-iteration}. Therefore, since $\mathcal F_n$ vanishes along $P^\mathfrak m|_\mathcal C$, Gr\"onwall's inequality implies 
\begin{equation*}
    |\mathcal F_n|\les \ve \left(\|\Psi_{n-1}-\Psi_{n-2}\|_{C^2(\mathcal R)}+\|Q_{n-1}-Q_{n-2}\|_{C^1(\mathcal R)}+\|f_{n-1}-f_{n-2}\|_{C^1_{\sigma}(P|_\mathcal R)}\right).
\end{equation*}
After choosing $\ve$ sufficiently small, the proof of \eqref{eq:iterate-Cauchy} is complete. 

Therefore, $(\Psi_n,Q_n,f_n)$ converges to a solution $(\Psi,Q,f)$ of the system \eqref{eq:B-1}--\eqref{eq:B-3} in $C^2\times C^1\times C^1$, which is moreover $C^\infty$ smooth by the higher order estimates proved in \cref{lem:B-iteration}. Uniqueness of the solution can be proved along the same lines as the proof of the estimate \eqref{eq:B-Cauchy-main} and is omitted.  
\end{proof}

\subsection{Proof of local well posedness for the Einstein--Maxwell--Vlasov system}\label{sec:proof-of-LWP}

In this section, we prove \cref{prop:char-IVP-Vlasov}, local well-posedness for the Einstein--Maxwell--Vlasov system in small characteristic rectangles. The proof has 3 steps: In the first step, we solve the wave equations \eqref{eq:r-wave} and \eqref{eq:Omega-wave}, the \emph{ingoing} Maxwell equation \eqref{eq:Max-u}, and the Maxwell--Vlasov equation \eqref{eq:SS-MV} using \cref{prop:B-main}. In order to directly quote \cref{prop:B-main}, we in fact consider a \emph{renormalized} system that fixes the location of the mass shell to the fixed family of hyperbolas $\tilde p^u\tilde p^v=\mathfrak m^2$. In step 2, we show that the \emph{outgoing} Maxwell equation holds as a result of conservation law \eqref{eq:SS-N-div}, which follows from the Maxwell--Vlasov equation \eqref{eq:SS-MV}. Finally, in step 3, we show that Raychaudhuri's equations \eqref{eq:Ray-u} and \eqref{eq:Ray-v} hold, using now the Bianchi identities \eqref{eq:Bianchi-1} and \eqref{eq:Bianchi-2}, which again follow from \eqref{eq:SS-MV}. Steps 2 and 3 may be thought of as \emph{propagation of constraints}, as they require the relevant equations to hold on initial data.

\begin{proof}[Proof of \cref{prop:char-IVP-Vlasov}]

\ul{Step 1.} Consider the wave-transport system
\begin{align}
 \label{eq:renorm-1}   \partial_u\partial_v r &= -\frac{\Omega^2}{4r}-\frac{\partial_ur \partial_vr}{r}+\frac{\pi r\Omega^2}{4}\int_{\tilde p^u\tilde p^v\ge \mathfrak m^2}\tilde p^u\tilde p^v \tilde f(u,v,\tilde p^u,\tilde p^v)\,d\tilde p^ud\tilde p^v,\\
    \nonumber \partial_u\partial_v{\log\Omega^2}&=\frac{\Omega^2}{2r^2}+\frac{2\partial_ur\partial_vr}{r^2}+\frac{\pi\Omega^2}{2} \int_{\tilde p^u\tilde p^v\ge \mathfrak m^2}\tilde p^u\tilde p^v \tilde f(u,v,\tilde p^u,\tilde p^v)\,d\tilde p^ud\tilde p^v\\
     &\quad -\frac{\pi\Omega^2}{2}\int_{\tilde p^u\tilde p^v\ge \mathfrak m^2}(\tilde p^u\tilde p^v-\mathfrak m^2)\tilde f(u,v,\tilde p^u,\tilde p^v)\,d\tilde p^ud\tilde p^v,\\
         \partial_u Q&= -\frac{\pi}{2}\mathfrak e r^2\Omega \int_{\tilde p^u\tilde p^v\ge \mathfrak m^2}\tilde p^u \tilde f(u,v,\tilde p^u,\tilde p^v)\,d\tilde p^ud\tilde p^v,\\
     \tilde X\tilde f&=0, \label{eq:renorm-4} 
\end{align}
where
\begin{multline} \nonumber
    \tilde X \doteq \tilde p^u\partial_u+\tilde p^v\partial_v-\left(\partial_u{\log\Omega} (\tilde p^u)^2-\partial_v{\log\Omega} \tilde p^u\tilde p^v+\frac{2\partial_vr}{r}(\tilde p^u\tilde p^v-\mathfrak m^2) + \mathfrak e\frac{\Omega Q}{r^2}\tilde p^u\right)\partial_{\tilde p^u}\\-\left(\partial_v{\log\Omega} (\tilde p^v)^2-\partial_u{\log\Omega} \tilde p^u\tilde p^v+\frac{2\partial_ur}{r}(\tilde p^u\tilde p^v-\mathfrak m^2) - \mathfrak e\frac{\Omega Q}{r^2}\tilde p^v\right)\partial_{\tilde p^v},
\end{multline}
with initial data $\mathring \Psi=(\log\mathring r,\log\mathring\Omega^2)$, $\mathring Q$, and $\mathring{\tilde f}(u,v,\tilde p^u,\tilde p^v)=\mathring f(u,v,\mathring\Omega^{-1}\tilde p^u,\mathring\Omega^{-1}\tilde p^v)$. Since $\log\mathring\Omega^2$ is bounded on $\mathcal C$ and either $\mathfrak m>0$ or $\ell\ge c_\ell$ on $\spt(\mathring f)$, there exists a $\kappa>0$ such that $\spt(\mathring{\tilde f})\subset H^\kappa$. Furthermore, the structural conditions \eqref{eq:B-assumption-on-xi} and \eqref{eq:tangency} are easily verified for $\tilde X$, so \cref{prop:B-main} produces a unique local smooth solution to \eqref{eq:renorm-1}--\eqref{eq:renorm-4} if $\ve_\loc$ is chosen sufficiently small. Making the change of variables $\tilde p^u\mapsto \Omega p^u$ and $\tilde p^v\mapsto \Omega p^v$, defining $f(u,v,p^u,p^v)=\tilde f(u,v,\Omega p^u,\Omega p^v)$, and observing that 
\begin{equation*}
    Xf=\Omega^{-1} \tilde X\tilde f=0,
\end{equation*}
we have obtained a unique local smooth solution $(r,\Omega^2,Q,f)$ for the system \eqref{eq:r-wave}, \eqref{eq:Omega-wave}, \eqref{eq:Max-u}, and \eqref{eq:SS-MV} on $\mathcal R$ which extends the initial data. 

\ul{Step 2.} We first aim to prove \eqref{eq:SS-N-div} using only \eqref{eq:SS-MV}. At this point, one could apply \cref{prop:SS-equiv} and \eqref{eq:Div-SS} to derive \eqref{eq:SS-N-div}, but we give now a direct proof.

First, using the definitions \eqref{eq:Nu} and \eqref{eq:Nv}, we have
\begin{align}
\nonumber\partial_u\left(\frac{r^2\Omega^2}{\pi}N^u\right)+\partial_v\left(\frac{r^2\Omega^2}{\pi}N^v\right)&=\partial_u\int_0^\infty \int_{\mathfrak m^2/(\Omega^2p^v)}^\infty r^2\Omega^4 p^uf\,dp^udp^v+\partial_v\int_0^\infty \int_{\mathfrak m^2/(\Omega^2p^u)}^\infty r^2\Omega^4 p^vf\,dp^vdp^u\\
\nonumber&= \int_{\Omega^2p^up^v\ge \mathfrak m^2}\left(2r\partial_ur\Omega^4p^uf+2r^2\Omega^4\partial_u{\log\Omega^2}p^uf+r^2\Omega^4p^u\partial_uf\right) dp^udp^v\\
\nonumber&\quad + \int_{\Omega^2p^up^v\ge \mathfrak m^2}\left(2r\partial_vr\Omega^4p^vf+2r^2\Omega^4\partial_v{\log\Omega^2}p^vf+r^2\Omega^4p^v\partial_vf\right) dp^udp^v\\
    &\quad +\int_0^\infty  \frac{\mathfrak m^4 r^2}{(p^v)^2}\partial_u{\log\Omega^2}f \,dp^v+\int_0^\infty  \frac{\mathfrak m^4 r^2}{(p^u)^2}\partial_v{\log\Omega^2}f \,dp^u.\label{eq:B-2.9-1}
\end{align}
 Adding both terms involving spatial derivatives of $f$ and using \eqref{eq:SS-MV} yields
\begin{align}
     \nonumber\int_{\Omega^2p^up^v\ge \mathfrak m^2}r^2\Omega^4\left(p^u\partial_uf+p^v\partial_v f\right)\,dp^udp^v&=\int_0^\infty\int^\infty_{\mathfrak m^2/(\Omega^2p^v)}r^2\Omega^4 \Xi^u\partial_{p^u}f\,dp^udp^v\\
     &\quad+\int_0^\infty\int_{\mathfrak m^2/(\Omega^2p^u)}^\infty r^2\Omega^4\Xi^v\partial_{p^v}f\,dp^vdp^u,\label{eq:B-f-integral-1}
\end{align}
where
\begin{align*}
    \Xi^u&\doteq\partial_u{\log\Omega^2}(p^u)^2+\frac{2\partial_vr}{r\Omega^2}(\Omega^2 p^up^v-\mathfrak m^2)+\mathfrak e\frac{Q}{r^2}p^u, \\
    \Xi^v&\doteq \partial_v{\log\Omega^2}(p^v)^2+\frac{2\partial_ur}{r\Omega^2}(\Omega^2 p^up^v-\mathfrak m^2)-\mathfrak e\frac{Q}{r^2}p^v.
\end{align*}
Integrating the first term on the right-hand side of \eqref{eq:B-f-integral-1} by parts, we find
\begin{align}
   \int_0^\infty\int^\infty_{\mathfrak m^2/(\Omega^2p^v)}r^2\Omega^4 \Xi^u\partial_{p^u}f\,dp^udp^v &=-\int_{\Omega^2 p^up^v\ge \mathfrak m^2}\left(2r^2\Omega^4\partial_u{\log\Omega^2}p^u+2r\partial_vr \Omega^4p^v+\mathfrak e\Omega^4 Q\right) f\,dp^udp^v\nonumber\\
    &\quad -\int_0^\infty \left(\frac{\mathfrak m^4r^2}{(p^v)^2}\partial_u{\log\Omega^2}+\mathfrak e \frac{\mathfrak m^2\Omega^2}{p^v}Q\right) f\,dp^v,\label{eq:B-2.9-2}
\end{align} where $f$ is evaluated at $p^u=\mathfrak m^2/(\Omega^2p^v)$,
and for the second term,
\begin{align}
   \int_0^\infty\int^\infty_{\mathfrak m^2/(\Omega^2p^u)}r^2\Omega^4 \Xi^v\partial_{p^v}f\,dp^vdp^u &=-\int_{\Omega^2 p^up^v\ge \mathfrak m^2}\left(2r^2\Omega^4\partial_v{\log\Omega^2}p^v+2r\partial_ur \Omega^4p^u-\mathfrak e\Omega^4 Q\right) f\,dp^udp^v\nonumber\\
    &\quad -\int_0^\infty \left(\frac{\mathfrak m^4r^2}{(p^u)^2}\partial_v{\log\Omega^2}-\mathfrak e \frac{\mathfrak m^2\Omega^2}{p^u}Q\right) f\,dp^u,\label{eq:B-2.9-3}
\end{align}
where $f$ is evaluated at $p^v=\mathfrak m^2/(\Omega^2p^u)$. Combining \eqref{eq:B-2.9-1}--\eqref{eq:B-2.9-3} yields \eqref{eq:SS-N-div} after noting that 
\begin{equation*}
    \int_0^\infty \mathfrak e\frac{\mathfrak m^2\Omega^2}{p^v}Qf\,dp^v= \int_0^\infty \mathfrak e\frac{\mathfrak m^2\Omega^2}{p^u}Qf\,dp^u.
\end{equation*}

We can now derive the ingoing Maxwell equation \eqref{eq:Max-v}. By \eqref{eq:Max-u}, we have
\begin{equation*}
    Q(u,v)=Q(U_0,v)-\int_{U_0}^u\tfrac 12\mathfrak er^2\Omega^2 N^v\,du'
\end{equation*}
on $\mathcal R$. We then derive
\begin{equation*}
     \partial_vQ(u,v)=\partial_vQ(U_0,v)-\int_{U_0}^u\partial_v( \tfrac 12\mathfrak er^2\Omega^2N^v)\,du'= \partial_vQ(U_0,v)+\int_{U_0}^u\partial_u(\tfrac 12\mathfrak er^2\Omega^2N^u)\,du'= \tfrac 12\mathfrak er^2\Omega^2N^v(u,v),
\end{equation*}
where in the final equality we used the fundamental theorem of calculus and the assumption that \eqref{eq:Max-v} holds on $\{U_0\}\times[V_0,V_1]$.

\ul{Step 3.} By a lengthy calculation which is very similar to the one performed in step 2, one may use \eqref{eq:SS-MV} to derive the Bianchi identities in the form \eqref{eq:Bianchi-1} and \eqref{eq:Bianchi-2}. Using the Maxwell equations \eqref{eq:Max-u} and \eqref{eq:Max-v}, this implies the Bianchi identities in the form \eqref{eq:Bianchi-general-1} and \eqref{eq:Bianchi-general-2}, where
\begin{equation*}
    \mathbf T^{uu}=T^{uu},\quad
    \mathbf T^{uv}=T^{uv}+\frac{Q^2}{\Omega^2 r^4},\quad
    \mathbf T^{vv}= T^{vv},\quad
\mathbf S=S+\frac{Q^2}{2r^4}.
\end{equation*}
By another lengthy calculation, using now also the wave equations \eqref{eq:r-wave} and \eqref{eq:Omega-wave}, one can derive the pair of identities
\begin{align*}
    \partial_v\left(r\partial_u^2r-r\partial_ur\partial_u{\log\Omega^2}+\frac 14 r^2\Omega^4 \mathbf T^{vv}\right)&=0,\\
    \partial_u\left(r\partial_v^2r-r\partial_vr\partial_v{\log\Omega^2}+\frac 14 r^2\Omega^4 \mathbf T^{uu}\right)&=0.
\end{align*}
These identities, together with the assumption that \eqref{eq:Ray-u} holds on $[U_0,U_1]\times\{V_0\}$ and \eqref{eq:Ray-v} holds on $\{U_0\}\times[V_0,V_1]$, prove that \eqref{eq:Ray-u}  and \eqref{eq:Ray-v} hold throughout $\mathcal R$. \end{proof}

\printbibliography[heading=bibintoc] 
\end{document}